\pdfoutput=1
\documentclass[journal, 10pt]{IEEEtran}
\IEEEoverridecommandlockouts
\usepackage{graphicx}
\usepackage{amssymb}
\usepackage{amsmath}
\usepackage{amsthm}
\usepackage{amsfonts}
\usepackage{hyperref,algorithm}
\usepackage{algpseudocode}
\usepackage{mathtools}
\usepackage{epstopdf}
\usepackage{dblfloatfix}
\usepackage{color}
\usepackage{array}
\usepackage{arydshln}
\usepackage{multirow}
\usepackage{booktabs}
\usepackage{balance}
\usepackage{setspace}
\newtheorem{definition}{\textbf{Definition}}

\newtheorem{remark}{\textbf{Remark}}
\newtheorem{lemma}{\textbf{Lemma}}
\newtheorem{corollary}{\textbf{Corollary}}
\newtheorem{proposition}{\textbf{Proposition}}

\newcommand{\R}{\mathbb{R}}
\renewcommand{\l}{\ell}
\newcommand{\Identity}{1\!\!1}
\newcommand{\G}{{\mathcal{G}}}
\newcommand{\E}{{\mathcal{E}}}

\newcommand{\V}{{\mathcal{V}}}

\usepackage{array}
\newcolumntype{L}[1]{>{\raggedright\let\newline\\\arraybackslash\hspace{0pt}}m{#1}}
\newcolumntype{C}[1]{>{\centering\let\newline\\\arraybackslash\hspace{0pt}}m{#1}}
\newcolumntype{R}[1]{>{\raggedleft\let\newline\\\arraybackslash\hspace{0pt}}m{#1}}
\renewcommand{\L}{\boldsymbol{\mathcal{L}}}
 
 \DeclareMathOperator*{\argmin}{argmin}
\renewcommand{\IEEEbibitemsep}{0pt plus 0.5pt}

\makeatletter
\IEEEtriggercmd{\reset@font\normalfont\fontsize{7.5pt}{8.0pt}\selectfont}
\makeatother
\IEEEtriggeratref{1}
\usepackage{dblfloatfix}

\begin{document}

\title{
Scalable $M$-Channel Critically Sampled Filter Banks for Graph Signals 
}
\author{\IEEEauthorblockN{{Shuni Li, Yan Jin, and David I Shuman}}
\thanks{Shuni Li is with the Department of Statistics, University of California, Berkeley, CA 94720, USA (email: shuni\_li@berkeley.edu). Yan Jin is with the Institute for Data, Systems, and Society, Massachusetts Institute of Technology, Cambridge, MA 02139, USA (email: yjin1@mit.edu). David I Shuman is with the Department of Mathematics, Statistics, and Computer Science, Macalester College, St. Paul, MN 55105, USA (email: dshuman1@macalester.edu).}
\thanks{This research has been funded in part by a grant to Macalester College from the Howard Hughes Medical Institute through the Precollege and Undergraduate Science Education Program.}
\thanks{The authors would like to thank Federico Poloni for providing a proof to Proposition \ref{Pr:mat_part}, and Andrew Bernoff for helpful discussions about the matrix partitioning problem discussed in Section \ref{Se:partition}.}
\thanks{MATLAB code for all numerical experiments in this paper is available at \url{http://www.macalester.edu/\textasciitilde dshuman1/publications.html}. It leverages the open access GSPBox \cite{gspbox}, into which it will soon be integrated.}}

\maketitle

\begin{abstract}
We investigate a scalable $M$-channel critically sampled filter bank for graph signals, where each of the $M$ filters is supported on a different subband of the graph Laplacian spectrum. For analysis, the graph signal is filtered on each subband and downsampled on a corresponding set of vertices. However, the classical synthesis filters are replaced with interpolation operators. For small graphs, we use a full eigendecomposition of the graph Laplacian to partition the graph vertices such that the $m^{th}$ set comprises a uniqueness set for signals supported on the $m^{th}$ subband. The resulting transform is critically sampled, the dictionary atoms are orthogonal to those supported on different bands, and graph signals are perfectly reconstructable from their analysis coefficients. We also investigate fast versions of the proposed transform that scale efficiently for large, sparse graphs. Issues that arise in this context include designing the filter bank to be more amenable to polynomial approximation, estimating the number of samples required for each band, performing non-uniform random sampling for the filtered signals on each band, and using efficient reconstruction methods. We empirically explore the joint vertex-frequency localization of the dictionary atoms, the sparsity of the analysis coefficients for different classes of signals, the reconstruction error resulting from the numerical approximations, and the ability of the proposed transform to compress piecewise-smooth graph signals. The proposed filter bank also yields a fast, approximate graph Fourier transform with a coarse resolution in the spectral domain.
\end{abstract}


\section{Introduction}

In graph signal processing \cite{shuman2013emerging}, transforms and filter banks can help exploit structure in the data, in order, for example, 
to compress a graph signal, remove noise, or fill in missing information. 
Broad classes of recently proposed transforms include graph Fourier transforms, vertex domain designs such as \cite{Crovella2003,wang}, top-down approaches such as \cite{szlam,gavish,irion}, diffusion-based designs such as \cite{coifman2006diffusion,Maggioni_biorthogonal}, spectral domain designs such as \cite{hammond2011wavelets}-\nocite{shuman2013spectrum}\cite{behjat2016signal}, windowed graph Fourier transforms  \cite{shuman2015vertex}, and generalized filter banks, 
the focus of this paper.

\begin{figure}[t]
\centerline{\includegraphics[width=3.2in]{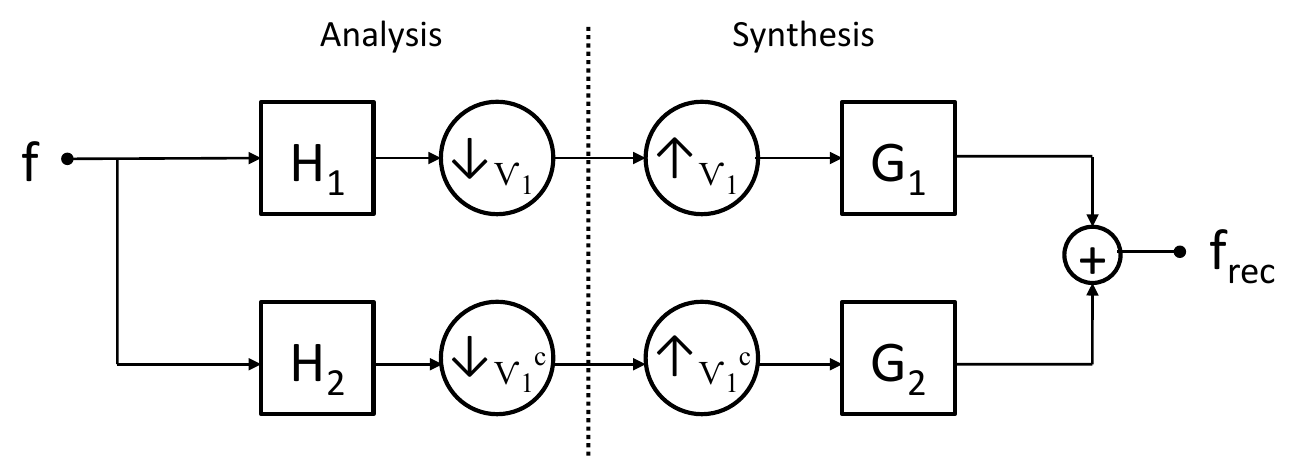}}
\caption{Two channel critically sampled graph filter bank. 
Here, $\mathbf{H}_1$ is a lowpass graph spectral filter, and $\mathbf{H}_2$ is a highpass graph spectral filter.}\label{Fig:two_channel}
\vspace{-.45cm}
\end{figure}

The extension of the classical two channel critically sampled filter bank to the graph setting is first proposed in \cite{narang_icip}. Fig.\ \ref{Fig:two_channel} shows the analysis and synthesis banks, where ${\bf H}_i$ and ${\bf G}_i$ are graph spectral filters \cite{shuman2013emerging}, and the lowpass and highpass bands are downsampled on complementary sets of vertices. For a general weighted, undirected graph, it is not straightforward 
how to design the downsampling sets and the four graph spectral filters to ensure perfect reconstruction. One approach is to separate the graph into a union of subgraphs, each of which has some regular structure. For example, \cite{narang2012perfect,narang_bior_filters} show that the normalized graph Laplacian eigenvectors of \emph{bipartite} graphs have a spectral folding property that make it possible to design analysis and synthesis filters to guarantee perfect reconstruction. They take advantage of this property by decomposing the graph into bipartite graphs and constructing a multichannel, separable filter bank, while \cite{sakiyama} adds vertices and edges to the original graph to form an approximating bipartite graph. References \cite{teke2016,teke2017ii} generalize this spectral folding property to $M$-block cyclic graphs, and leverage it to construct $M$-channel graph filter banks. Another class of regular structured graphs is \emph{shift invariant} graphs \cite[Chapter 5.1]{grady}. These graphs have a circulant graph Laplacian and their eigenvectors are the columns of the discrete Fourier transform matrix. Any graph can be written as the sum of circulant graphs, and \cite{ekambaram_icip,ekambaram2013globalsip,kotzagiannidis2016icassp} take advantage of this fact in designing critically sampled graph filter banks with perfect reconstruction. Another approach is to use architectures other than the critically sampled filter bank, such as lifting transforms \cite{jansen,narang_lifting_graphs} or pyramid transforms \cite{shuman_TSP_multiscale}.

Our approach in this paper  
is to replace the synthesis filters with interpolation operators on each subband of the graph spectrum. While this idea  
is suggested independently in \cite{chen2015discrete} for small graphs (say 5,000 or fewer vertices), we investigate it in more detail here, and extend it to large, sparse graphs.

We develop three variants of $M$-channel critically sampled filter banks ($M$-CSFB) for graph signals. The exact $M$-CSFB (initially presented in \cite{jin_conf}) features dictionary atoms that are jointly localized in the vertex and graph spectral domains, and therefore can compactly represent graph signals with localized singularities \cite{hammond2011wavelets}. The fast $M$-CSFB and fast, signal-adapted $M$-CSFB transforms we propose here 
share the same general structure as the exact $M$-CSFB, but scale efficiently for large, sparse graphs; i.e., they do not require a full eigendecomposition of the graph Laplacian.
Specific contributions of our work include:
 \begin{enumerate}
 \item A constructive method to partition the graph into uniqueness sets for any given partition of the spectrum. This is a key step in the exact $M$-CSFB (Section \ref{Se:fb_design}).
 \item A new filter bank design that is 
 both adapted to an efficient estimate of the distribution of the graph Laplacian spectrum and more amenable to polynomial approximation (Section \ref{Se:fast_mcsfb}).
 \item A \emph{scalable} method to sample and interpolate 
 bandpass and highpass graph signals. While we leverage the recent flurry of work in sampling and reconstruction of graph signals \cite{chen2015discrete}-\nocite{pesenson_paley,narang2013interpolation,anis2014towards,gadde2015probabilistic,shomorony,PuyTGV15,chen2015signal,tsitsvero2016uncertainty,chen2016signal,anis2016efficient}\cite{di2017sampling}, the prior literature focuses on methods that require a full eigendecomposition or assume the graph signals are smooth (lowpass). We use efficient convex optimization methods with a novel penalty term to perform the interpolation (Section \ref{Se:fast_mcsfb}). 
 \item The idea to adapt both the non-uniform sampling weights and the number of samples allocated to each band to the actual 
 signal being analyzed, in addition to the 
 graph structure, 
 because there is less benefit from taking 
 samples in areas of the graph where the filtered signal does not have much energy.  
 The fast, signal-adapted $M$-CSFB in Section \ref{Se:signal_adapted} is based on this concept. 
 \item Empirical explorations of the 
 exact
$M$-CSFB, fast $M$-CSFB, and signal-adapted fast $M$-CSFB transforms. In Section \ref{Se:ill2}, we investigate computation times, the reconstruction error resulting from the numerical approximations, tradeoffs involved in choosing the parameters, and applications such as compression and fast, approximate graph Fourier transforms.
 \end{enumerate}

\section{$M$-Channel Critically Sampled Filter Bank} \label{Se:fb_design}

\subsection{Notation}
We consider graph signals ${\bf f} \in \R^N$ residing on a weighted, undirected graph $\G=\{\V,\E,\mathbf{W}\}$, where $\V$ is the set of $N$ vertices, $\E$ is the set of edges, and $\mathbf{W}$ is the weighted adjacency matrix. 
Throughout, we take $\L$ to be the unnormalized graph Laplacian $\mathbf{D}-\mathbf{W}$, where $\mathbf{D}$ is the diagonal matrix of vertex degrees. However, our theory and proposed transform also apply to the normalized graph Laplacian $\mathbf{I}-\mathbf{D}^{-\frac{1}{2}}\mathbf{W}\mathbf{D}^{-\frac{1}{2}}$, or any other Hermitian operator. We can diagonalize the graph Laplacian as $\L=\mathbf{U}{\boldsymbol \Lambda}\mathbf{U}^{*}$, where ${\boldsymbol \Lambda}$ is the diagonal matrix of eigenvalues $\lambda_0,\lambda_1,\ldots,\lambda_{N-1}$ of $\L$, and the columns ${\bf u}_0,{\bf u}_1,\ldots,{\bf u}_{N-1}$ of $\mathbf{U}$ are the associated eigenvectors of $\L$. The graph Fourier transform of a signal is $\hat{\bf{f}}=\mathbf{U}^{*}{\bf f}$, and ${h}(\L){\bf f}=\mathbf{U}{h}({\boldsymbol \Lambda})\mathbf{U}^{*}{\bf f}$ applies the filter ${h}: [0,\lambda_{\max}] \rightarrow \R$ to the graph signal ${\bf f}$. We  
let $\mathbf{U}_{{\cal R}}$  
denote the submatrix formed by taking the columns of $\mathbf{U}$ associated with the Laplacian eigenvalues indexed by ${\cal R} \subseteq \{0,1,\ldots,N-1\}$, 
and  
$\mathbf{U}_{{\cal S},{\cal R}}$  
denote the submatrix  formed by taking the rows of $\mathbf{U}_{{\cal R}}$ associated with the vertices indexed by the set ${\cal S} \subseteq \{1,2,\ldots,N\}$.

\subsection{Architecture}
We start by constructing an ideal filter bank of $M$ graph spectral filters, where for 
band endpoints 
$0=\tau_0 < \tau_1 < \ldots < \tau_{M-1} \leq \tau_M$ (with $\tau_M > \lambda_{\max}$), the $m^{th}$ filter is defined as
\begin{align} \label{Eq:bandpass}
{h}_m(\lambda)=
\begin{cases}
1,& \tau_{m-1} \leq \lambda < \tau_m  \\
0,&\hbox{ otherwise}
\end{cases},~m=1,2,\ldots,M.
\end{align}
 Fig. \ref{Fig:fb} shows an example of such an ideal filter bank. Note that for each $\l \in \{0,1,\ldots,N-1\}$, ${h}_m(\lambda_\l)=1$ for exactly one $m$. 
Equivalently, we are forming a partition $\{{\cal R}_1,{\cal R}_2,\ldots,{\cal R}_M\}$ of $\{0,1,\ldots,N-1\}$ and setting
\begin{align*}
{h}_m(\lambda_{\l})=
\begin{cases}
1,&\hbox{ if } {\l} \in {\cal R}_m \\
0,&\hbox{ otherwise}
\end{cases},~m=1,2,\ldots,M.
\end{align*}

\begin{figure}[t]
\centerline{\includegraphics[width=2.6in]{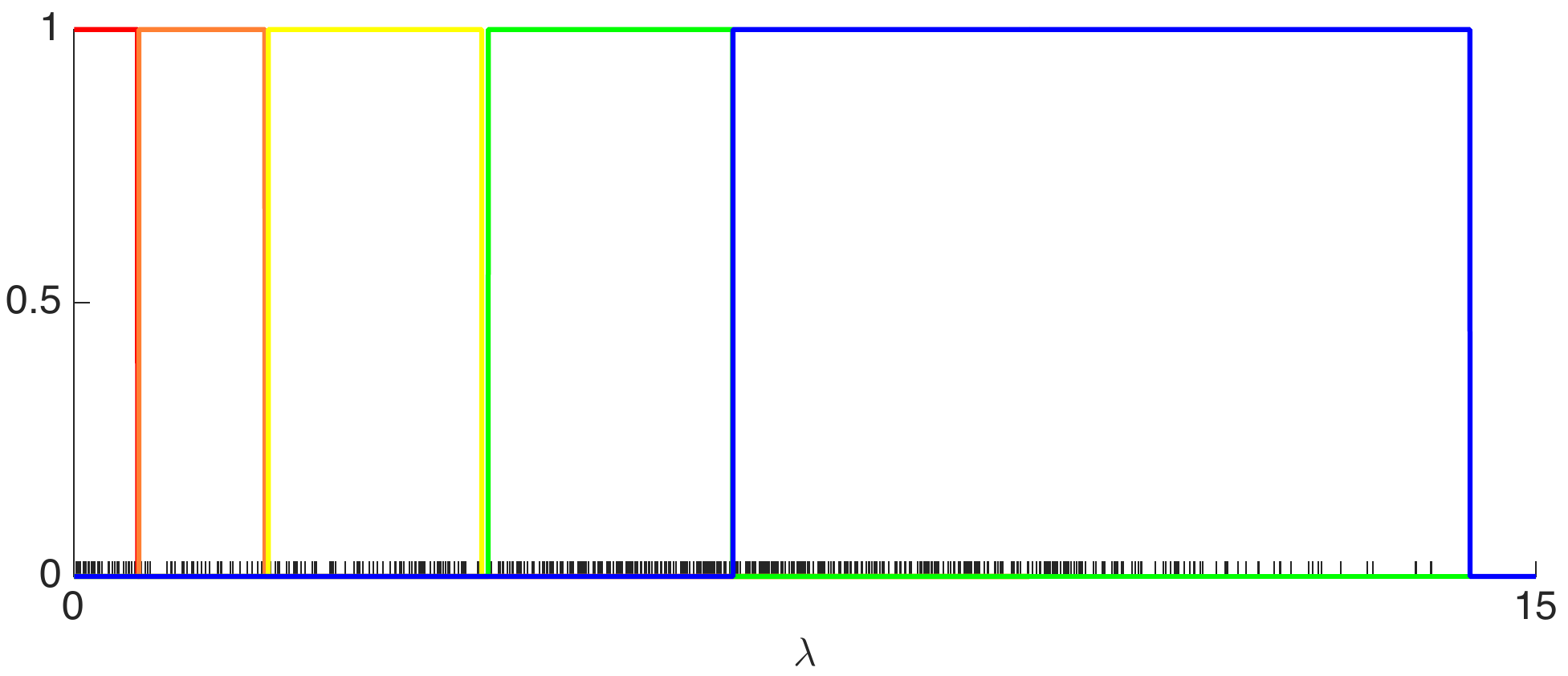}}
\caption{Example ideal filter bank. The red, orange, yellow, green, and blue filters span 31, 31, 63, 125, and 250 graph Laplacian eigenvalues, respectively, on a 500 node sensor network with 
a maximum graph Laplacian eigenvalue of 14.3. The tick marks on the x-axis represent the locations of the graph Laplacian eigenvalues.}\label{Fig:fb}
\vspace{-.3cm}
\end{figure}

\begin{figure}[t]
\centerline{\includegraphics[width=2.4in]{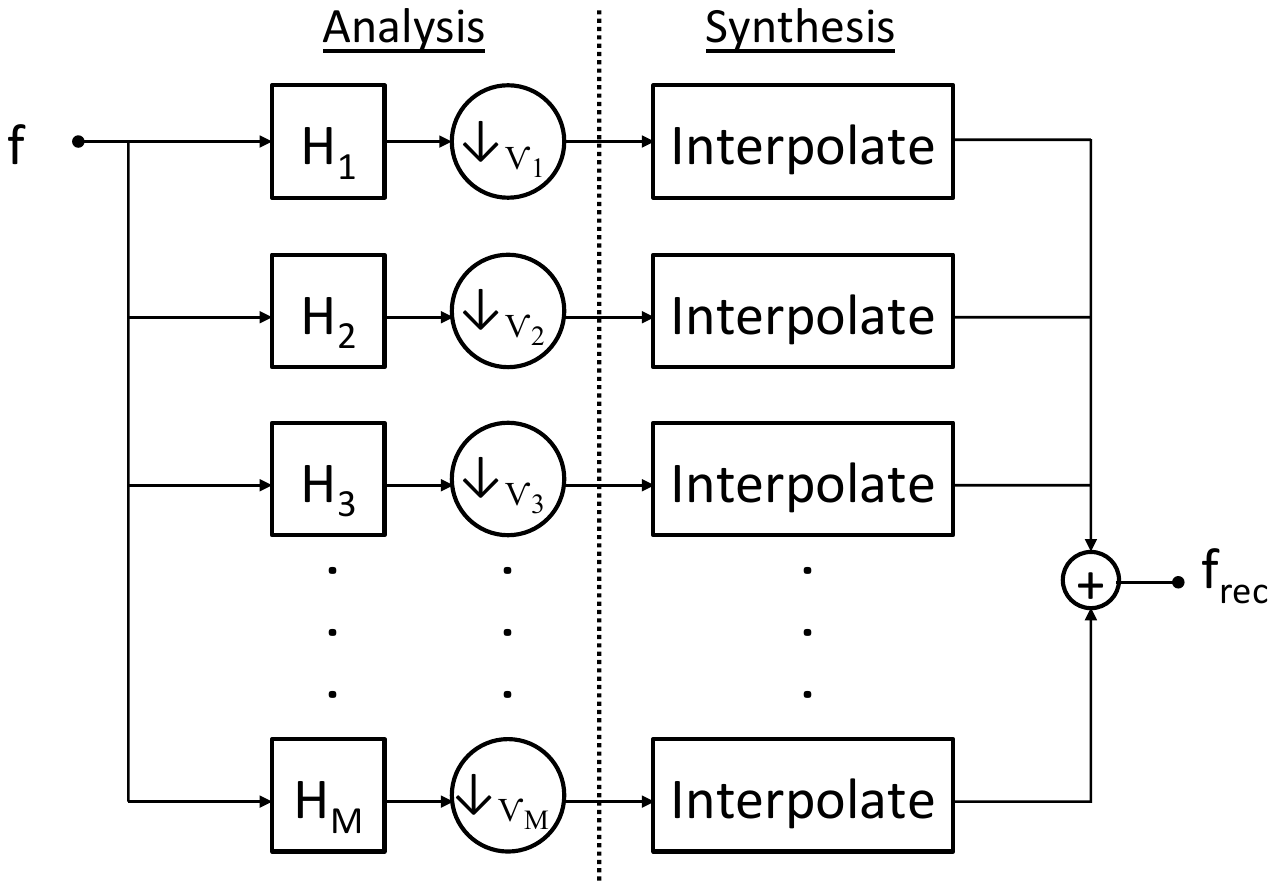}}
\caption{The $M$-channel critically sampled filter bank architecture. The sets $\V_1, \V_2, \ldots,\V_M$ form a partition of the set $\V$ of vertices, where each set $\V_m$ is a uniqueness set for graph signals supported on a different subband in the graph spectral domain. \vspace{-.1in}}\label{Fig:arch}
\vspace{-.2cm}
\end{figure}

The next step, which we discuss in detail in Section \ref{Se:partition}, is to partition the vertex set $\V$ into subsets $\V_1,\V_2,\ldots,\V_M$ such that $\V_m$ forms a uniqueness set for $\mbox{col}\left(\mathbf{U}_{{\cal R}_m}\right)$.
\begin{definition}[Uniqueness set \cite{pesenson_paley}] 
Let ${\cal P}$ be a subspace of $\R^n$. 
Then a subset $\V_s$ of the vertices $\V$ is a uniqueness set for 
${\cal P}$ if  and only if for all 
${\bf f},{\bf g} \in {\cal P}$,
${\bf f}_{\V_s}={\bf g}_{\V_s}$ implies ${\bf f}={\bf g}$. That is, if two signals in ${\cal P}$ have the same values on the vertices in the uniqueness set $\V_s$, then they must be the same signal.
\end{definition}
The following equivalent characterization of a uniqueness set is often useful.
\begin{lemma}[\cite{anis2014towards}, \cite{shomorony}]\label{Le:eq_uniq}
The set $\mathcal{S}$ of $k$ vertices is a uniqueness set for $\mbox{col}({\mathbf{U}}_{\mathcal T})$ if and only if the matrix whose columns are ${\bf u}_{{\mathcal T}_1},{\bf u}_{{\mathcal T}_2},\ldots,{\bf u}_{{\mathcal T}_k}, {\boldsymbol \delta}_{{\mathcal{S}}^c_1}, {\boldsymbol \delta}_{{\mathcal{S}}^c_2}, \ldots, {\boldsymbol \delta}_{{\mathcal{S}}^c_{n-k}}$ is nonsingular,
where ${\bf u}_{{\mathcal T}_i}$ is the $i$th column of ${\mathbf{U}}_{\mathcal T}$, and each ${\boldsymbol \delta}_{{\mathcal{S}}^c_i}$ is a Kronecker delta centered on a vertex not included in $\mathcal{S}$.
\end{lemma}

The $m^{th}$ channel of the analysis 
bank 
filters the graph signal by an ideal 
filter on subband ${\cal R}_m$, and downsamples the result 
onto the vertices in $\V_m$. For synthesis, we  
interpolate 
from the samples on $\V_m$ to $\mbox{col}\left(\mathbf{U}_{{\cal R}_m}\right)$. Denoting the analysis coefficients 
(i.e., the filtered and downsampled signal) of the 
$m^{th}$ branch by ${\bf y}_{\V_m}$, we have
\begin{align} \label{Eq:synth}
{\bf f}_{rec} = \sum_{m=1}^M  \mathbf{U}_{{\cal R}_m} \mathbf{U}_{\V_m,{\cal R}_m}^{-1} {\bf y}_{\V_m}.
\end{align}

\noindent If there is no error in the coefficients, then the reconstruction is perfect, because $\V_m$ is a uniqueness set for $\mbox{col}\left(\mathbf{U}_{{\cal R}_m}\right)$, ensuring $\mathbf{U}_{\V_m,{\cal R}_m}$ is full rank. Fig.\ \ref{Fig:arch} shows the architecture of the proposed $M$-channel critically sampled filter bank with interpolation on the synthesis side. \emph{Critically sampled} refers to the fact that the number of analysis coefficients is equal to the length of the original signal; that is, $\sum_{m=1}^M |\V_m| = N$.

\subsection{Partitioning the graph into uniqueness sets for different frequency bands}\label{Se:partition}

In this section, we show how to partition the set of vertices into uniqueness sets for different subbands of the graph Laplacian eigenvectors. We start with the easier case of $M=2$ and then examine the general case.

\subsubsection{$M=2$ channels}
First we show that if a set of vertices is a uniqueness set for a set of signals contained in a band of spectral frequencies, then the complement set of vertices is a uniqueness set for the set of signals with no energy in that band of spectral frequencies.
\begin{proposition}\label{Le:highpass_uniqueness}
On a graph ${\mathcal G}$ with $N$ vertices, let 
${\mathcal T} \subseteq \{0,1,\ldots,N-1\}$ denote a subset of the graph Laplacian eigenvalue indices, and let ${\mathcal T}^c=\{0,1,\ldots,N-1\} \setminus {\mathcal T}$.
Then 
$\mathcal{S}^c$ is a uniqueness set for $\mbox{col}({\mathbf{U}}_{{\mathcal T}^c})$
 if and only if
$\mathcal{S}$ is a uniqueness set for 
$\mbox{col}({\mathbf{U}}_{{\mathcal T}})$.
\end{proposition}
This fact follows from either the CS decomposition \cite[Equation (32)]{paige})
or the nullity theorem \cite[Theorem 2.1]{strangInterplay}.  
We also provide a standalone proof  
in the Appendix that only requires that the space spanned by the first $k$ columns of ${\mathbf{U}}$ is orthogonal to the space spanned by the last $N-k$ columns, not that ${\mathbf{U}}$ is an orthogonal matrix.
The Steinitz exchange lemma \cite{steinitz} guarantees that we can find the uniqueness set $\mathcal{S}$ (and thus $\mathcal{S}^c$), and the graph signal processing literature contains methods such as Algorithm 1 of \cite{shomorony} to do so.

\subsubsection{$M>2$ channels}
The issue with using the methods of Proposition \ref{Le:highpass_uniqueness} for the case of $M>2$ is that while the submatrix ${\mathbf{U}}_{{\mathcal S^c},{\mathcal T^c}}$ is nonsingular, it is not necessarily orthogonal, and so we cannot proceed with an inductive argument. The following proposition and corollary circumvent this issue by only using the nonsingularity of the original matrix. The proof of the following proposition is due to Federico Poloni \cite{poloni}, 
and we later discovered the same method in \cite{greeneMultiple}, \cite[Theorem 3.3]{greene_magnanti}.
\begin{proposition}\label{Pr:mat_part}
Let $\mathbf{A}$ be an $N \times N$ nonsingular matrix, and $\beta=\{\beta_1,\beta_2,\ldots,\beta_M\}$ be a partition of $\{1,2,\ldots,N\}$. Then there exists another partition $\alpha=\{\alpha_1,\alpha_2,\ldots,\alpha_M\}$ of $\{1,2,\ldots,N\}$ with $|\alpha_i|=|\beta_i|$ for all $i$ such that the $M$ square submatrices $\mathbf{A}_{\alpha_i,\beta_i}$ are all nonsingular.
\end{proposition}
\begin{proof}
First consider the case $M=2$, and let $k=|\beta_1|$. Then by the generalized Laplace expansion \cite{gle}, 
\begin{align}\label{Eq:det2}
\mbox{det}(\mathbf{A})~=~\sum_{\mathclap{\{\alpha_1 \subset \{1,2,\ldots,N\}: |\alpha_1|=k \}}}~~\sigma_{\alpha_1,\beta_1} \mbox{det}(\mathbf{A}_{\alpha_1,\beta_1})\mbox{det}(\mathbf{A}_{\alpha_1^c,\beta_2}),
\end{align}
where the sign $\sigma_{\alpha_1,\beta_1}$ of the permutation determined by $\alpha_1$ and $\beta_1$ is equal to 1 or -1. Since $\mbox{det}(\mathbf{A})\neq 0$, one of the terms in the summation of \eqref{Eq:det2} must be nonzero, ensuring a choice of $\alpha_1$ such that the submatrices $\mathbf{A}_{\alpha_1,\beta_1}$ and $\mathbf{A}_{\alpha_1^c,\beta_2}$ are nonsingular. We can choose $\{\alpha_1,\alpha_1^c\}$ as the desired partition. For $M>2$, by induction, we have
\begin{align}\label{Eq:detM}
\mbox{det}(\mathbf{A})~=~\sum_{\mathclap{\{\hbox{Partitions } \alpha \hbox{ of }\{1,2,\ldots,N\}: |\alpha_i|=|\beta_i|~\forall i \}}}~~\sigma_{\alpha}~{\textstyle \prod_{i=1}^M}~\mbox{det}(\mathbf{A}_{\alpha_i,\beta_i}),
\end{align}
where again $|\sigma_\alpha|=1$,  
and one of the terms in the summation in \eqref{Eq:detM} must be nonzero, yielding the desired partition. 
\end{proof}

\begin{corollary}\label{Co:part_uniq}
For any 
partition $\{{\cal R}_1,{\cal R}_2,\ldots {\cal R}_M \}$ of the graph Laplacian eigenvalue indices $\{0,1,\ldots,N-1\}$ into $M$ subsets, there exists a partition $\{\V_1,\V_2,\ldots,\V_M\}$ of the graph vertices into $M$ subsets 
such that for every $m \in \{1,2,\ldots,M\}$, $|\V_m|=|{\cal R}_m|$ and 
$\V_m$ is a uniqueness set for $\mbox{col}\left(\mathbf{U}_{{\cal R}_m}\right)$. 
\end{corollary}
\begin{proof}
By Proposition \ref{Pr:mat_part}, we can find a partition such that $\mathbf{U}_{\V_m,{\cal R}_m}$ is nonsingular for all $m$. Let $\mathbf{E}_m$ be the matrix formed by joining the $k_m$ columns of  $\mathbf{U}$ indexed by ${\cal R}_m$ with $N-k_m$ Kronecker deltas centered on all vertices not included in $\V_m$. By Lemma \ref{Le:eq_uniq}, it suffices to show that the matrices $\mathbf{E}_m$ are all nonsingular. Yet, for all $m$, we have
$|\mbox{det}(\mathbf{E}_m)|=|\mbox{det}(\mathbf{U}_{\V_m,{\cal R}_m})| \neq 0.$
\end{proof}

\setlength{\textfloatsep}{12pt}
\begin{algorithm}[t] 
\caption{Partition the vertices into uniqueness sets for each frequency band}
\begin{algorithmic}
\State \textbf{Input} $\mathbf{U}$,~a partition $\{{{\cal R}_1,{\cal R}_2,\ldots,{\cal R}_M}\}$ 
\State $\mathcal{S} \gets \emptyset$
\For {$m=1,2,\ldots,M$} 
\State Find sets $\gamma_1, \gamma_2 \subset {\mathcal{S}}^c$ s.t. $\mathbf{U}_{\gamma_1,{\cal R}_m}$ and $\mathbf{U}_{\gamma_2,{\cal R}_{m+1:M}}$ are nonsingular
\While {$\gamma_1 \cap \gamma_2 \neq \emptyset$}
\State Find a chain of pivots from an element $y \in {\mathcal S}^c \setminus (\gamma_1 \cup \gamma_2)$
to an element $z \in \gamma_1 \cap \gamma_2$ (c.f. \cite{greene_magnanti} for details)
\State Update $\gamma_1$ and $\gamma_2$ by carrying out a series of exchanges resulting with $y$ and $z$ each appearing in exactly one of $\gamma_1$ or $\gamma_2$
\EndWhile
\State $\V_m \gets \gamma_1$
\State $\mathcal{S} \gets \mathcal{S} \cup \gamma_1$ 
\EndFor
\State \textbf{Output} the partition $\{\V_1,\V_2,\ldots,\V_M$\}
\end{algorithmic}
\label{Al:uniqueness}
\end{algorithm}

Corollary \ref{Co:part_uniq} ensures the existence of the desired partition, and the proof of Proposition \ref{Pr:mat_part} suggests that we can find it inductively. However, given a partition of the columns of ${\mathbf{A}}$ into two sets $\mathcal{T}$ and $\mathcal{T}^c$, Proposition \ref{Pr:mat_part} does not provide a constructive method to partition the rows of ${\mathbf{A}}$ into two sets $\mathcal{S}$ and $\mathcal{S}^c$ such that the submatrices ${\mathbf{A}}_{{\mathcal{S}},{\mathcal{T}}}$ and ${\mathbf{A}}_{{\mathcal{S}^c},{\mathcal{T}}^c}$ are nonsingular. This problem is studied in the more general framework of matroid theory in \cite{greene_magnanti}, which gives an algorithm to find the desired row partition. 
We summarize this method in Algorithm \ref{Al:uniqueness}, which takes in a partition $\{{\cal R}_1,{\cal R}_2,\ldots {\cal R}_M \}$ of the spectral indices and constructs the partition $\{\V_1,\V_2,\ldots,\V_M\}$ of the vertices.  In Fig. \ref{Fig:part_examples}, we show two examples of the resulting partitions.

\begin{figure}[t]
\begin{minipage}[m]{0.43\linewidth}
\centerline{\includegraphics[width=.88\linewidth]{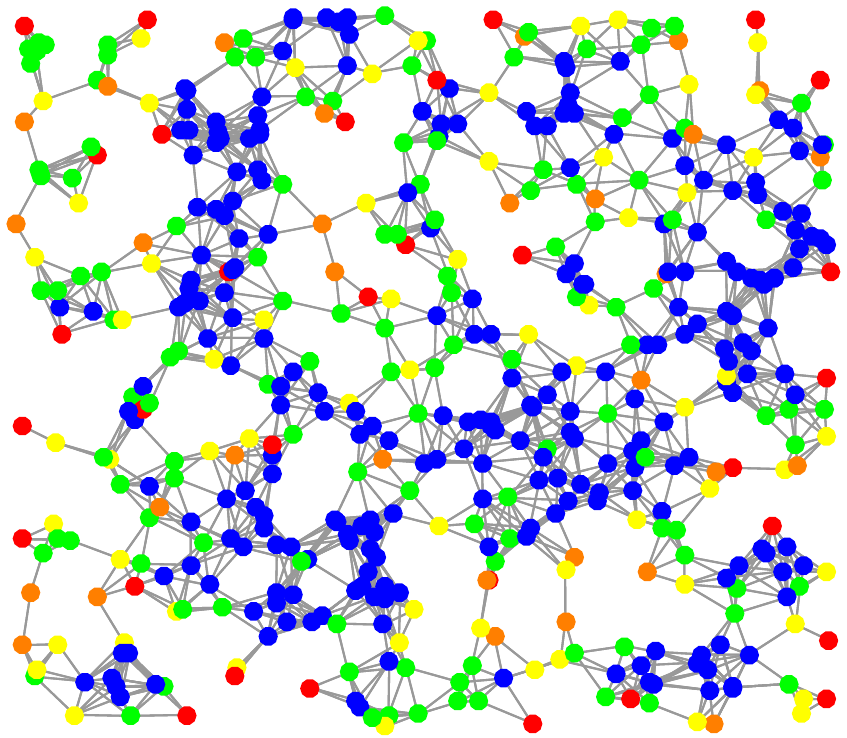}}
\end{minipage}
\hspace{.003\linewidth}
\begin{minipage}[m]{0.05\linewidth}
\centerline{\includegraphics[width=.85\linewidth]{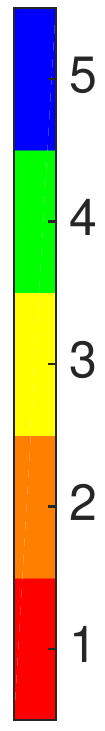}}
\end{minipage}
\hspace{.003\linewidth}
\begin{minipage}[m]{0.46\linewidth}
\centerline{\includegraphics[width=.9\linewidth]{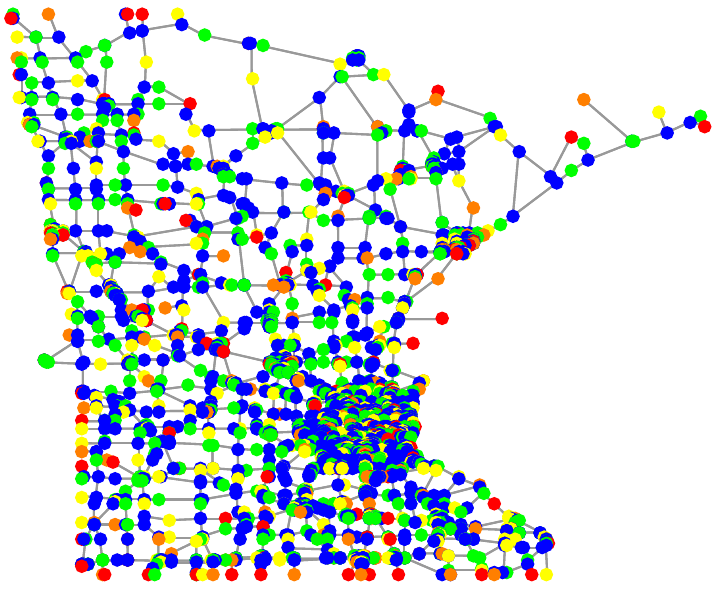}}
\end{minipage}
\caption{Partitions of a 500 node random sensor network and the Minnesota road network \cite{gleich} into uniqueness sets for five different sectral bands, with the indices increasing from lowpass bands (1) to highpass bands (5).} \label{Fig:part_examples}
\vspace{-.4cm}
\end{figure} 

\begin{remark}
While Algorithm \ref{Al:uniqueness} always finds a partition into uniqueness sets, such a partition is usually not unique. The initial choices of $\gamma_i$ in each loop play a significant role in the final 
partition. In the numerical experiments, 
we use the greedy algorithm in \cite[Algorithm 1]{shomorony}  to find an initial choice for $\gamma_1$, permute the complement of $\gamma_1$ to the top,
and then perform row reduction 
to find an initial choice for $\gamma_2$.  
\end{remark}

\subsection{Transform properties}
\subsubsection{Dictionary atoms}
Let ${\bf M}_m \in \R^{|{\cal V}_m| \times N}$ be the downsampling matrix for the $m^{th}$ channel. That is, ${\bf M}_m(i,j)=1$ if vertex $j$ is the $i^{th}$ element of ${\cal V}_m$, and 0 otherwise. The proposed transform is a linear mapping ${\cal F}: \R^N \rightarrow \R^N$ by 
${\cal F}{\bf f} = \boldsymbol{\Phi}^{\top} {\bf f}$, where the 
resulting dictionary is of the form 
$\boldsymbol{\Phi}:=\Bigl[h_1(\L){\bf M}_1^{\top} \mid 
h_2(\L){\bf M}_2^{\top} \mid 
\cdots  
\mid 
h_M(\L){\bf M}_M^{\top}
\Bigr]$.
While the transform is not orthogonal, each atom (column of $\boldsymbol{\Phi}$) is orthogonal to all atoms concentrated on other spectral bands. This is because the atoms are projections of Kronecker deltas onto the orthogonal subspaces spanned by the Laplacian eigenvectors of each band. 
That is, each atom is of the form $h_m(\L){\boldsymbol \delta}_i$, where vertex $i$ is in ${\cal V}_m$. If $m \neq m^{\prime}$, then the inner product of two atoms from different bands is given by
\begin{align} \label{Eq:ortho_bands}
&\langle h_m(\L){\boldsymbol \delta}_i , h_{m^{\prime}}(\L){\boldsymbol \delta}_{i^{\prime}} \rangle \nonumber \\
&~~~~~~
= {\boldsymbol \delta}_i^{\top} {\bf U} h_m({\boldsymbol \Lambda}) {\bf U}^* {\bf U} h_{m^{\prime}}({\boldsymbol \Lambda}) {\bf U}^* {\boldsymbol \delta}_{i^{\prime}}= 0,
\end{align}
since ${\bf U}^* {\bf U}={\bf I}$ and $h_m(\lambda)h_{m^{\prime}}(\lambda)=0$ for all $\lambda$ by design.
Note also that the wavelet atoms at all scales ($m>1$) have mean zero, as they have no energy at eigenvalue zero.

\begin{figure}[tbh]
\begin{minipage}[m]{0.48\linewidth}
\centerline{\includegraphics[width=.9\linewidth]{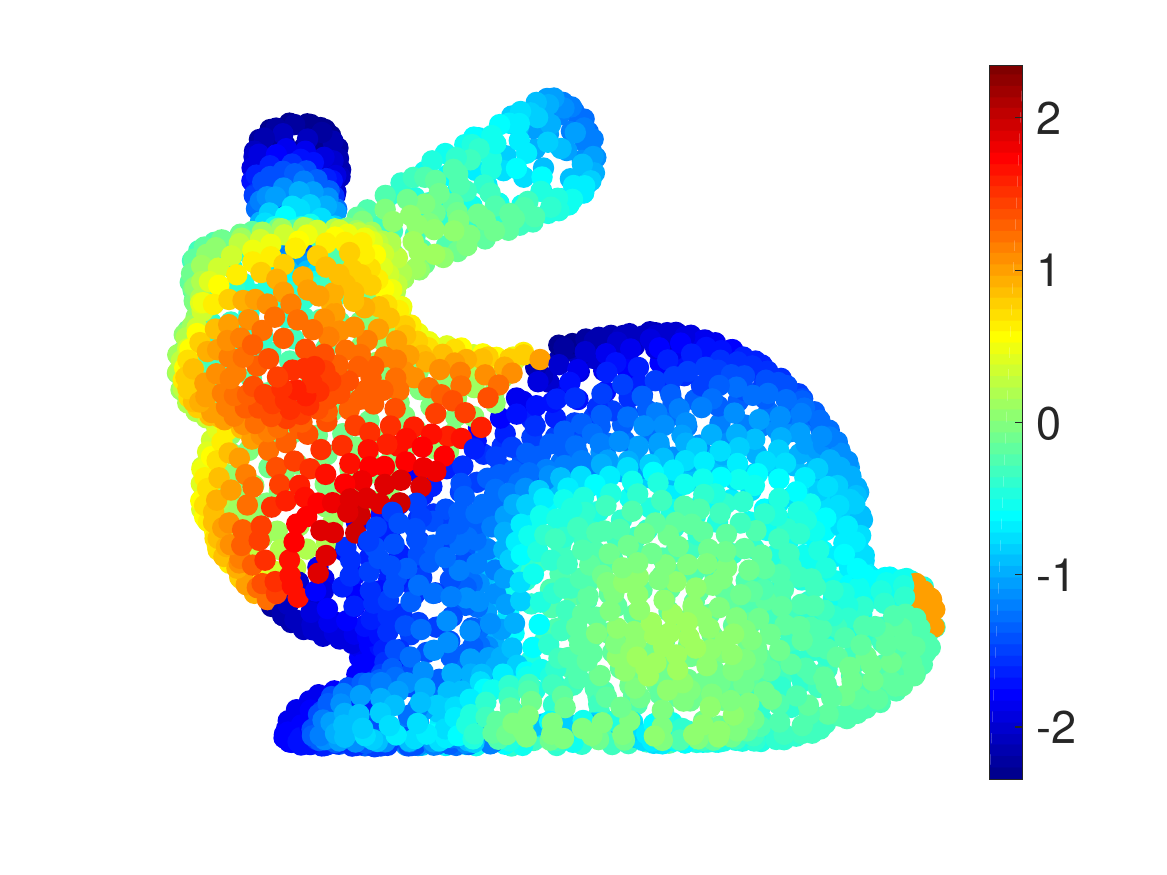}}
\centerline{\small{(a)}}
\end{minipage}
\begin{minipage}[m]{0.48\linewidth}
\centerline{\includegraphics[width=.9\linewidth]{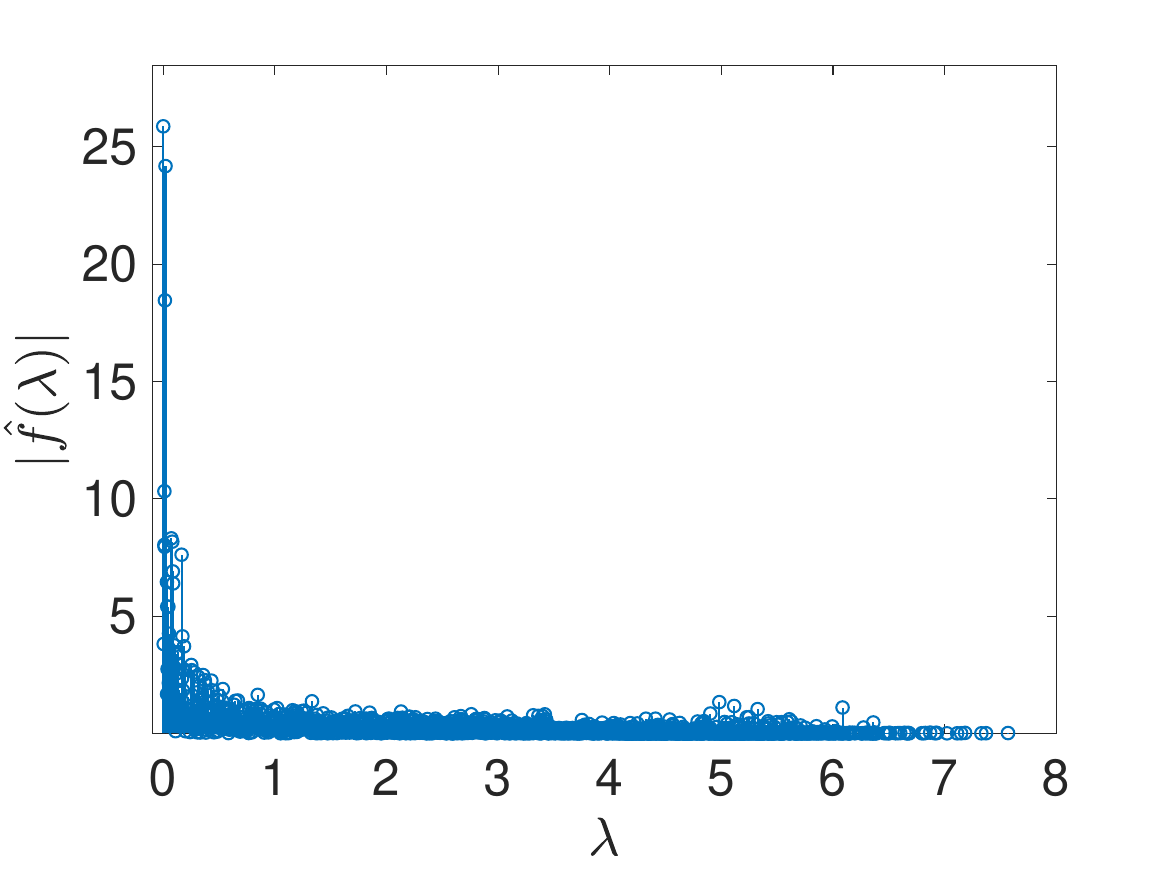}~}
\centerline{\small{(b)}}
\end{minipage} \\
\begin{minipage}[m]{0.48\linewidth}
\centerline{\includegraphics[width=.9\linewidth]{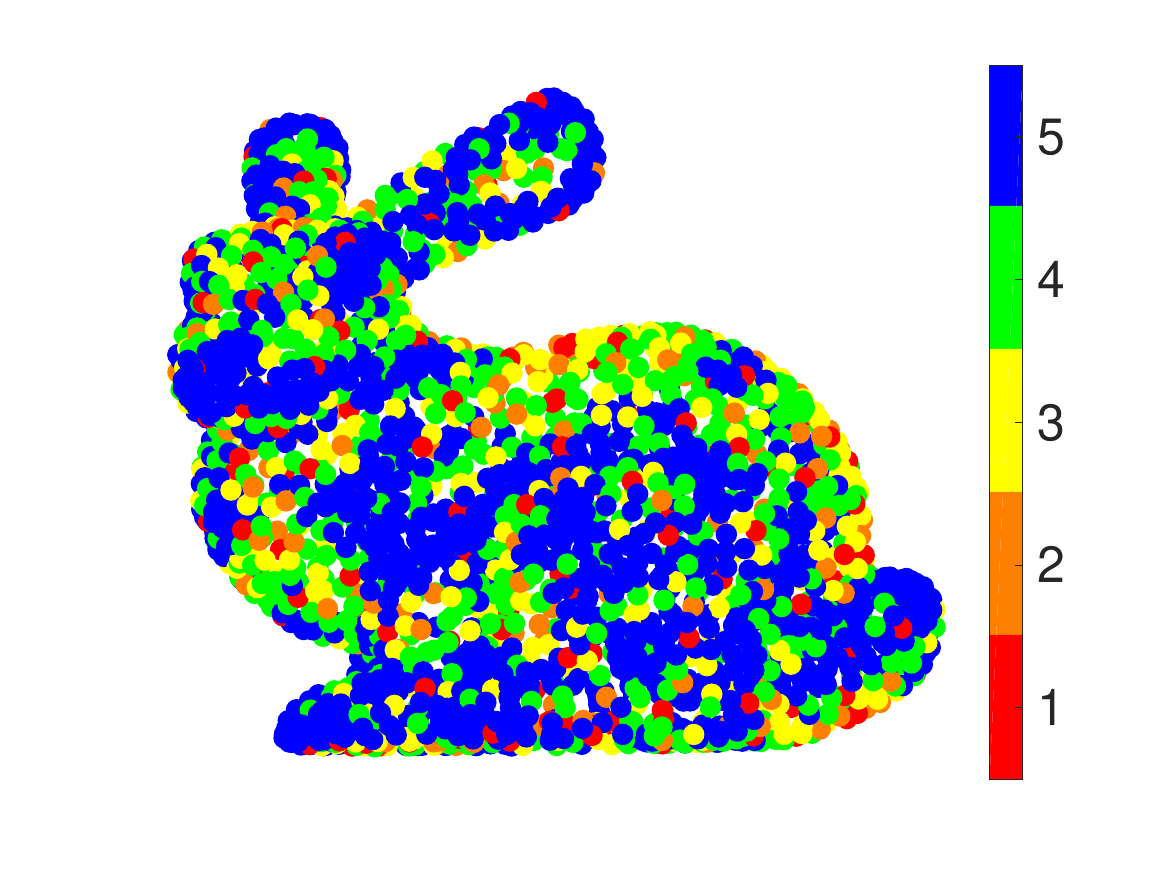}}
\centerline{\small{(c)}}
\end{minipage}
\begin{minipage}[m]{0.48\linewidth}
\centerline{\includegraphics[width=.9\linewidth]{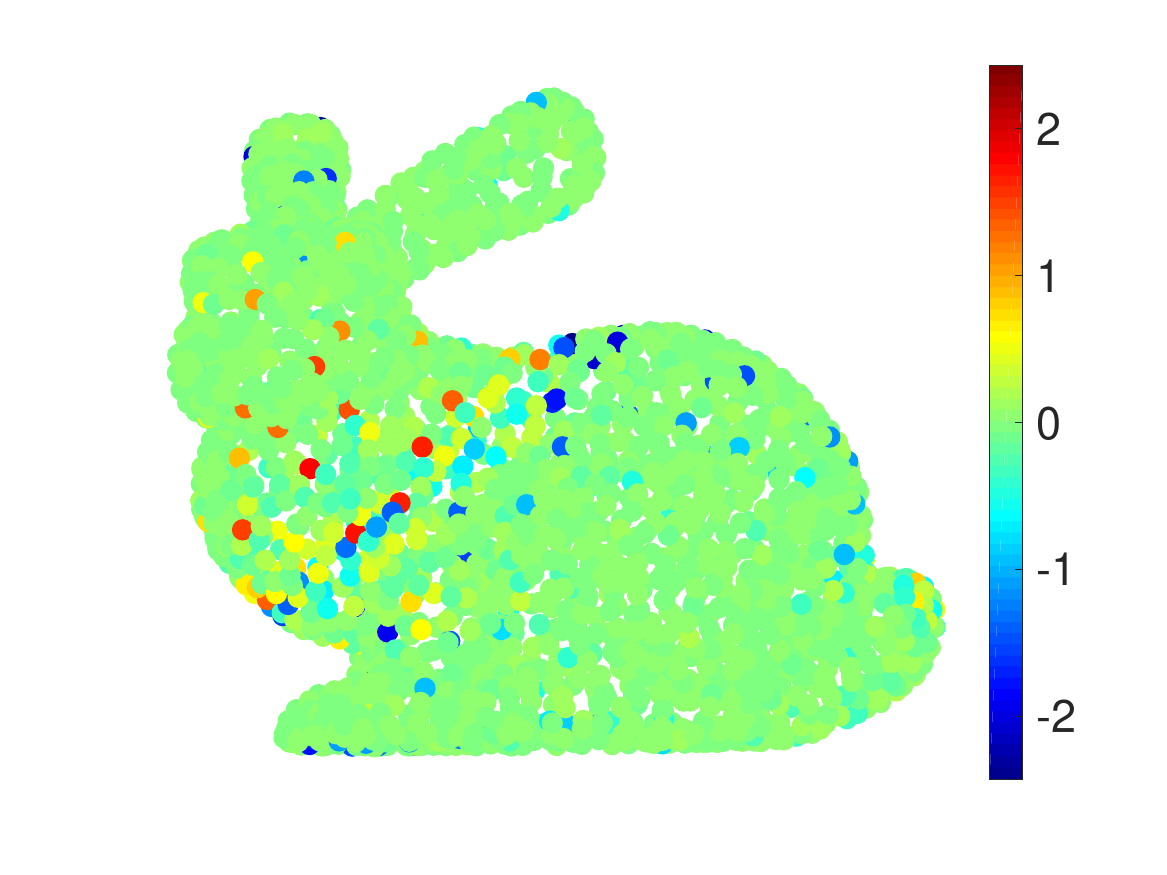}}
\centerline{\small{(d)}}
\end{minipage}
\caption{(a)-(b) Piecewise smooth signal on the Stanford bunny graph \cite{bunny} in the vertex and graph spectral domains, respectively. (c) Partition of the graph into uniqueness sets for five different spectral bands. (d) $M$-channel filter bank analysis coefficients of the signal shown in (a) and (b). \vspace{-.1in}}\label{Fig:bunny_signal}
\end{figure}

\begin{figure*}[t] 
\begin{minipage}[m]{0.16\linewidth}
~
\end{minipage}
\begin{minipage}[m]{0.16\linewidth}
\centerline{\small{~~~Scaling Functions}}
\end{minipage}
\hspace{.01\linewidth}
\begin{minipage}[m]{0.16\linewidth}
\centerline{\small{Wavelet Scale 1~}}
\end{minipage}
\begin{minipage}[m]{0.16\linewidth}
\centerline{\small{Wavelet Scale 2~}}\end{minipage}
\begin{minipage}[m]{0.16\linewidth}
\centerline{\small{Wavelet Scale 3~}}\end{minipage}
\begin{minipage}[m]{0.16\linewidth}
\centerline{\small{Wavelet Scale 4~}}\end{minipage} \\
\begin{minipage}[m]{0.16\linewidth}
\centerline{\small{Example Atom}}
\end{minipage}
\begin{minipage}[m]{0.16\linewidth}
\centerline{~~\includegraphics[width=1\linewidth]{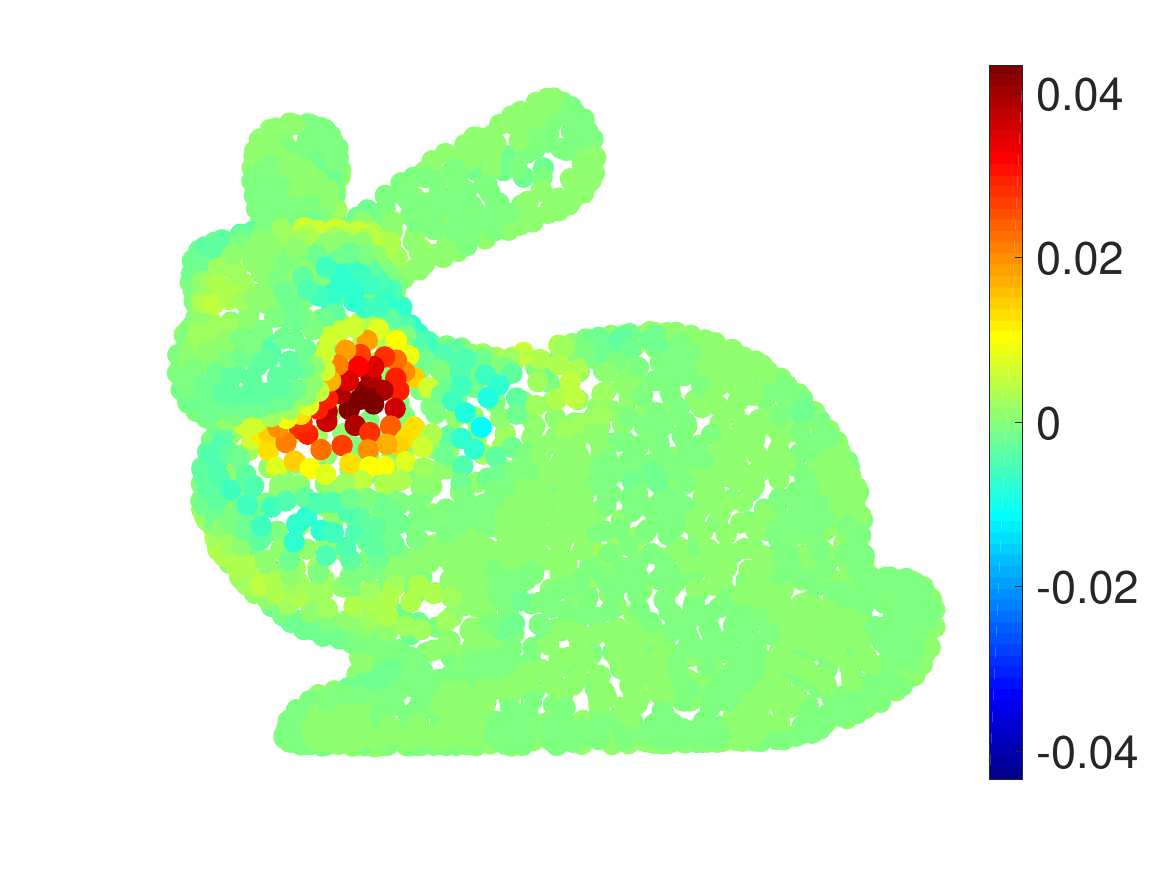}}
\end{minipage}
\begin{minipage}[m]{0.16\linewidth}
\centerline{~~\includegraphics[width=1\linewidth]{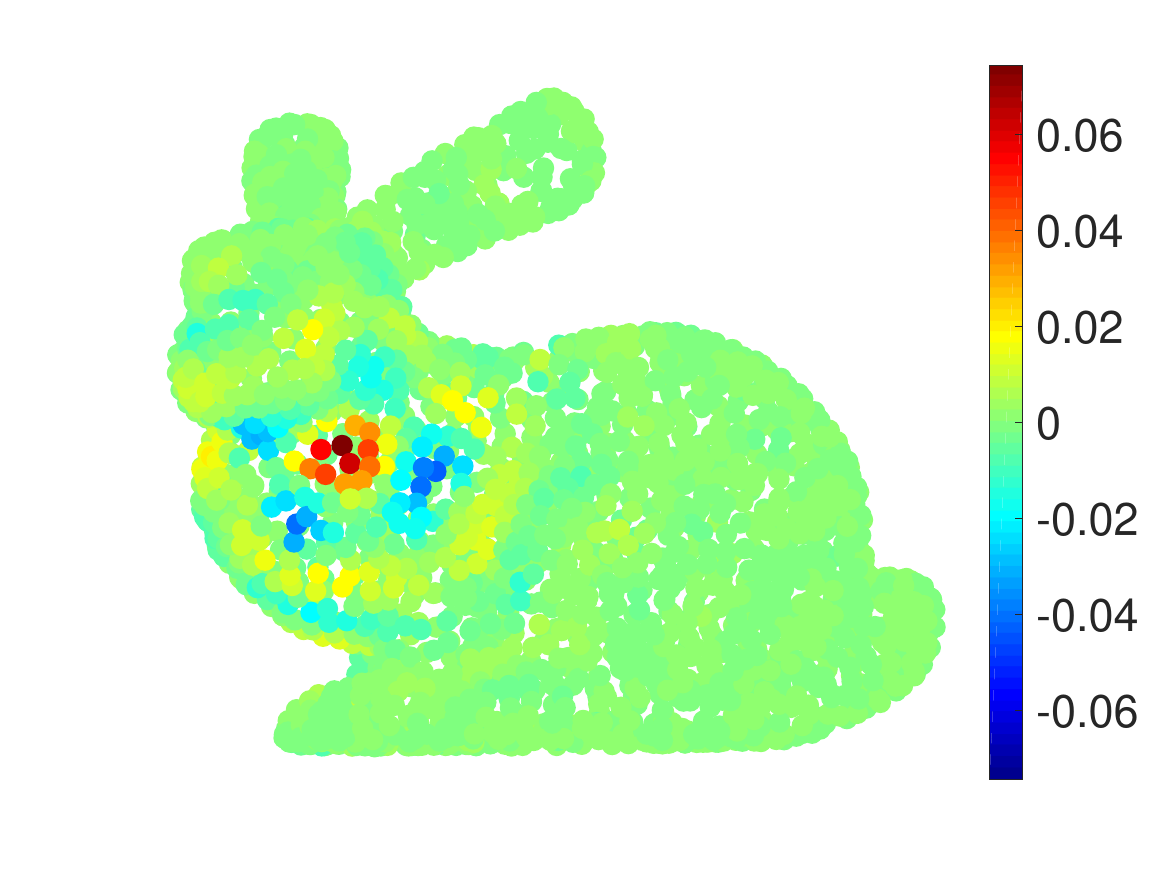}}
\end{minipage}
\begin{minipage}[m]{0.16\linewidth}
\centerline{~~\includegraphics[width=1\linewidth]{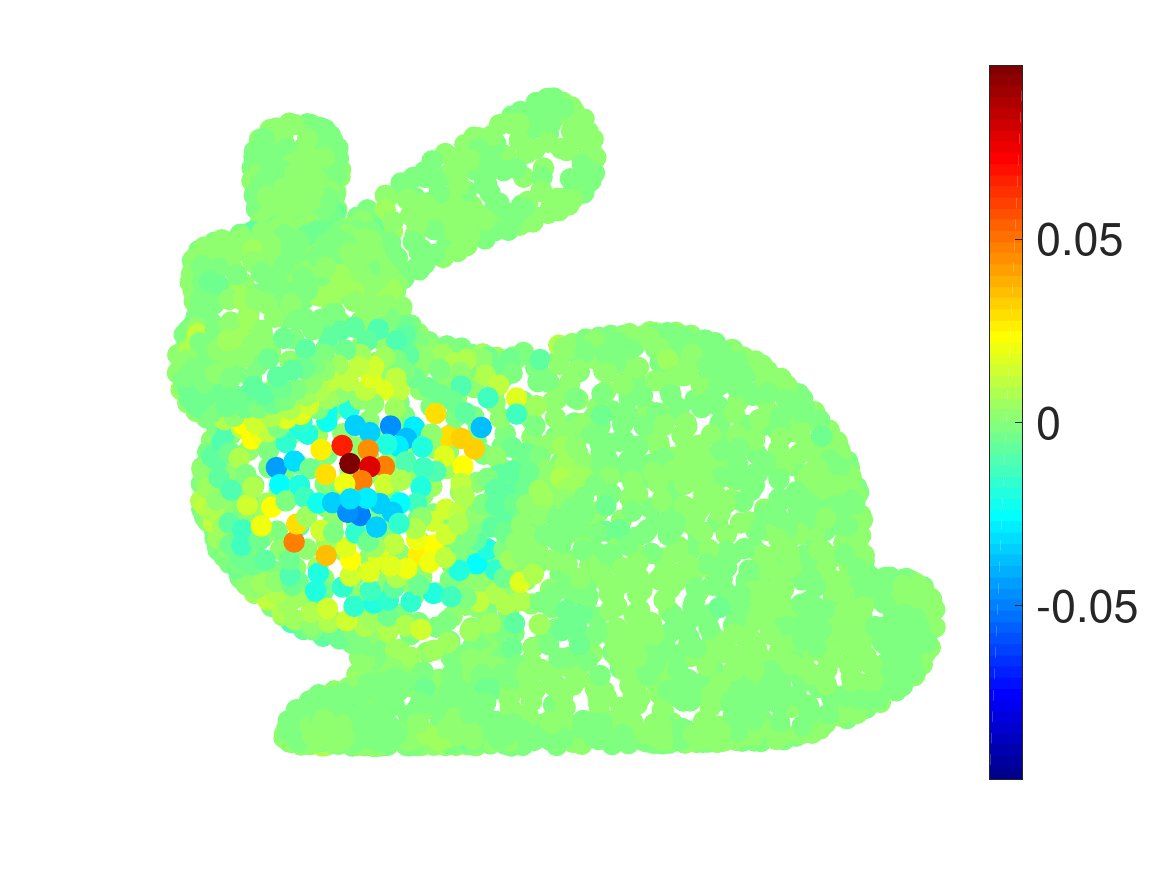}}
\end{minipage}
\begin{minipage}[m]{0.16\linewidth}
\centerline{~~\includegraphics[width=1\linewidth]{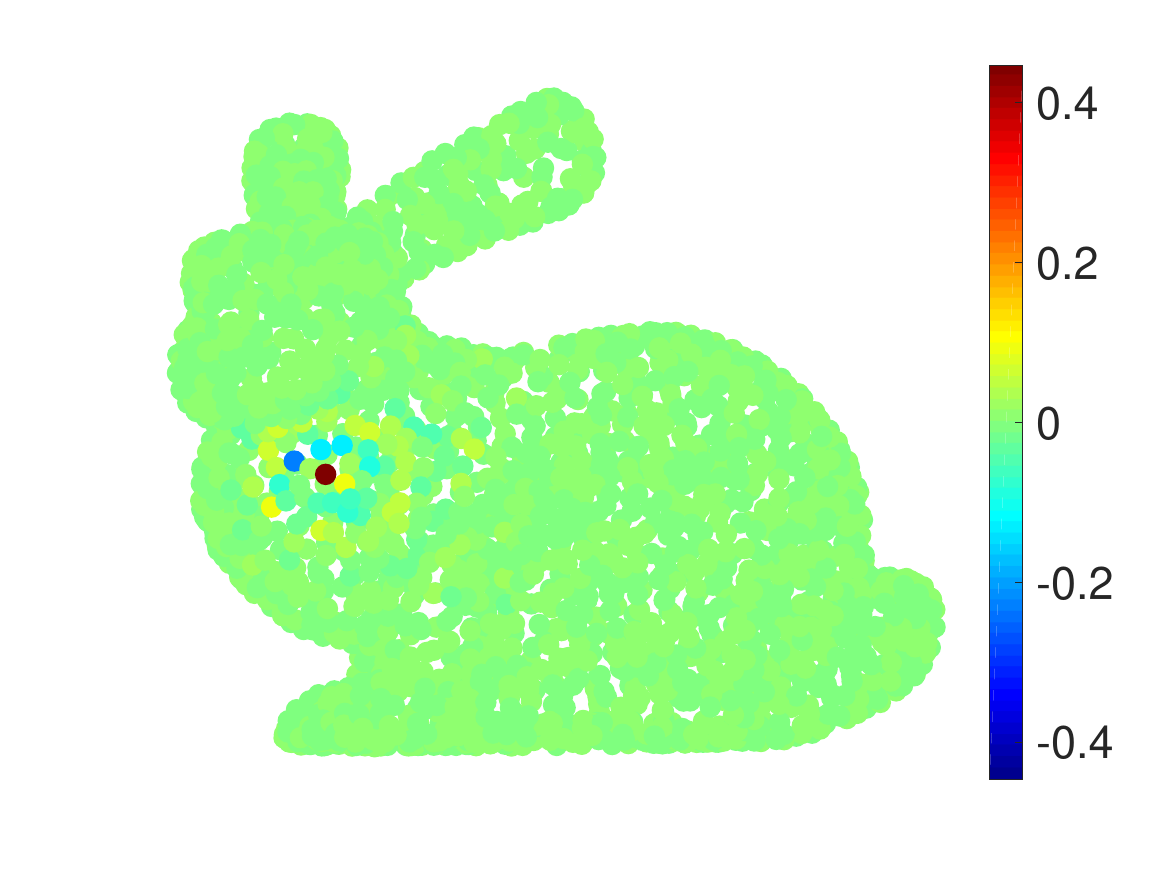}}
\end{minipage}
\begin{minipage}[m]{0.16\linewidth}
\centerline{~~\includegraphics[width=1\linewidth]{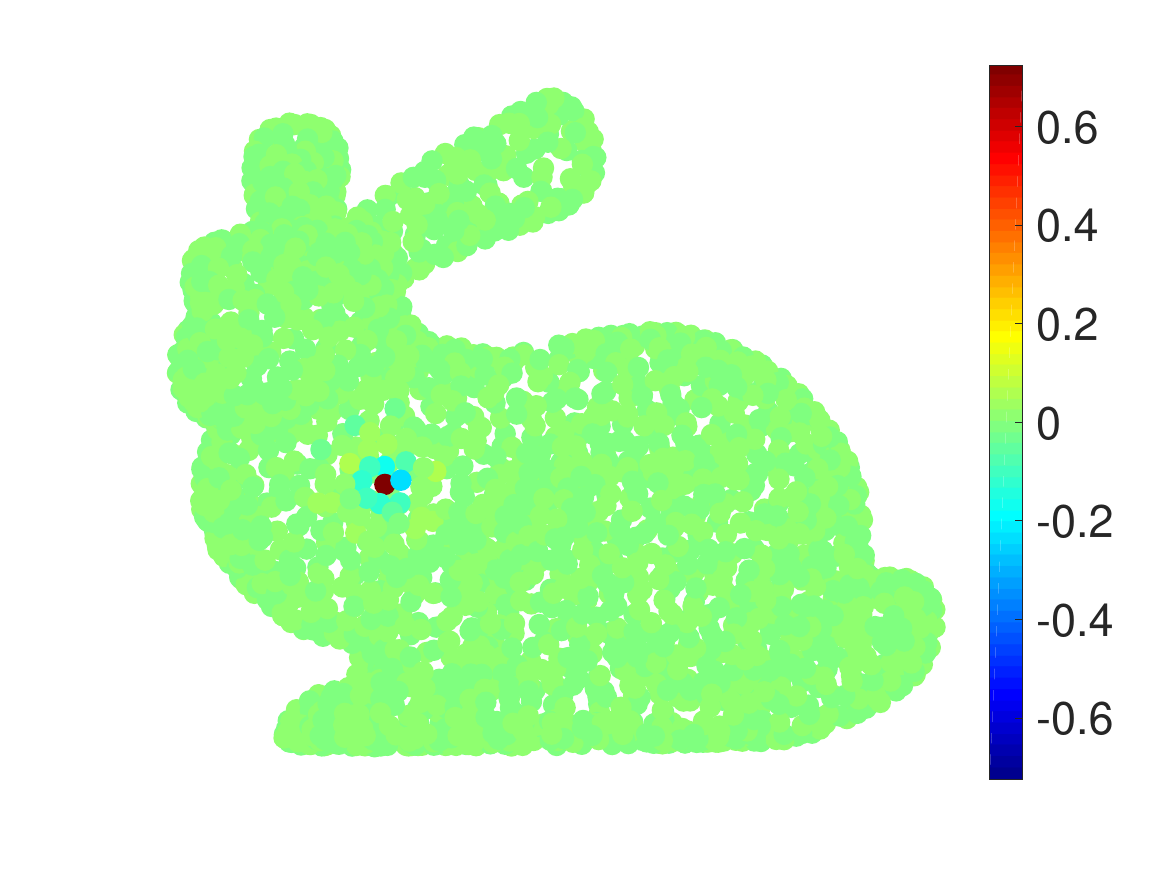}}
\end{minipage}\\
\begin{minipage}[m]{0.16\linewidth}
\centerline{\small{Spectral Content}}
\centerline{\small{of All Atoms}}
\end{minipage}
\begin{minipage}[m]{0.16\linewidth}
\centerline{\includegraphics[width=1\linewidth]{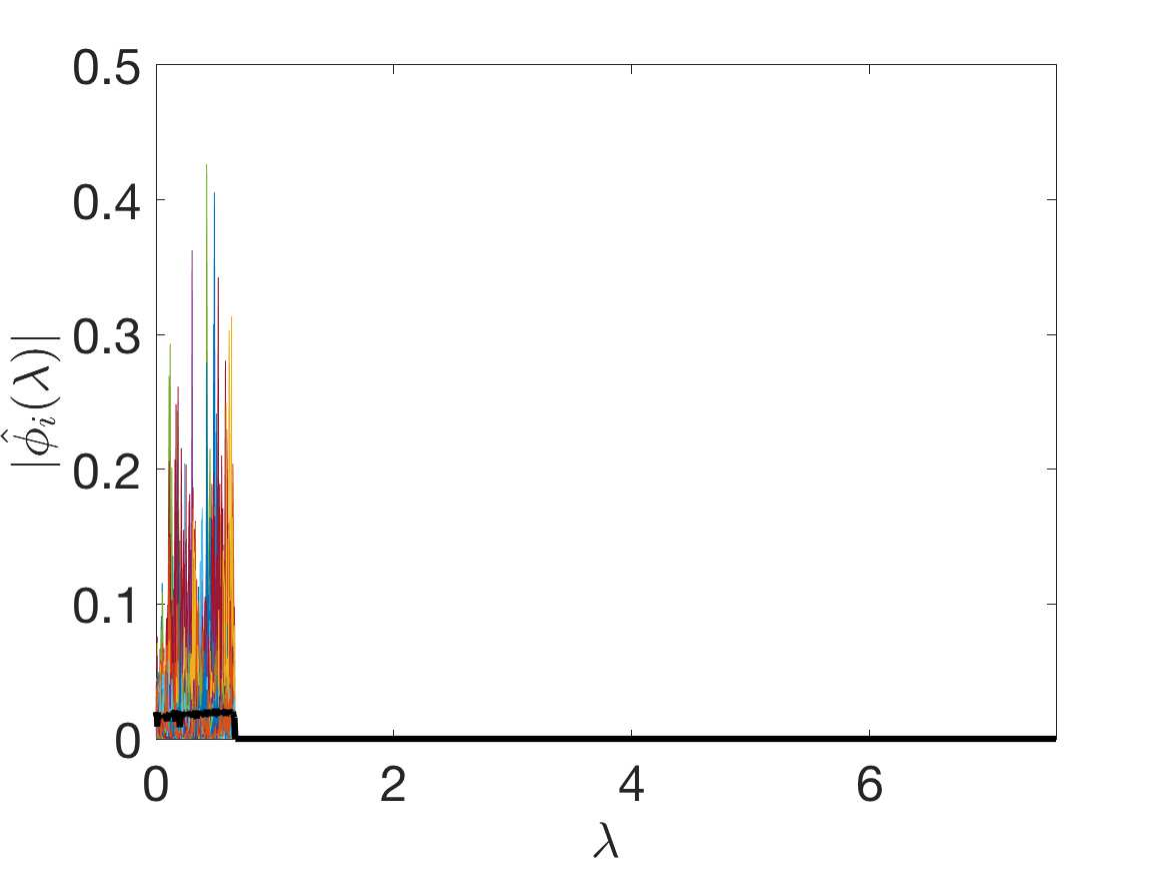}}
\end{minipage}
\begin{minipage}[m]{0.16\linewidth}
\centerline{\includegraphics[width=1\linewidth]{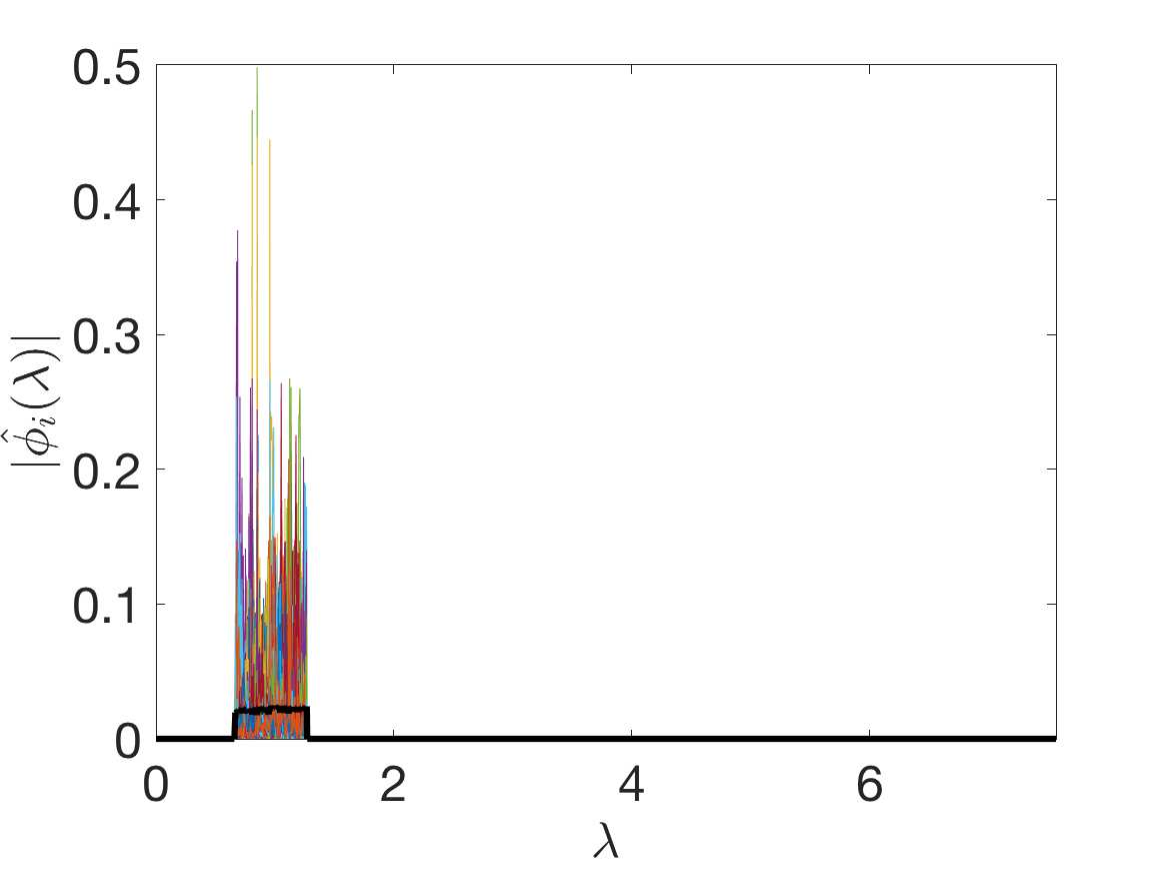}}
\end{minipage}
\begin{minipage}[m]{0.16\linewidth}
\centerline{\includegraphics[width=1\linewidth]{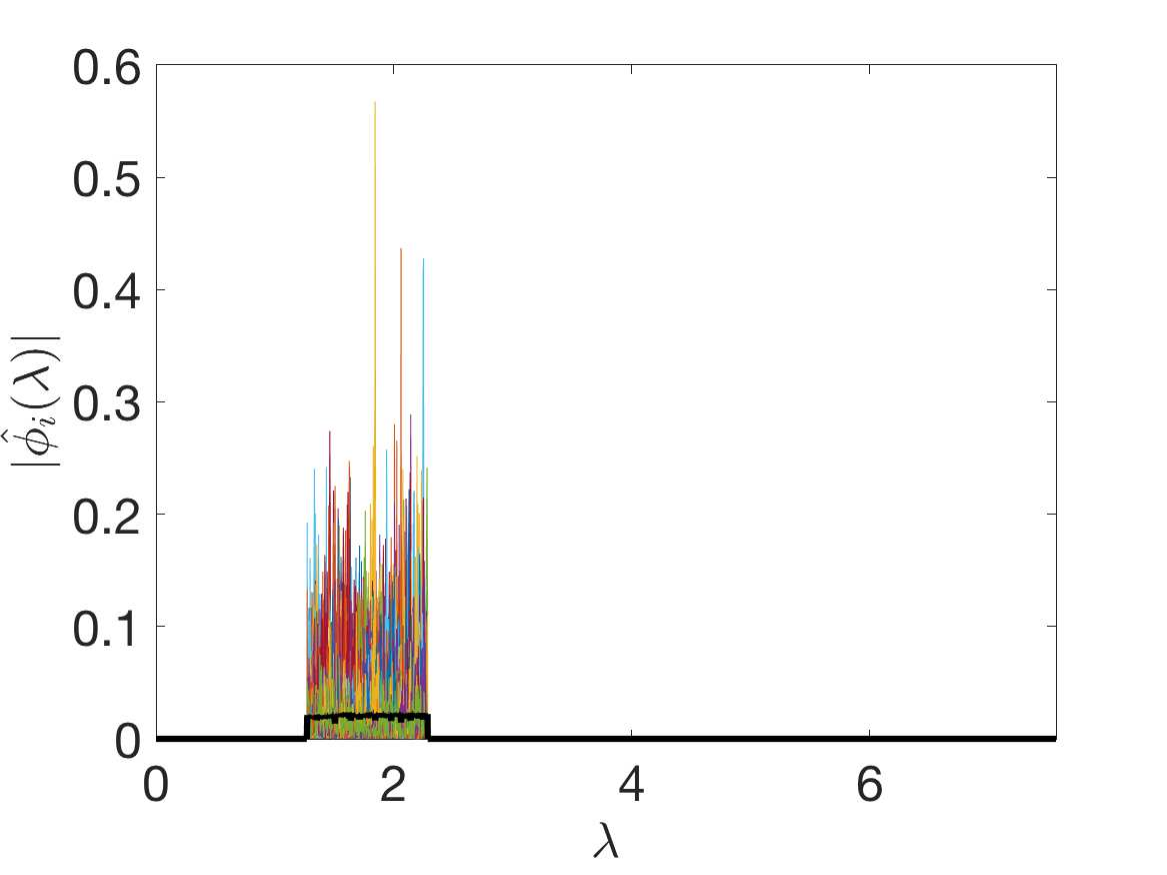}}
\end{minipage}
\begin{minipage}[m]{0.16\linewidth}
\centerline{\includegraphics[width=1\linewidth]{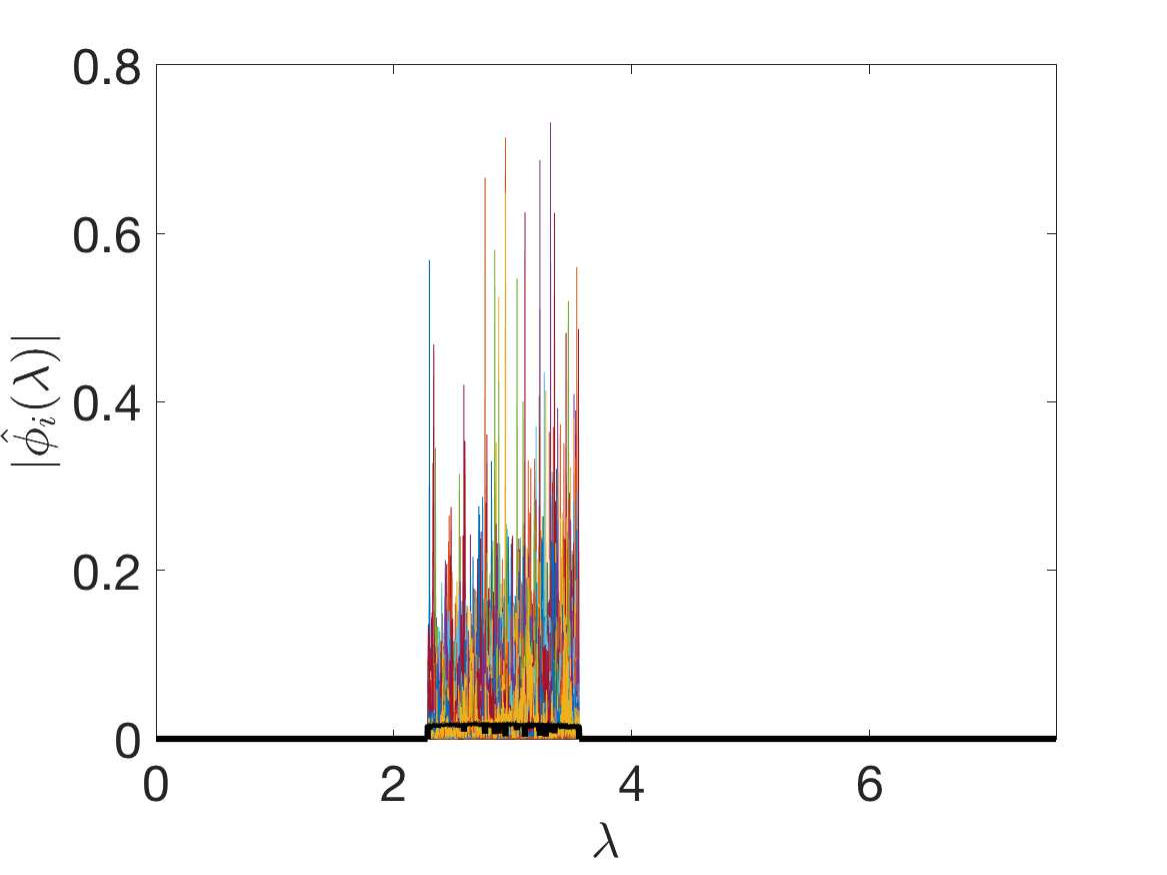}}
\end{minipage}
\begin{minipage}[m]{0.16\linewidth}
\centerline{\includegraphics[width=1\linewidth]{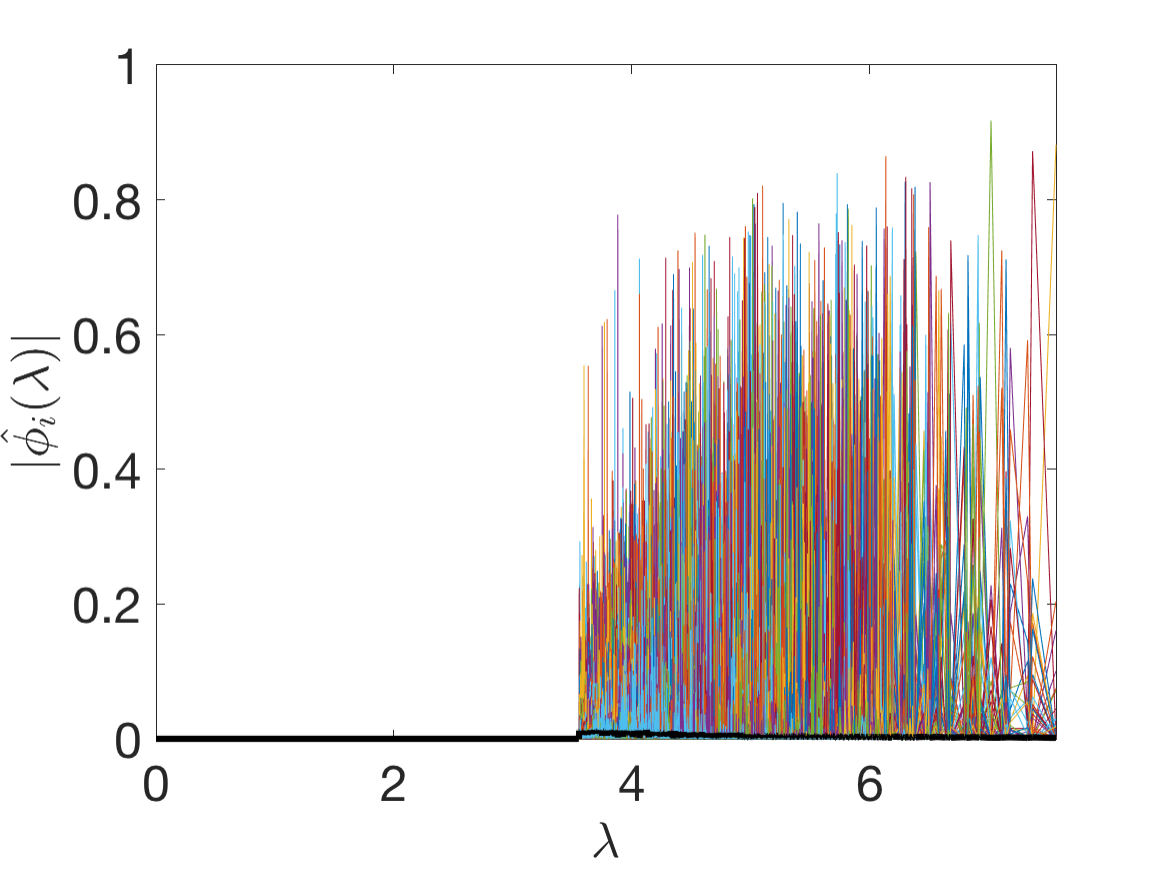}}
\end{minipage}\\
\begin{minipage}[m]{0.16\linewidth}
\centerline{\small{Analysis}}
\centerline{\small{Coefficients}}
\end{minipage}
\begin{minipage}[m]{0.16\linewidth}
\centerline{\includegraphics[width=1\linewidth]{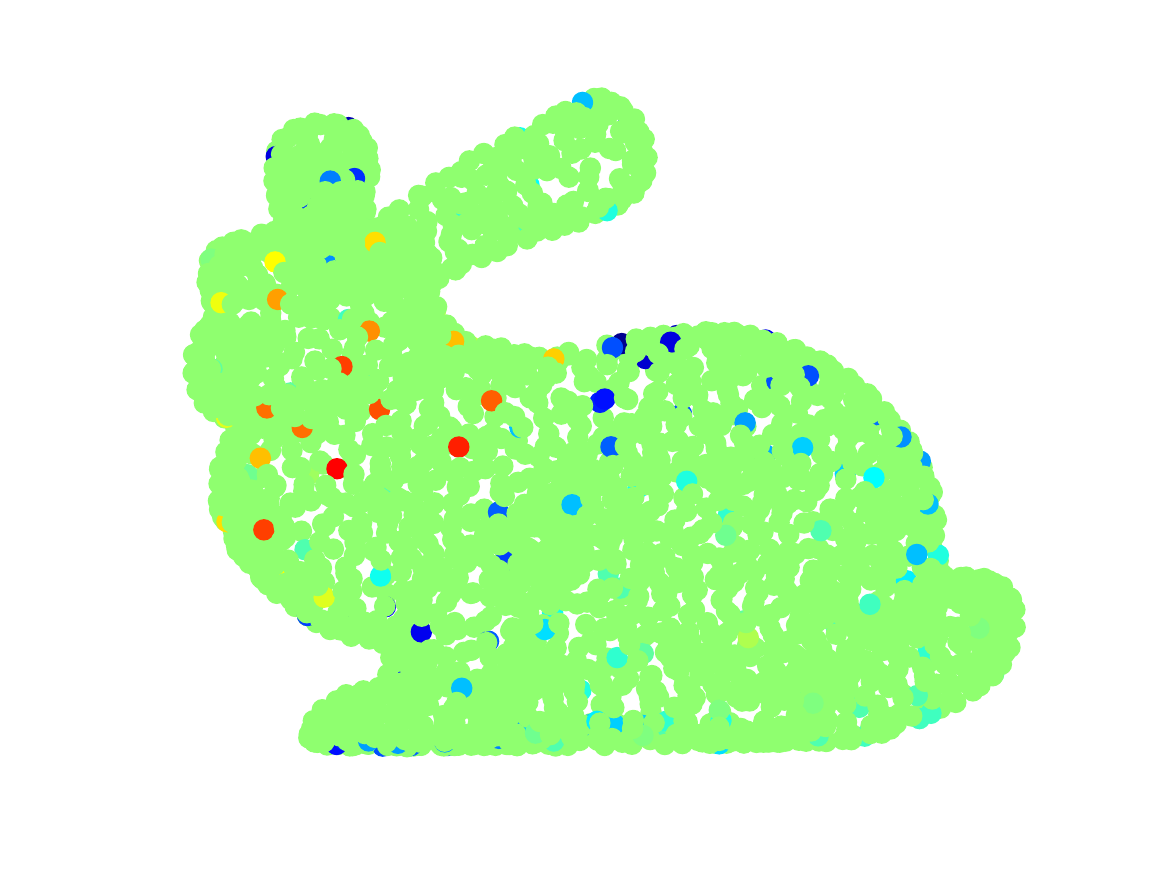}}
\end{minipage}
\begin{minipage}[m]{0.16\linewidth}
\centerline{\includegraphics[width=1\linewidth]{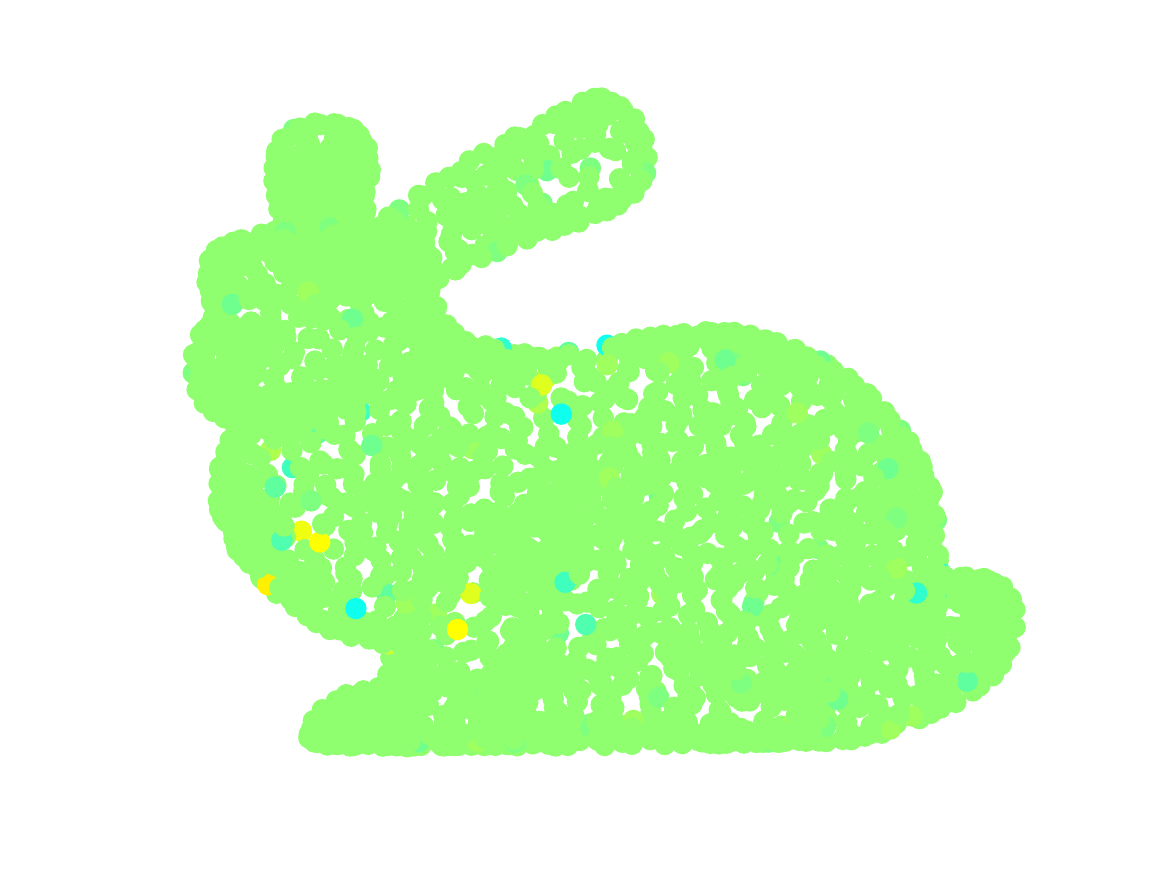}}
\end{minipage}
\begin{minipage}[m]{0.16\linewidth}
\centerline{\includegraphics[width=1\linewidth]{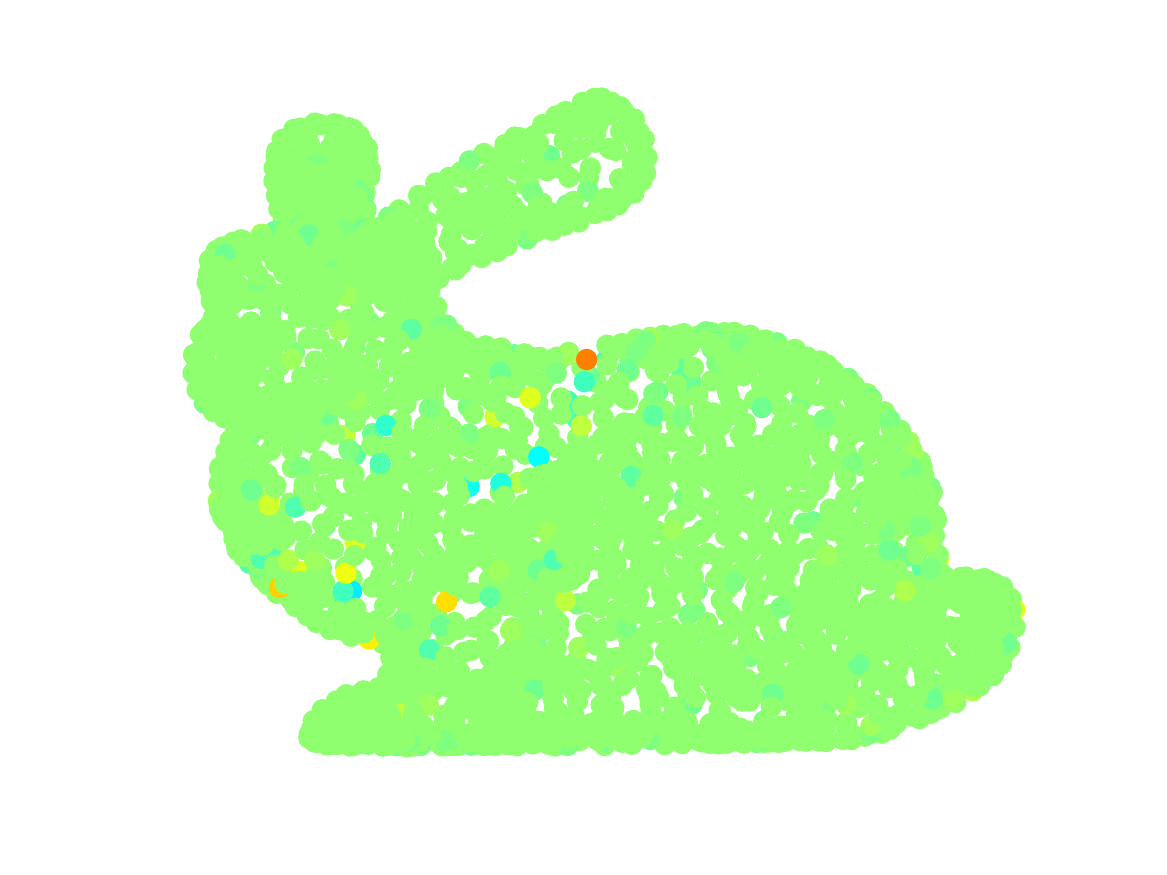}}
\end{minipage}
\begin{minipage}[m]{0.16\linewidth}
\centerline{\includegraphics[width=1\linewidth]{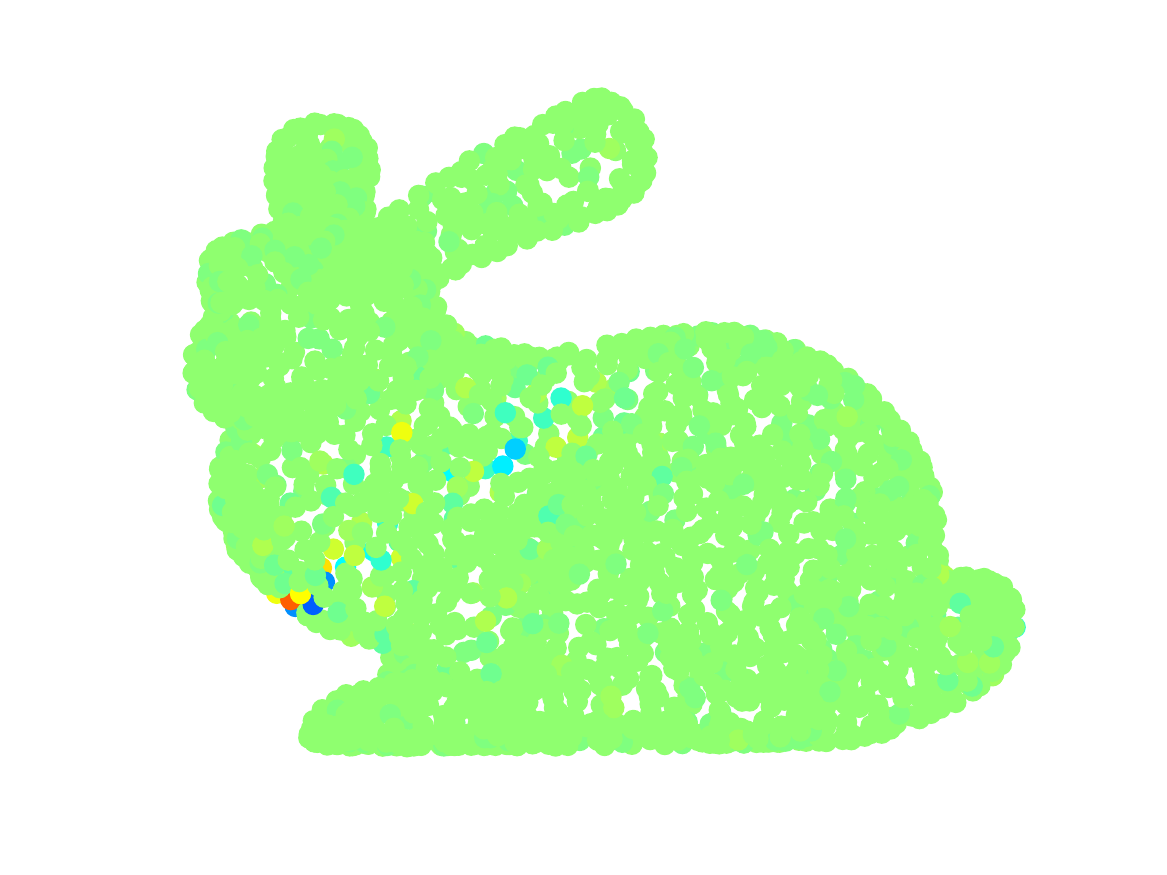}}
\end{minipage}
\begin{minipage}[m]{0.16\linewidth}
\centerline{\includegraphics[width=1\linewidth]{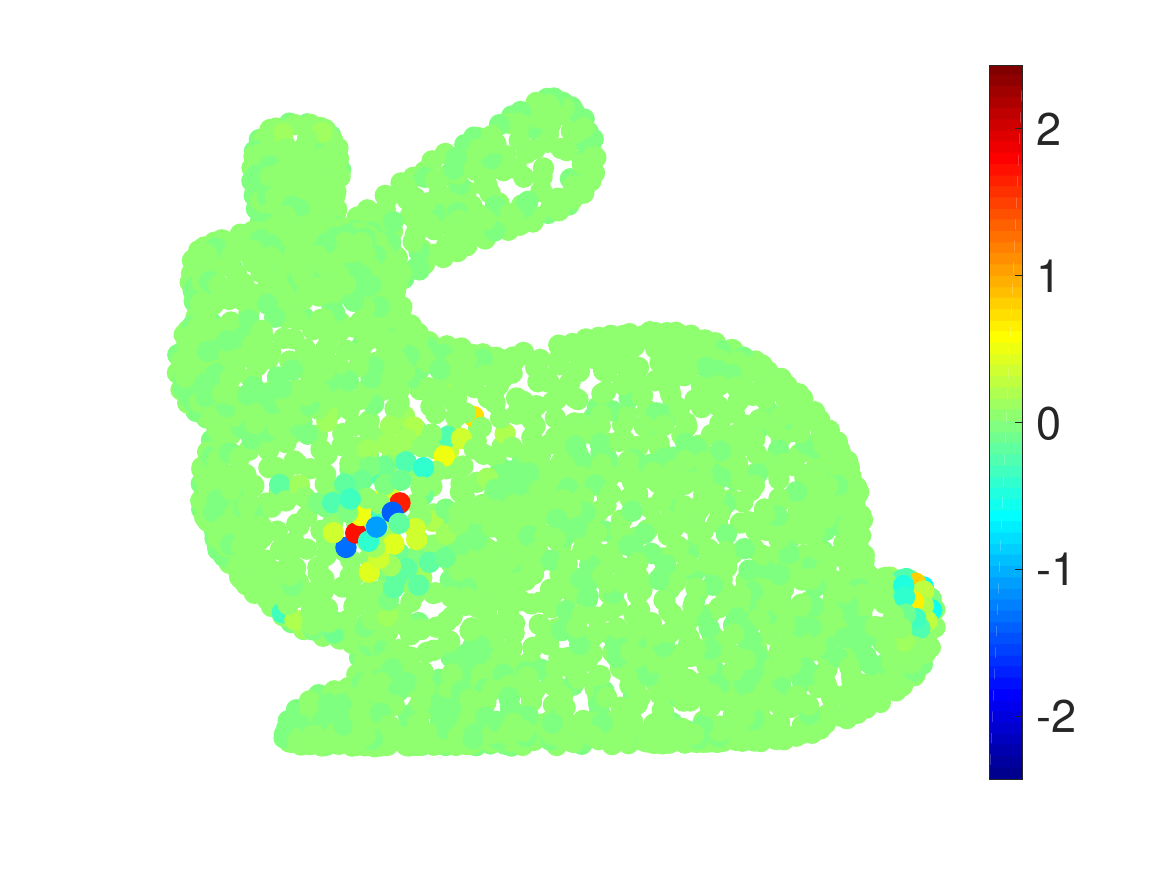}}
\end{minipage}\\
\begin{minipage}[m]{0.16\linewidth}
\centerline{\small{Reconstruction}}
\centerline{\small{by Band}}
\end{minipage}
\begin{minipage}[m]{0.16\linewidth}
\centerline{\includegraphics[width=1\linewidth]{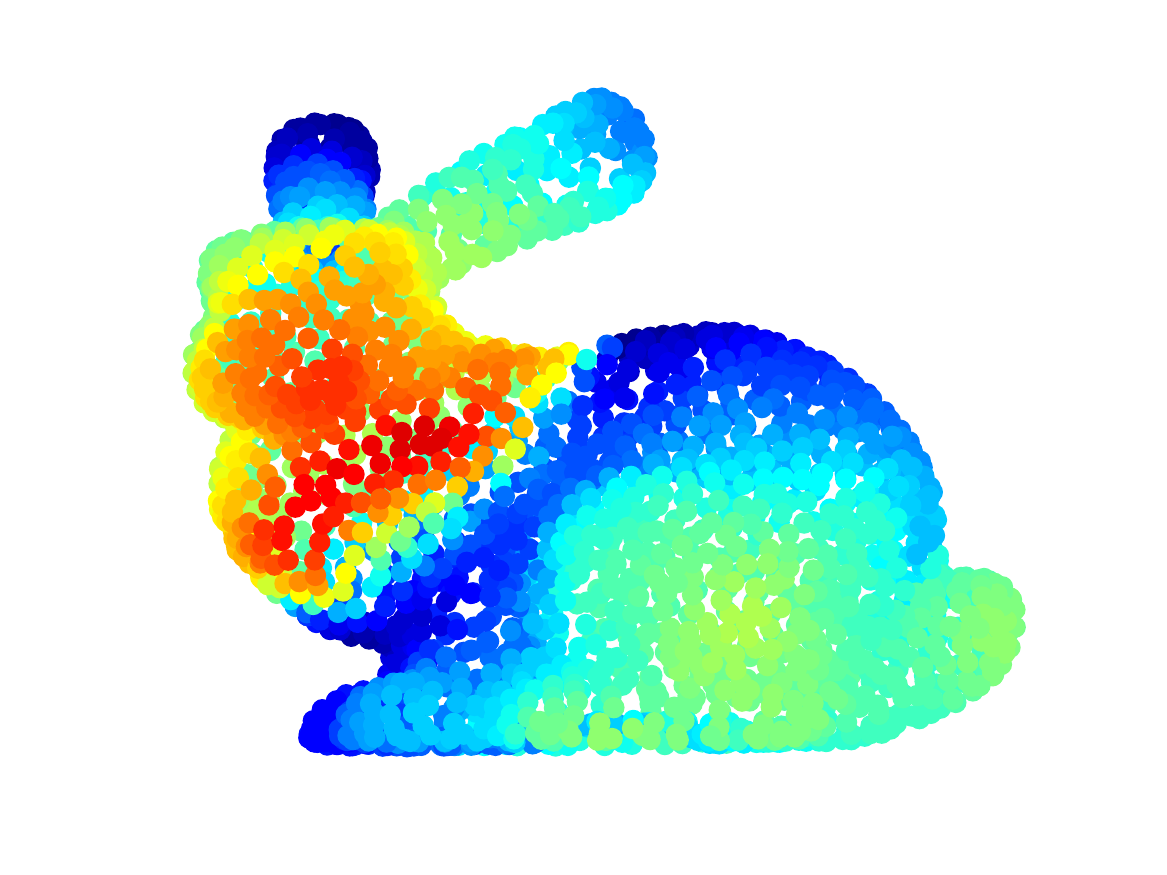}}
\end{minipage}
\begin{minipage}[m]{0.16\linewidth}
\centerline{\includegraphics[width=1\linewidth]{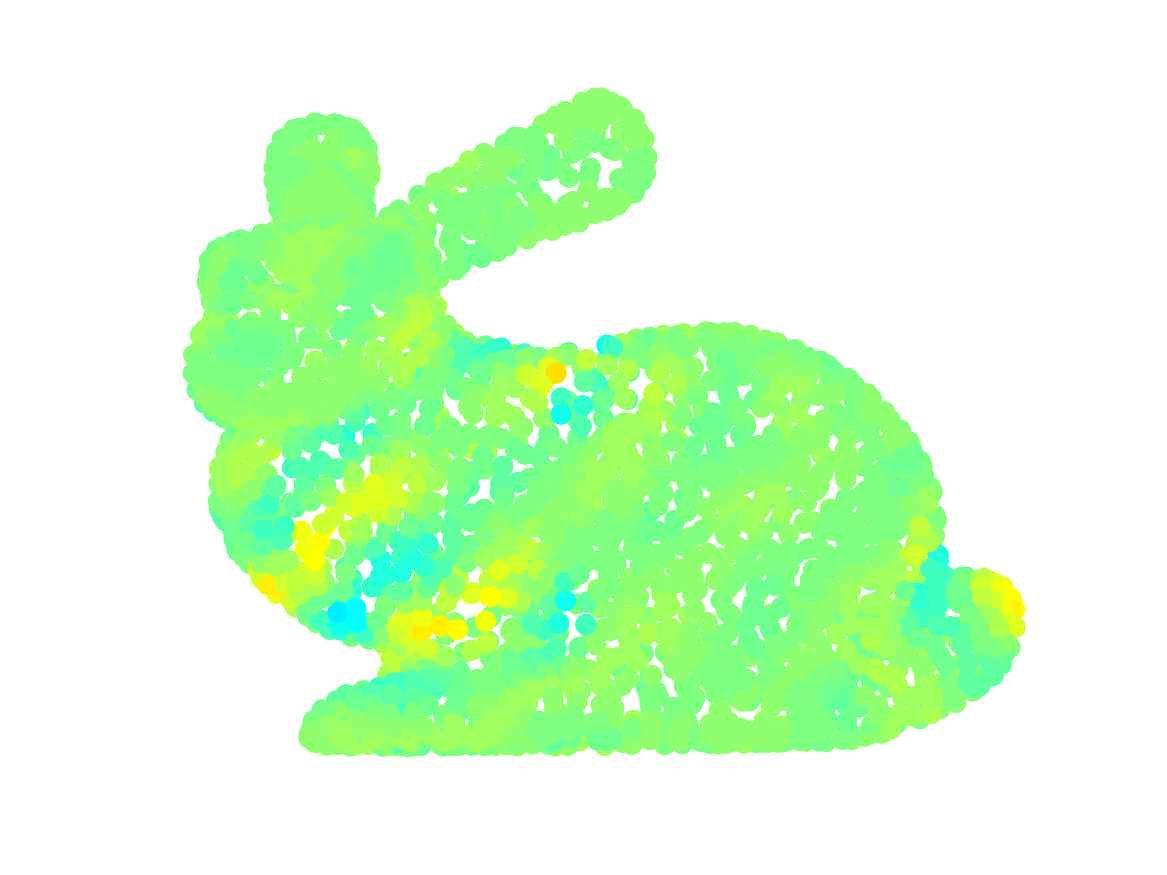}}
\end{minipage}
\begin{minipage}[m]{0.16\linewidth}
\centerline{\includegraphics[width=1\linewidth]{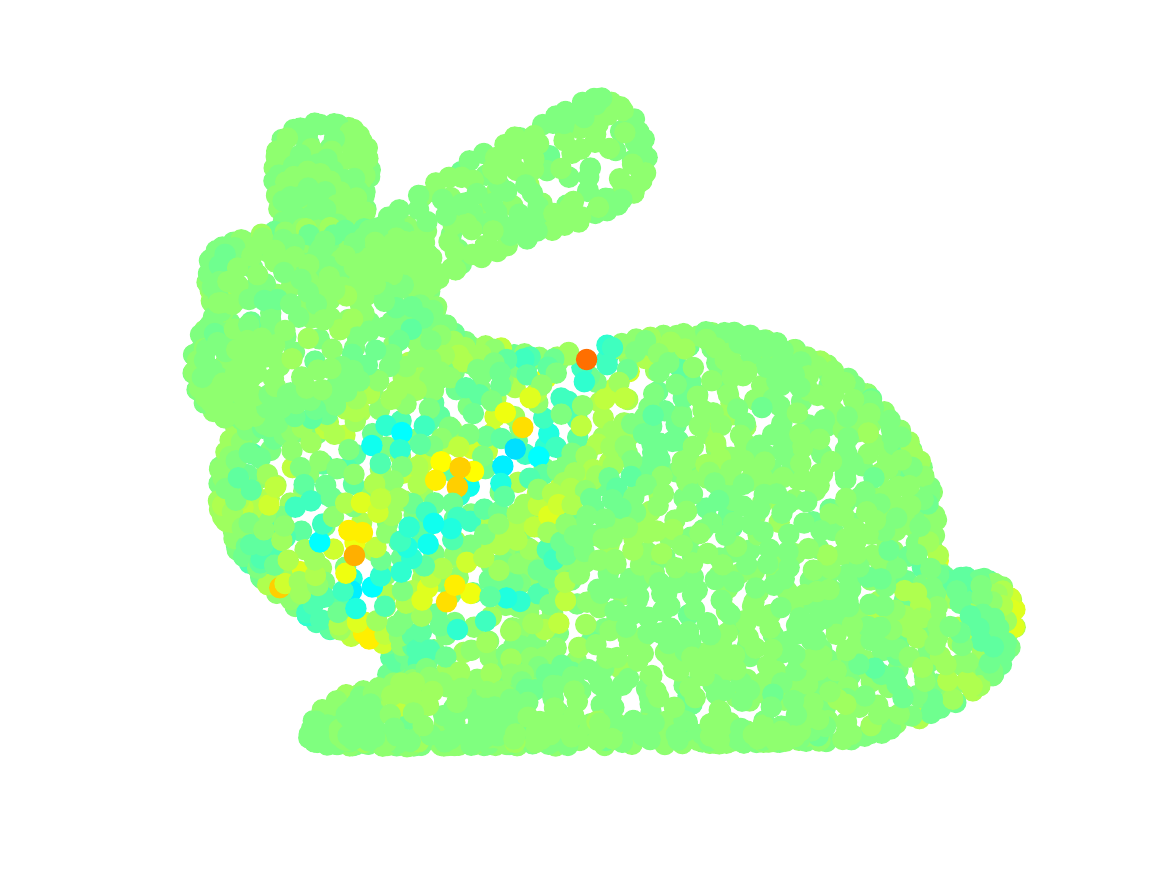}}
\end{minipage}
\begin{minipage}[m]{0.16\linewidth}
\centerline{\includegraphics[width=1\linewidth]{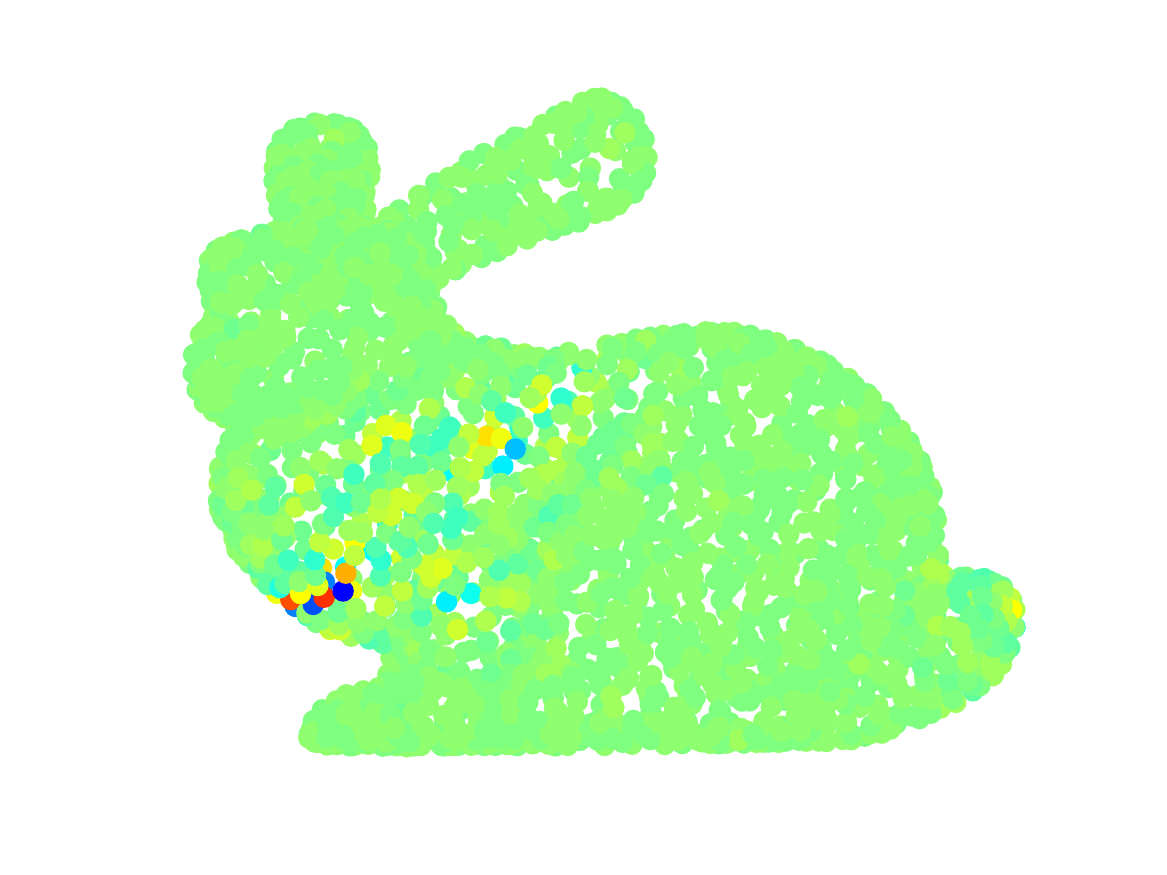}}
\end{minipage}
\begin{minipage}[m]{0.16\linewidth}
\centerline{\includegraphics[width=1\linewidth]{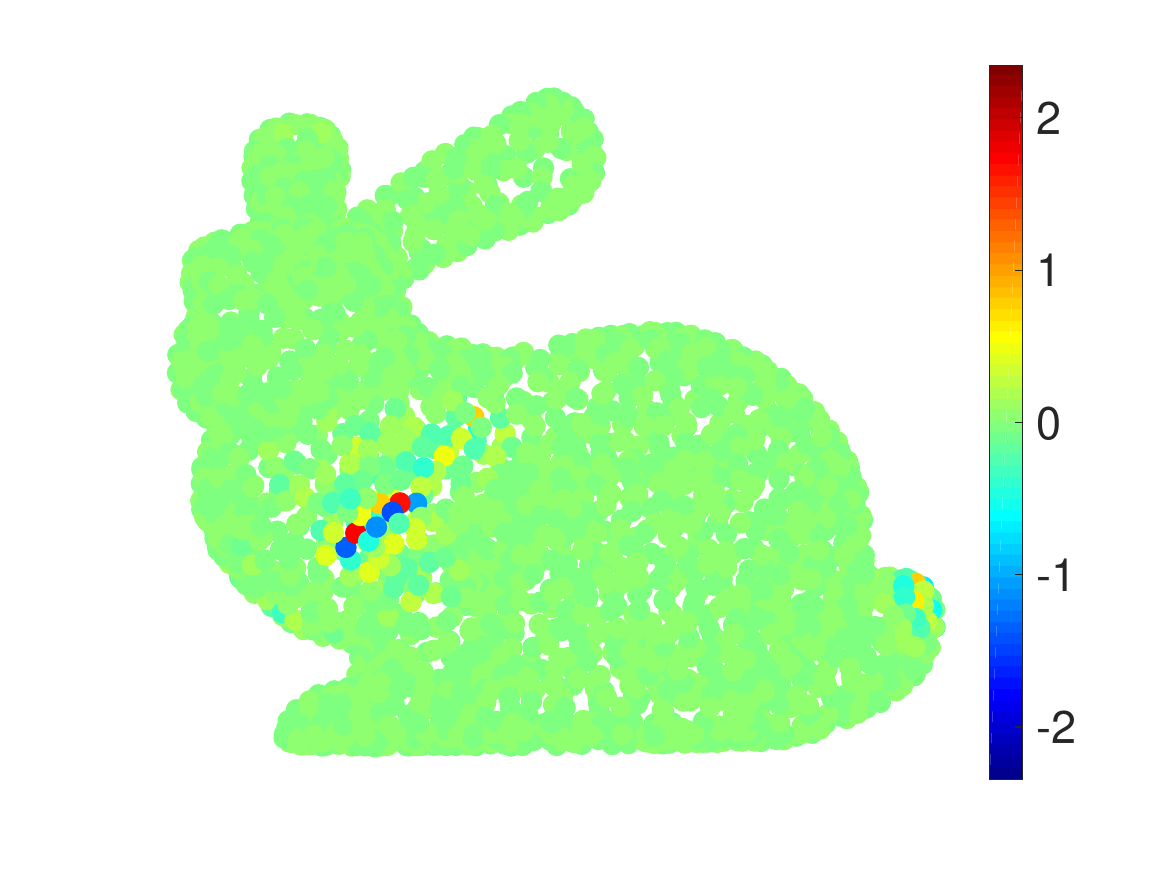}}
\end{minipage}
\caption{$M$-channel filter bank example. The first row shows example atoms in the vertex domain. The second row shows all atoms in the spectral domain, with an average of the atoms in each band shown by the thick black lines. The third row shows the analysis coefficients of Fig.\ \ref{Fig:bunny_signal}(d) by band, and the last row is the interpolation by band from those coefficients.}\label{Fig:bunny_coef}
\vspace{-.5cm}
\end{figure*}

\subsubsection{Signals that are sparsely represented by the M-CSFB transform}
Globally smooth signals trivially lead to sparse analysis coefficients because the coefficients are only nonzero for the first set(s) of vertices in the partition. More generally, signals that are concentrated in the graph Fourier domain 
have sparse M-CSFB analysis coefficients, because the coefficients for any channel whose filter does not overlap with the support of the signal in the graph Fourier domain are all equal to zero.

\subsection{Joint vertex-frequency localization of atoms and example analysis coefficients} \label{Se:ill1}
Next, we empirically examine the joint localization of the dictionary atoms 
in the vertex and graph spectral domains, which is key for their ability to compactly represent localized phenomena (e.g., discontinuities, edges). On the Stanford bunny graph \cite{bunny} with 2503 vertices, we partition the spectrum into five bands, and show the resulting partition into uniqueness sets 
in Fig.\ \ref{Fig:bunny_signal}(c). The first row of Fig.\ \ref{Fig:bunny_coef} shows five example atoms whose energies are concentrated on different spectral bands.  These atoms are also generally localized in the vertex domain,
with the wavelets becoming more localized at higher scales, as expected.
The second row of Fig.\ \ref{Fig:bunny_coef} shows the localization of the spectral content of all atoms in each band, with the averages  
represented by 
thick black lines. 

We then apply the proposed transform to a piecewise-smooth graph signal ${\bf f}$ that is shown in the vertex domain in Fig. \ref{Fig:bunny_signal}(a), and in the graph spectral domain in Fig. \ref{Fig:bunny_signal}(b). The full set of analysis coefficients is shown in Fig.\ \ref{Fig:bunny_signal}(d), and these are separated by band in the third row of Fig.\ \ref{Fig:bunny_coef}. We see that with the exception of the lowpass channel, the coefficients are clustered around the two main discontinuities (around the midsection and tail of the bunny). The bottom row of Fig.\ \ref{Fig:bunny_coef} shows the interpolation of these coefficients onto the corresponding spectral bands. If we sum these reconstructions together, we recover  
 the original signal in Fig.\ \ref{Fig:bunny_signal}(a).

\section{Fast M-CSFB Transform} \label{Se:fast_mcsfb}
In the numerical examples in the previous sections, we have computed a full eigendecomposition of the graph Laplacian and 
used it for all three of the filtering, sampling, and interpolation operations; however, such an eigendecomposition does not scale well with the size of the graph, as it requires ${\cal O}(N^3)$ operations with naive methods. 
In 
this section, we develop a fast approximate version of the proposed transform that scales more efficiently for large, sparse graphs.

\subsection{Approximation by polynomial filters} \label{Se:poly_approx}

Fast, approximate methods for computing $h_m(\L){\bf f}$, a function of sparse matrix times a vector, include approximating the function $h_m(\cdot)$ by a polynomial (e.g., via a truncated Chebyshev or Legendre expansion), approximating $h_m(\cdot)$ by a rational function, Krylov space methods (Lanczos in our case of a symmetric matrix $\L$), and the matrix version of the Cauchy integral theorem (see, e.g., \cite{higham}\nocite{davies2005computing, frommer}-\cite{dubious} and references therein). The first three of these methods have been examined in graph signal processing settings \cite{hammond2011wavelets,PuyTGV15},  \cite{shuman_distributed_sipn}\nocite{susnjara, shi2015infinite}-\cite{loukas2015distributed}. Here, to approximate the analysis side filters, 
we focus on order $K$ Chebyshev polynomial approximations of the form
\begin{align}\label{Eq:cheb}
\tilde{h}(\L){\bf f} := \sum_{k=0}^K \alpha_k \bar{T}_k(\L){\bf f}.
\end{align}
In \eqref{Eq:cheb}, $\bar{T}_k(\cdot)$ are Chebyshev polynomials shifted to the interval $[0,\lambda_{\max}]$. Thus,  $\bar{T}_0(\L){\bf f} = {\bf f}$, $\bar{T}_1(\L){\bf f} = \frac{2}{\lambda_{\max}}\L{\bf f}-{\bf f}$, and for $k\geq 2$, by the three term recurrence relation of Chebyshev polynomials, we have
\begin{align}\label{Eq:Tkbar_rec}
\hspace{-.05in}\bar{T}_k(\L){\bf f} = \frac{4}{\lambda_{\max}}\left(\L-\frac{\lambda_{\max}}{2}{\bf I}\right)\bar{T}_{k-1}(\L){\bf f}-\bar{T}_{k-2}(\L){\bf f}.
\end{align}

The coefficients in \eqref{Eq:cheb} are often taken to be $\alpha_0=\frac{1}{2}c_0$ and $\alpha_k=c_k$ for $k=1,2,\ldots,K$, where  $\{c_k\}_{k=0,\ldots,K}$ are the truncated Chebyshev expansion coefficients 
\begin{align}\label{Eq:cheb_coeff}
\hspace{-.1in}c_{k}:=\langle h , \bar{T}_k \rangle = \frac{2}{\pi}\int_{0}^{\pi}\cos(k\phi)\hspace{.02in}h\Bigl(\frac{\lambda_{\max}}{2} \bigl(\cos(\phi) +1\bigr)\Bigr)d\phi.\hspace{-.08in}
\end{align}
However, the oscillations that arise in Chebyshev polynomial approximations of bandpass filters may result in larger values of $\tilde{h}_m(\lambda)\tilde{h}_{m^{\prime}}(\lambda)$, even when the ideal filters $h_m(\cdot)$ and $h_{m^{\prime}}(\cdot)$ have supports that do not come close to overlapping. This negates the orthogonality of the atoms across bands shown in \eqref{Eq:ortho_bands}. In an attempt to at least preserve \emph{near} orthogonality across bands, we 
therefore use the Jackson-Chebyshev polynomial approximations from \cite{di2016efficient,puy_structured_sampling} that damp the Gibbs oscillations appearing in Chebyshev expansions. With the damping, 
\begin{align}\label{Eq:damping1}
\alpha_0=\frac{1}{2}c_0 \hbox{ and }\alpha_k=\gamma_{k,K}c_k \hbox{ for }k=1,2,\ldots,K, 
\end{align}
where, as presented in \cite{di2016efficient}, 
\begin{align}
&\gamma_{k,K} =\label{Eq:damping2}\\
& \frac{(1-\frac{k}{K+2})\sin(\frac{\pi}{K+2})\cos(\frac{k\pi}{K+2})+\frac{1}{K+2}\cos(\frac{\pi}{K+2})\sin(\frac{k\pi}{K+2})}{\sin(\frac{\pi}{K+2})}.  \nonumber
\end{align}
Fig. \ref{Fig:approx_filtering_error} shows the Jackson-Chebyshev polynomial approximations for ideal bandpass filters of two different graphs.\footnote{For the net25 graph, we have removed the self loops and added a single edge to connect the two connected components.}

\begin{figure}[tb]
\begin{minipage}[m]{0.49\linewidth}
\centerline{\includegraphics[width=.9\linewidth]{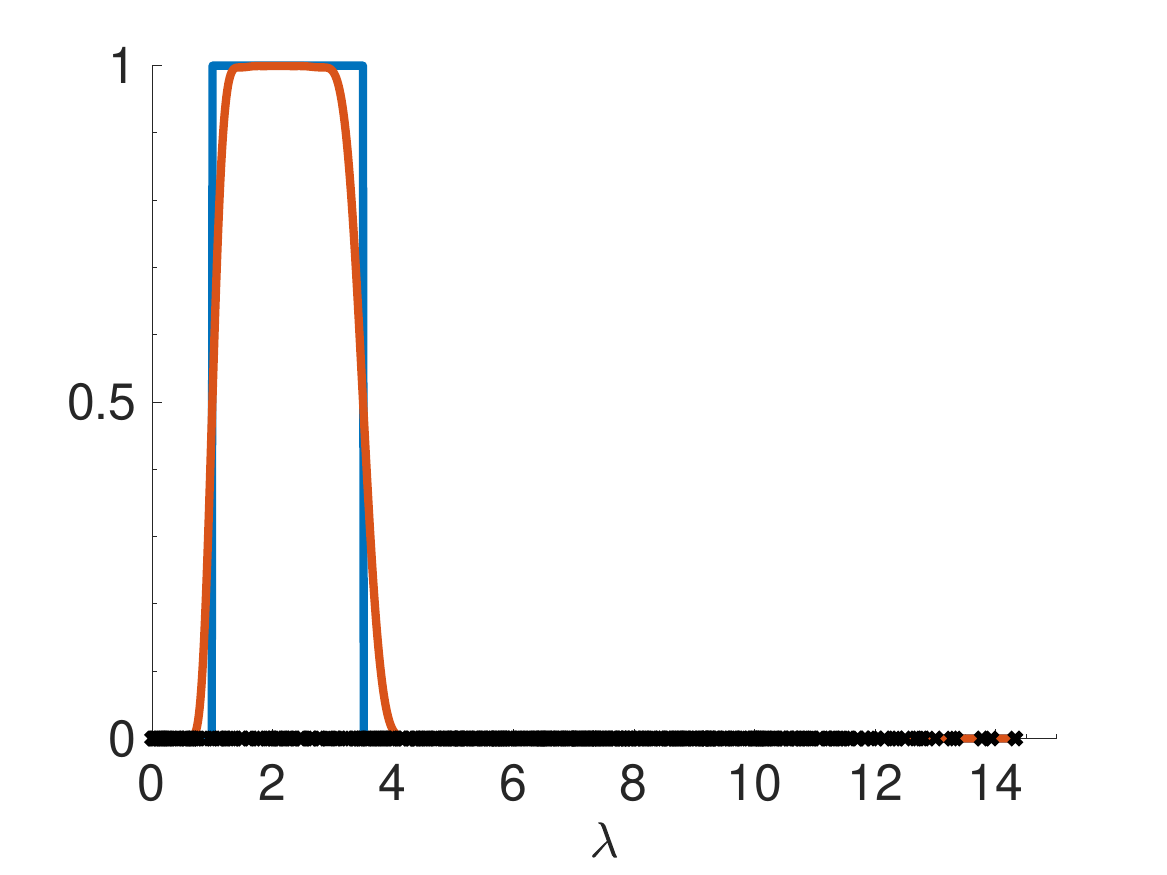}}
\centerline{~~\small{(a)}}
\end{minipage}
\begin{minipage}[m]{0.49\linewidth}
\centerline{\includegraphics[width=.9\linewidth]{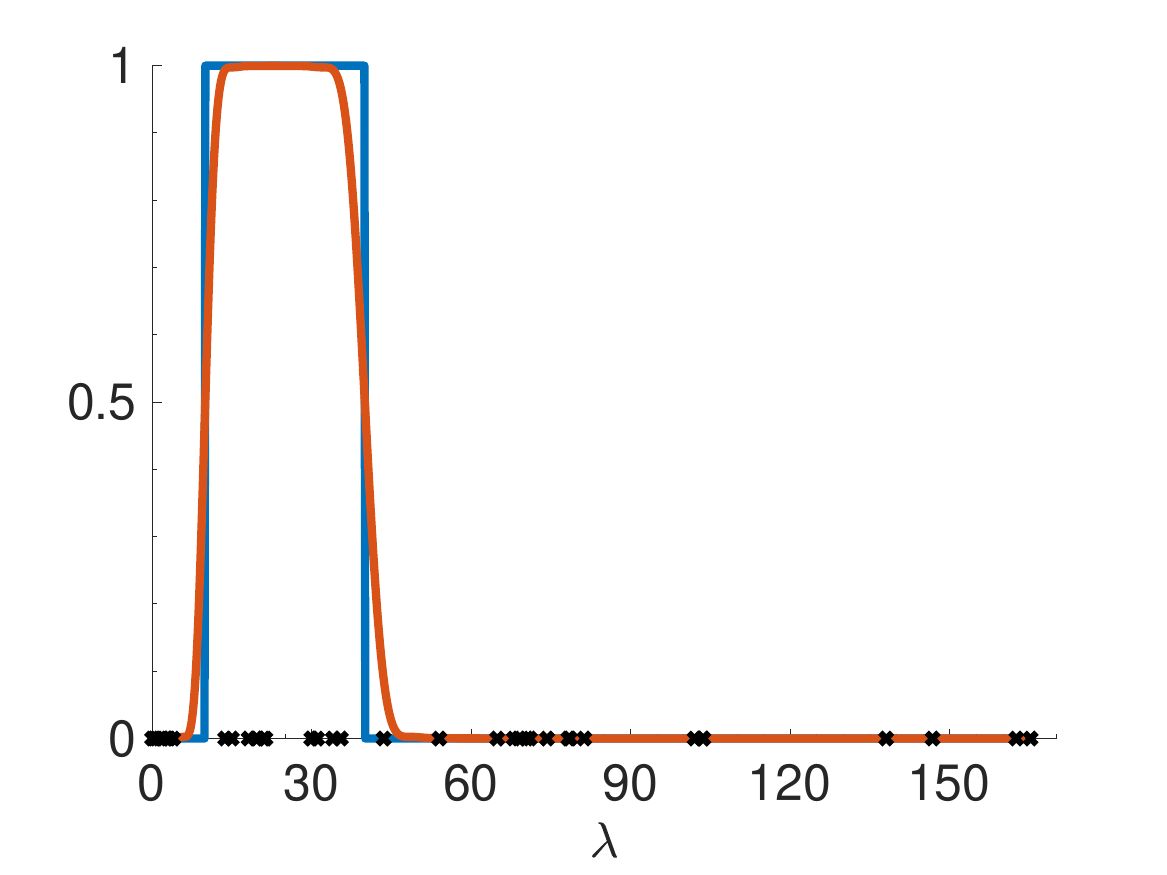}}
\centerline{~~\small{(b)}}
\end{minipage} \\
\begin{minipage}[m]{0.49\linewidth}
\centerline{\includegraphics[width=.9\linewidth]{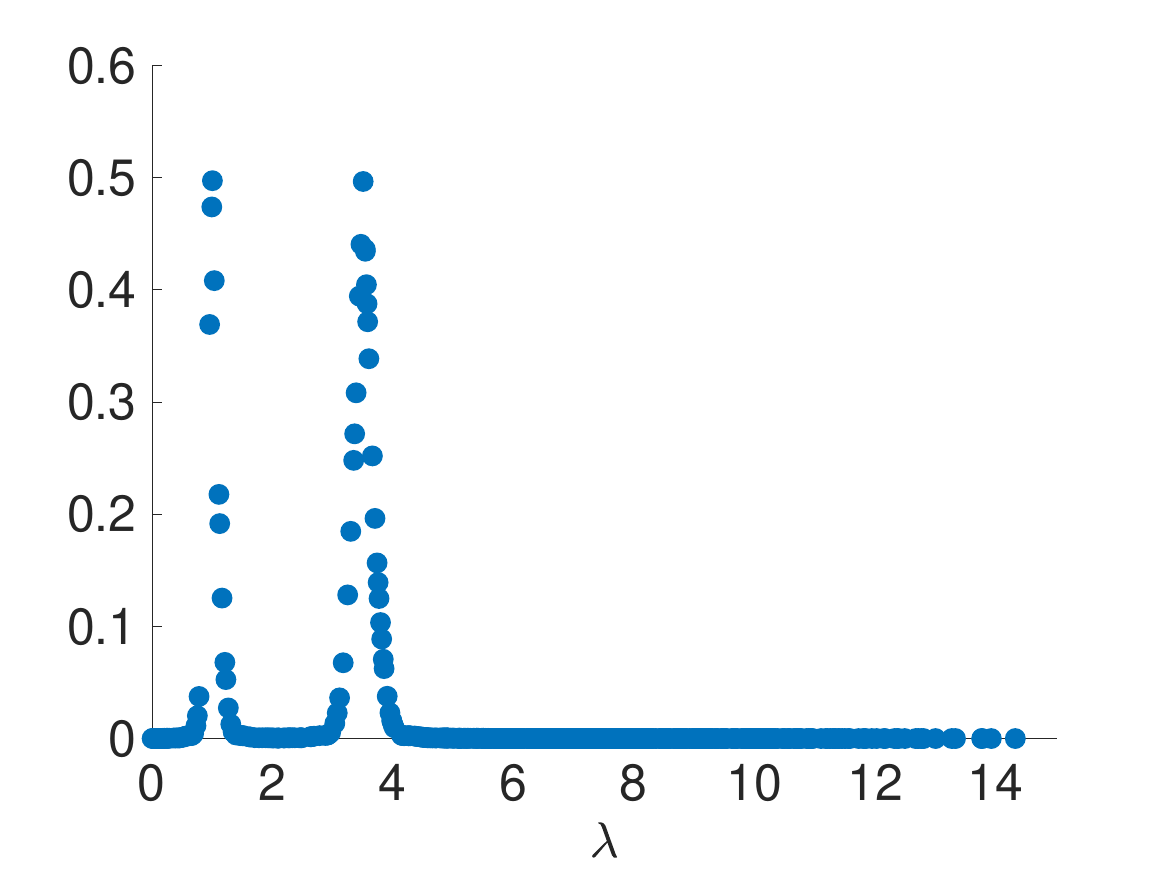}}
\centerline{~~\small{(c)}}
\end{minipage}
\begin{minipage}[m]{0.49\linewidth}
\centerline{\includegraphics[width=.9\linewidth]{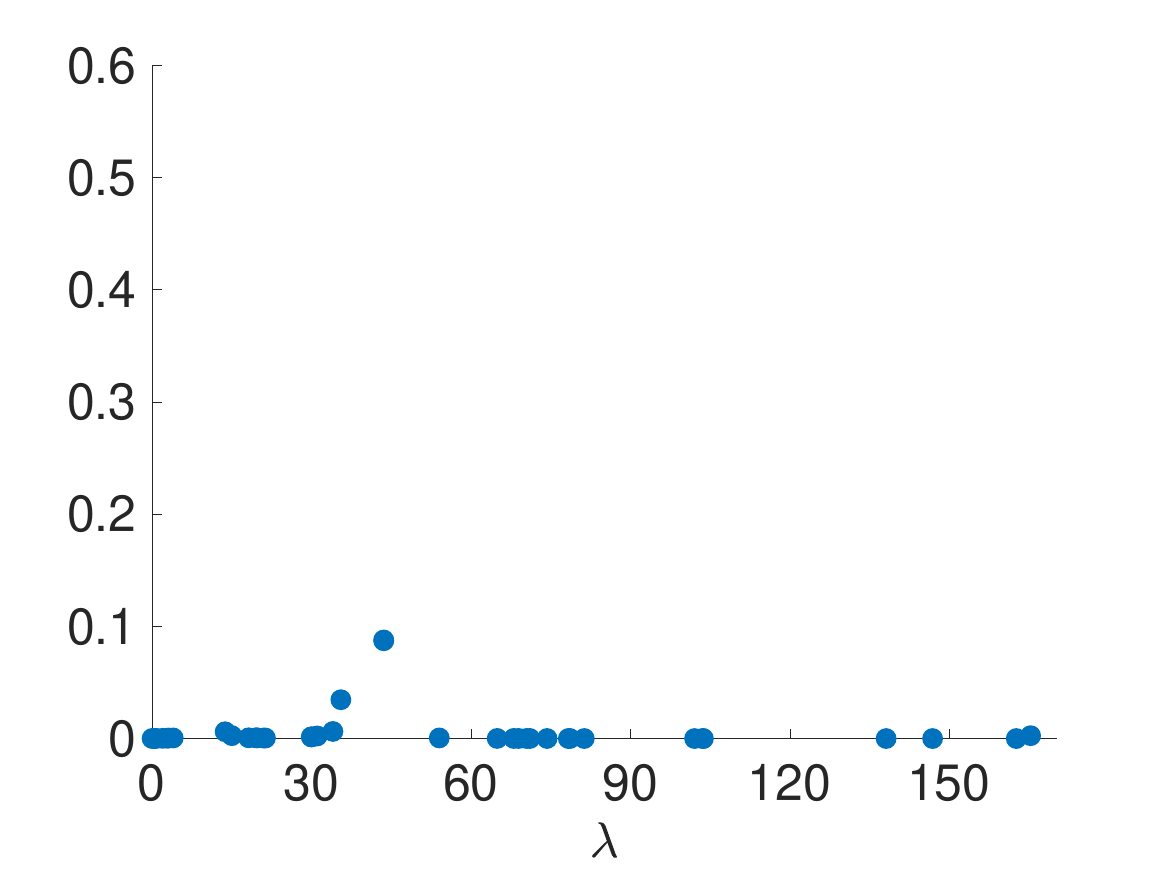}}
\centerline{~~\small{(d)}}
\end{minipage}
\caption{Degree 80 Jackson-Chebyshev polynomial approximations for ideal bandpass filters on the (a) 500 vertex random sensor network of Fig. \ref{Fig:part_examples} and (b) Andrianov net25 graph from \cite{davis2011university} with 9,520 vertices. In (c) and (d), we show the errors $|\tilde{h}(\lambda_{\l})-h(\lambda_{\l})|$ at each of the Laplacian eigenvalues of the corresponding graphs in (a) and (b). \vspace{-.2in}}
\label{Fig:approx_filtering_error}
\vspace{.05cm}
\end{figure}

\subsection{Filter bank design}

We can quantify the worst case error introduced when approximating  $h_m(\cdot)$ by an approximant $\tilde{h}_m(\cdot)$ as follows:
\begin{align} \label{Eq:poly_approx_error}
||\tilde{h}_m(\L)-h_m(\L) ||_2   &= \max_{\l \in \{0,1,\ldots,N-1\}}|\tilde{h}_m(\lambda_{\l})-h_m(\lambda_{\l}) |  \nonumber \\
 &\leq \max_{\lambda \in [0,\lambda_{\max}]}|\tilde{h}_m(\lambda)-h_m(\lambda) |.
\end{align}
While approximation theory often aims to minimize the upper bound in \eqref{Eq:poly_approx_error}, only the errors exactly at the graph Laplacian eigenvalues affect the overall approximation error $||\tilde{h}_m(\L)-h_m(\L) ||_2$. Since, as seen in Fig. \ref{Fig:approx_filtering_error}, the errors of the Jackson-Chebyshev polynomial approximation are concentrated around the discontinuities of $h_m(\cdot)$, a guiding principle when designing the filter bank to be more amenable to fast approximation is to choose the endpoints $\{\tau_m\}_{m=1,\ldots,M-1}$ of the bandpass filters to be in gaps in the graph Laplacian spectrum. Unfortunately, we do not have access to the exact graph Laplacian eigenvalues (the reason for introducing this approximation in the first place is that they are too expensive to compute for large graphs); however, we can efficiently estimate the density of the spectrum in order to design the filters 
have the endpoints close to fewer eigenvalues of $\L$.

\subsubsection{Estimating the spectral density} \label{Se:spectral_density}

Lin et al. \cite{lin_spectral_density} provide an excellent overview of methods to approximate the \emph{spectral density function} \cite[Chapter 6]{van_mieghem}) (also called the \emph{Density of States} or \emph{empirical spectral distribution} \cite[Chapter 2.4]{tao_random_matrix}) of a matrix, which in our context for the graph Laplacian $\L$ is the probability measure 
$p_{\lambda}(s):=\frac{1}{N}\sum_{\l=0}^{N-1} \Identity_{\left\{\lambda_{\l}=s\right\}}.$
Here, we use a variant of the Kernel Polynomial Method \cite{silver1994densities}\nocite{silver1996kernel}-\cite{wang1994calculating} described in \cite{lin_spectral_density} to estimate the \emph{cumulative spectral density function} or \emph{empirical spectral cumulative distribution}
\begin{align}
P_{\lambda}(z):=\frac{1}{N}\sum_{\l=0}^{N-1} \Identity_{\left\{\lambda_{\l}\leq z\right\}}.
\end{align}
The procedure 
starts by estimating $\lambda_{\max}$, for example via the power iteration. Then for each of $T$ linearly spaced points $\xi_i$ between 0 and $\lambda_{\max}$, we use Hutchinson's stochastic trace estimator \cite{hutchinson} to estimate $\eta_i$, the number of eigenvalues less than or equal to $\xi_i$. Defining the Heaviside function $\Theta_{\xi_i}(\lambda):=\Identity_{\left\{\lambda \leq \xi_i\right\}},$ we have
\begin{align}
\eta_i =\mbox{tr}\Bigl(\Theta_{\xi_i}(\L)\Bigr) 
&=\mathbb{E}[{\bf x}^{\top}\Theta_{\xi_i}(\L){\bf x}] \label{Eq:hutch1}\\
&\approx \frac{1}{J} \sum_{j=1}^J {{\bf x}^{(j)}}^{\top}\Theta_{\xi_i}(\L){\bf x}^{(j)} \label{Eq:hutch2} \\
&\approx \frac{1}{J} \sum_{j=1}^J {{\bf x}^{(j)}}^{\top}\tilde{\Theta}_{\xi_i}(\L){\bf x}^{(j)}. \label{Eq:hutch3}
\end{align}
In \eqref{Eq:hutch1}, ${\bf x}$ is a random vector with each component having an independent and identical standard normal distribution. Each vector  ${\bf x}^{(j)}$ in \eqref{Eq:hutch2} is chosen according to this same distribution, and in our experiments, we take the default number of vectors to be $J=30$. In \eqref{Eq:hutch3}, $\tilde{\Theta}_{\xi_i}$ is the Jackson-Chebyshev approximation to ${\Theta}_{\xi_i}$ discussed in Section \ref{Se:poly_approx}.
If we place the $J$ random vectors into the columns of an $N \times J$ matrix ${\bf X}$, the computational cost of estimating the spectral distribution is dominated by computing 
\begin{align} \label{Eq:hutch4}
\tilde{\Theta}_{\xi_i}(\L){\bf X}=\sum_{k=0}^K \alpha_k \bar{T}_k(\L){\bf X}
\end{align}
 for each $\xi_i$. Yet, 
we only need to compute $\{ \bar{T}_k(\L){\bf X}\}_{k=0,1,\ldots,K}$ recursively once, as this sequence can be reused for each $\xi_i$, with different choices of the $\alpha_k$'s. Therefore, the overall computational cost is ${\cal O}(KJ|\E|)$.

As in \cite{shuman2013spectrum}, once we compute the eigenvalue count estimates $\{\eta_i\}$, we approximate the empirical spectral cumulative distribution $P_{\lambda}(\cdot)$ by performing monotonic piecewise cubic interpolation \cite{fritsch} on the series of points $\left\{\left(\xi_i,\frac{\eta_i}{N}\right)\right\}_{i=1,2,\ldots,T}$. We denote the result as $\tilde{P}_{\lambda}(\cdot)$. Algorithm \ref{Al:spectral_density} summarizes these computations.

\setlength{\textfloatsep}{12pt}
\begin{algorithm}[tb]
\caption{Spectral density approximation}
\begin{algorithmic}
\State \textbf{Input} graph $\G$, estimate for $\lambda_{\max}$, degree $K$, number of random vectors $J$
\State Generate an $N \times J$ matrix $\bf{X}$  whose columns are i.i.d. standard normal random vectors
\State Compute $\{\bar{T}_k(\L){\bf X}\}_{k=0,1,\ldots,K}$  via \eqref{Eq:Tkbar_rec}
\State Choose $T$ linearly spaced points $\{\xi_i\}_{i=1,\ldots,T}$ between 0 and $\lambda_{\max}$
\For {$i=1,2,\ldots,T$}
\State Approximate $\eta_{i}$ via \eqref{Eq:hutch3} and \eqref{Eq:hutch4} 
\EndFor
\State Estimate the spectral density function $\tilde{P}_{\lambda}$ by performing monotonic cubic interpolation on the set of points $\{\xi_i, \frac{\eta_i}{N} \}$
\State \textbf{Output} $\bf{X}$, $\{\bar{T}_k(\L){\bf X}\}_{k=0,1,\ldots,K}$, and $\tilde{P}_{\lambda}$
\end{algorithmic}
\label{Al:spectral_density}
\end{algorithm}

\subsubsection{Choosing initial band ends}

When selecting the band ends $\{\tau_m\}$ for each of the $M$ ideal filters, we consider two factors: spectrum-adaptation and spacing. In our implementation, the filter bank can either be adapted to the spectral distribution or just to the support of the spectrum $[0,\lambda_{\max}]$, and it can be either evenly or logarithmically spaced (four options in all). For example, if the filter bank is only adapted to the support of the spectrum and is evenly spaced, then $\tau_m=\frac{m}{M}\lambda_{\max}$. Fig. \ref{Fig:fb_design}(a)-(b) show a spectrum-adapted, logarithmically spaced choice with $\tau_m=\tilde{P}^{-1}_{\lambda}\Bigl(\frac{1}{2}^{M-m}\Bigr)$ for $m=1,2,\ldots,M$, such that approximately half of the eigenvalues are in the highest band, a quarter in the next highest band, and so forth. 

\begin{figure}[tb]
\begin{minipage}[m]{0.49\linewidth}
\centerline{\includegraphics[width=1.1\linewidth]{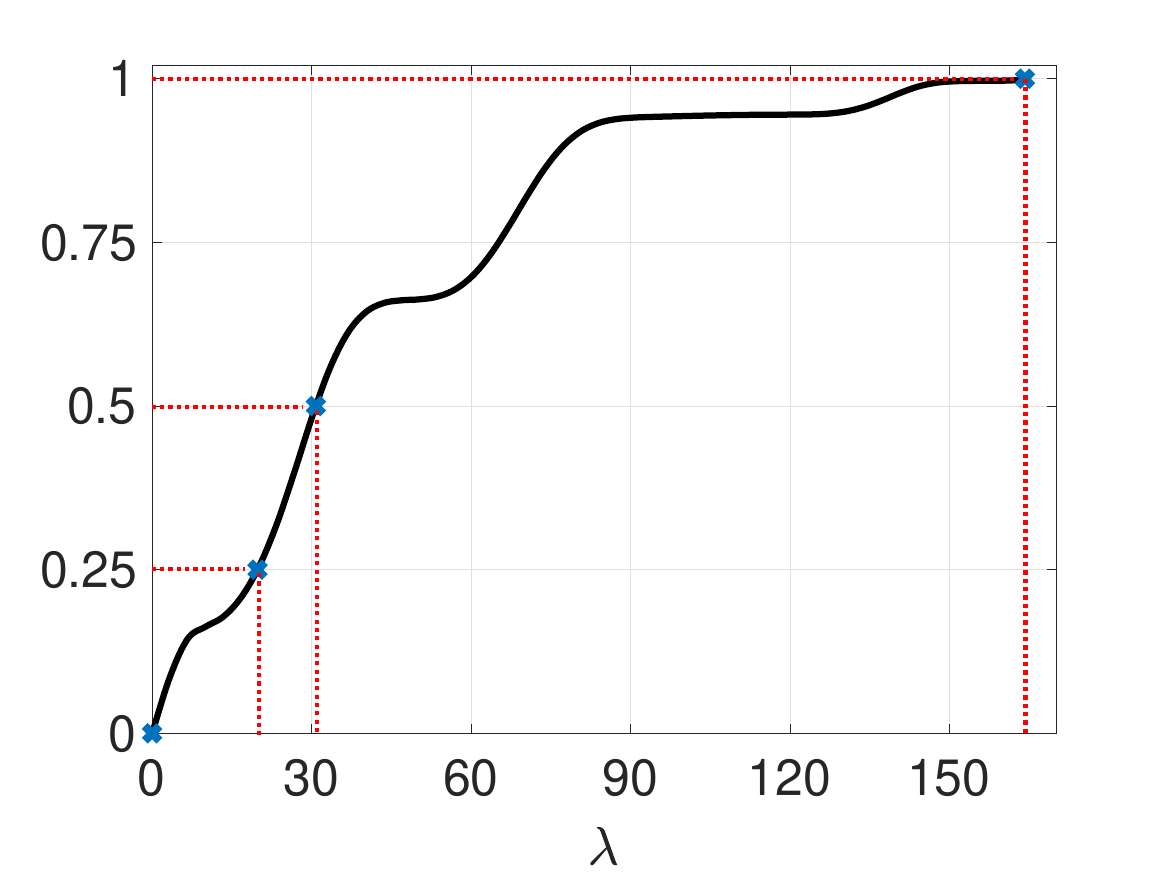}}
\centerline{~~\small{(a)}}
\end{minipage}
\begin{minipage}[m]{0.49\linewidth}
\centerline{\includegraphics[width=1.1\linewidth]{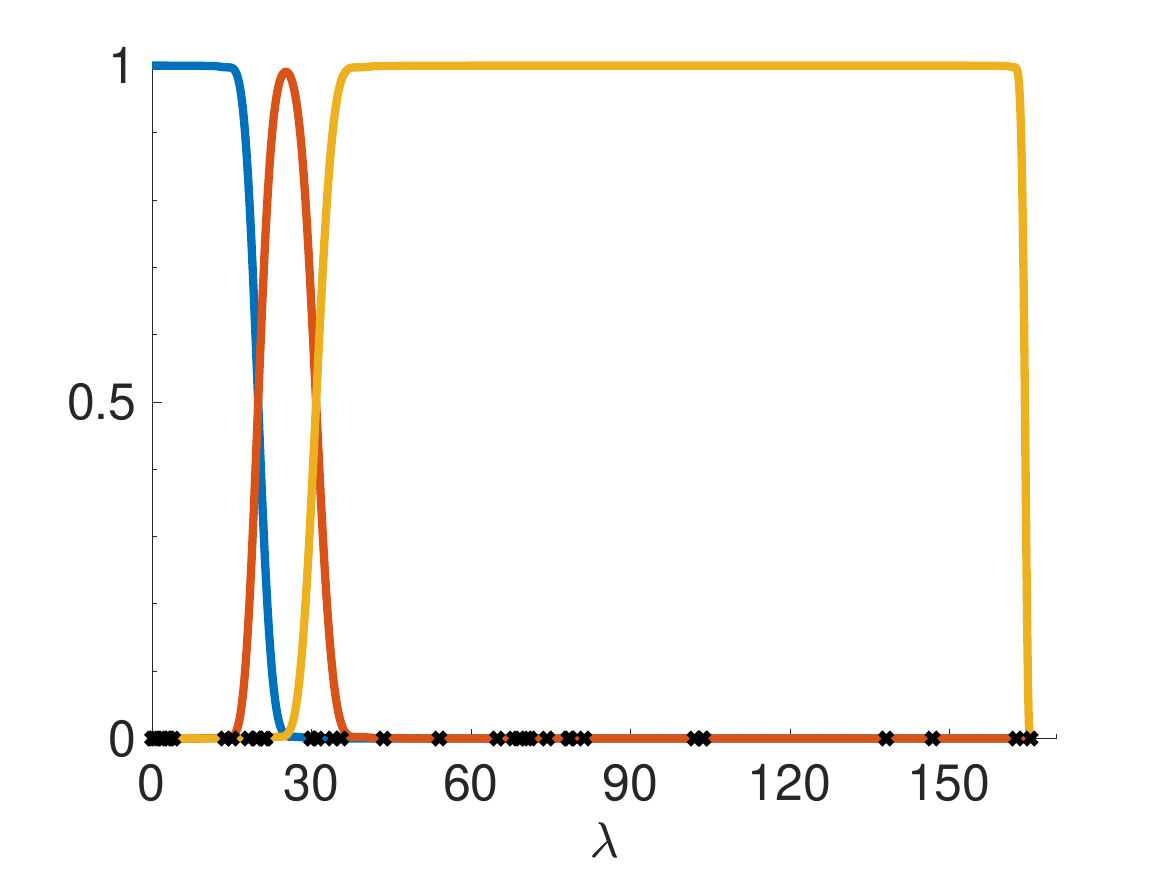}}
\centerline{~~\small{(b)}}
\end{minipage} \\
\begin{minipage}[m]{0.49\linewidth}
\centerline{\includegraphics[width=1.1\linewidth]{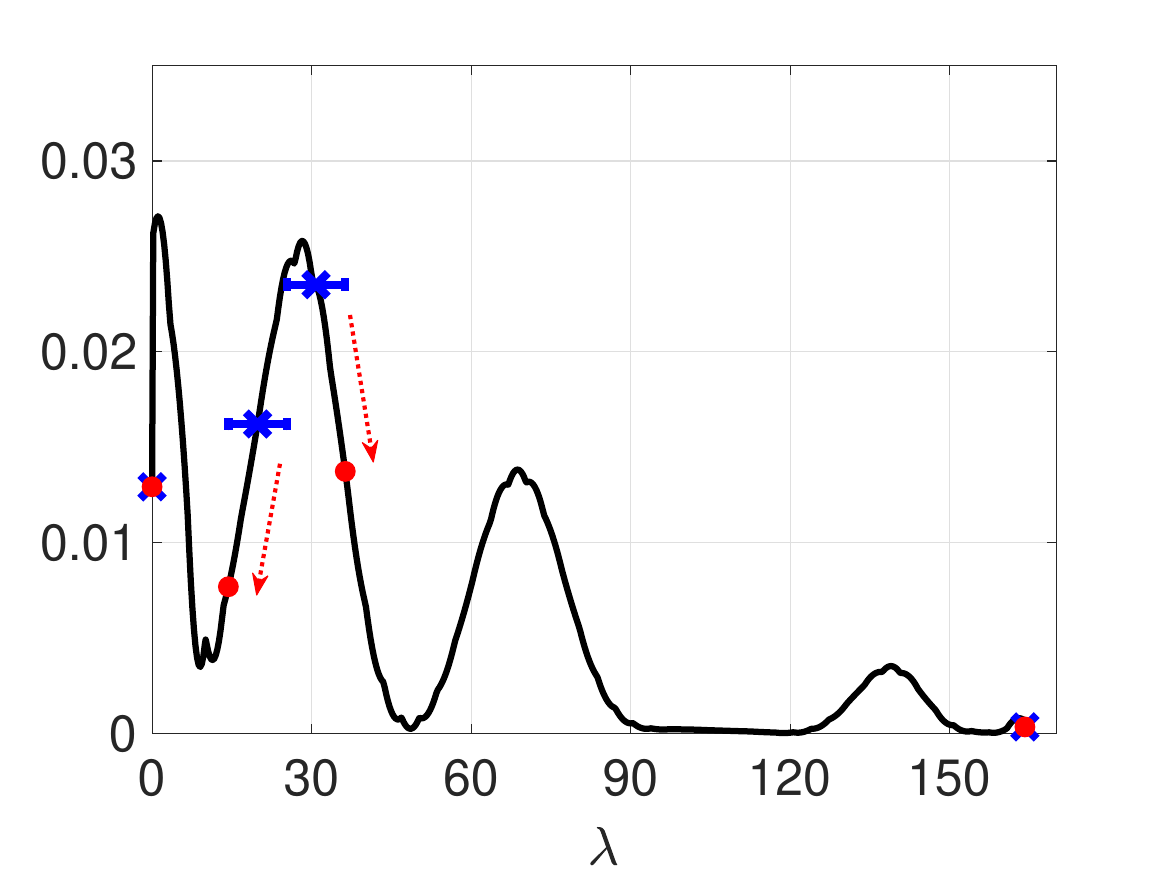}}
\centerline{~~\small{(c)}}
\end{minipage}
\begin{minipage}[m]{0.49\linewidth}
\centerline{\includegraphics[width=1.1\linewidth]{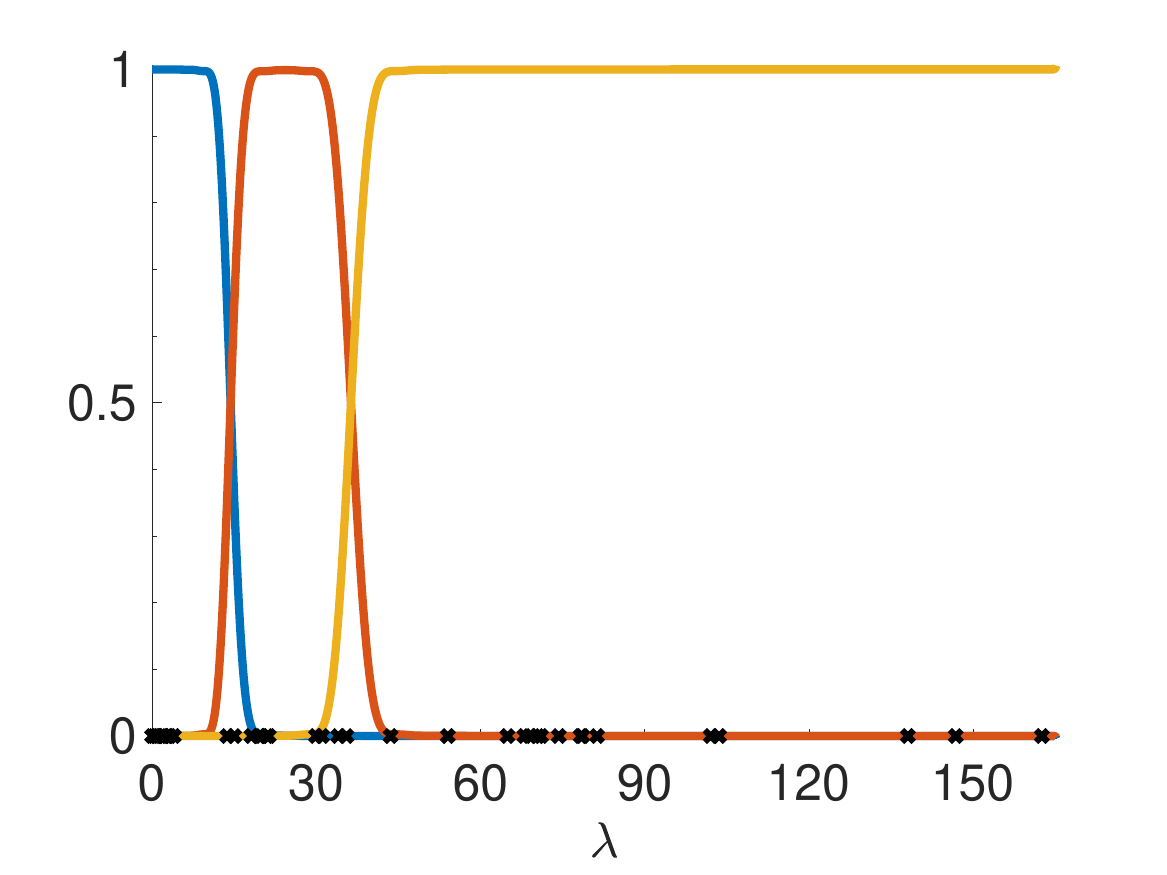}}
\centerline{~~\small{(d)}}
\end{minipage}
\caption{{(a) The approximate cumulative spectral density function, $\tilde{P}_{\lambda}(\cdot),$ for the net25 graph described in Fig. \ref{Fig:approx_filtering_error}. The blue X marks correspond to the initial choice of band endpoints computed by taking the inverse of logarithmically spaced points on the vertical axis. (b) The degree 80 Jackson-Chebyshev approximations to the ideal filters defined by the initial choice of band ends from (a). (c) The objective function of \eqref{Eq:adjustment} (a discrete approximation of the spectral density function $p_{\lambda}(\cdot)$) with $\Delta=.1$. The blue horizontal lines correspond to the search intervals ${\cal I}_1$ and ${\cal I}_2$, and the red circles represent the adjusted band ends $\{\tau_m\}_{m=0,1,2,3}$. (d) The degree 80 Jackson-Chebyshev approximations to the ideal filters defined by the adjusted choice of band ends. Note that the errors between the approximate filters and ideal bandpass filters are concentrated in regions with fewer eigenvalues.}\vspace{-.1in}}\label{Fig:fb_design}
\vspace{.2cm}
\end{figure}

\subsubsection{Adjusting the band ends} \label{Se:adjust}

In order to make the filters more amenable to approximation, we then adjust the initial choice of band endpoints so that they lie in lower density regions of the spectrum. Specifically, for each $m=1,2,\ldots,M-1$ and some 
 $\Delta>0$, we let the final endpoint be
\begin{align} \label{Eq:adjustment}
\tau_m^* = \argmin_{\tau \in {\cal I}_m} \left\{\frac{\tilde{P}_{\lambda}(\tau+\Delta)-\tilde{P}_{\lambda}(\tau-\Delta)}{2\Delta}\right\},
\end{align}
where ${\cal I}_m$ is an interval around the initial choice of $\tau_m$. Fig. \ref{Fig:fb_design}(c) shows the objective function in \eqref{Eq:adjustment}, along with the initial band ends, search intervals, and adjusted band ends. Comparing Fig. \ref{Fig:fb_design}(b) and Fig. \ref{Fig:fb_design}(d), 
the band end adjustments
lead to fewer eigenvalues falling close to the filter borders, 
reducing the error incurred by the polynomial approximation 
process. Algorithm \ref{Al:filter_bank} summarizes the filter bank design in the 
case of spectrum-adapted and logarithmically spaced filters.

\begin{algorithm}[tb]
\caption{Spectrum-adapted and logarithmically spaced filter bank design}
\begin{algorithmic}
\State \textbf{Input} Estimate for $\lambda_{\max}$, approximate spectral density $\tilde{P}_{\lambda}$, number of bands $M$, $\Delta>0$, degree $K$

\For {$m=1,2,\ldots,M$}
\State Compute the initial band end: $\tau_m = \tilde{P}_{\lambda}^{-1}(\frac{1}{2}^{M-m})$
\EndFor
\State Set $\tau_0^* = \tau_0 = 0$, $\tau_M^* = \tau_M=\lambda_{\max}$, 
\For {$m = 1,\cdots, M-1$}
\State Set the search radius: $$r = \text{min} \Bigl\{\frac{\tau_m-\tau_{m-1}}{2}, \frac{\tau_{m+1}-\tau_m}{2}\Bigr\}$$
\State Set the search interval: ${\cal I}_m = [\tau_m-r, \tau_m+r]$
\State Update the band ends: $$\tau_m^* = \underset{\tau \in {\cal I}_m}{\text{argmin}} \left\{\frac{\tilde{P}_{\lambda}(\tau+\Delta)-\tilde{P}_{\lambda}(\tau-\Delta)}{2\Delta}\right\}$$
\EndFor
\For {$m=1,2,\ldots,M$}
\State Construct the ideal filter $h_m(\lambda)$ according to \eqref{Eq:bandpass} using $\tau_{m-1}^*$ and $\tau_m^*$ 
\State Construct the polynomial filter approximation $\tilde{h}_m(\lambda)$ and the corresponding Jackson-Chebyshev polynomial coefficients $\alpha_{m,k}$ 
via \eqref{Eq:cheb_coeff}-\eqref{Eq:damping2}
\EndFor
\State \textbf{Output} 
Degree $K$ Jackson-Chebyshev polynomial filters $\{ \tilde{h}_1(\lambda), \tilde{h}_2(\lambda), \cdots, \tilde{h}_M(\lambda)\}$, and associated coefficients $\{\alpha_{m,k}\}_{m=1,2,\ldots,M; k=0,1,\ldots,K}$
\end{algorithmic}
\label{Al:filter_bank}
\end{algorithm}

\subsection{Non-uniform random sampling distribution}
The partitioning of the vertices into uniqueness sets described in Algorithm 1 requires a full eigendecomposition of the graph Laplacian to compute the matrix ${\bf U}$. Two broad approaches to more efficient sampling have recently been investigated: greedy methods \cite{chen2015discrete,anis2014towards,tsitsvero2016uncertainty,anis2016efficient} and random sampling methods \cite{shomorony,PuyTGV15,chen2016signal}, which have close connections to \emph{leverage score} sampling in the statistics and numerical linear algebra literature 
(see, e.g., \cite{drineas2012fast}\nocite{mahoney2009cur}-\cite{mahoney2011randomized}). 
Reference \cite{anis2016efficient} has a nice review of the computational complexities of the various greedy routines for identifying uniqueness sets. Most of these are designed specifically for lowpass signals. 
 
We adapt the non-uniform random sampling method of \cite{PuyTGV15}, which scales more efficiently than greedy methods. 
For the $m^{th}$ band, we identify the downsampling set $\V_m$ by sampling the vertices $\V$ without replacement according to a discrete probability distribution ${\boldsymbol \omega}_m$. To minimize the graph weighted coherence, it is ideal to take ${\boldsymbol \omega}_m(i) \propto ||{\bf U}_{{\cal R}_m}^{\top} {\boldsymbol \delta}_i ||_2^2$ \cite{PuyTGV15}; however, we do not have access to ${\bf U}_{{\cal R}_m}$. Instead, we take 
\begin{align}\label{Eq:samp_dist}
{\boldsymbol \omega}_m(i) \propto ||(\tilde{h}_m(\L){\bf X})^{\top} {\boldsymbol \delta}_i ||_2^2,
\end{align}
which \cite{PuyTGV15} shows is an unbiased estimator of $||{\bf U}_{{\cal R}_m}^{\top} {\boldsymbol \delta}_i ||_2^2$ when ${\bf X}$ is the random matrix from \eqref{Eq:hutch4}. Since we already compute and store the series of matrices $\{\bar{T}_k(\L){\bf X}\}$ for the spectral density estimation of Section \ref{Se:spectral_density}, we just need to compute the polynomial approximation coefficients $\{\alpha_{m,k}\}$ in \eqref{Eq:hutch4} for $h_m(\lambda)$ in order to compute $\tilde{h}_m(\L){\bf X}$.

Intuitively, the sampling distribution approximates the energy of the selected eigenvectors concentrated on each vertex. In the extreme case that the selected eigenvectors are completely concentrated on a single vertex or small neighborhood of vertices, sampling signal values outside of this set provides no additional information, justifying the zero weight in the sampling distribution. For eigenvectors whose energy is equally spread across the graph, this results in uniform sampling. In particular, for any walk-regular graph, a class that includes vertex-transitive graphs, which in turn include shift-invariant graphs such as the cycle graph, $||{\bf U}_{{\cal R}}^{\top} {\boldsymbol \delta}_i ||_2^2$ is constant across vertices $i$ for any choice of eigenvectors ${\cal R}$ \cite[Corollary 3.2]{chan1997symmetry}, resulting in uniform random sampling for all bands. As discussed in \cite[Section 5.1.2]{PuyTGV15}, non-uniform sampling is particularly beneficial for bands with localized eigenvectors, 
 which most commonly occur at the middle and upper ends of the spectrum. For the low end of the spectrum with smooth eigenvectors, the intuition is that it is easier to interpolate missing values in highly connected regions of the graph, and therefore there are slightly higher weights on the less connected vertices (e.g., near the boundaries in Fig. \ref{Fig:temperature}(g)).

\begin{algorithm}[tb]
\caption{Construct the downsampling sets}
\begin{algorithmic}
\State \textbf{Input} graph $\G$,  ${\bf X}$, $\{\bar{T}_k(\L){\bf X}\}_{k=0,1,\ldots,K}$, Jackson-Chebyshev coefficients $\{\alpha_{m,k}\}_{m=1,2,\ldots,M; k=0,1,\ldots,K}$ for the polynomial filters $\{ \tilde{h}_1(\lambda), \tilde{h}_2(\lambda), \cdots, \tilde{h}_M(\lambda)\}$, signal ${\bf f}$ (optional)
\For {$m=1,2,\ldots,M$}
\State Compute $\tilde{h}_m(\L){\bf X}=\sum_{k=0}^K \alpha_{m,k}\bar{T}_k(\L){\bf X}$
\State Set the weight for each vertex $i \in \V$: 
$${\boldsymbol \omega}_m(i) = ||(\tilde{h}_m(\L){\bf X})^{\top} {\boldsymbol \delta}_i||^2_2$$
\If{signal-adapted weights}
\State Compute $\tilde{h}_m(\L){\bf f}$ via \eqref{Eq:cheb} with the same $\{\alpha_{m,k}\}$
\State Adapt the weights: $${\boldsymbol \omega}_m(i) = {\boldsymbol \omega}_m(i) \cdot \log(1+|(\tilde{h}_m(\L){\bf f})(i)|)$$
\EndIf
\State Normalize the weights: ${\boldsymbol \omega}_m(i)=\frac{{\boldsymbol \omega}_m(i)}{\sum_{i=1}^N {\boldsymbol \omega}_m(i)}$

\State Set the initial number of samples based on \eqref{Eq:hutch3}: 
$$n_m = \frac{1}{J} \mathrm{Trace}({\bf X}^\top \tilde{h}_m(\L)\bf{X})$$
\EndFor
\If{signal-adapted number of samples}
\For {$m=1,2,\ldots,M$}
\State Set $n_m=n_m \cdot \log(1+||\tilde{h}_m(\L){\bf f}||)$
\EndFor
\EndIf
\State Compute total initial number of samples: $N_0 = \sum_m n_m$
\For {$m=1,2,\ldots,M$}
\State Normalize the number of total samples: $$n_m = \hbox{round}\Bigl(\frac{n_m}{N_0}N_T\Bigr),$$
where $N_T$ is the target number of samples (e.g., $N_T=N$ for critical sampling)
\EndFor
\State Adjust to meet target number of samples:  
\If{$\sum_m n_m > N_T$} 
\State Set $n_M = n_M - (\sum_m n_m -N_T)$
\ElsIf{$\sum_m n_m < N_T$}
\State Set $n_1 = n_1 + (N_T-\sum_m n_m)$
\EndIf
\For {$m=1,2,\ldots,M$}
\State Choose the downsampling set $\V_m$ by randomly sampling $n_m$ vertices according to the distribution ${\boldsymbol \omega}_m$
\EndFor
\State \textbf{Output} downsampling sets $\{\V_1,\V_2, \ldots, \V_M\}$, sampling distributions $\{{\boldsymbol \omega}_1,{\boldsymbol \omega}_2,\ldots,{\boldsymbol \omega}_M\}$
\end{algorithmic}
\label{Al:downsampling}
\end{algorithm}

\subsection{Number of samples}

One option to ensure critical sampling is to choose the number of samples for each band 
according to the initial filter bank design. For example, if the filter bank is designed to be adapted to the spectrum with logarithmic spacing, we can choose $\frac{N}{2}$ samples for the highest band, $\frac{N}{4}$ for the next highest, and so forth. However, the adjustments we make in Section \ref{Se:adjust} affect the number of eigenvalues contained in each band. Since we have an estimate of the cumulative spectral distribution, one approximation for the number of samples in the adjusted $m^{th}$ band is to round $N \cdot (\tilde{P}_{\lambda}(\tau_{m})-\tilde{P}_{\lambda}(\tau_{m-1}))$. As a band end $\tau_m$ may fall at a point where $\tilde{P}_\lambda$ has been interpolated via cubic functions, another option is to estimate the number of eigenvalues between $\tau_{m-1}$ and $\tau_m$, once again with the stochastic trace estimator in \eqref{Eq:hutch3}, except using the bandpass filter $h_m(\lambda)$ from \eqref{Eq:bandpass}. We already compute $\tilde{h}_m(\L){\bf X}$ to calculate the sampling distribution in \eqref{Eq:samp_dist}. We can substitute the columns $\tilde{h}_m(\L){\bf x}^{(j)}$ of this matrix into 
\eqref{Eq:hutch3} for an estimate of the number of eigenvalues in the $m^{th}$ band. An added benefit of this extra step is that the thresholds $\{\tau_m\}$ are chosen to be in areas of low spectral density, which improves the accuracy of the eigenvalue count estimate \cite{di2016efficient}. 

We make small adjustments to ensure the total number of samples is equal to some target $N_T$. 
In our experiments, we take
$N_T=N$  to ensure critical sampling. Our default is to add samples to the lowest band if the normalized total is 
below $N_T$, and remove samples from the highest band if the normalized total is above 
$N_T$.  Algorithm \ref{Al:downsampling} summarizes
the proposed method to choose the downsampling sets.

Note that $\hbox{dim(col}(\tilde{h}_m(\L))) \geq \hbox{dim(col}({h}_m(\L)))$, 
with the difference depending on the number of Laplacian eigenvalues just outside the end points of $h_m(\cdot)$ and the degree of approximation used for $\tilde{h}_m(\cdot)$. Therefore, we expect that to perfectly reconstruct signals in $\hbox{col}(\tilde{h}_m(\L))$, we need more samples than the number of eigenvalues in the support of $h_m(\cdot)$. In Section \ref{Se:ill2}, we explore how the reconstruction error is reduced as we increase the number of samples in each band.

\subsection{Interpolation}
The exact interpolation \eqref{Eq:synth} requires the eigenvector matrix ${\bf U}$, and in case ${\bf U}_{\V_m,{\cal R}_m}$ is not full rank, the standard least squares reconstruction for the $m^{th}$ channel
\begin{align}
{\bf f}_{m,{rec}}=\mathbf{U}_{{\cal R}_m}(\mathbf{U}_{\V_m,{\cal R}_m}^{\top}\mathbf{U}_{\V_m,{\cal R}_m})^{-1}\mathbf{U}_{\V_m,{\cal R}_m}^{\top}{\bf y}_{\V_m}
\end{align}
also requires ${\bf U}_{{\cal R}_m}$. One option explored in \cite{halko,paratte} is to leverage $\{\bar{T}_k(\L){\bf X}\}$ again to approximate the column space of ${\bf U}_{{\cal R}_m}$ by filtering at least $|{\cal R}_m|$ standard normal random vectors with the filter $\tilde{h}_m(\cdot)$, possibly followed by orthonormalization via QR factorization.

A second approach suggested in \cite{PuyTGV15} to efficiently reconstruct lowpass signals is to relax the optimization problem
\begin{align*}
\min_{{\bf z} \in \hbox{col}({\bf U}_{{\cal R}_m})} ||{\boldsymbol \Omega}_{m,\V_m}^{-\frac{1}{2}}\left({\bf M_m z}-{{\bf y}_{\V_m}} \right) ||_2^2
\end{align*} 
to 
\begin{align}\label{Eq:approx_rec_opt}
\min_{{\bf z} \in \R^N} \left\{{\bf z}^{\top}\varphi_m(\L){\bf z} + \kappa ||{\boldsymbol \Omega}_{m,\V_m}^{-\frac{1}{2}}\left({\bf M_m z}-{{\bf y}_{\V_m}} \right) ||_2^2\right\},
\end{align} 
where ${\boldsymbol \Omega}_{m,\V_m}$ is a $|\V_m| \times |\V_m|$  diagonal matrix with the $m^{th}$ channel sampling weights of $\V_m$ along the diagonal, and $\kappa>0$ is a parameter to trade off the two optimization objectives. The regularization term ${\bf z}^{\top}\varphi_m(\L){\bf z}$ in \eqref{Eq:approx_rec_opt} penalizes reconstructions with support outside of the desired spectral band. For lowpass signals, Puy et al. \cite{PuyTGV15} take the penalty function $\varphi_m(\lambda)$ to be a nonnegative, nondecreasing polynomial, such as $\lambda^l$, with $l$ a positive integer. For more general classes of signals (i.e., the midpass and highpass signals output from the higher bands of the proposed filter bank), it is important to keep the nonnegativity property, in order to ensure that $\varphi_m(\L)$ is positive semi-definite and the optimization problem \eqref{Eq:approx_rec_opt} is convex. However, we can drop the nondecreasing requirement, and instead choose penalty functions concentrated outside the $m^{th}$ spectral band. Options we explore include (i) the polynomial filter $\varphi_m(\lambda)=1-\tilde{h}(\lambda)$; (ii) the rational filter $\varphi_m(\lambda)=\frac{1}{\tilde{h}(\lambda)+\epsilon}-\frac{1}{1+\epsilon}$; 
and (iii) a polynomial approximation of a penalty function constructed as a piecewise cubic spline, an approach explored in \cite{chen_saad}. See Fig. \ref{Fig:penalty} for example graphs of these penalty functions.

\begin{figure}[tb]
\begin{minipage}[m]{0.49\linewidth}
\centerline{\small{$K=20$}}
\centerline{\includegraphics[width=1.1\linewidth]{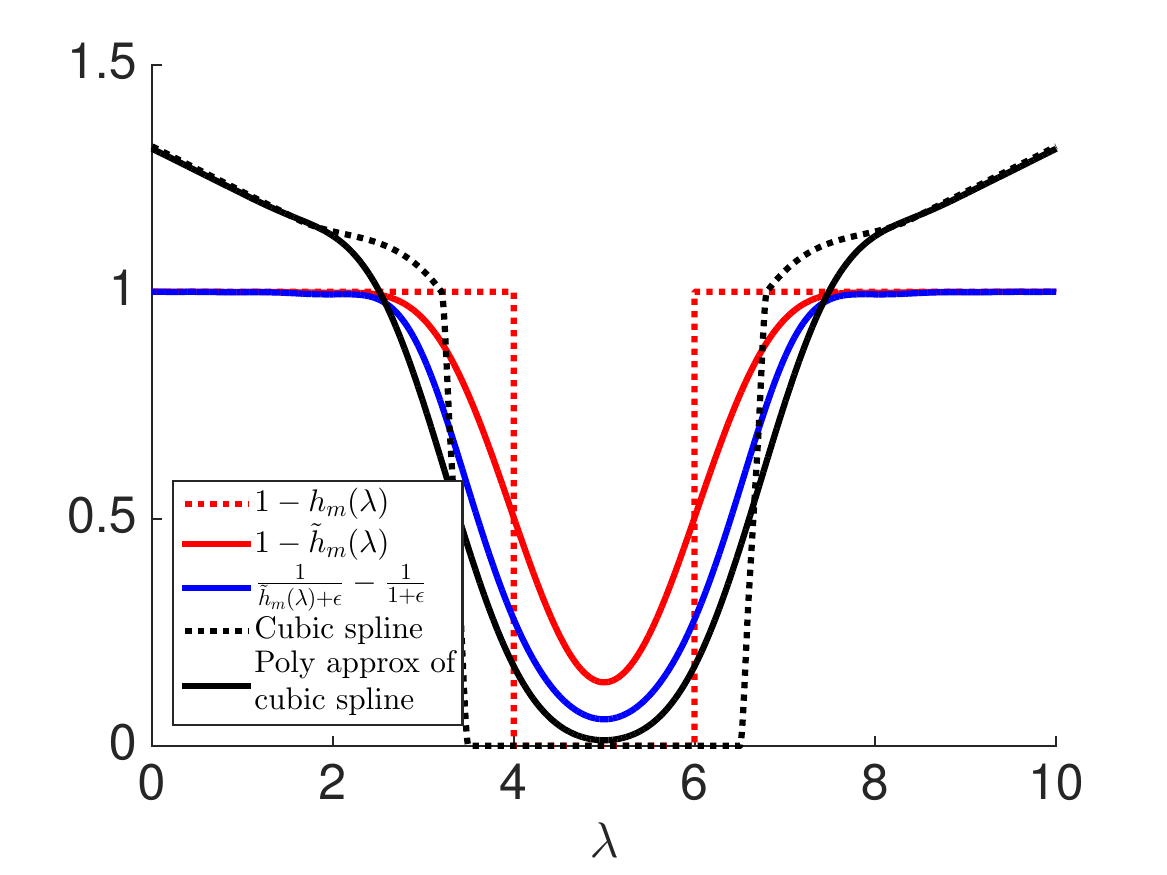}}
\centerline{~~\small{(a)}}
\end{minipage}
\begin{minipage}[m]{0.49\linewidth}
\centerline{\small{$K=50$}}
\centerline{\includegraphics[width=1.1\linewidth]{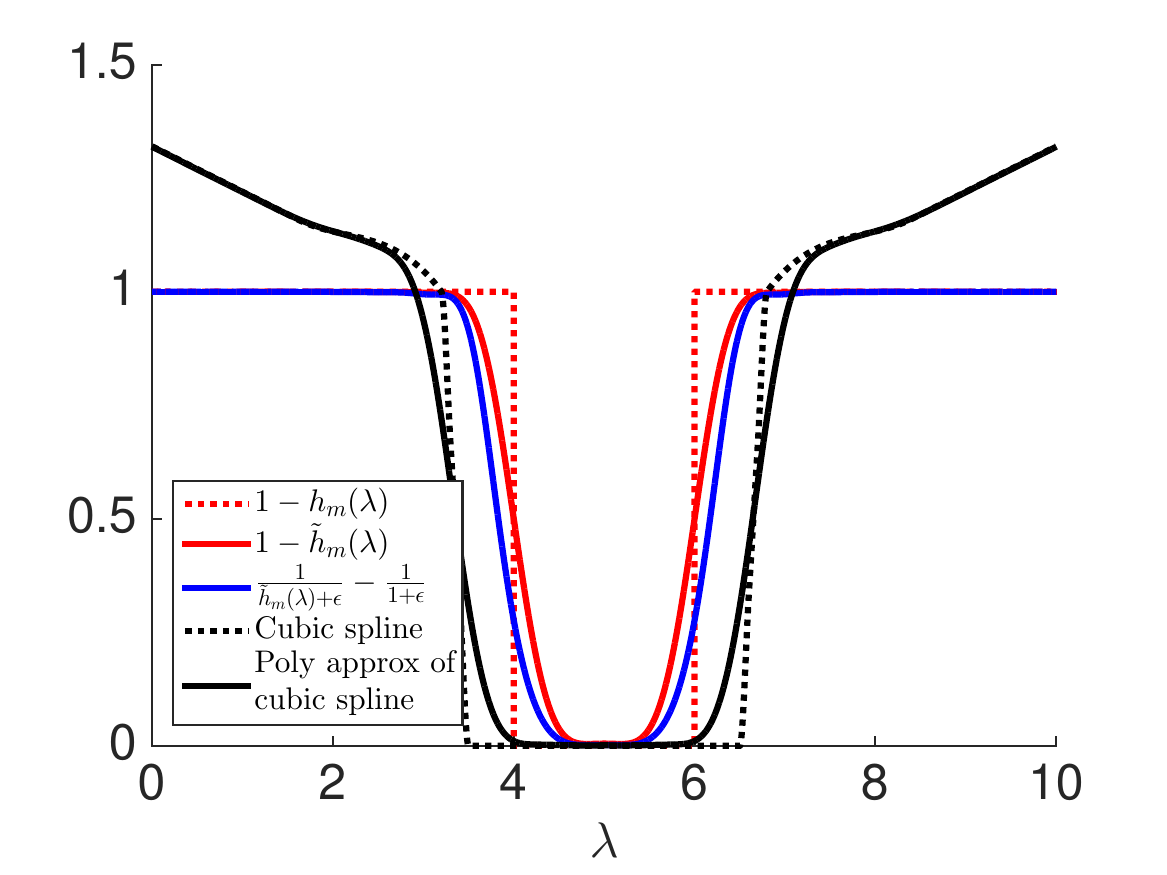}}
\centerline{~~\small{(b)}}
\end{minipage} 
\caption{Example penalty filters $\varphi_m$ for the regularization term in \eqref{Eq:approx_rec_opt}, with $\epsilon=\frac{\sqrt{5}-1}{2}$. Here, $\tilde{h}_m$ is a Jackson-Chebyshev polynomial approximation of $h_m$ of degree 20 and 50 in (a) and (b), respectively. \vspace{-.15in} }\label{Fig:penalty}
\vspace{.2cm}
\end{figure}

From the first-order optimality conditions, the solution to \eqref{Eq:approx_rec_opt} is the solution to the linear system of equations
\begin{align}\label{Eq:approx_rec_sol}
\Bigl(\kappa M_m^{\top}{\boldsymbol \Omega}_{m,\V_m}^{-1}M_m+\varphi_m(\L)\Bigr){\bf  z}=\kappa M_m^{\top}{\boldsymbol \Omega}_{m,\V_m}^{-1}{\bf y}_{\V_m},
\end{align}
which can be solved, for example, with the preconditioned conjugate gradient method. For the preconditioner, we use a diagonal matrix whose $i^{th}$ element is equal to 1 if $i \notin {\V_m}$ and $1+\frac{\kappa}{{\boldsymbol \omega}_m(i)}$ if $i \in {\V_m}$, which serves as an approximation to the matrix $\kappa M_m^{\top}{\boldsymbol \Omega}_{m,\V_m}^{-1}M_m+\varphi_m(\L)$ in \eqref{Eq:approx_rec_sol}.

\subsection{Summary and properties of the fast M-CSFB transform} \label{Se:fast_prop}
In summary, as shown in the flow chart in Fig. \ref{Fig:flow_chart}, the set up for the fast M-CSFB consists of approximating the spectral density of the graph Laplacian, designing the filter bank via Algorithm \ref{Al:filter_bank} and choosing the downsampling sets via Algorithm \ref{Al:downsampling}. To analyze a signal, we apply each of the $M$ Jackson-Chebyshev polynomial filters output from Algorithm \ref{Al:filter_bank} to the signal, and then downsample on the corresponding set of vertices output from Algorithm \ref{Al:downsampling}. To synthesize a signal from its transform coefficients, we solve \eqref{Eq:approx_rec_sol} for each band and sum the results. The complexity of the set up is dominated by the computation of $\{\bar{T}_k(\L){\bf X}\}_{k=0,1,\ldots,K}$  via \eqref{Eq:Tkbar_rec}, which has computational complexity ${\cal O}(JK|\E|)$. The computational complexity of the analysis is ${\cal O}(K|\E|)$. So, if $N$ is large, the number of random vectors $J$ and degree of polynomial approximation $K$ are small compared to $N$, and the graph is sparse ($|\E|$ is roughly a small constant times $N$), the setup and analysis scale linearly with the number of vertices. If each $\varphi_m$ is taken to be an order $K$ polynomial and the conjugate gradient is run for at most $I$ iterations, the bottleneck computation of the synthesis has complexity ${\cal O}(IMK|\E|)$. In practice, the required number of iterations and corresponding computation time depend on the conditioning of the matrix on the left-hand side of \eqref{Eq:approx_rec_sol} and the choice of preconditioner. 

When the degree of approximation $K$ is small, the energies 
of the fast, approximate transform atoms, which are of the form $\tilde{h}_m(\L){\boldsymbol \delta}_i$, may be slightly more spread in the spectral domain than those of the atoms of the form ${h}_m(\L){\boldsymbol \delta}_i$, due to the polynomial approximation of the ideal filter. However, the energy is also guaranteed to be supported completely within a radius of $K$ hops from the center vertex $i$ \cite{hammond2011wavelets,shuman2015vertex}. Thus, we have better control over the spread in the vertex domain, which, as we saw in the scale 3 wavelet atom in Fig. \ref{Fig:bunny_coef}, may be larger with the ideal filters.

\begin{figure}[t]
\centerline{\includegraphics[width=1\linewidth]{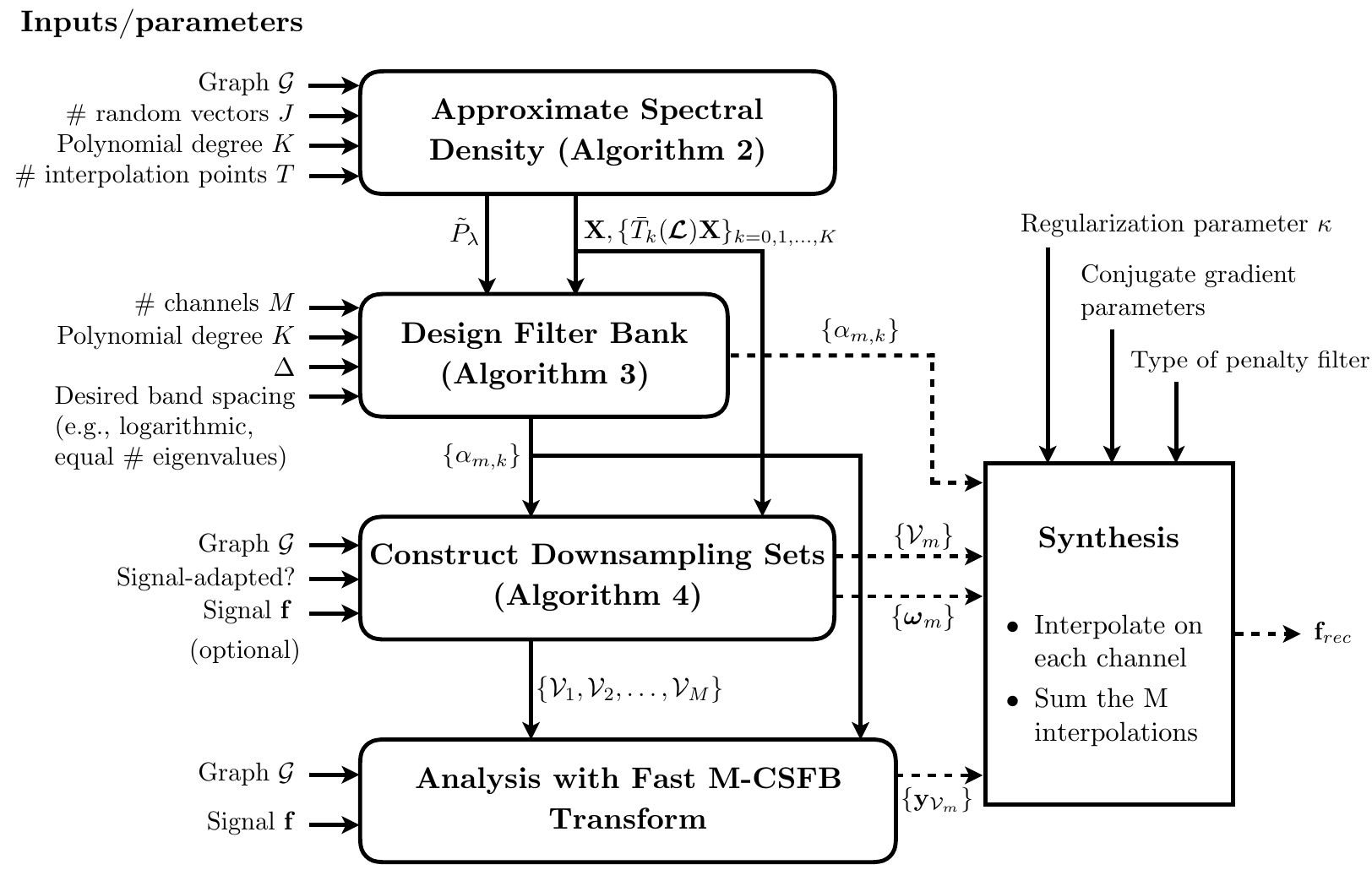}}
\caption{Flow chart of the set up, analysis, and synthesis for the fast $M$-CSFB transform. The only differences for the signal-adapted transform are in the construction of the downsampling sets.} \label{Fig:flow_chart}
\vspace{-.08in}
\end{figure}

\section{Signal-Adapted Fast M-CSFB Transform} \label{Se:signal_adapted}
Just as it is helpful for interpolation to sample more signal values at vertices where the energy of the selected eigenvectors is concentrated, it is also helpful to sample more values where the energy of the filtered signals is concentrated. 
This motivates three adaptations to the fast M-CSFB transform. 

First, we subtract the mean of the signal (i.e., let ${\bf f}={\bf f}-\frac{{\bf 1}^{\top}{\bf f}}{N} {\bf 1}$) before sending it into the filter bank, and then add this constant back to every vertex when summing the interpolations from the $M$ channels. To ensure critical sampling, we only allow $N-1$ total samples in addition to this mean.

Second, we adapt the sampling weights by setting ${\boldsymbol \omega}_m(i) = {\boldsymbol \omega}_m(i) \cdot  \log\left(1+|(\tilde{h}_m(\L){\bf f})(i)|\right)$. Thus, if a filtered signal on a given band is concentrated on a certain region of the graph, the sampling weights 
are concentrated on the intersection of that region and the set of vertices where the energies of the selected eigenvectors are concentrated. 

Third, beyond the distribution of samples within each band, we need to decide how many samples to allocate to each band. In the exact computation (small graph) case, allocating the samples according to the number of eigenvalues contained in the disjoint bands ensures perfect reconstruction. 
However, with approximate computations, it is beneficial to the overall 
reconstruction error to do a better job of interpolation on the bands whose filtered signals have the most energy. 
We set the initial number of samples by multiplying the estimate 
of the number of eigenvalues in the band with $\log(1+||\tilde{h}_m(\L){\bf f}||)$, 
 as shown in Algorithm \ref{Al:downsampling}. In the extreme case of a filtered signal with no energy, this choice leads to zero measurements and a reconstruction of the all zero vector.

\begin{table*}[tb]
{\footnotesize
\tabcolsep=0.11cm
\begin{center}
\begin{tabular}{l|C{.8cm}C{.8cm}C{.9cm}|C{.8cm}C{.8cm}C{.9cm}|C{.8cm}C{.8cm}C{.9cm}|C{.8cm}C{.8cm}C{.9cm}|C{.8cm}C{.8cm}C{.9cm}|}
\cline{2-16}
 & \multicolumn{3}{ c| }{Sensor Network} & \multicolumn{3}{ c| }{Bunny} & \multicolumn{3}{ c| }{Andrianov net25 Graph} & \multicolumn{3}{ c| }{Community Graph} & \multicolumn{3}{ c| }{Temperatures} \\ 
  &  \multicolumn{3}{ C{3cm}| }{$N=500$ \newline $|\E|=2,050$} & \multicolumn{3}{ C{3cm}| }{$N=2,503$ \newline $|\E|=13,726$}&  \multicolumn{3}{ C{3cm}| }{$N=9,520$ \newline $|\E|=195,841$}&  \multicolumn{3}{ C{3cm}| }{$N=25,000$ \newline $|\E|=480,459$}  & \multicolumn{3}{ C{3cm}| }{$N=469,404$ \newline $|\E|=1,865,415$} \\ 
\cline{2-16}
& Anal. \newline Time &  Synth. \newline Time & Rec. \newline NMSE  & Anal. \newline Time &  Synth. \newline Time & Rec. \newline NMSE & Anal. \newline Time &  Synth. \newline Time & Rec. \newline NMSE & Anal. \newline Time &  Synth. \newline Time & Rec. \newline NMSE & Anal. \newline Time & Synth. \newline Time & Rec. \newline NMSE   \\ 
\cline{1-16}
\multicolumn{1}{|L{2.6cm}|}{Graph Fourier Transform}  & 0.1 & 0.01 & 5.4e-30 & 9.8 & 0.02 & 2.5e-29 & 295.7 & 0.08 & 1.4e-28 & 8544.8 & 0.6 & 4.5e-28 & NA & NA & NA \\
\cline{1-16}
\multicolumn{1}{ |L{2.6cm}| }{Exact $M$-CSFB}& 2.2 & 0.06& 7.8e-30 & 380.4  & 0.1  & 7.8e-23 & NA & NA & NA & NA & NA & NA & NA& NA & NA \\
\cline{1-16}
\multicolumn{1}{|L{2.6cm}|}{Diffusion Wavelets \cite{coifman2006diffusion}}  & 8.5 & 0.03 & 1.2e-30 & 313.9 & 0.02 & 1.2e-29 & 14354 & 0.3 & 1.0e-26 & NA & NA & NA & NA & NA & NA \\
\cline{1-16}

\multicolumn{1}{|L{2.6cm}| }{Graph-QMF \cite{narang2012perfect}} & 0.6 & 0.1& 5.4e-8 & 4.9 & 3.4 & 3.2e-8 & 38.4 & 21.0 & 3.3e-9 & 1062.7 & 978.0 & 6.0e-8 & NA & NA & NA \\
\cline{1-16}

\multicolumn{1}{|L{2.6cm}| }{Fast $M$-CSFB \newline (Scenario A: faster)}& 0.6 & 0.5& 6.8e-2 & 0.8 & 0.9 & 8.2e-2 &2.3 & 3.1 &1.6e-1 & 2.8 & 12.4 & 2.2e-1 & 55.1 & 94.5 & 1.4e-2 \\
\cline{1-16}
\multicolumn{1}{|L{2.6cm}|}{Fast $M$-CSFB \newline  (Scenario B: more accurate)} & 0.7& 1.0& 9.2e-2 & 0.9 & 3.7 & 3.3e-2 & 1.4 & 12.1 & 1.4e-1 & 4.4 & 71.7 & 1.5e-1 & 91.6 & 874.3 & 7.0e-3 \\
\cline{1-16} 
\multicolumn{1}{|L{2.6cm}|}{Signal-Adapted \newline Fast $M$-CSFB \newline (Scenario A: faster)} & 0.7 & 0.5 & 3.8e-2 & 0.8 & 0.9 & 3.4e-2 &0.8 & 2.2 & 6.7e-2 & 2.8 & 9.9 & 1.2e-1 & 47.6 & 98.4 & 1.7e-3 \\
\cline{1-16} 
\multicolumn{1}{|L{2.6cm}| }{Signal-Adapted \newline Fast $M$-CSFB \newline  (Scenario B: more accurate)} & 0.7 & 1.1& 2.4e-2 & 0.9 & 3.6 & 1.2e-2 & 1.3 &9.7& 7.7e-2 & 4.4 & 71.1 & 7.9e-2 & 81.2 & 976.0 & 6.6e-4 \\
\cline{1-16} 
\end{tabular}
\end{center}
}
\caption{Comparison of computation times (seconds) and reconstruction errors} 
\label{Ta:comp_times}
\vspace{-.65cm}
\end{table*}

Note that when analyzing a single signal, these adaptations do not add significantly to the computational complexity
of the transform. However, if we are repeating the transform on 
many different signals residing on the same graph, we do need 
to rerun the random selection of vertices for each signal.

\section{Numerical Experiments} \label{Se:ill2}

\subsection{Scalability}

The one indisputable advantage of the proposed fast $M$-CSFB transforms over other critically sampled transforms for graph signals is their scalability to sparse graphs with a large number of vertices. In Table \ref{Ta:comp_times}, 
we compare the computation times of the proposed transform to those of the exact graph Fourier transform (i.e., a full diagonalization of $\L$); diffusion wavelets \cite{coifman2006diffusion} with five scales (one scaling and four wavelets) and a precision of $\epsilon=1\hbox{e-4}$; and a graph quadrature mirror filter (QMF) bank \cite{narang2012perfect} with polynomial approximation order of $K=50$. For the fast $M$-CSFB (original and signal adapted versions), we consider two scenarios. Scenario A is faster, but less accurate, with $K=25$, a conjugate gradient (CG) tolerance of 1e-8, and a maximum of 100 CG iterations. Scenario B is slower, but more accurate, with $K=50$, a CG tolerance of 1e-10, and a maximum of 250 iterations. For all fast $M$-CSFB cases, we let $M=5$, $J=30$ and $\kappa=1$, and subtract out the mean of the signal before applying the filter bank. For larger graphs, calculations such as a full diagonalization are either not possible due to memory limits or would take in excess of a day to compute. We denote these by NA. Note that the graph-QMF transform is slower due to two additional bottlenecks. First, it requires a graph coloring via algorithms with complexities of ${\cal O}(N^3)$ or ${\cal O}(N^4)$ \cite{klotz2002graph}.  Second, whereas for the $M$-CSFB we can reuse the single sequence of vectors $\{\bar{T}_k(\L){\bf f}\}$ in computing the filtered signal $\tilde{h}_m(\L){\bf f}$ on each channel, the iterative nature of the graph-QMF filter bank leads to different input signals at each level, resulting in a complexity increase from ${\cal O}(K|\E|)$ (assuming $M$ is on the order of the average degree of the graph) to  ${\cal O}(K|\E|(2^L-1))$, where $L$ is the number of bipartite subgraphs in the graph-QMF transform. The second bottleneck is more significant in the computation times for smaller graphs, while the first becomes prohibitive for extremely large graphs. 

We apply all transforms to the signals shown above for the sensor network and bunny graph, 
to Gaussian random vectors with independent entries for the 
net25 and community (100 communities) graphs, and to a temperature signal discussed in detail in the next subsection. Note that in all of these examples, we are simply performing analysis followed by synthesis, without doing any compression or other adjustments to the analysis coefficients (compression examples are included in Section \ref{Se:compression}). Of particular note in Table \ref{Ta:comp_times} is that the fast $M$-CSFB analysis times are under a minute for a graph with almost a half of a million vertices and 2 million edges. 

\begin{figure*}[tb] 
\begin{minipage}[m]{0.24\linewidth}
\centerline{\includegraphics[width=.9\linewidth]{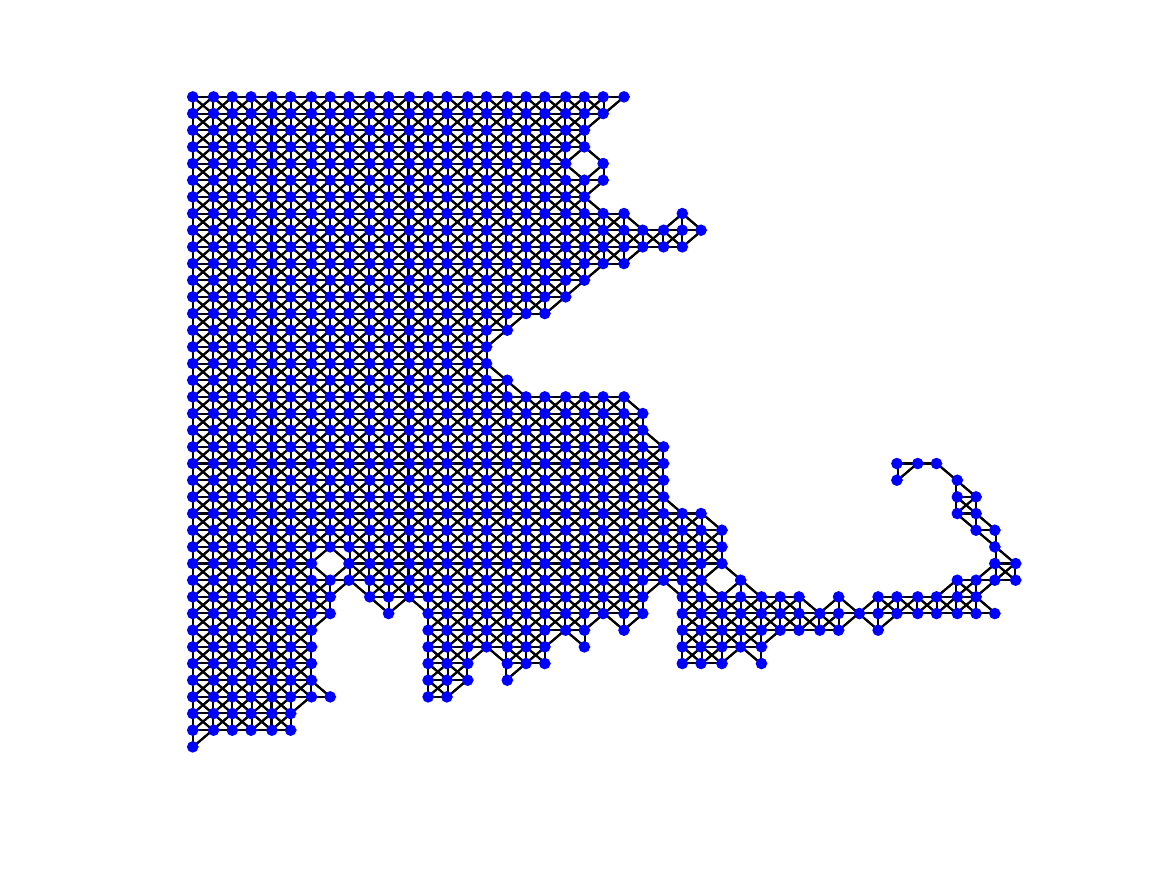}}
\vspace{-.2in}
\centerline{\small{(a)}}
\end{minipage}
\begin{minipage}[m]{0.24\linewidth}
\centerline{\includegraphics[width=.9\linewidth]{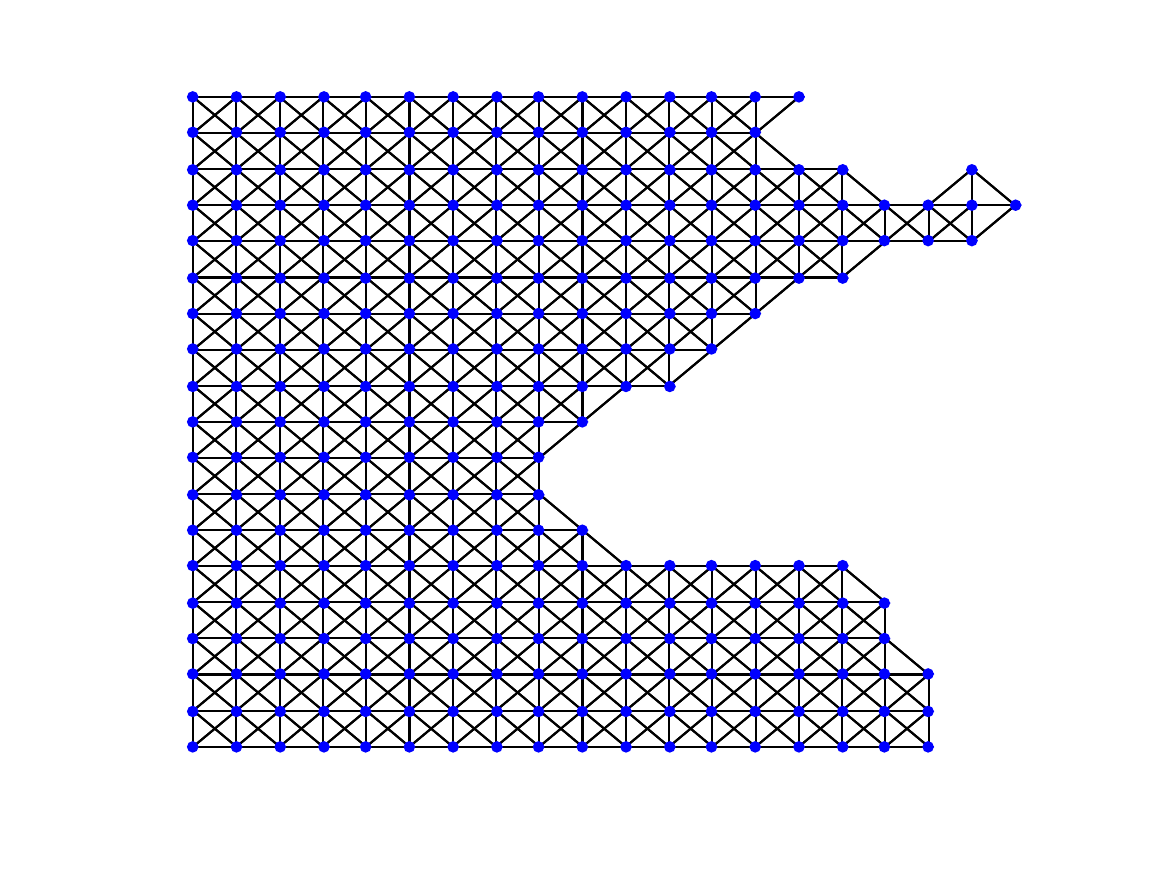}}
\vspace{-.2in}
\centerline{\small{(b)}}
\end{minipage} 
\begin{minipage}[m]{0.24\linewidth}
\centerline{\includegraphics[width=.9\linewidth]{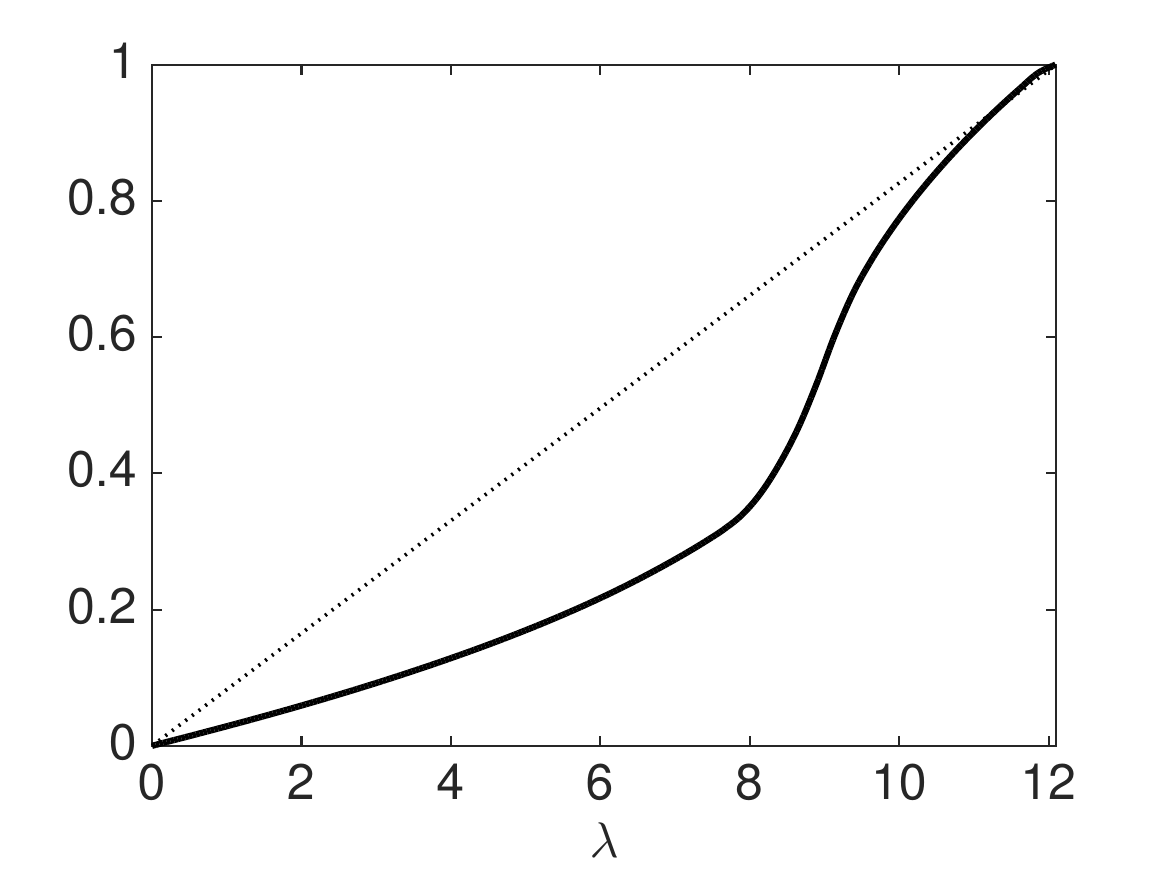}}
\centerline{~\small{(c)}}
\end{minipage} 
\begin{minipage}[m]{0.24\linewidth}
\centerline{\includegraphics[width=.9\linewidth]{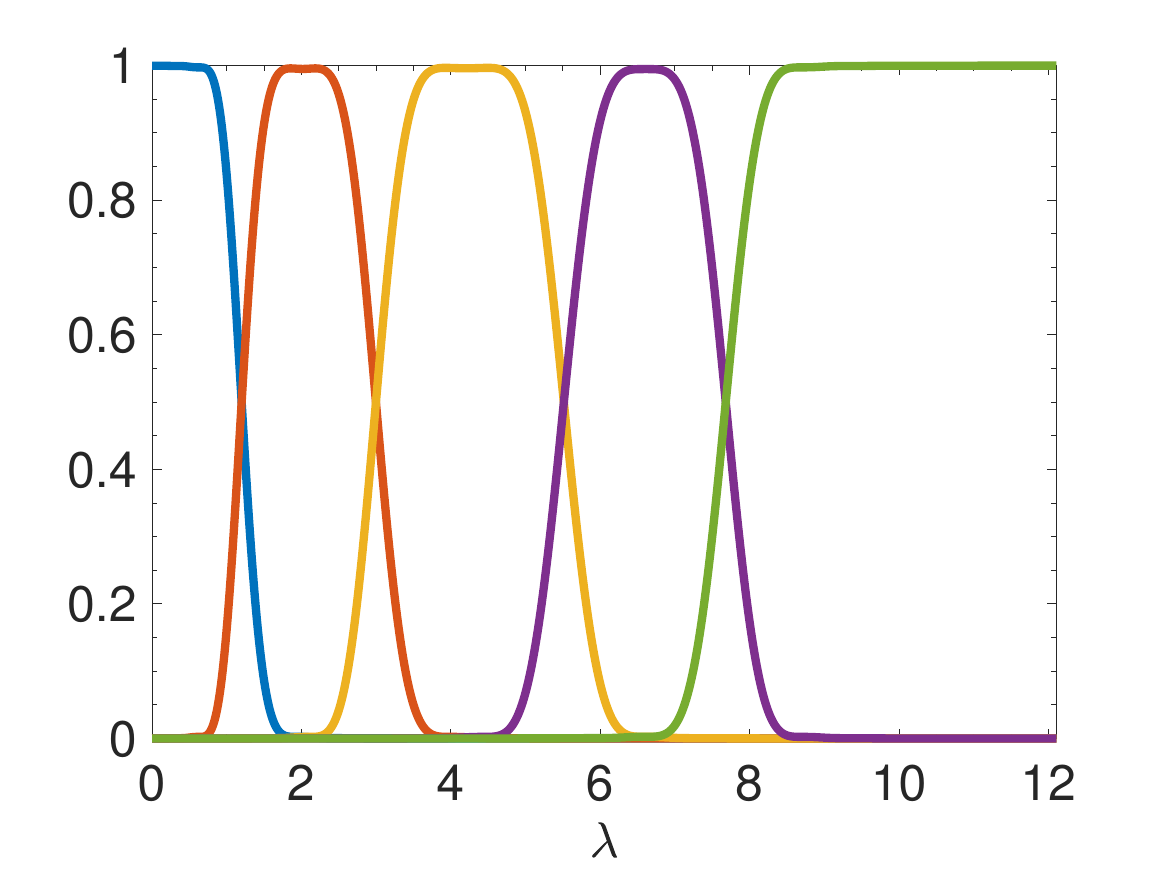}}
\centerline{~\small{(d)}}
\end{minipage} \medskip \\
\begin{minipage}[m]{0.24\linewidth}
\centerline{\includegraphics[width=.9\linewidth]{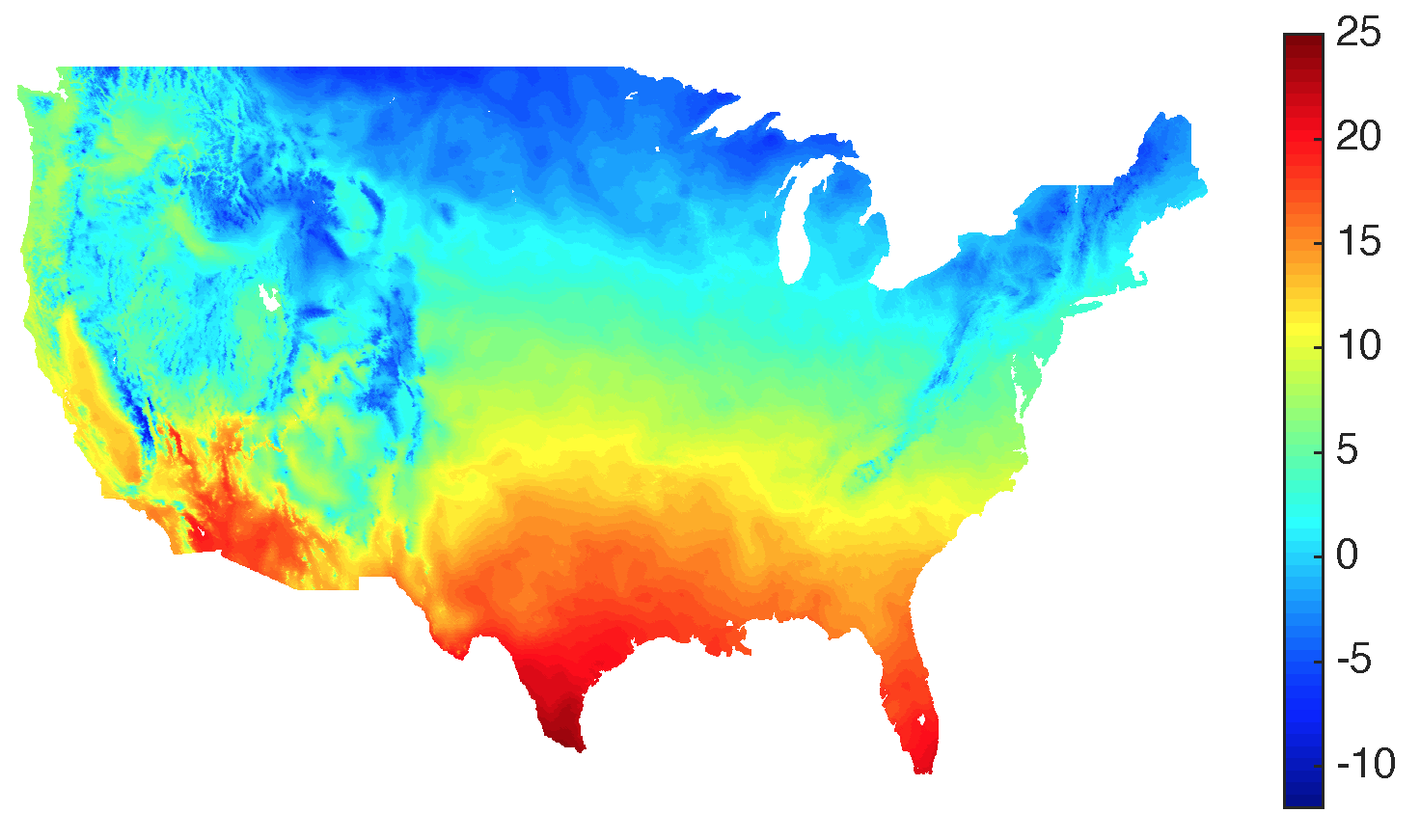}}
\centerline{\small{(e)}}
\end{minipage}
\begin{minipage}[m]{0.24\linewidth}
\centerline{\includegraphics[width=.9\linewidth]{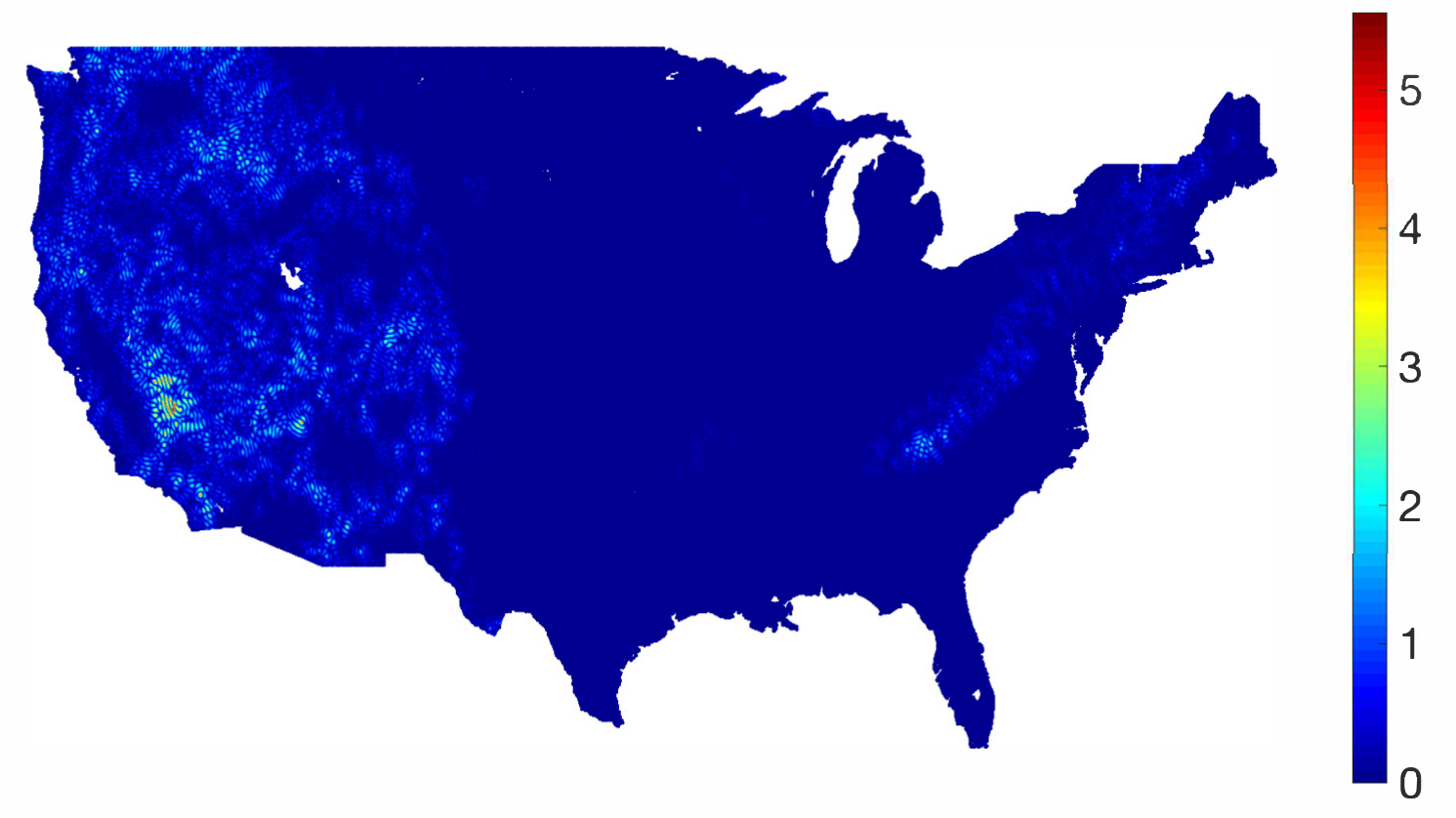}}
\centerline{\small{(f)}}
\end{minipage} 
\begin{minipage}[m]{0.24\linewidth}
\centerline{\includegraphics[width=.85\linewidth]{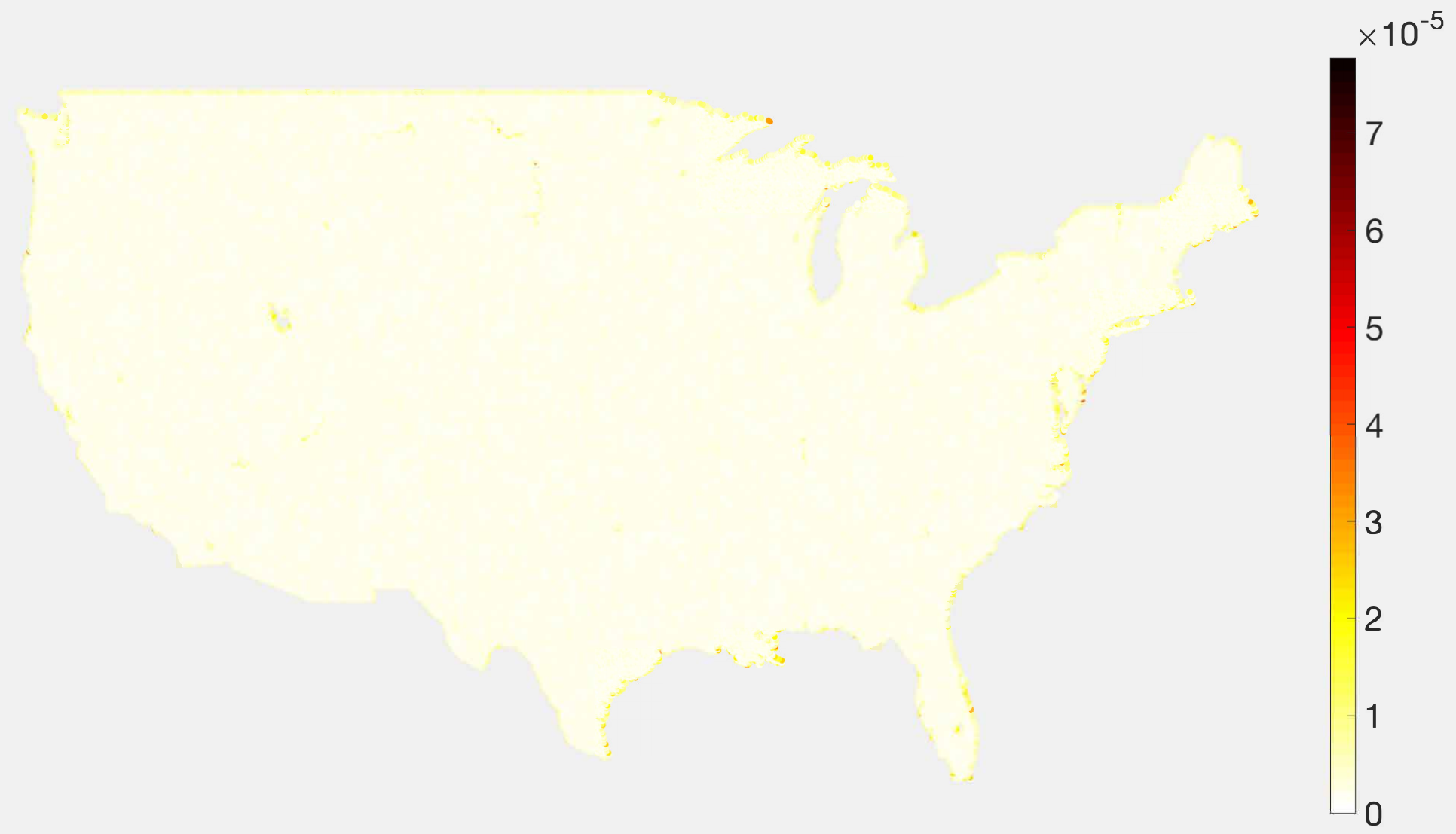}}
\centerline{\small{(g)}}
\end{minipage}
\begin{minipage}[m]{0.24\linewidth}
\centerline{\includegraphics[width=.84\linewidth]{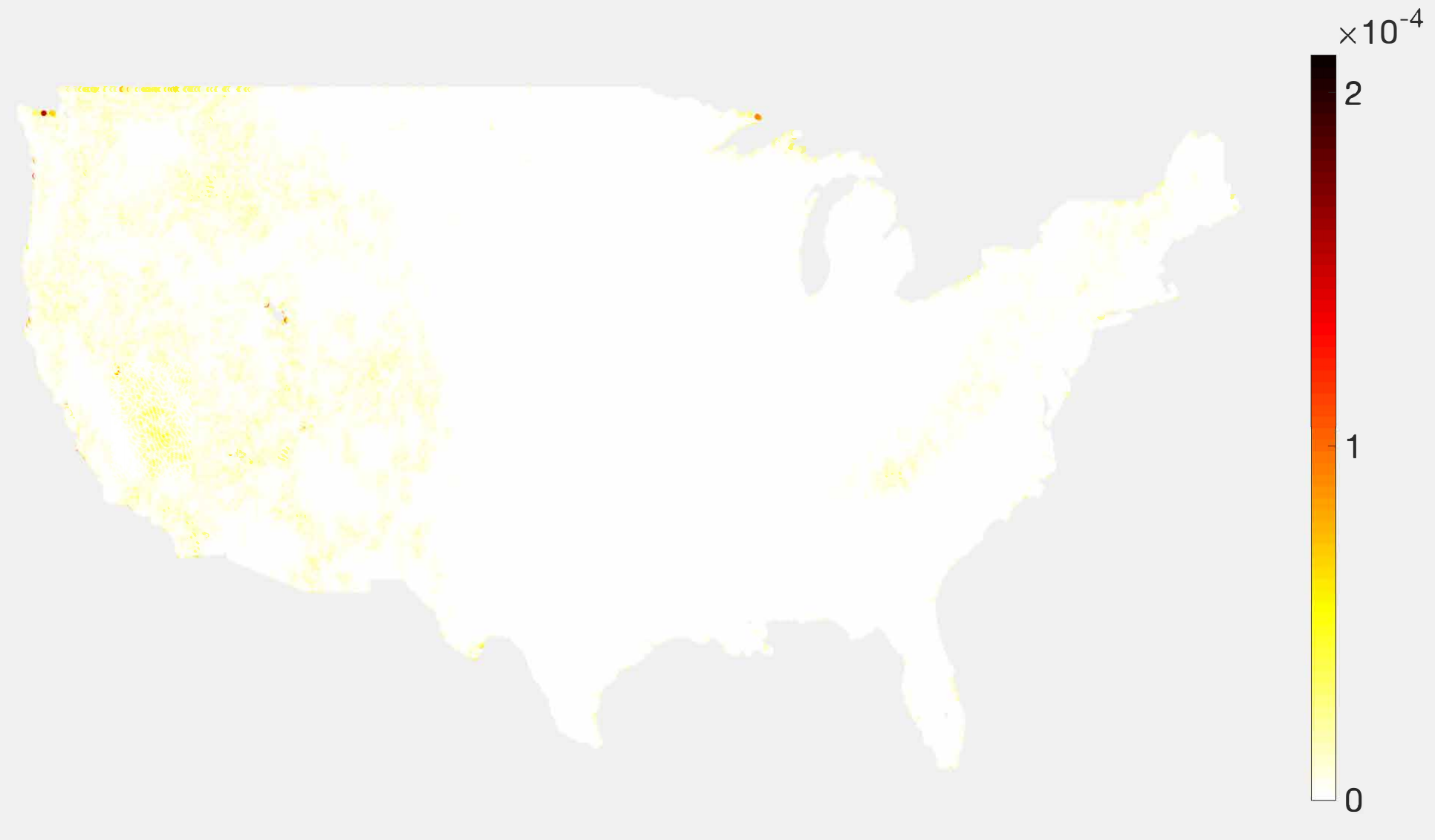}}
\centerline{\small{(h)}}
\end{minipage} \medskip  \\
\begin{minipage}[m]{0.24\linewidth}
\centerline{\includegraphics[width=.84\linewidth]{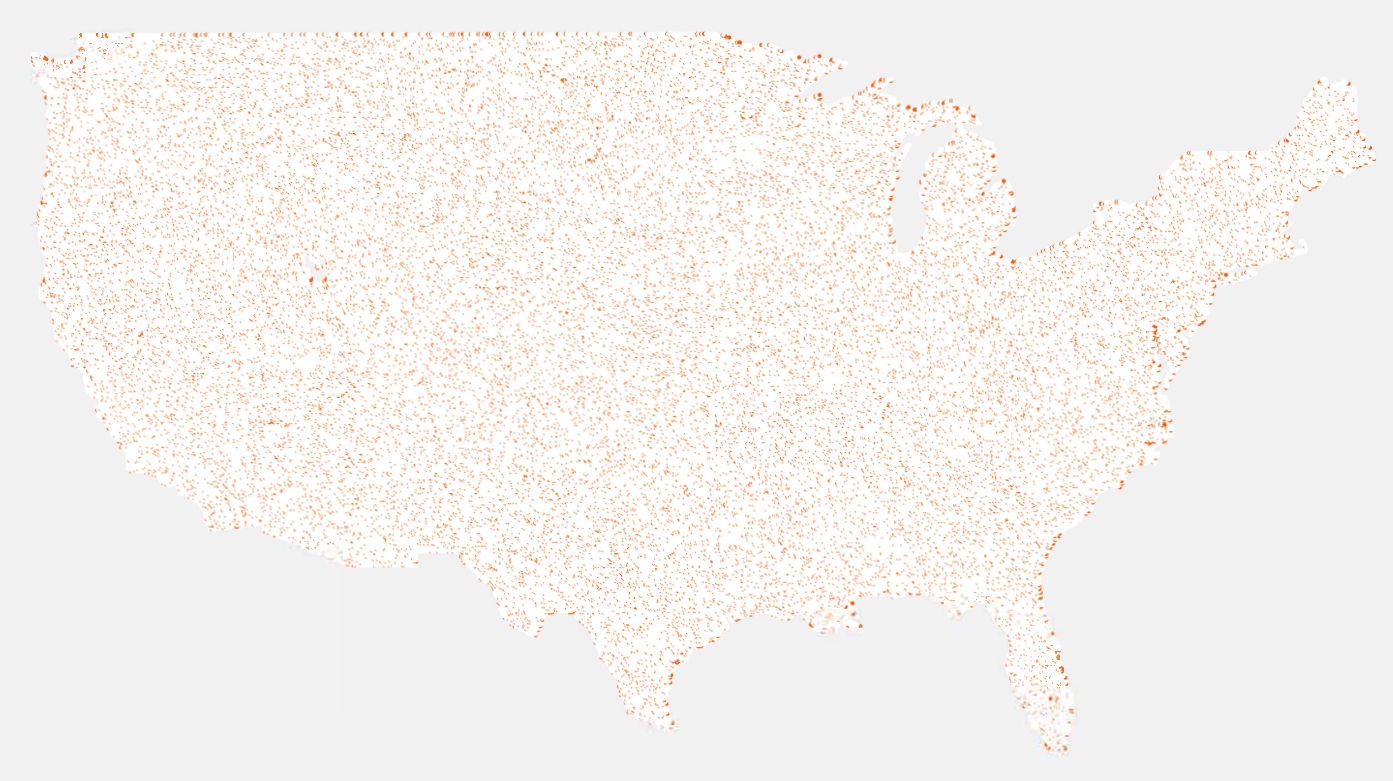}}
\centerline{\small{(i)}}
\end{minipage}
\begin{minipage}[m]{0.24\linewidth}
\centerline{\includegraphics[width=.84\linewidth]{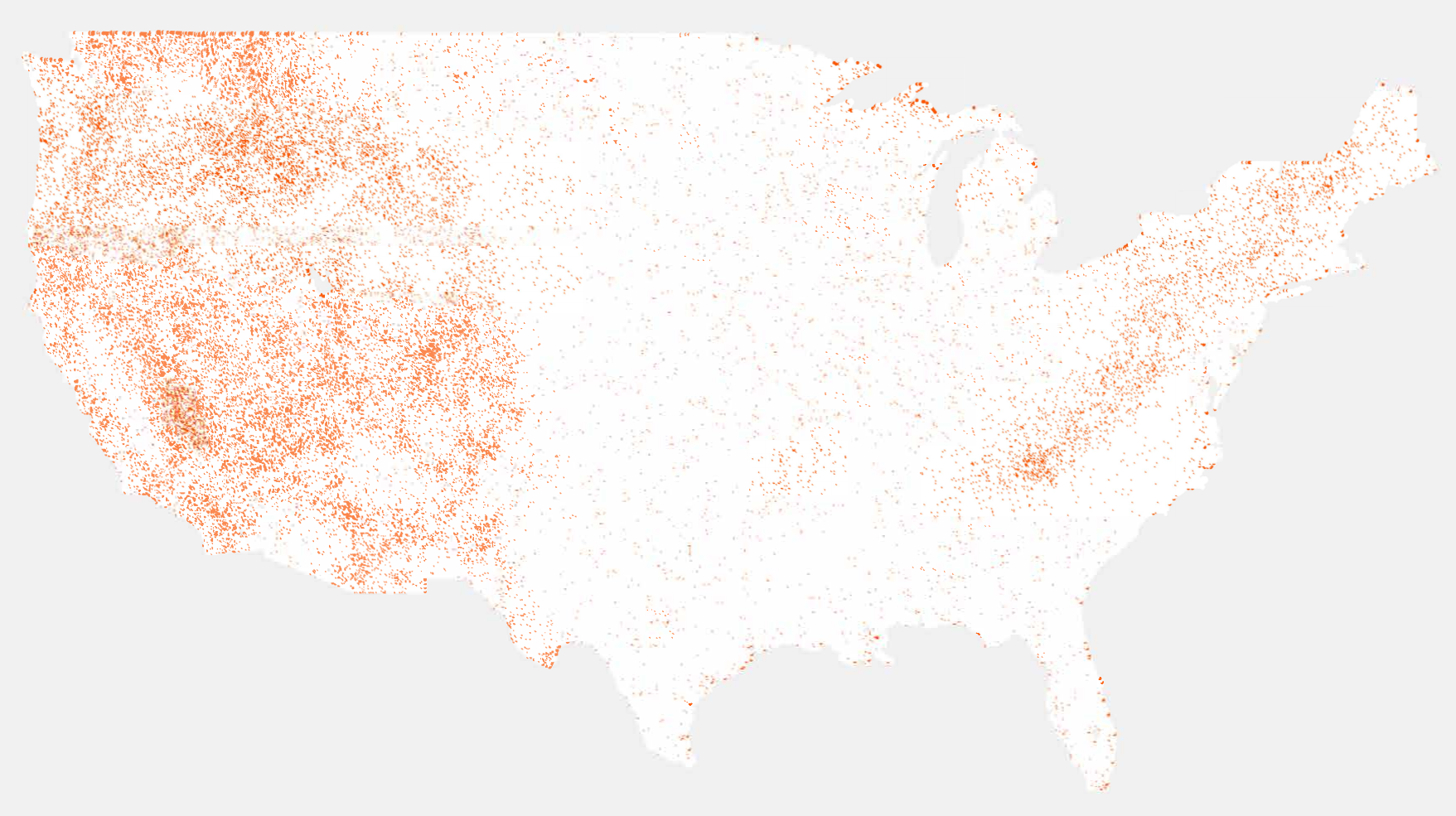}}
\centerline{\small{(j)}}
\end{minipage}
\begin{minipage}[m]{0.24\linewidth}
\centerline{\includegraphics[width=.9\linewidth]{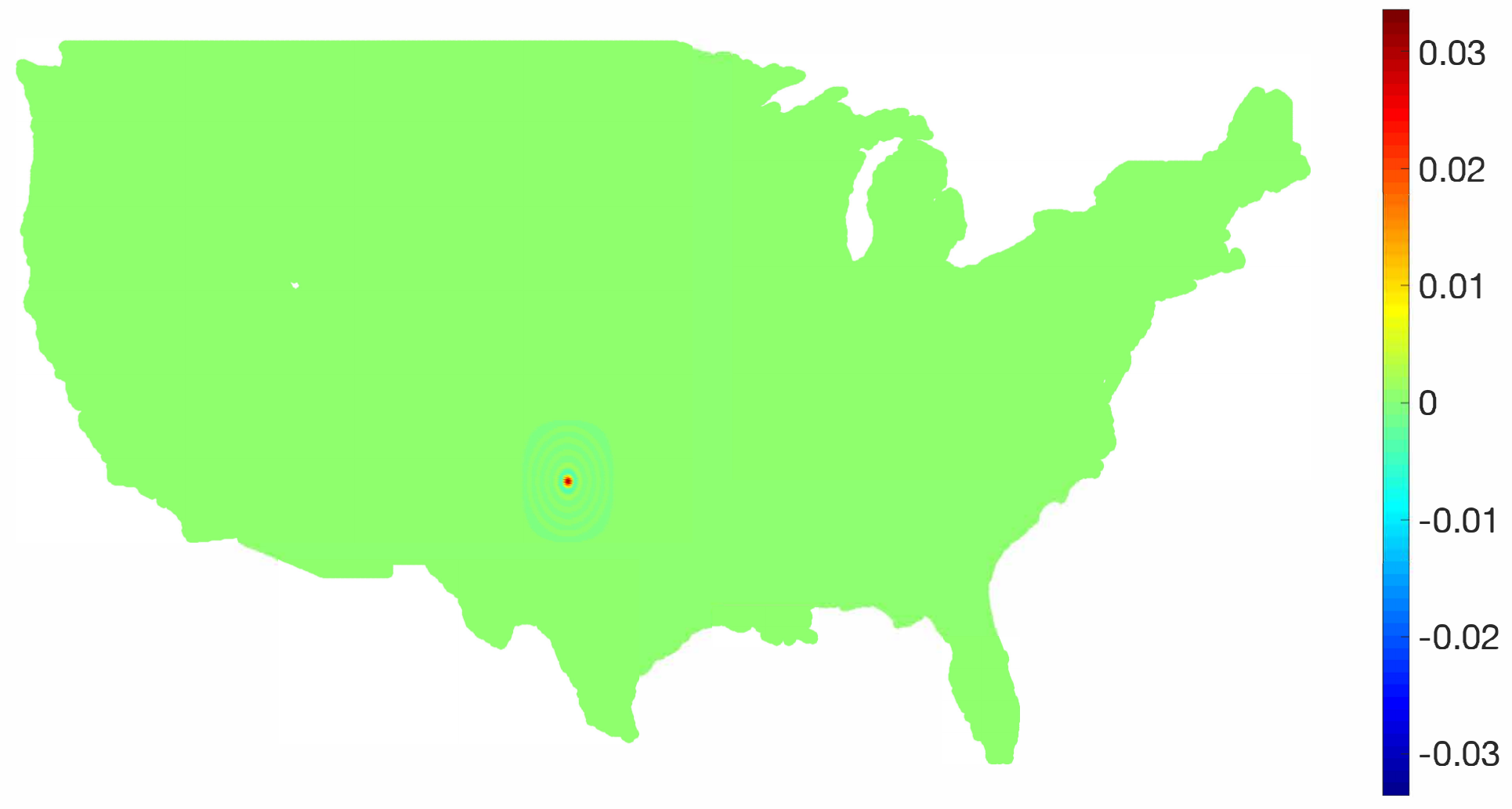}}
\centerline{\small{(k)}} 
\end{minipage} 
\begin{minipage}[m]{0.24\linewidth}
\centerline{\includegraphics[width=.9\linewidth]{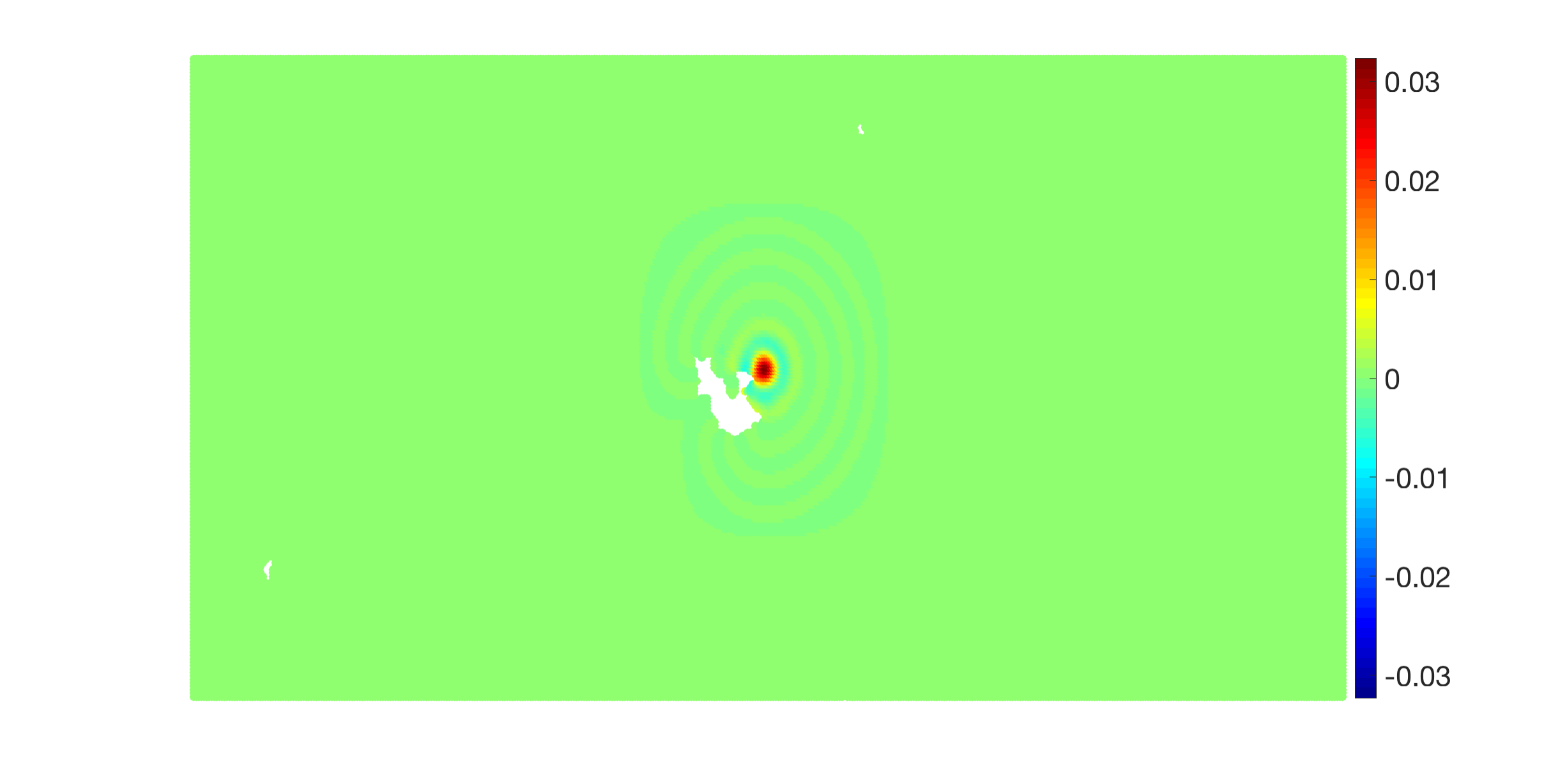}}
\centerline{\small{(l)}}
\end{minipage} \medskip \\
\begin{minipage}[m]{0.24\linewidth}
\centerline{\includegraphics[width=.9\linewidth]{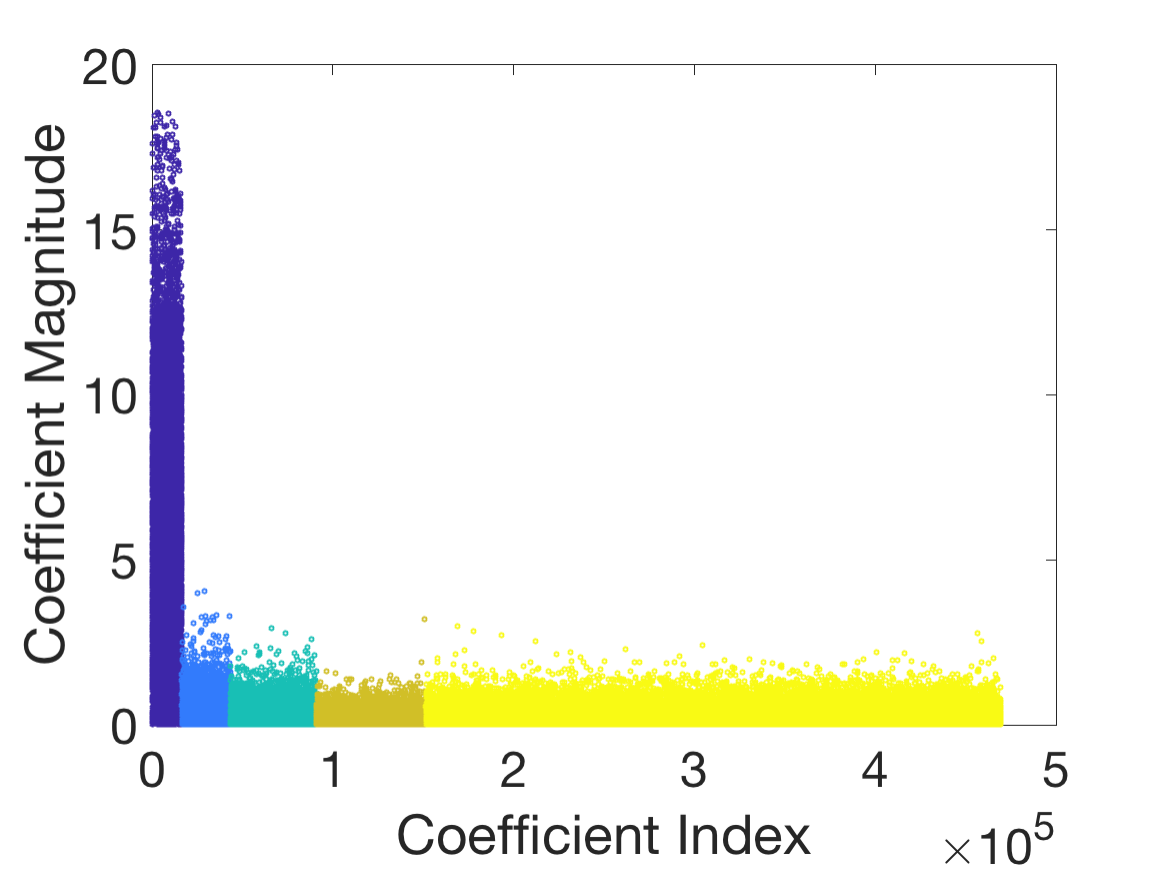}}
\centerline{\small{(m)}}
\end{minipage}
\begin{minipage}[m]{0.24\linewidth}
\centerline{\includegraphics[width=.9\linewidth]{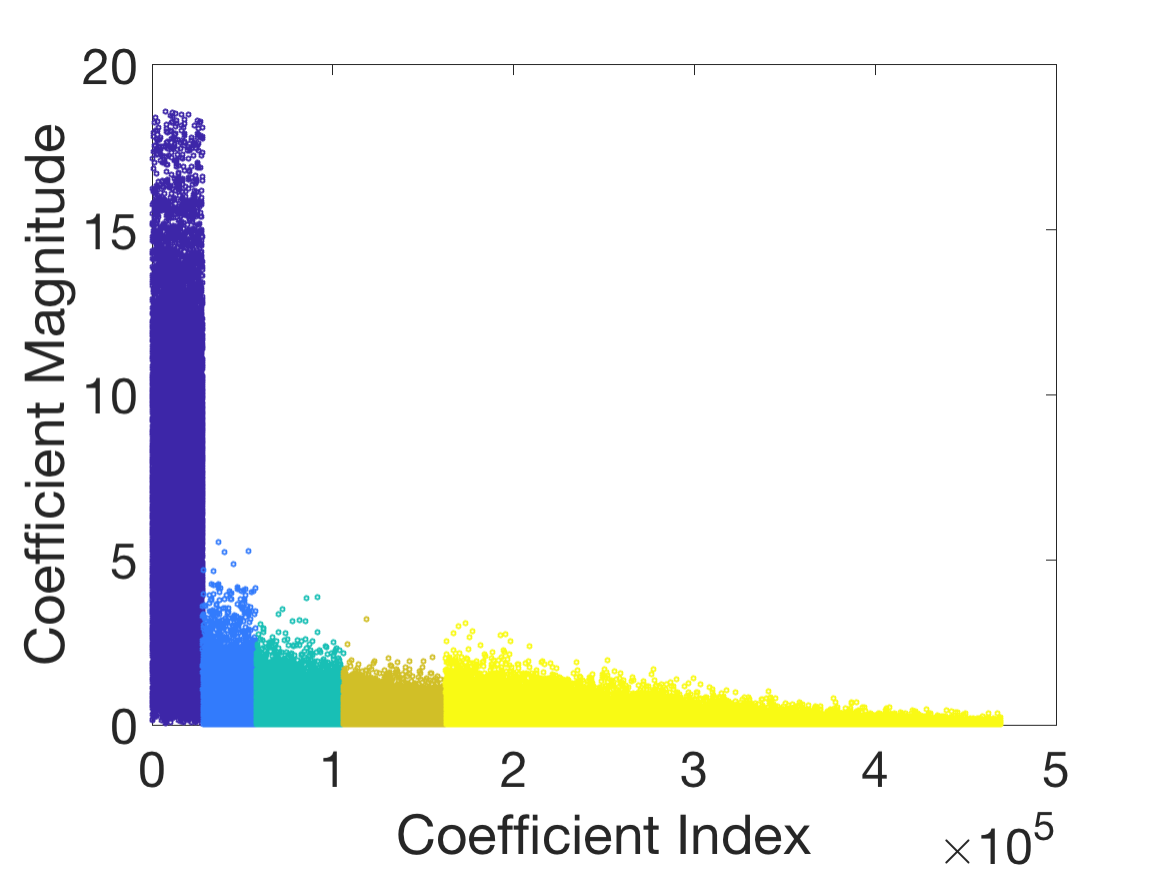}}
\centerline{\small{(n)}}
\end{minipage}
\begin{minipage}[m]{0.24\linewidth}
\centerline{\includegraphics[width=.85\linewidth]{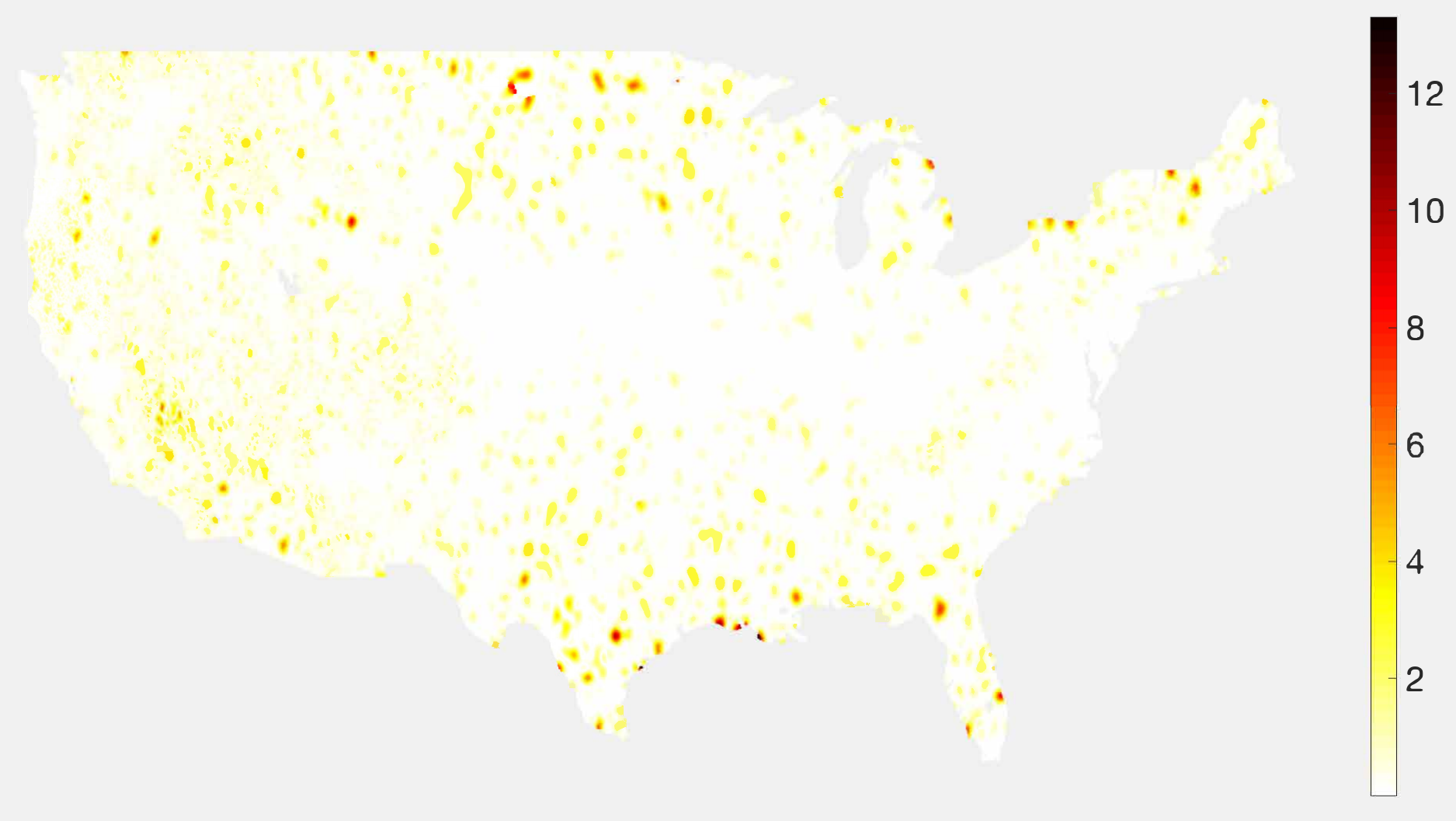}}
\centerline{\small{(o)}}
\end{minipage}
\begin{minipage}[m]{0.24\linewidth}
\centerline{\includegraphics[width=.85\linewidth]{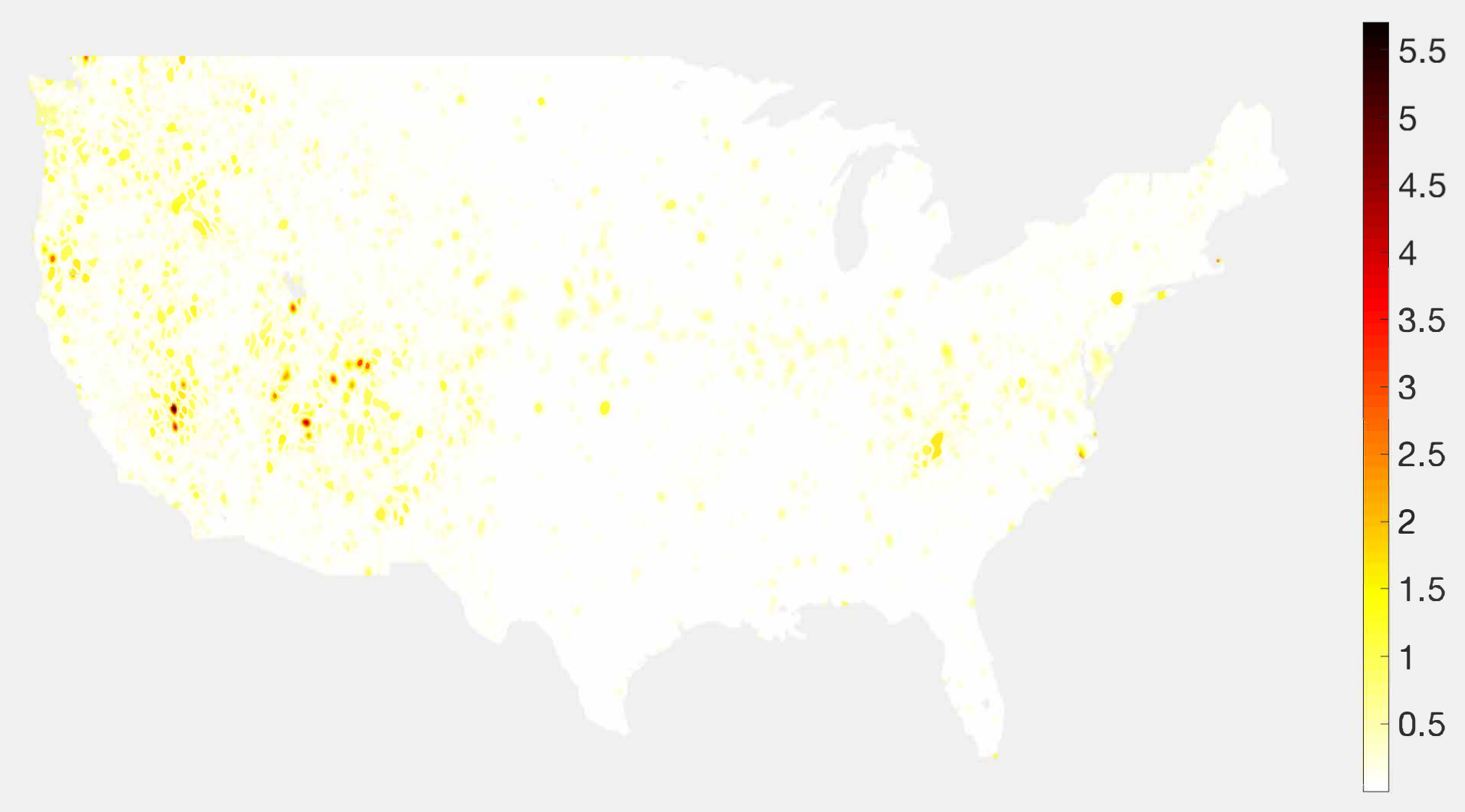}}
\centerline{\small{(p)}}
\end{minipage}
\caption{Fast $M$-channel filter bank example. (a)-(b) The eight-neighbor local graph structure for the Eastern Massachusetts and Boston regions. 
(c) Estimate of the cumulative spectral distribution function of the graph. (d) Approximate filter bank with five bands and degree 50 polynomial approximations. (e) Average temperatures for March 2018, measured at 469,404 locations. (f) Magnitudes of the filtered signal on the second band, $|\tilde{h}_2(\L){\bf f}|$. (g)-(h) Weights of the non-uniform sampling distribution for the second band, without and with adaptation to the filtered signal in (f). (i) The 27,021 vertices randomly selected according to the distribution in (g) to be included in $\V_2$. (j) The 29,516 vertices randomly selected according to the signal-adapted distribution in (h) to be included in $\V_2$. (k)-(l) Example scaling functions, with the latter zoomed in to see the effect of the missing measurements over Great Salt Lake. (m)-(n) 
Magnitudes of all 469,404 fast M-CSFB analysis coefficients for the 
 signal in (e), colored by band, for the non-adapted and signal-adapted transforms, respectively. (o)-(p) Absolute values of the differences between the reconstructions and the original 
signal, 
for the non-adapted and adapted transforms, respectively. \vspace{-.1in}
} \label{Fig:temperature}
\end{figure*}

\subsection{Illustrative Example: Temperatures} \label{Se:temp}

The last column of Table  \ref{Ta:comp_times} is based on  
the average temperatures for March 2018, taken from the Gridded 5km GHCN-Daily Temperature and Precipitation Dataset (nClimGrid) \cite{nClimGrid1,nClimGrid2} of the United States National Oceanic and Atmospheric Administration (NOAA). The measurements are  
on a grid with spacing of $\frac{1}{24}$ of a degree for both latitude and longitude. We form an unweighted graph by connecting each measurement location to its eight neighbors on the grid, if they contain measurements, as shown in Fig. \ref{Fig:temperature}(a)-(b). We 
 eliminate isolated vertices and small components (e.g., islands), yielding a connected graph with 469,404 vertices (measurement locations). Fig. \ref{Fig:temperature}(c)-(e) show the estimated 
  spectral distribution, 
the 
approximate filter bank with five bands and degree 50 polynomial approximations, and the temperatures. 

An intuitive explanation of the distribution for the second band, 
 shown in Fig. \ref{Fig:temperature}(g), is that we want to sample with higher probability near the edges, as there are fewer neighbors from whom to interpolate the local average value. The 27,021 vertices selected for $\V_2$, shown in Fig. \ref{Fig:temperature}(i), are spread across the graph, with a higher density near the boundaries. The corresponding non-uniform sampling distribution and realization for the second band in the signal-adapted case are shown in Fig. \ref{Fig:temperature}(h) and Fig. \ref{Fig:temperature}(j). The non-zero analysis coefficients of the bandpass filters tend to coincide with topographical changes, as shown Fig. \ref{Fig:temperature}(f). Accordingly, the signal-adapted transform selects more samples from these regions.

Two scaling functions for the fast M-CSFB are shown in Fig. \ref{Fig:temperature}(k)-(l), with the image for the latter one zoomed in near the Great Salt Lake. Because $K=50$, the atoms are localized within 50 hops of the center vertex. When the center is in a region where each vertex has eight neighbors, the atom is symmetric, resembling the scaling function of a 2D-wavelet transform. When the center is near an edge, like the one shown in Fig. \ref{Fig:temperature}(l), the atom adapts to the shape of the underlying graph. While 
the signal adaptation changes the distribution of center locations, it does not change the shape of the atoms; that is, if a vertex $i$ is chosen as a center location for a given scale in both the non-adapted and signal-adapted versions of the transform, the corresponding atoms are identical. 

The magnitudes of the analysis coefficients, shown in Fig. \ref{Fig:temperature}(m)-(n), decay quickly, 
as the signal is generally smooth with some discontinuities that tend to coincide with topographical changes. 
The signal-adapted transform allocates more samples to the bands on the lower end of the spectrum (57,539 total for the first two bands, as opposed to 43,384 in the non-adapted case).
Fig. \ref{Fig:temperature}(o)-(p) show the reconstruction errors for the two versions of the fast $M$-CSFB transform, with the same parameters used in Scenario B above. The MSE of the signal-adapted transform is 0.0479, as opposed to 0.4994 without the signal adaptation. The first driver of this reduction is the allocation of additional samples to the first band, where 0.4586 of the 0.4994 MSE is incurred when not adapted to the signal. 
Additionally, more of this band's samples are taken in the upper and lower thirds of the country, where the energy of the filtered signal is concentrated. 

\subsection{Parameter Choices and Tradeoffs}

For the next set of numerical experiments, we apply a 4-band filter bank to the piecewise smooth signal 
from Fig. \ref{Fig:bunny_signal}(a), and consider the outputs of the first (lowpass) and third (bandpass) filters, shown in the the top two rows of Fig. \ref{Fig:samp_tradeoff}. We start by examining two key design choices for random sampling and interpolation of these two 
filtered signals: the sampling distribution and the number of samples. Unless otherwise noted, we use the settings of Scenario B above. 

\begin{figure}[H] 
\begin{minipage}[m]{0.16\linewidth}
~
\end{minipage}
\begin{minipage}[m]{0.4\linewidth}
\centerline{\small{Lowpass}}
\end{minipage}
\begin{minipage}[m]{0.4\linewidth}
\centerline{\small{Bandpass}}
\end{minipage} \\
\begin{minipage}[m]{0.16\linewidth}
\centerline{\small{Filtered}}
\centerline{\small{Signal}}
\centerline{\small{(Vertex)}}
\end{minipage}
\begin{minipage}[m]{0.4\linewidth}
\centerline{\includegraphics[width=.85\linewidth]{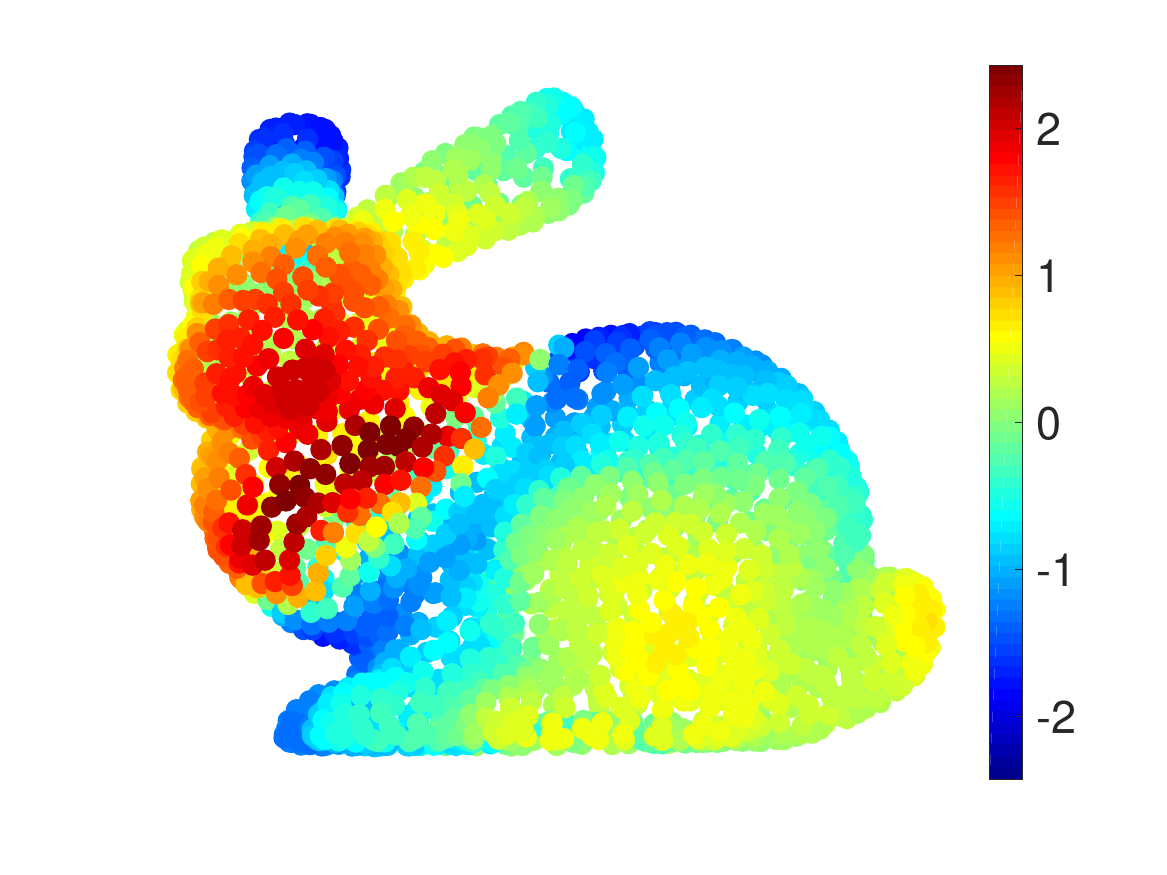}}
\end{minipage}
\begin{minipage}[m]{0.4\linewidth}
\centerline{\includegraphics[width=.85\linewidth]{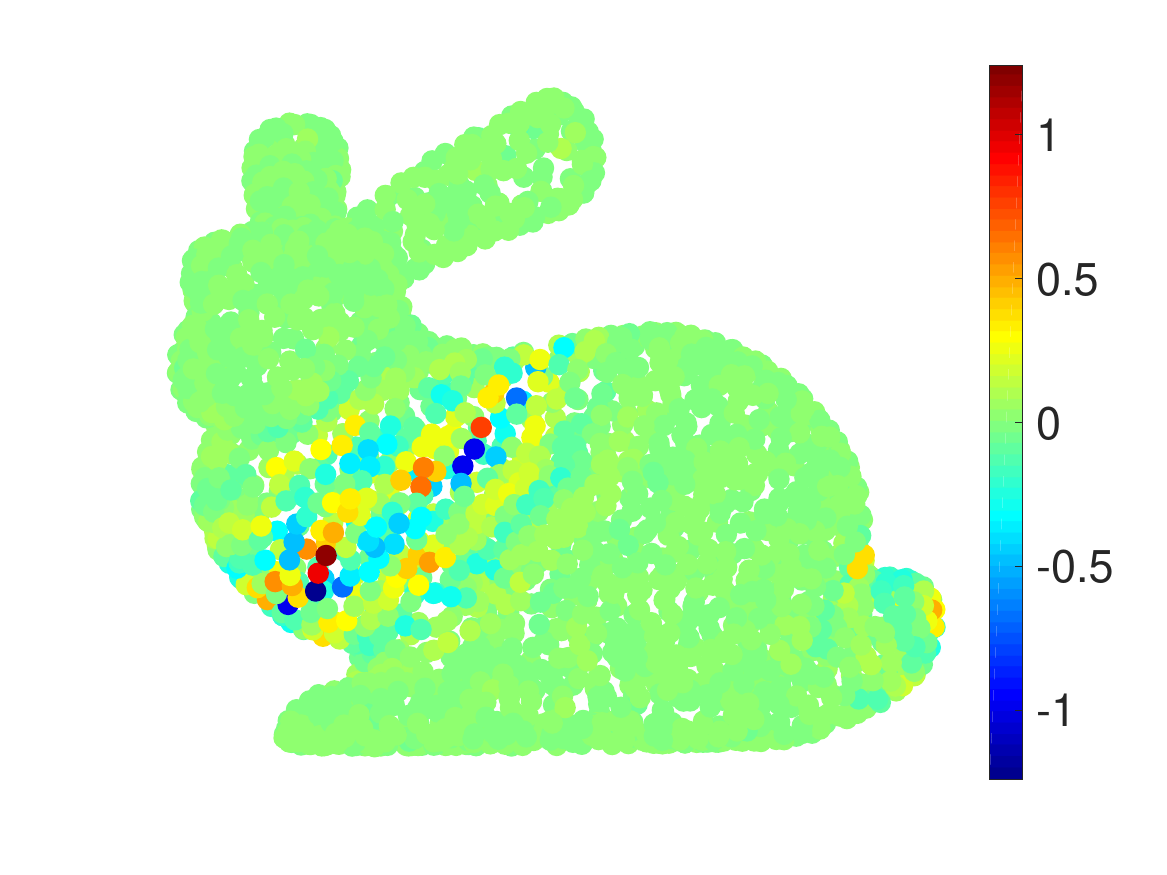}}
\end{minipage} \\
\begin{minipage}[m]{0.16\linewidth}
\centerline{\small{Filtered}}
\centerline{\small{Signal}}
\centerline{\small{(Spectral)}}
\end{minipage}
\begin{minipage}[m]{0.4\linewidth}
\centerline{\includegraphics[width=.85\linewidth]{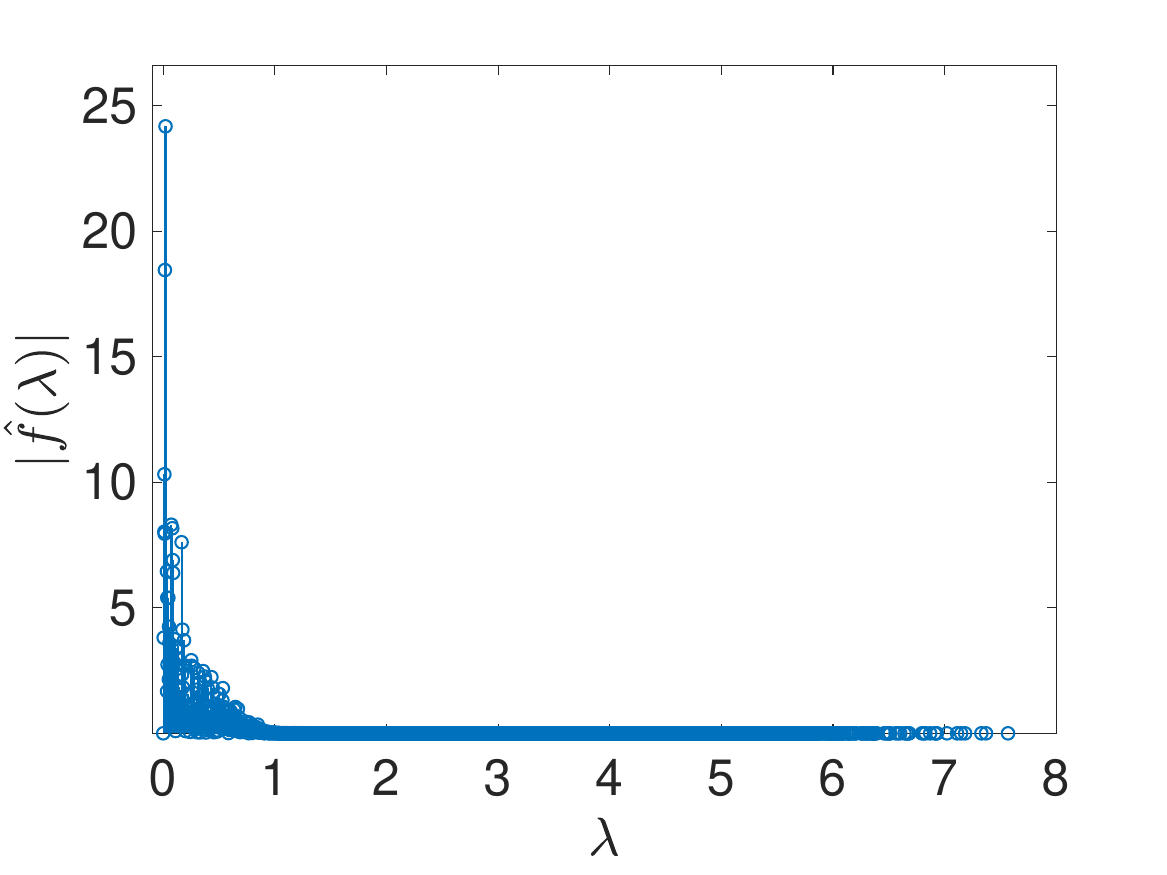}}
\end{minipage}
\begin{minipage}[m]{0.4\linewidth}
\centerline{\includegraphics[width=.85\linewidth]{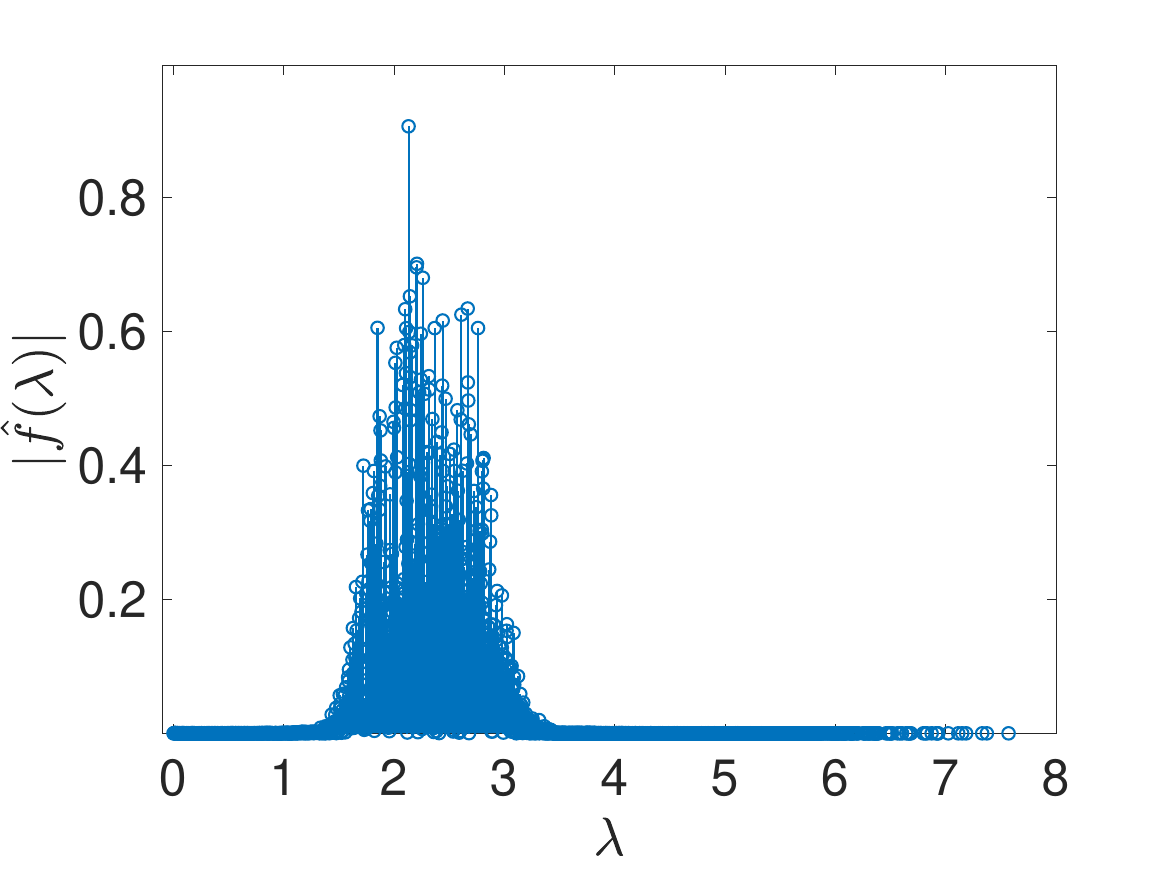}}
\end{minipage} \\ 
\begin{minipage}[m]{0.16\linewidth}
\centerline{\small{Sampling}}
\centerline{\small{Weights}}
\end{minipage}
\begin{minipage}[m]{0.4\linewidth}
\centerline{\includegraphics[width=.85\linewidth]{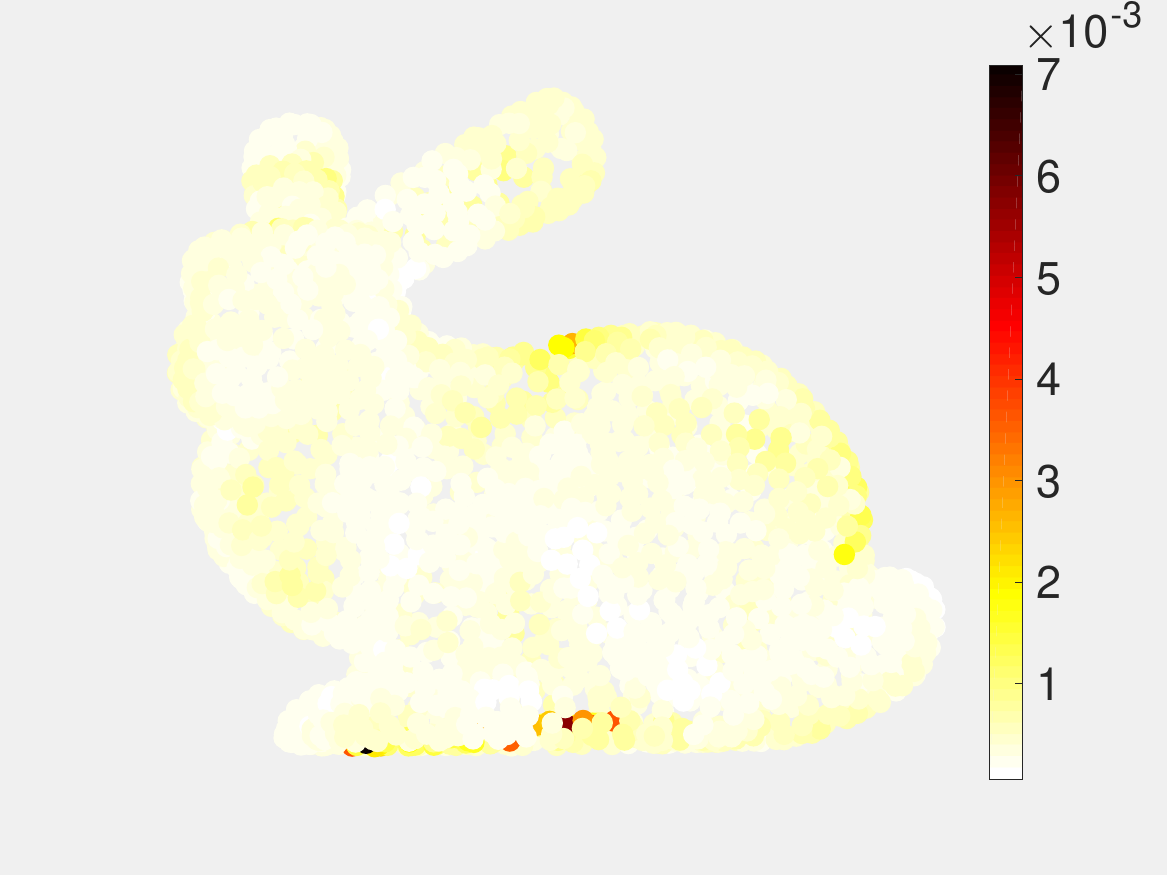}}
\end{minipage}
\begin{minipage}[m]{0.4\linewidth}
\centerline{\includegraphics[width=.85\linewidth]{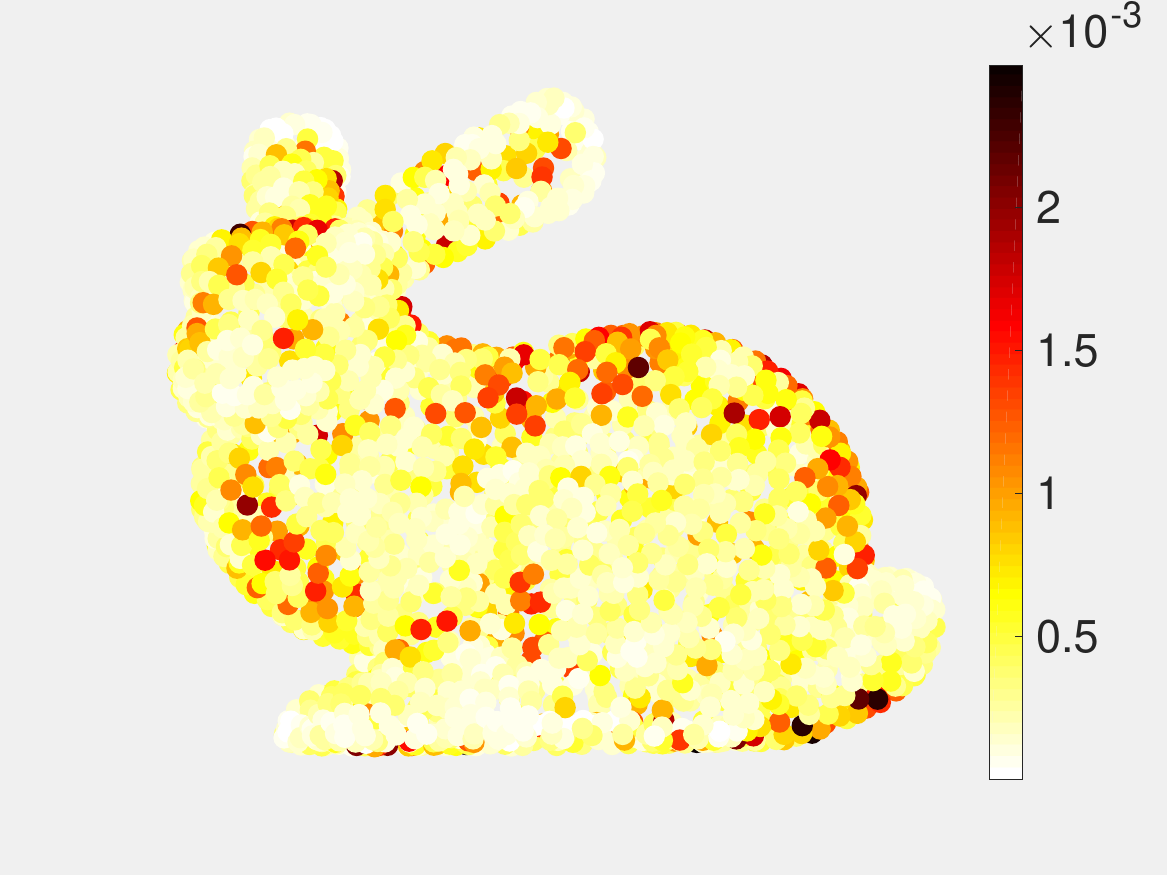}}
\end{minipage} \\
\begin{minipage}[m]{0.16\linewidth}
\centerline{\small{Sampling}}
\centerline{\small{Weights}}
\centerline{\small{(Adapted)}}
\end{minipage}
\begin{minipage}[m]{0.4\linewidth}
\centerline{\includegraphics[width=.85\linewidth]{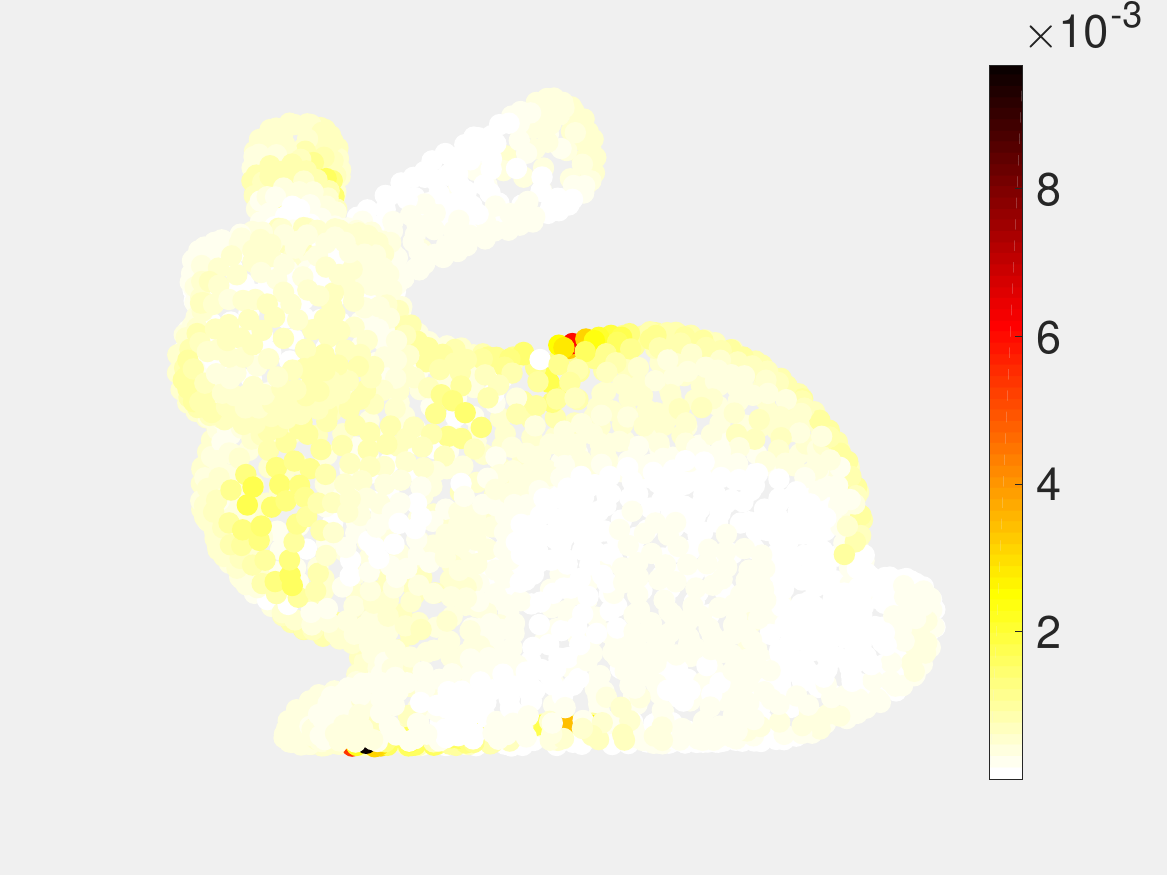}}
\end{minipage}
\begin{minipage}[m]{0.4\linewidth}
\centerline{\includegraphics[width=.85\linewidth]{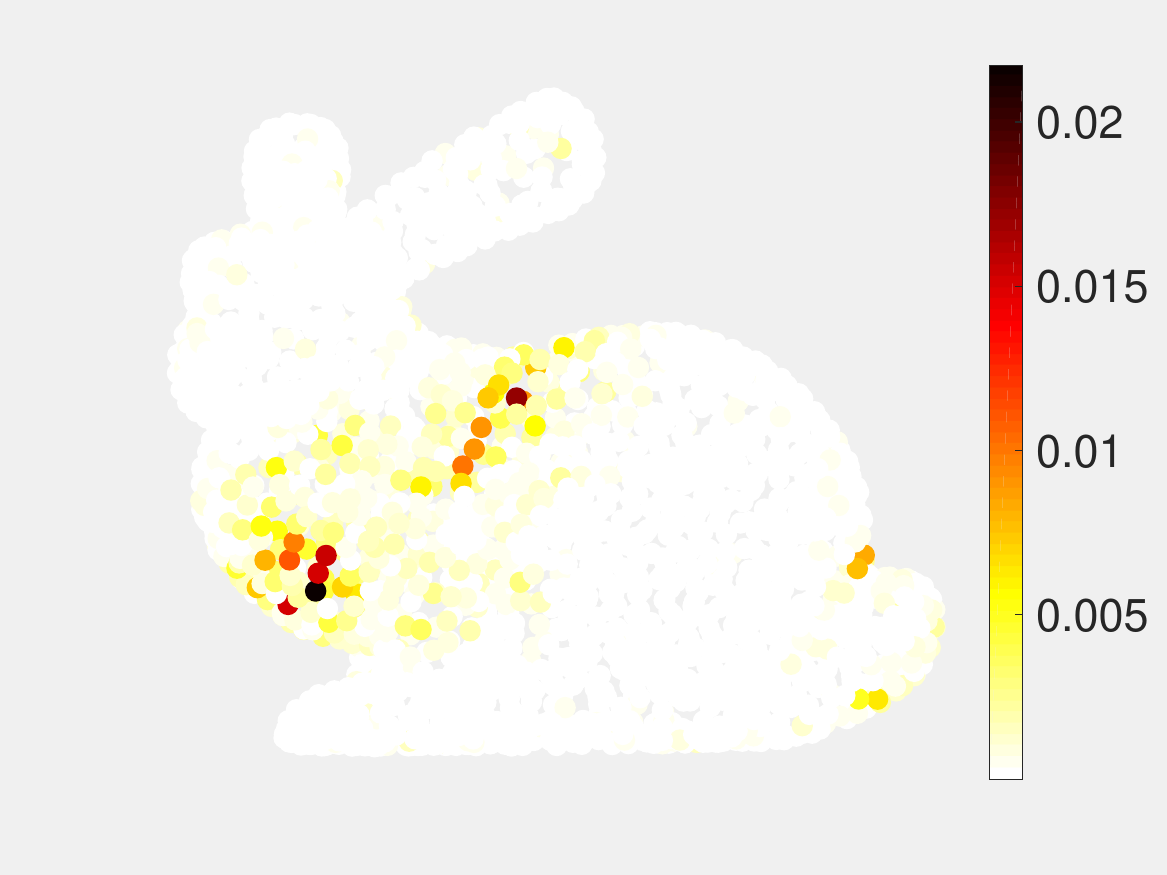}}
\end{minipage} \\
\begin{minipage}[m]{0.16\linewidth}
\centerline{\small{Single}}
\centerline{\small{Sampling Set}}
\centerline{\small{Realization}}
\end{minipage}
\begin{minipage}[m]{0.4\linewidth}
\centerline{\includegraphics[width=.85\linewidth]{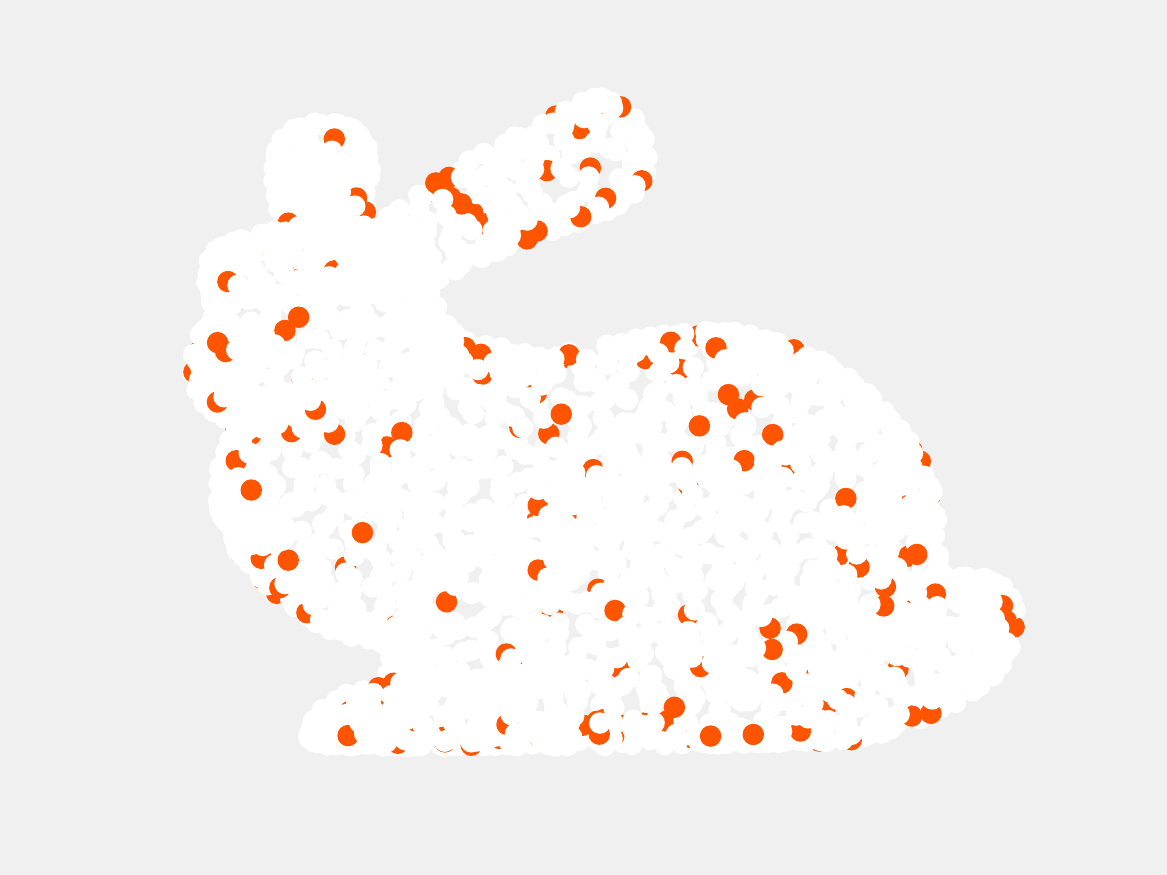}}
\end{minipage}
\begin{minipage}[m]{0.4\linewidth}
\centerline{\includegraphics[width=.85\linewidth]{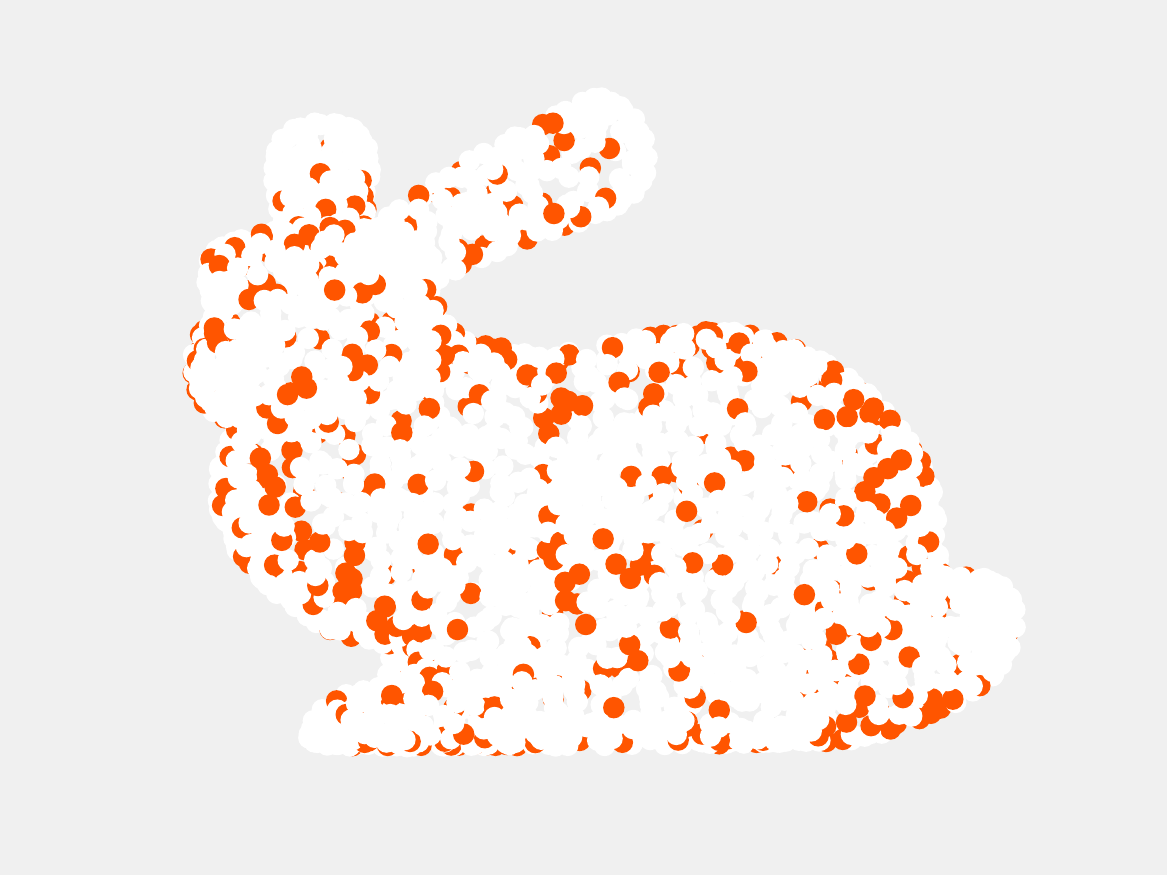}}
\end{minipage} \\
\begin{minipage}[m]{0.16\linewidth}
\centerline{\small{Single}}
\centerline{\small{Sampling Set}}
\centerline{\small{Realization}}
\centerline{\small{(Adapted}}
\centerline{\small{Weights)}}
\end{minipage}
\begin{minipage}[m]{0.4\linewidth}
\centerline{\includegraphics[width=.85\linewidth]{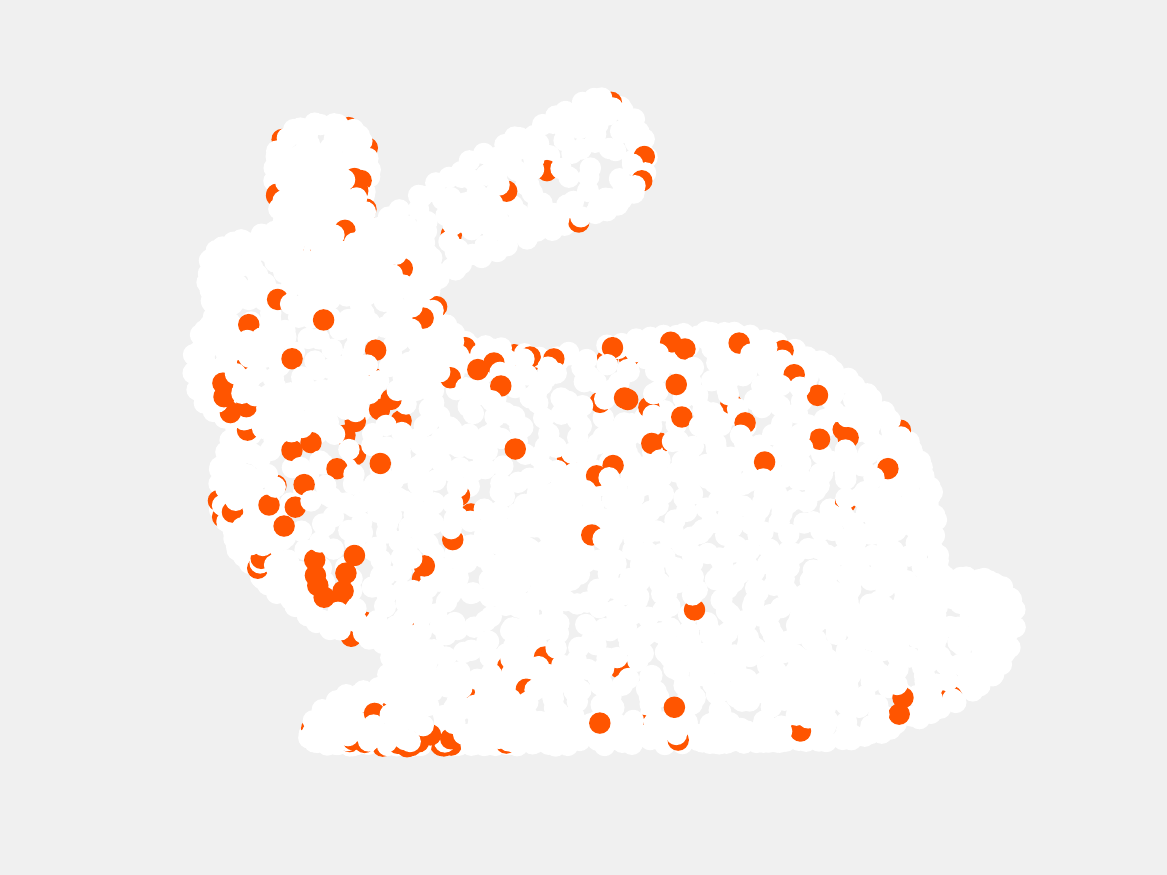}}
\end{minipage}
\begin{minipage}[m]{0.4\linewidth}
\centerline{\includegraphics[width=.85\linewidth]{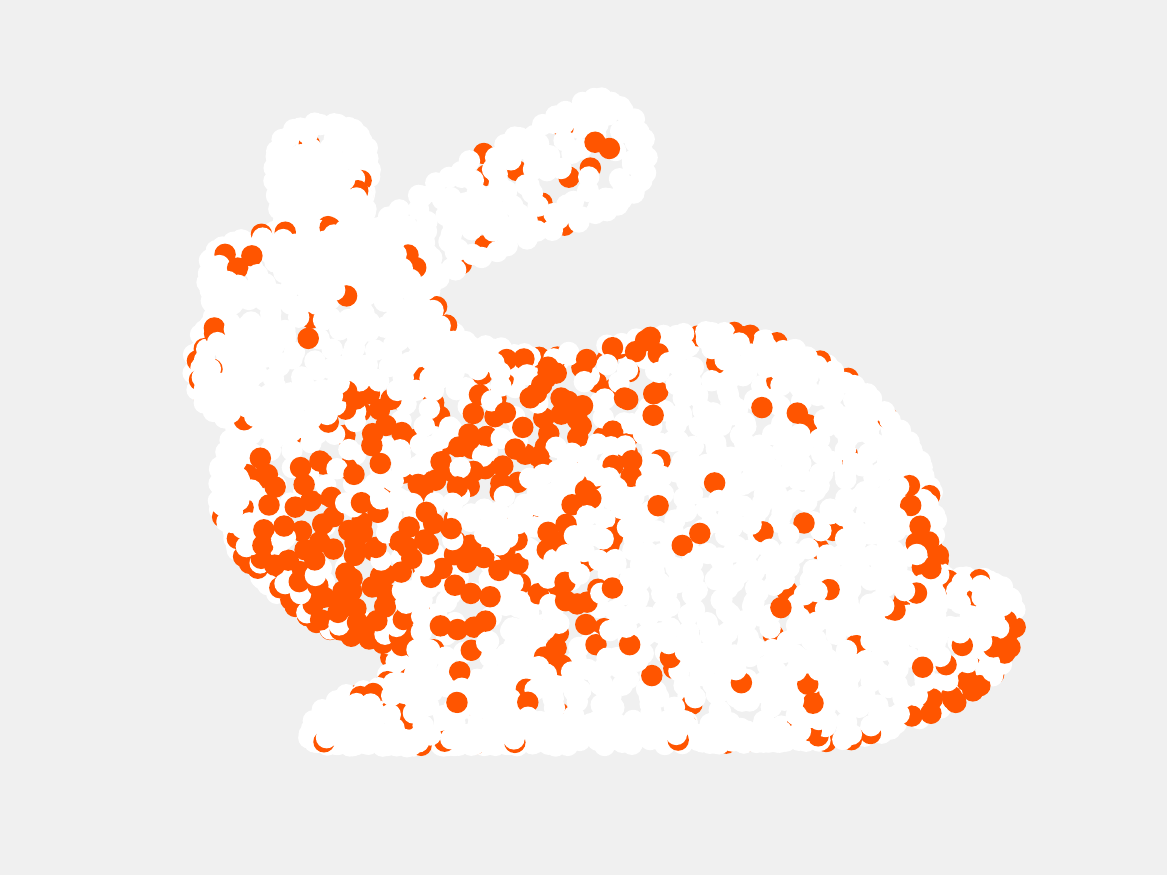}}
\end{minipage} \\
\begin{minipage}[m]{0.16\linewidth}
\centerline{\small{Average}}
\centerline{\small{Error}}
\end{minipage}
\begin{minipage}[m]{0.4\linewidth}
\centerline{\includegraphics[width=.85\linewidth]{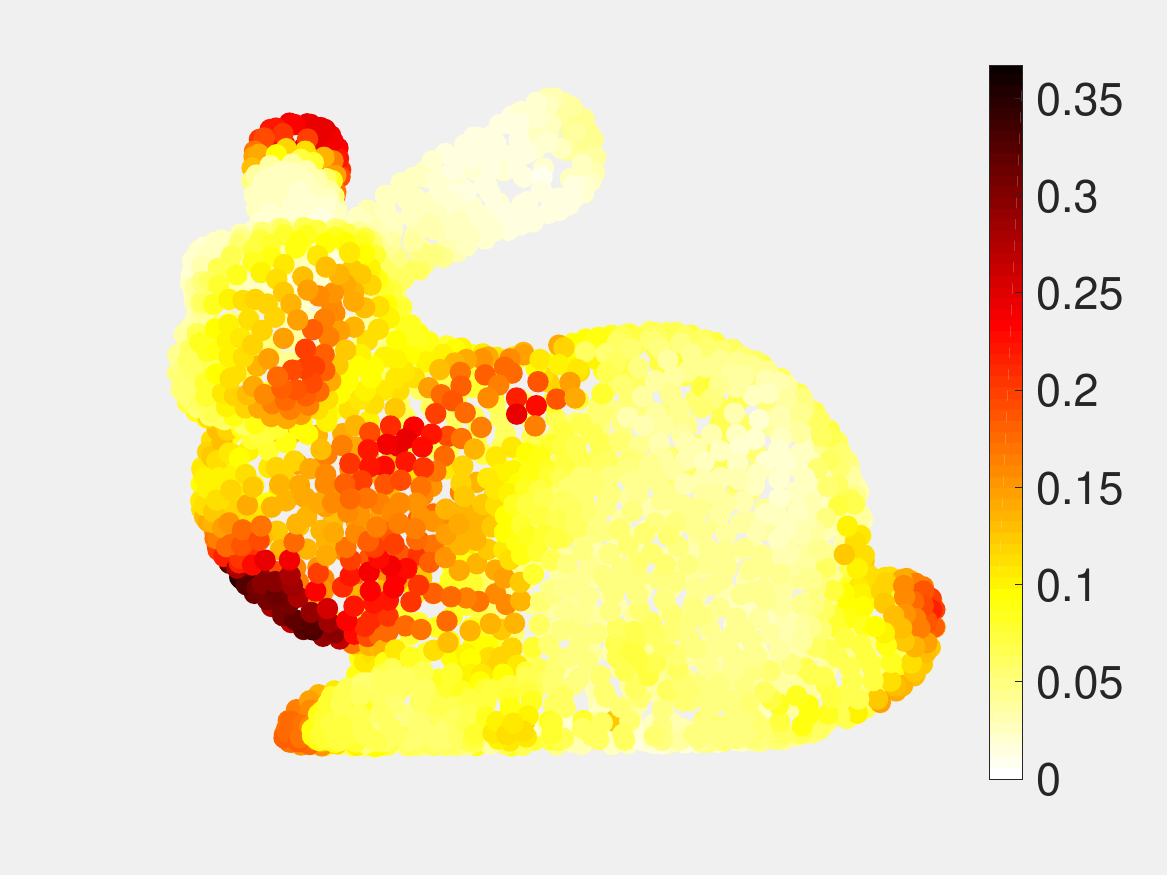}}
\end{minipage}
\begin{minipage}[m]{0.4\linewidth}
\centerline{\includegraphics[width=.85\linewidth]{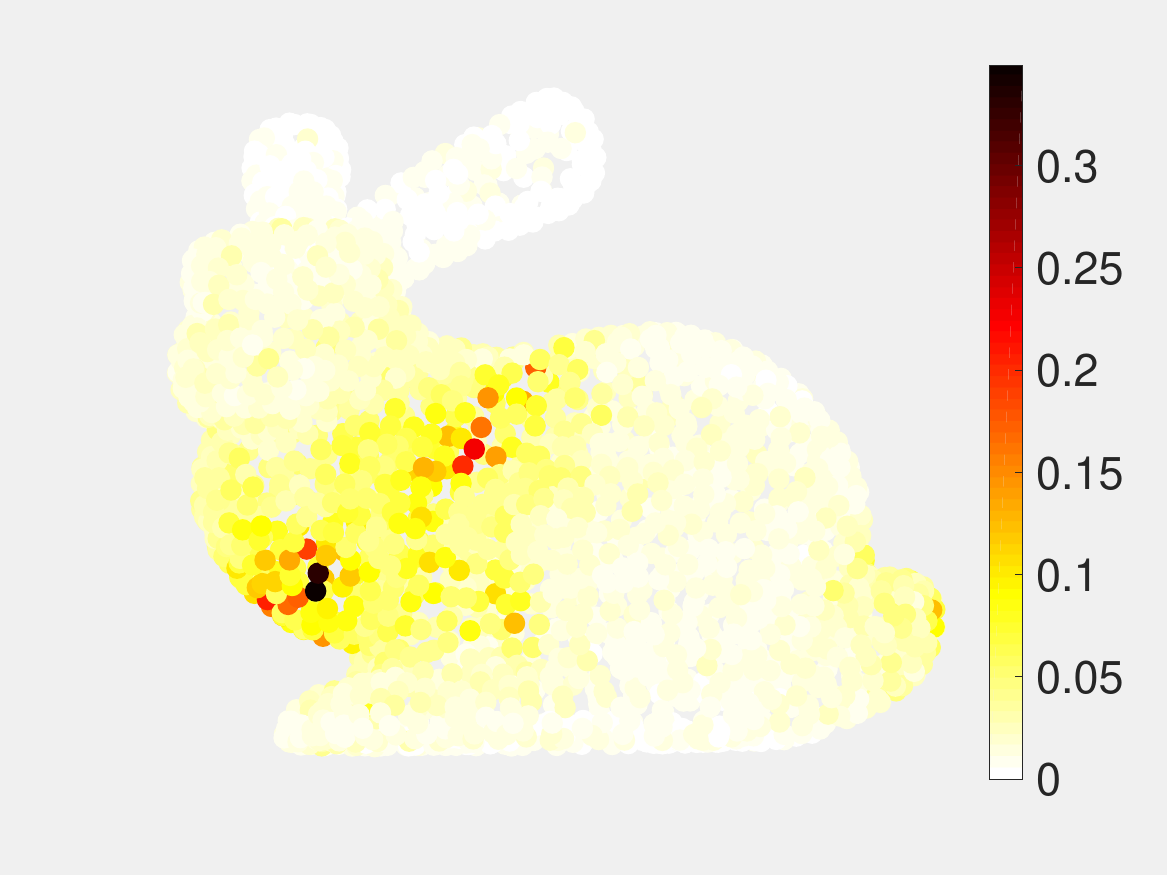}}
\end{minipage} \\
\begin{minipage}[m]{0.16\linewidth}
\centerline{\small{Average}}
\centerline{\small{Error}}
\centerline{\small{(Adapted}}
\centerline{\small{Weights)}}
\end{minipage}
\begin{minipage}[m]{0.4\linewidth}
\centerline{\includegraphics[width=.85\linewidth]{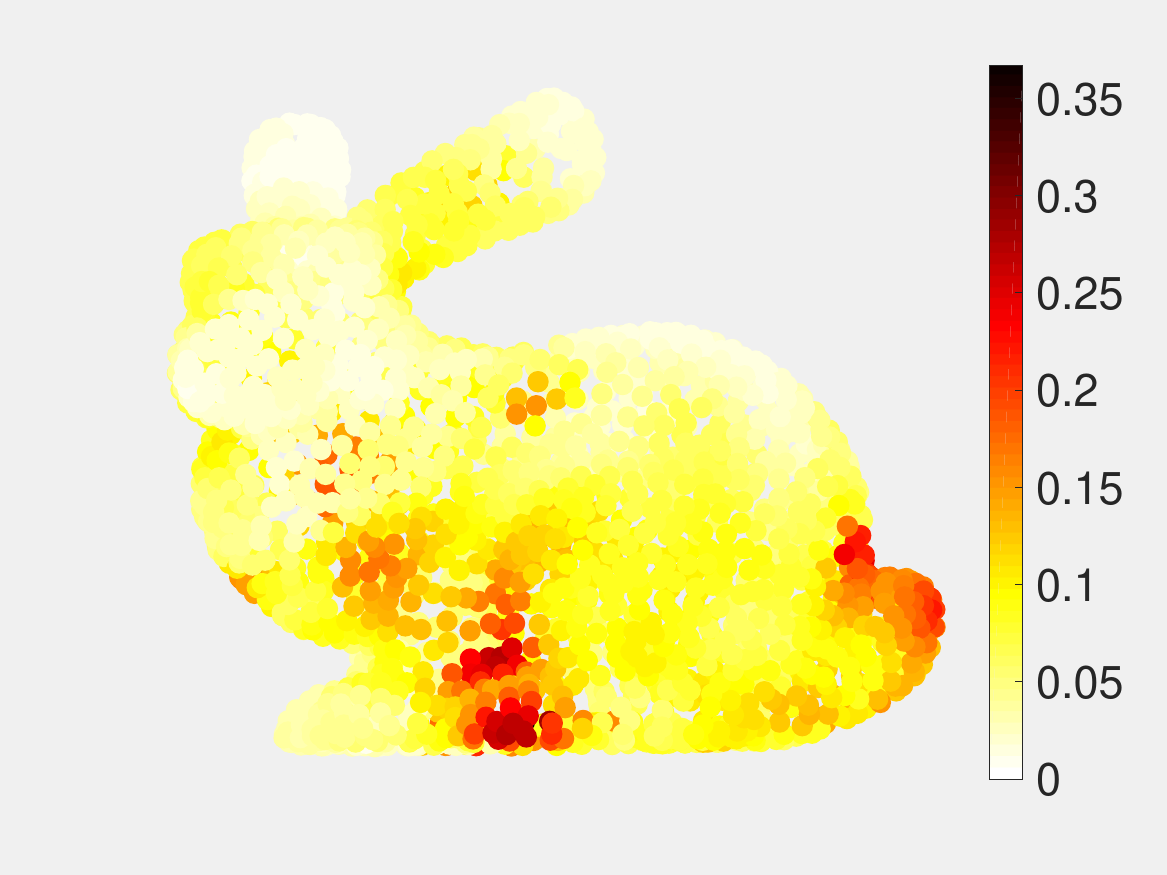}}
\end{minipage}
\begin{minipage}[m]{0.4\linewidth}
\centerline{\includegraphics[width=.85\linewidth]{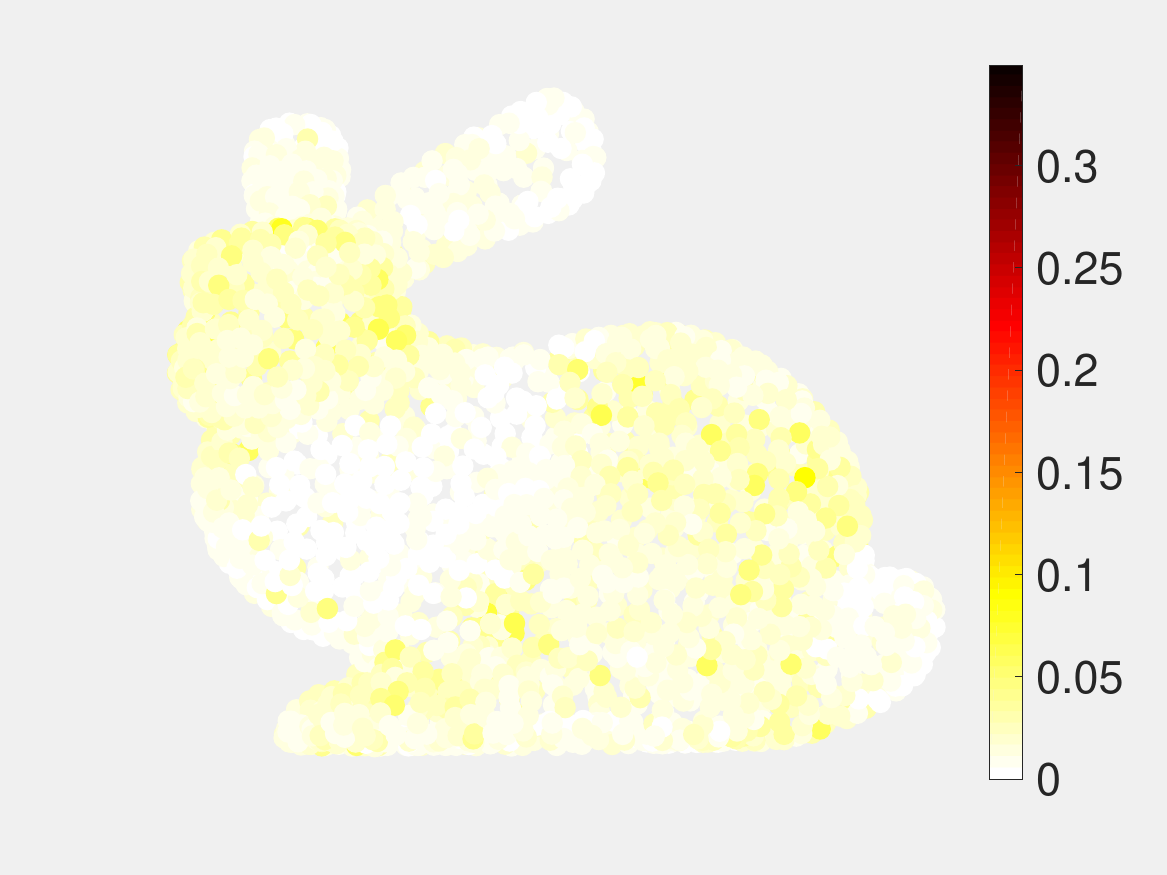}}
\end{minipage} \\
\begin{minipage}[m]{0.16\linewidth}
\centerline{\small{NMSE vs.}}
\centerline{\small{\# Samples}}
\end{minipage} 
\begin{minipage}[m]{0.4\linewidth}
\centerline{\includegraphics[width=.85\linewidth]{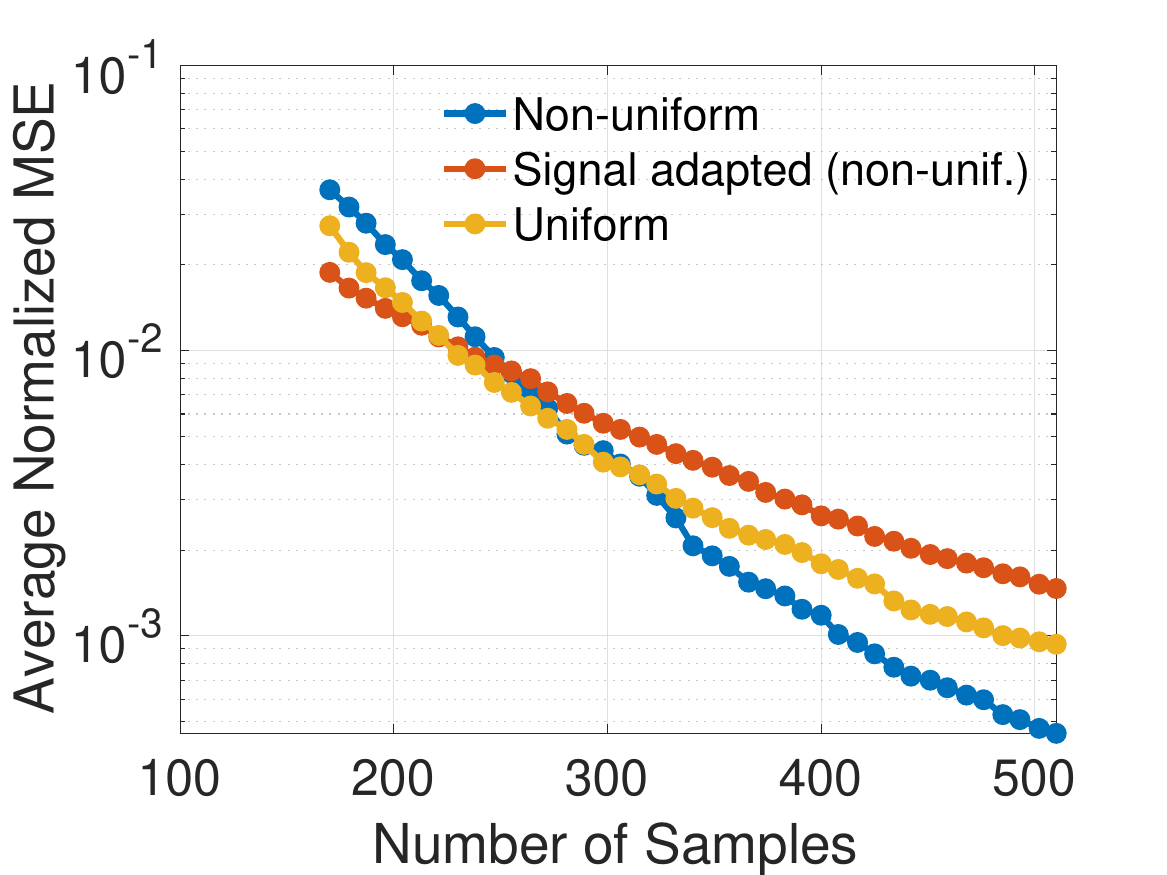}}
\end{minipage}
\begin{minipage}[m]{0.4\linewidth}
\centerline{\includegraphics[width=.85\linewidth]{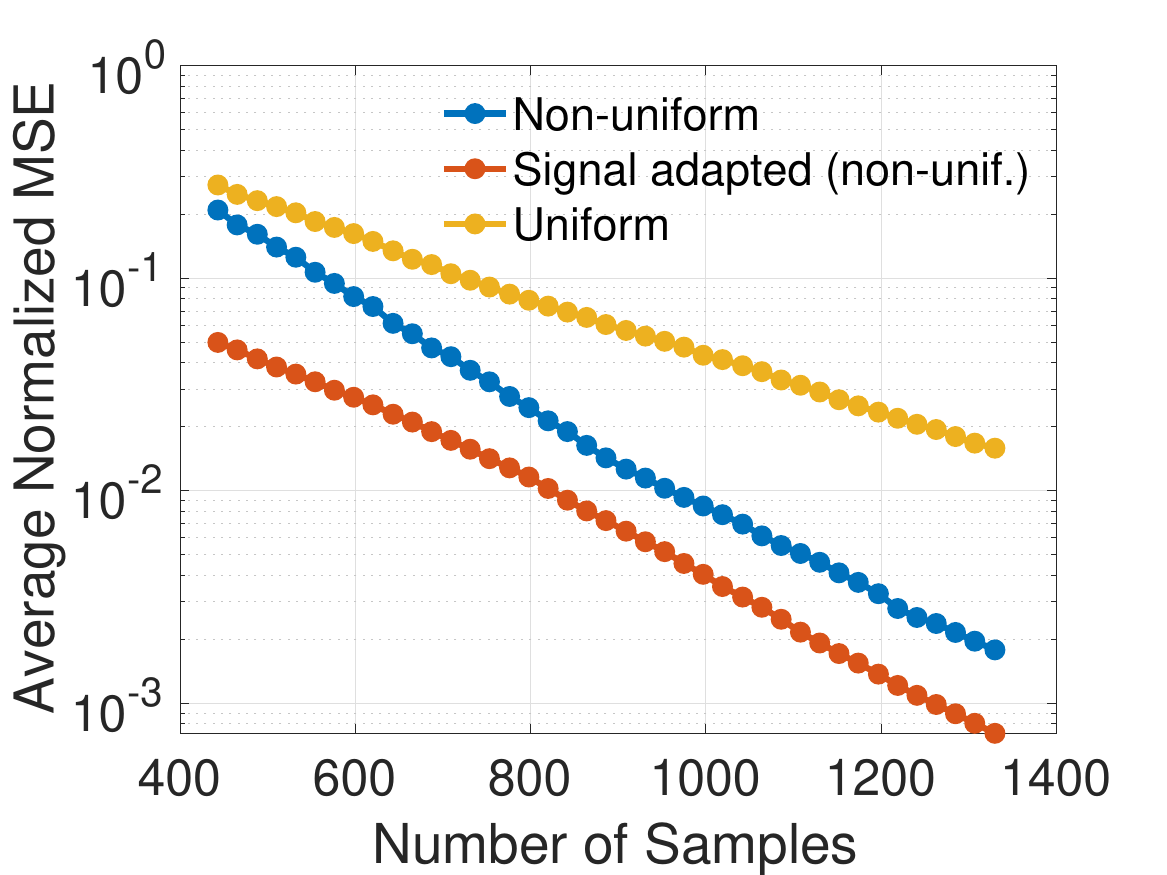}}
\end{minipage}
\caption{Tradeoff between reconstruction error and the number of samples, with and without adapting the sampling weights to the signal.}\label{Fig:samp_tradeoff}
\end{figure}

\subsubsection{Sampling distributions}

In the third through eighth rows of Fig. \ref{Fig:samp_tradeoff}, we examine the difference between using a non-uniform sampling distribution that is only adapted to the graph 
and a non-uniform sampling distribution that is adapted to both the graph and the signal, as discussed in Section \ref{Se:signal_adapted}. Because the energy of the lowpass signal is fairly evenly distributed across the bunny, the two sampling distributions are not that different for the first band, with the largest differences in the lower, rear region of the bunny. The energy of the bandpass signal is heavily concentrated around the midsection and tail of the bunny, the locations of the discontinuities in the original signal. The signal-adapted sampling distribution therefore places heavier weights in those areas. The seventh and eighth rows in Fig. \ref{Fig:samp_tradeoff} show the absolute values of the reconstruction errors, averaged over 50 trials of the random sampling, when the number of random samples is equal to the estimated number of eigenvalues in the specified band (170 in the lowpass case and 443 in the bandpass case). The benefit of the additional samples near the midsection for the bandpass channel is evident, as the average error is lower in this area. 

Both 
non-uniform sampling distributions consistently outperform uniform sampling in our 
experiments. The graph Laplacian eigenvectors associated with lower eigenvalues tend to be less localized, resulting in non-uniform sampling weights that are closer to uniform weights. Subsequently, there is less benefit from performing the non-uniform sampling on the first band (c.f., bottom row of Fig. \ref{Fig:samp_tradeoff}), consistent with the prior literature on sampling and interpolation of graph signals, such as \cite{chen2016signal}. The benefit of non-uniform sampling is greater for bands that include more eigenvectors whose energy is concentrated in certain regions of the graph. 

\subsubsection{Number of samples vs. reconstruction error}  Note that the polynomial approximated filters have wider supports as compared to the ideal filters. By performing critical sampling based on the estimated supports of the ideal filters, we may not have enough samples to reconstruct signals from the (wider) filtered subspace. To get a better reconstruction, we could include more samples for each band. 
In the 
bottom row of Fig. \ref{Fig:samp_tradeoff}, we explore the tradeoff between the number of samples and the reconstruction error, increasing the number of samples for each band 
to three times the estimated number of eigenvalues.

\subsubsection{Polynomial approximation order}
To examine the effect of the polynomial order $K$, we plot in Fig. \ref{Fig:poly_tradeoff} the reconstruction NMSEs for two channels of the fast $M$-CSFB transform applied to the piecewise smooth bunny signal with the parameters of Scenario B, averaged over 50 trials each for uniform sampling, non-uniform sampling, 
and non-uniform sampling adapted to the signal as well as the graph. One key takeaway is that the polynomial degree  plays a more important role when the number of bands is larger, as each filter is narrower, and thus more difficult to approximate by lower order polynomials.

\begin{figure}[t] 
\begin{minipage}[m]{0.16\linewidth}
~
\end{minipage} 
\begin{minipage}[m]{0.4\linewidth}
\centerline{\small{Lowpass}}
\end{minipage}
\begin{minipage}[m]{0.4\linewidth}
\centerline{\small{Bandpass}}
\end{minipage}  \\
\begin{minipage}[m]{0.16\linewidth}
\centerline{\small{$M=4$}}
\end{minipage} 
\begin{minipage}[m]{0.4\linewidth}
\centerline{\includegraphics[width=1\linewidth]{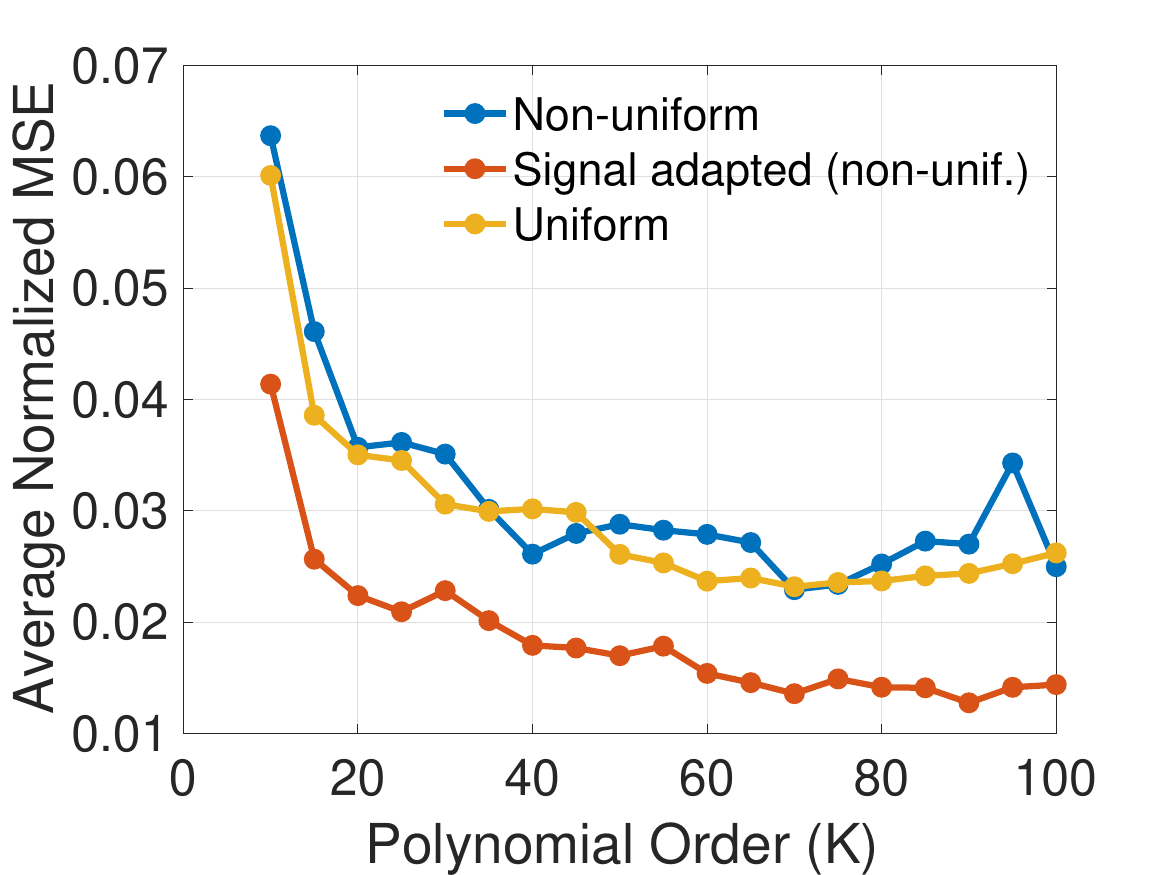}}
\end{minipage}
\begin{minipage}[m]{0.4\linewidth}
\centerline{\includegraphics[width=1\linewidth]{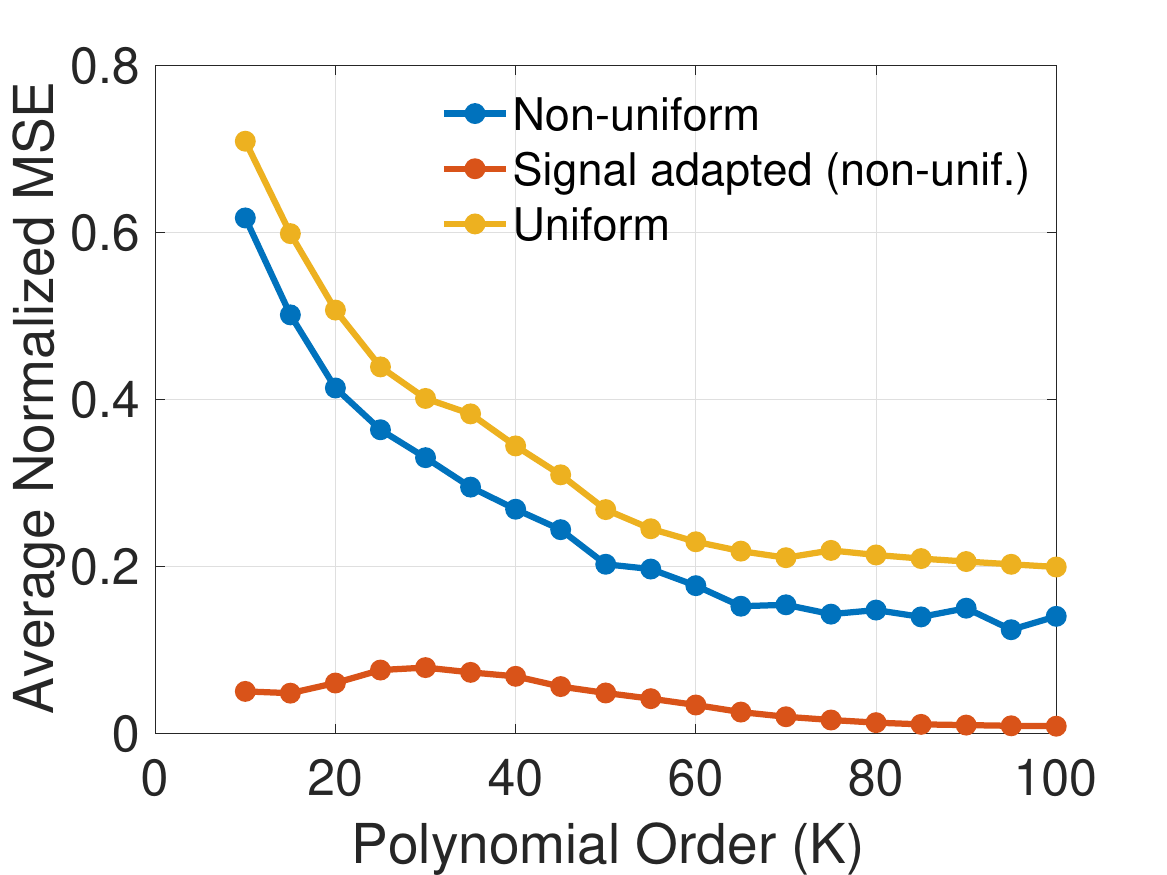}}
\end{minipage} \\
\begin{minipage}[m]{0.16\linewidth}
\centerline{\small{$M=10$}}
\end{minipage} 
\begin{minipage}[m]{0.4\linewidth}
\centerline{\includegraphics[width=1\linewidth]{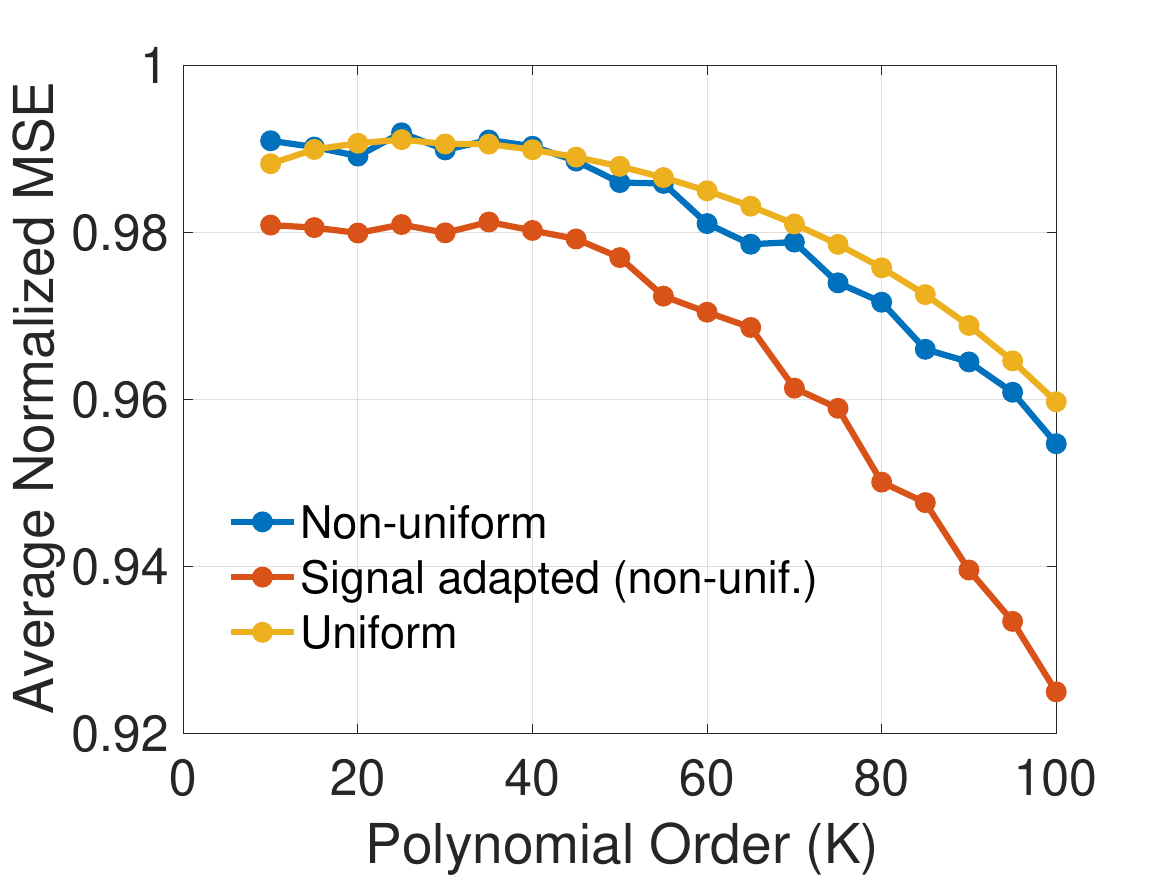}}
\end{minipage}
\begin{minipage}[m]{0.4\linewidth}
\centerline{\includegraphics[width=1\linewidth]{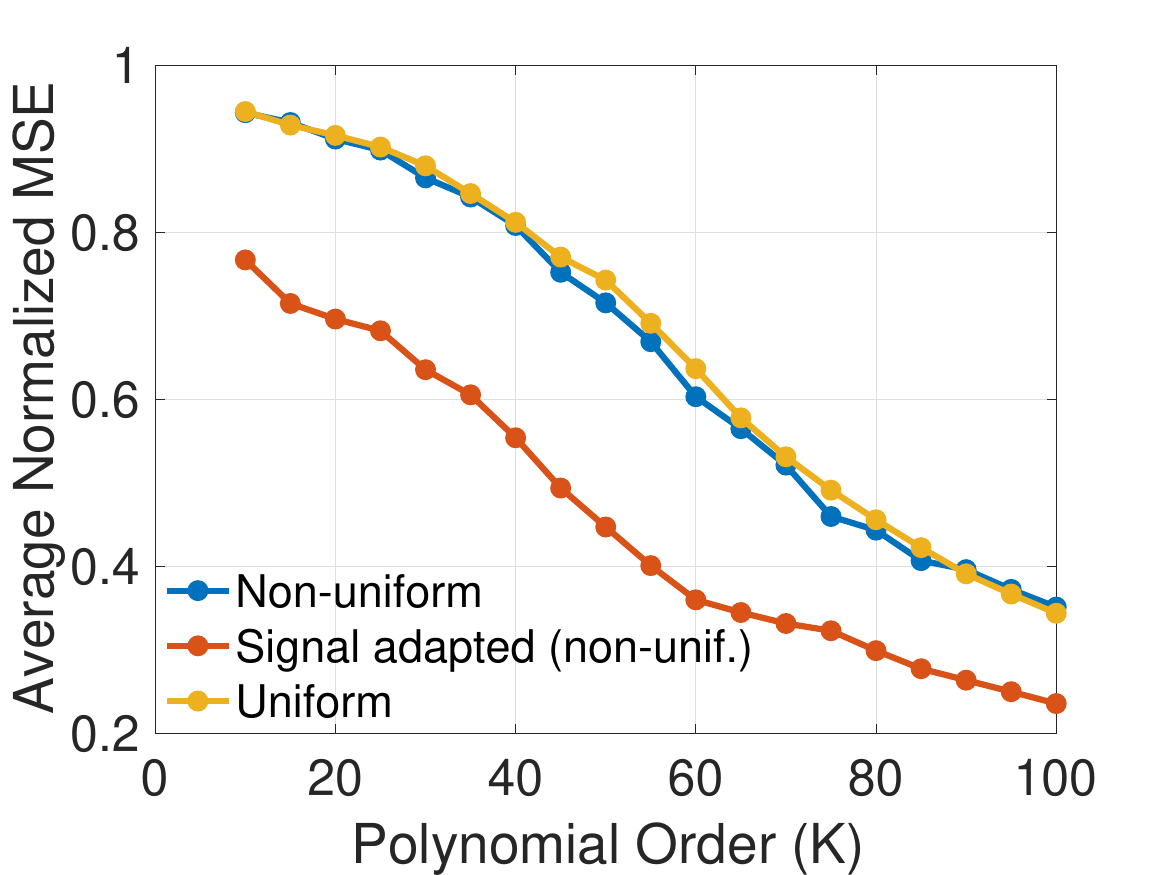}}
\end{minipage} 
\caption{The role of the polynomial order $K$ in the average reconstruction error for two different channels of the fast $M$-CSFB transform.}\label{Fig:poly_tradeoff}
\vspace{-.3cm}
\end{figure}

\subsubsection{Allocation of samples}
In addition to adapting the non-uniform sampling distributions, the signal-adapted fast $M$-CSFB transform 
redistributes the $N$ samples between the channels, allocating more samples to bands where more of the signal's energy resides. Again for a 4-band $M$-CSFB transform of the piecewise smooth bunny graph signal, Table \ref{Ta:sample_distribution} breaks down the improvement in NMSE (averaged over 50 trials). 
 We see that both adaptations reduce the reconstruction error. 
Adapting the sampling distributions and allocation of samples 
to the signal also reduces the NMSE in each of the 
five examples 
 and two scenarios shown in Table \ref{Ta:comp_times}.

\begin{table}[tbh]
{\footnotesize
\tabcolsep=0.11cm
\begin{center}
\begin{tabular}{l|C{1.5cm}C{1.5cm}C{1.5cm}|}
\cline{2-4}
 & \multicolumn{1}{ c| }{} & \multicolumn{1}{ c| }{} & \multicolumn{1}{ c| }{Sampling}  \\ 
  & \multicolumn{1}{ c| }{No} & \multicolumn{1}{ c| }{Sampling} & \multicolumn{1}{ c| }{Distributions \&}  \\ 
    & \multicolumn{1}{ c| }{Signal} & \multicolumn{1}{ c| }{Distributions} & \multicolumn{1}{ c| }{Allocations}  \\ 
        & \multicolumn{1}{ c| }{Adaptation} & \multicolumn{1}{ c| }{Adapted} & \multicolumn{1}{ c| }{Adapted}  \\ 
\cline{1-4}
\multicolumn{1}{|L{2.2cm}| }{Scenario A: faster}& 0.0399 & 0.0218 & 0.0106   \\
\cline{1-4}
\multicolumn{1}{|L{2.2cm}|}{Scenario B: more accurate} & 0.0318& 0.0144& 0.0052  \\
\cline{1-4} 
\end{tabular}
\end{center}
}
\caption{Average normalized mean square reconstruction error for bunny signal with 4-band fast $M$-CSFB transform variants}
\label{Ta:sample_distribution}
\end{table}

\begin{figure}[b]
\begin{minipage}[m]{0.48\linewidth}
\centerline{~~\includegraphics[width=.8\linewidth]{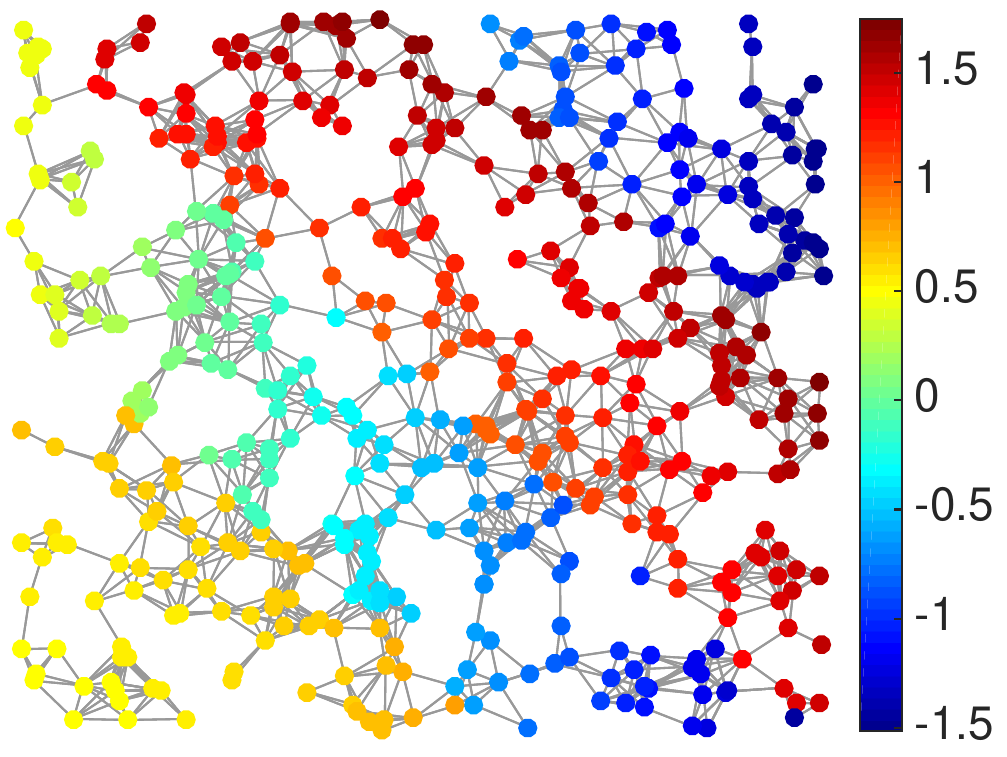}}
\vspace{.05in}
\centerline{~~\small{(a)}}
\end{minipage}
\begin{minipage}[m]{0.48\linewidth}
\vspace{.02in}
\centerline{\includegraphics[width=.8\linewidth]{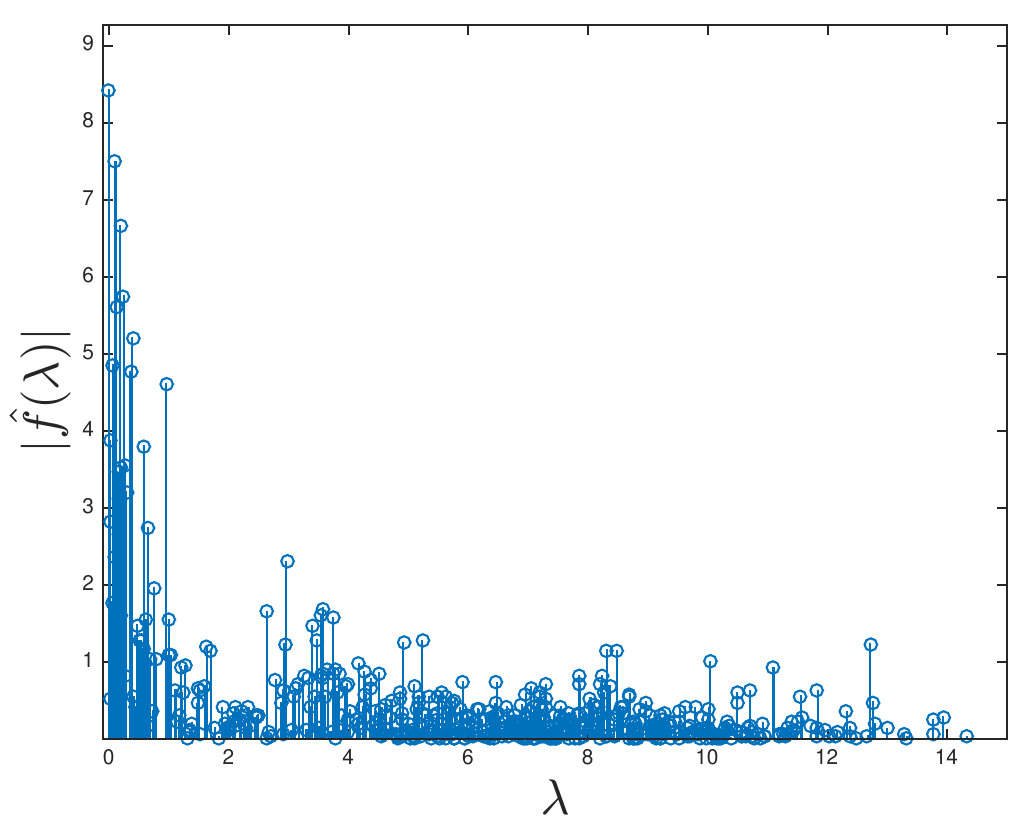}~}
\centerline{~~\small{(b)}}
\end{minipage} \\
\vspace{.07in}
\begin{minipage}[m]{0.48\linewidth}
\centerline{\includegraphics[width=.86\linewidth]{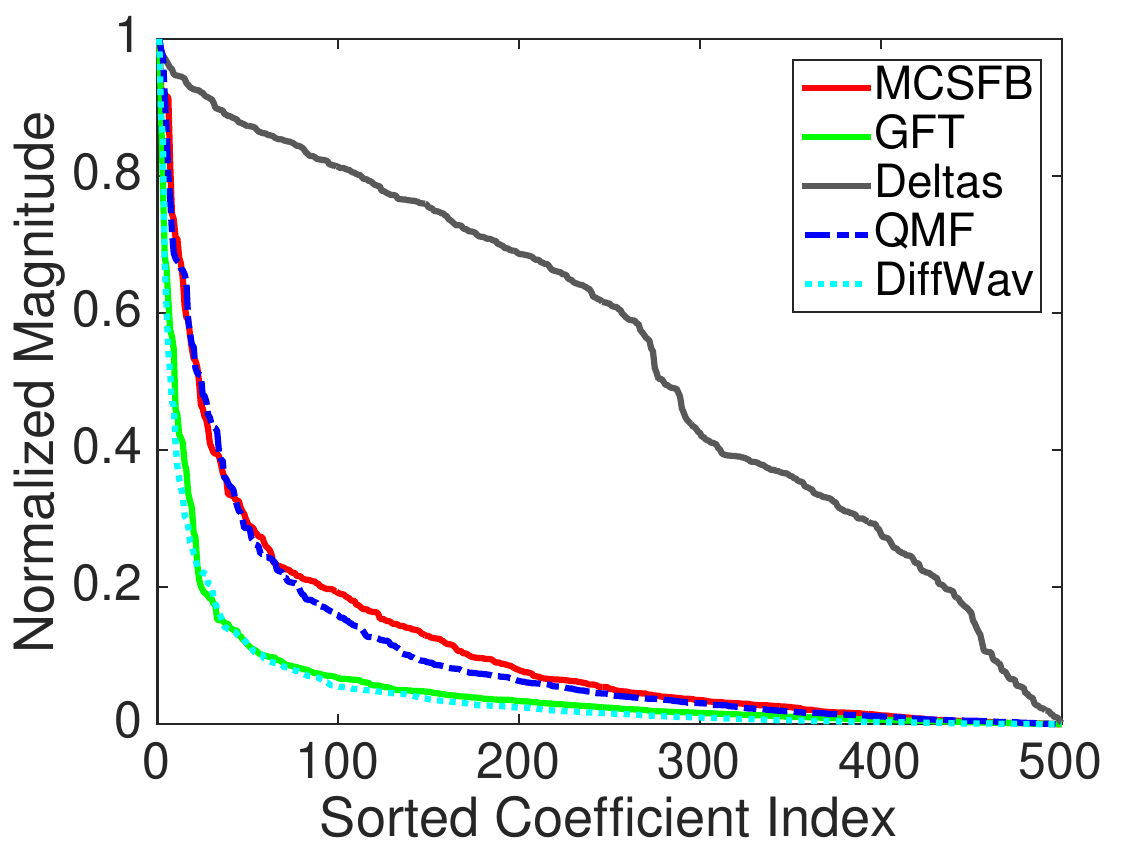}}
\centerline{\small{(c)}}
\end{minipage}
\begin{minipage}[m]{0.48\linewidth}
\centerline{\includegraphics[width=.95\linewidth]{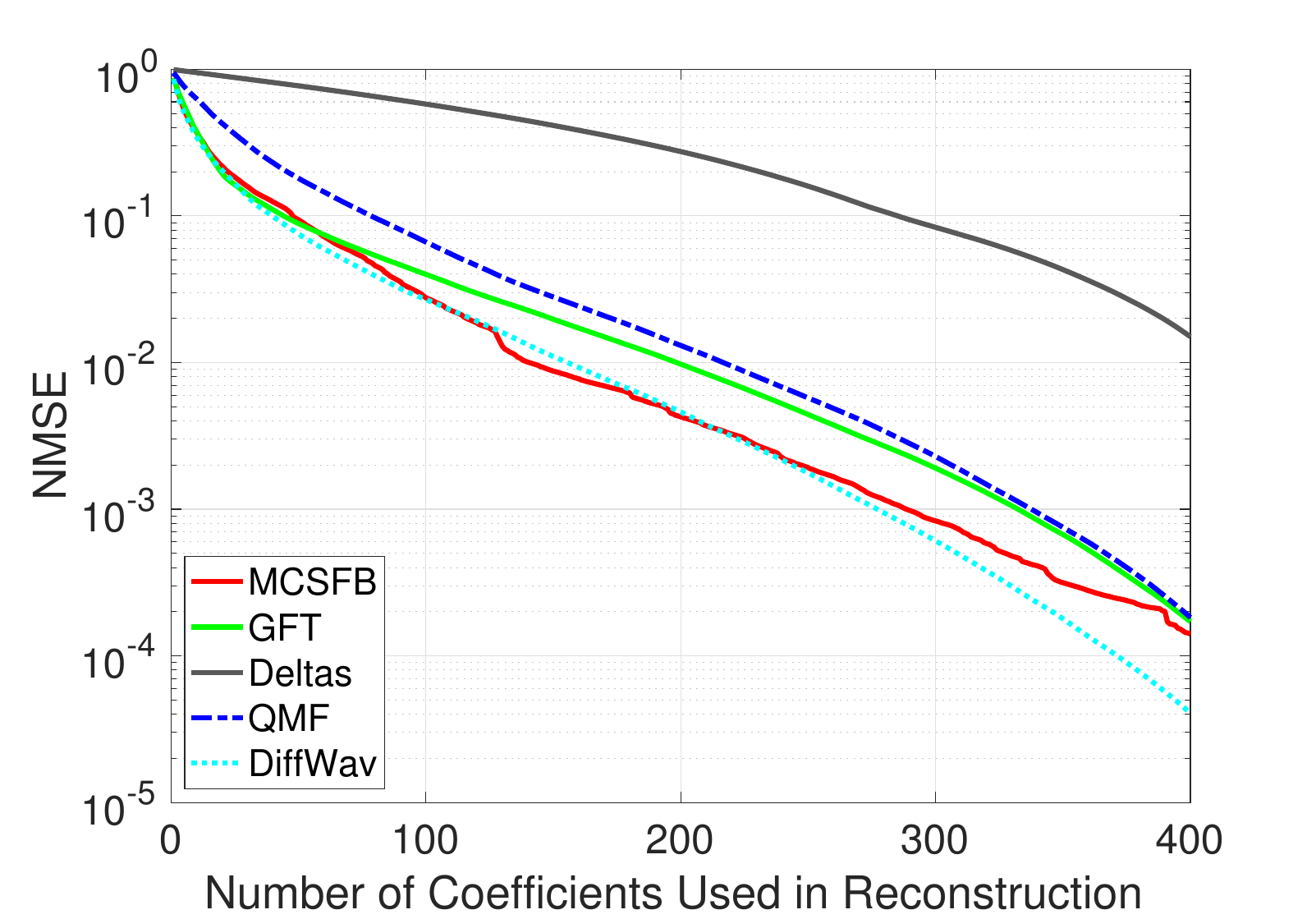}}
\centerline{\small{(d)}}
\end{minipage}
\caption{Compression example. (a)-(b) Piecewise-smooth signal from \cite[Fig. 11]{shuman_TSP_multiscale} in the vertex and graph spectral domains. (c) The normalized sorted magnitudes of the transform coefficients for the proposed $M$-CSFB (exact version), 
the graph Fourier transform, the basis of Kronecker deltas, the quadrature mirror filterbank \cite{narang2012perfect}, and the diffusion wavelet transform \cite{coifman2006diffusion}. (d) The reconstruction errors as a function of the sparsity threshold $T$ in \eqref{Eq:sparse_coding}.} \label{Fig:comp}
\vspace{-.6cm}
\end{figure}

\begin{figure}[tbh] 
\begin{minipage}[m]{0.15\linewidth}
\centerline{\footnotesize{Original}}
\centerline{\footnotesize{Signal}}
\end{minipage}
\begin{minipage}[m]{0.4\linewidth}
\centerline{\includegraphics[width=1\linewidth]{fig_temp_signal}}
\end{minipage}
\begin{minipage}[m]{0.3\linewidth}
\centerline{\footnotesize{~}}
\end{minipage} 
\begin{minipage}[m]{0.08\linewidth}
\centerline{\footnotesize{~}}
\end{minipage}
\\
\vspace{.05in}

\begin{minipage}[m]{0.15\linewidth}
~
\end{minipage}
\begin{minipage}[m]{0.35\linewidth}
\centerline{\footnotesize{Reconstruction~~~}}
\end{minipage}
\begin{minipage}[m]{0.35\linewidth}
\centerline{\footnotesize{Reconstruction Error~~}}
\end{minipage} 
\begin{minipage}[m]{0.08\linewidth}
\centerline{\footnotesize{~~NMSE}}
\end{minipage}
\\
\vspace{-.025in}

\begin{minipage}[m]{0.15\linewidth}
\centerline{\footnotesize{No}}
\centerline{\footnotesize{Compression}}
\end{minipage}
\begin{minipage}[m]{0.35\linewidth}
\centerline{\includegraphics[width=1\linewidth]{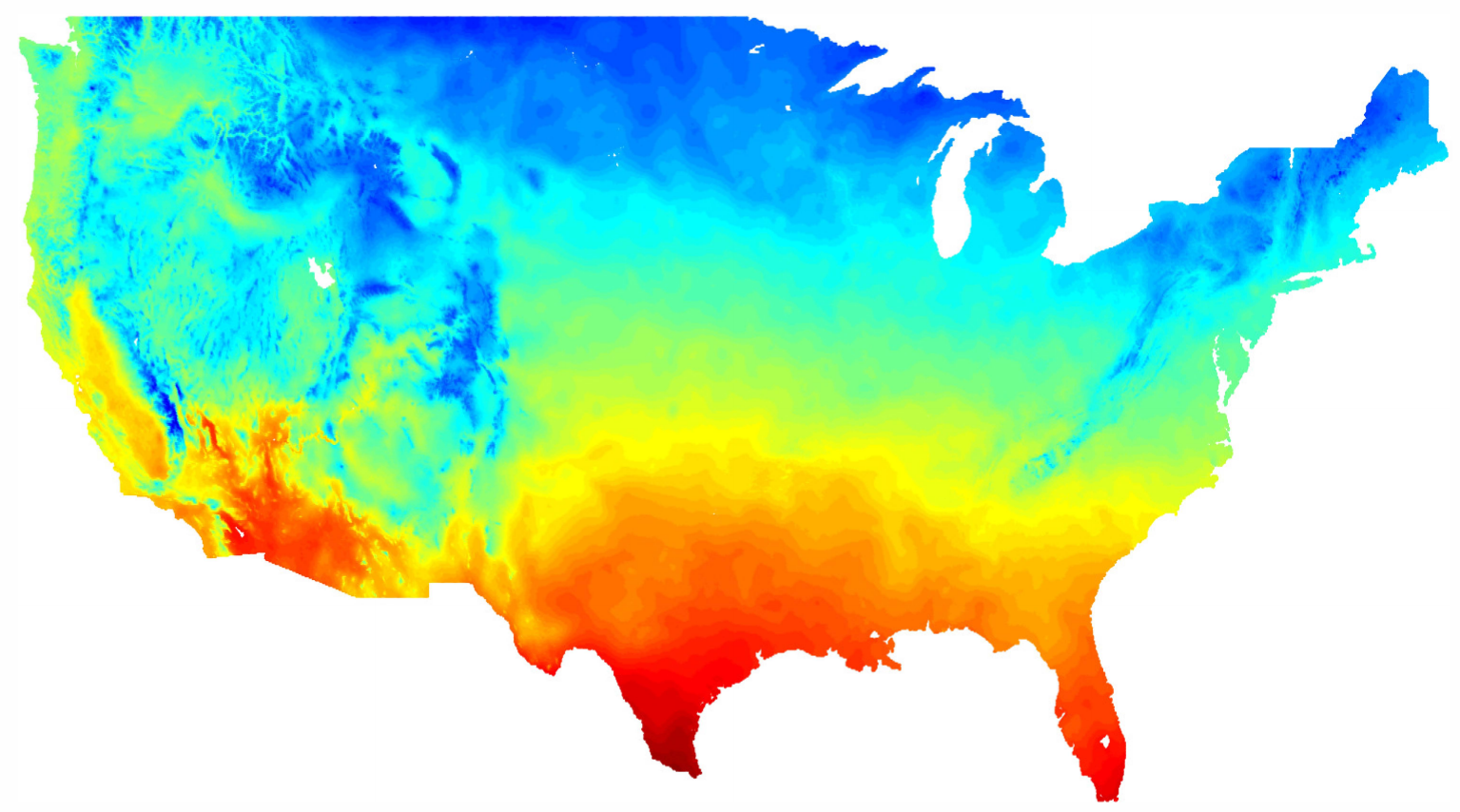}}
\end{minipage}
\begin{minipage}[m]{0.35\linewidth}
\centerline{\includegraphics[width=1\linewidth]{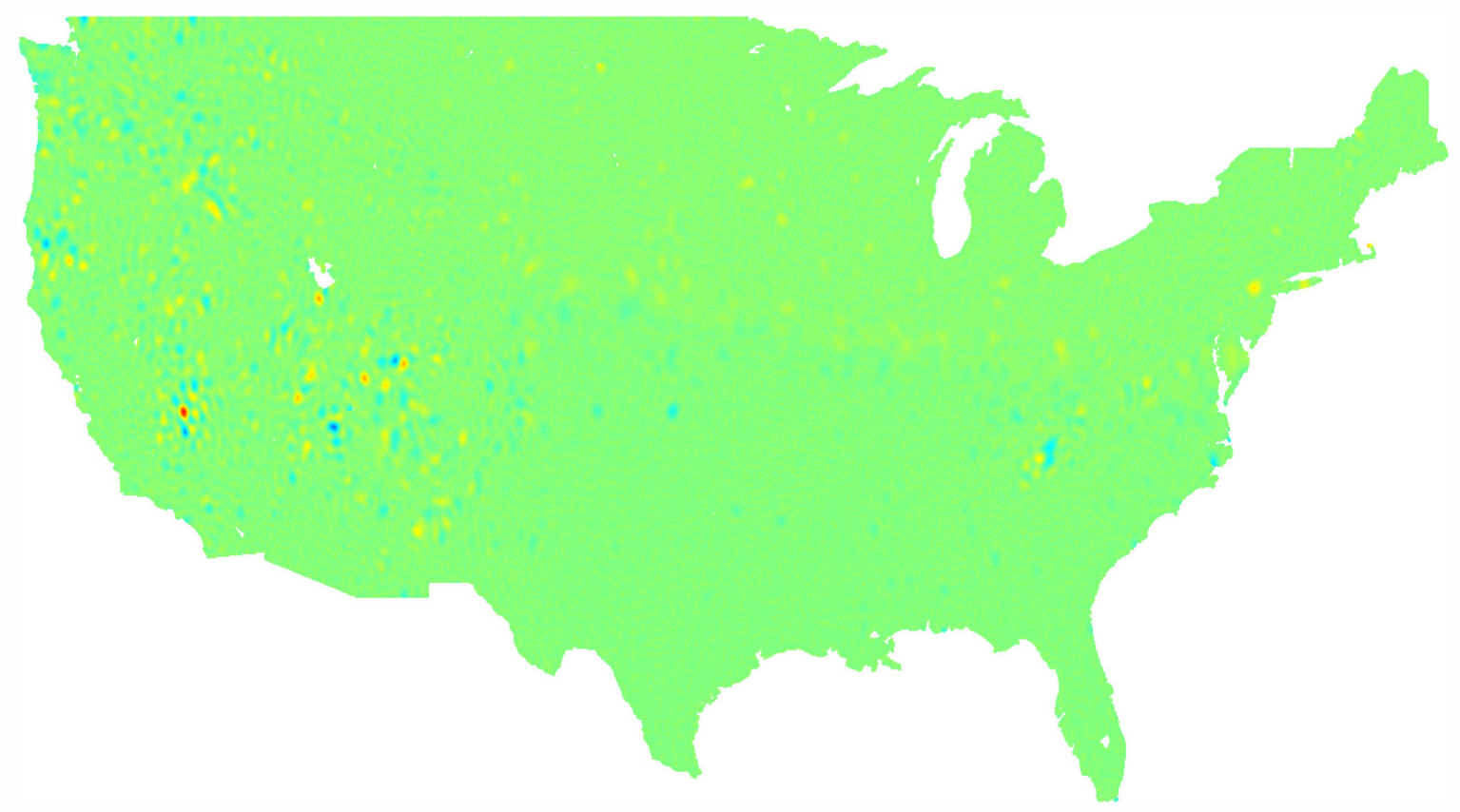}}
\end{minipage} 
\begin{minipage}[m]{0.08\linewidth}
\centerline{\footnotesize{~6.62e-4}}
\end{minipage}
\\
\begin{minipage}[m]{0.15\linewidth}
\centerline{\footnotesize{Compression}}
\centerline{\footnotesize{Ratio}}
\centerline{\footnotesize{1.25:1}}
\end{minipage}
\begin{minipage}[m]{0.35\linewidth}
\centerline{\includegraphics[width=1\linewidth]{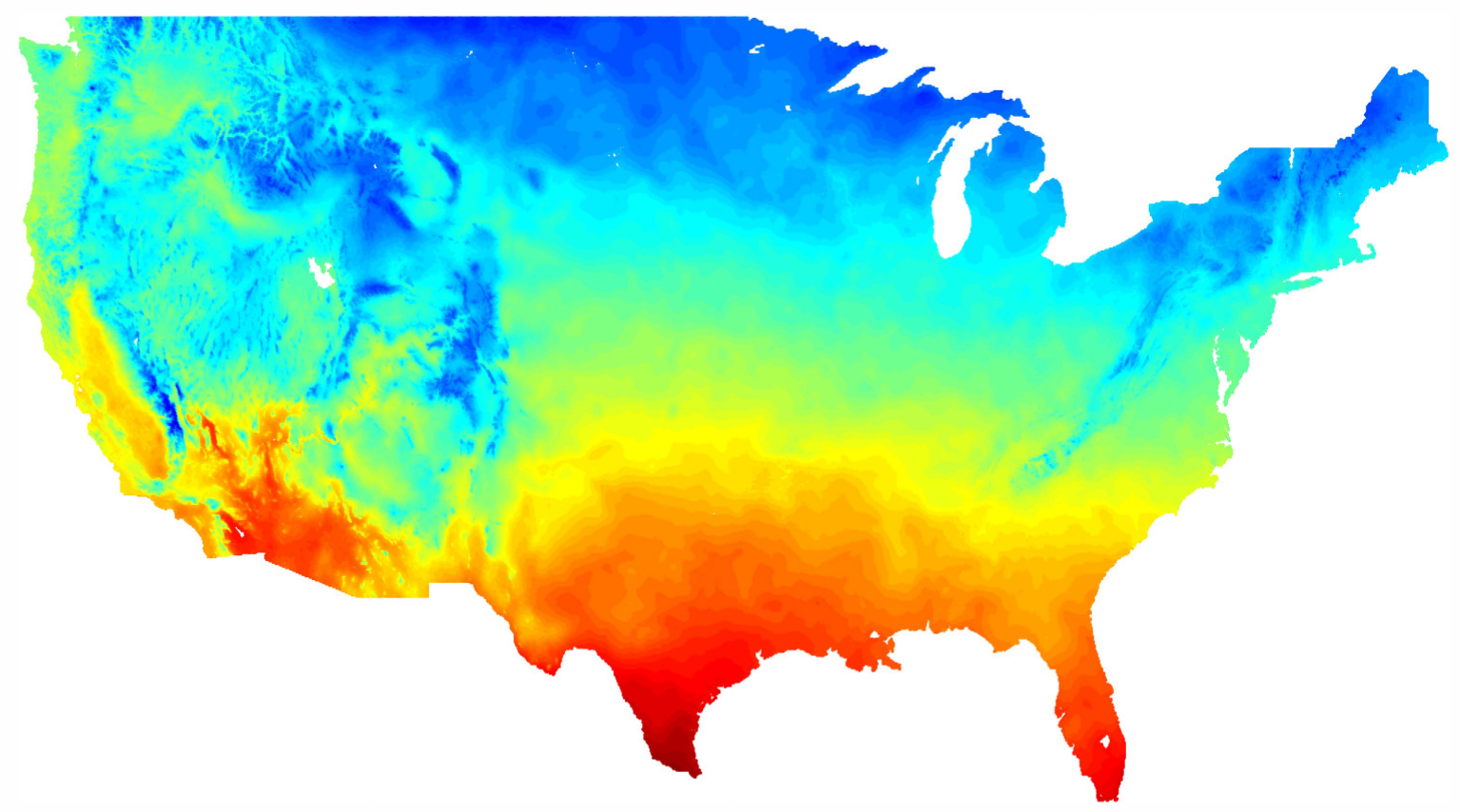}}
\end{minipage}
\begin{minipage}[m]{0.35\linewidth}
\centerline{\includegraphics[width=1\linewidth]{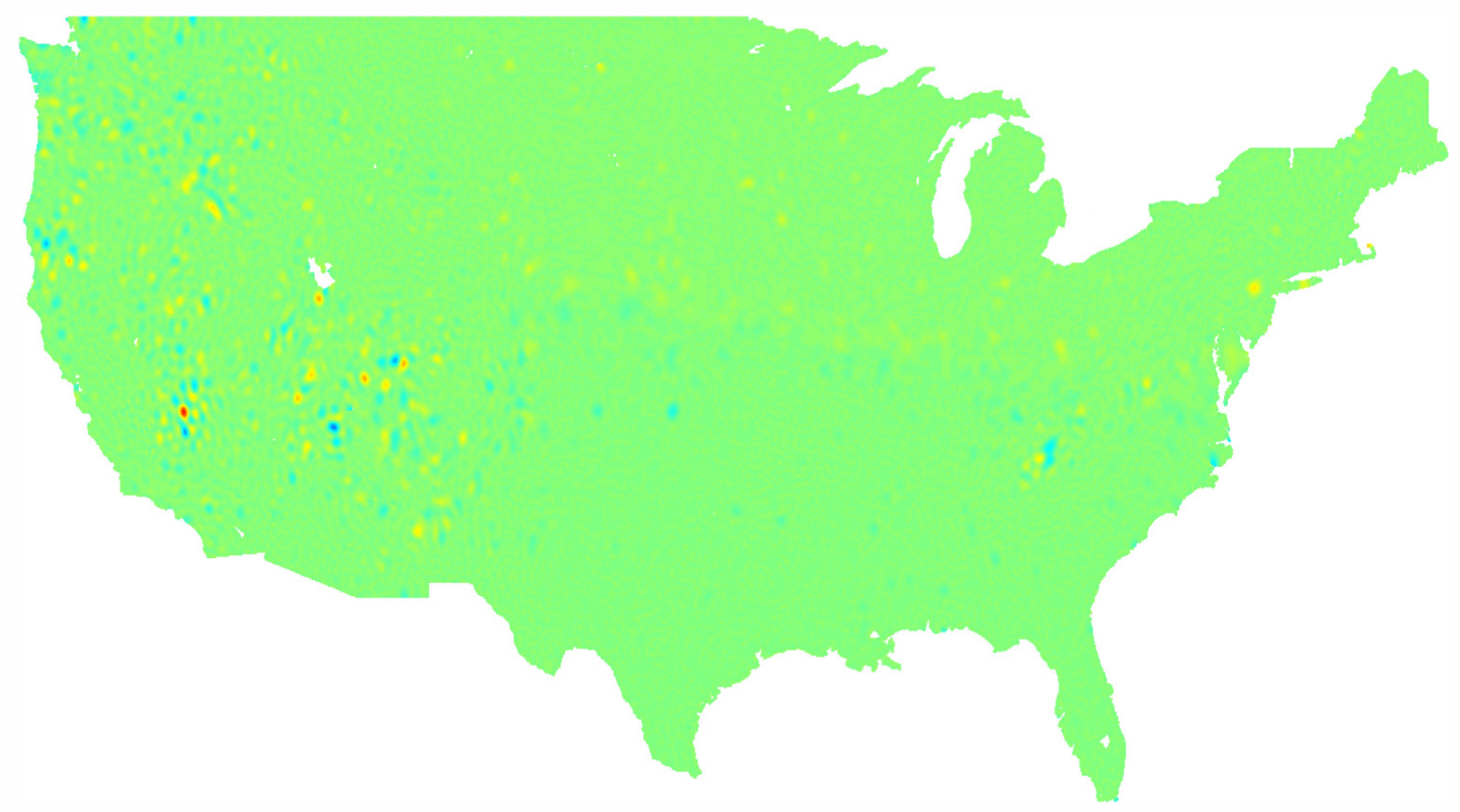}}
\end{minipage} 
\begin{minipage}[m]{0.08\linewidth}
\centerline{\footnotesize{~6.65e-4}}
\end{minipage}
\\
\begin{minipage}[m]{0.15\linewidth}
\centerline{\footnotesize{Compression}}
\centerline{\footnotesize{Ratio}}
\centerline{\footnotesize{2:1}}
\end{minipage}
\begin{minipage}[m]{0.35\linewidth}
\centerline{\includegraphics[width=1\linewidth]{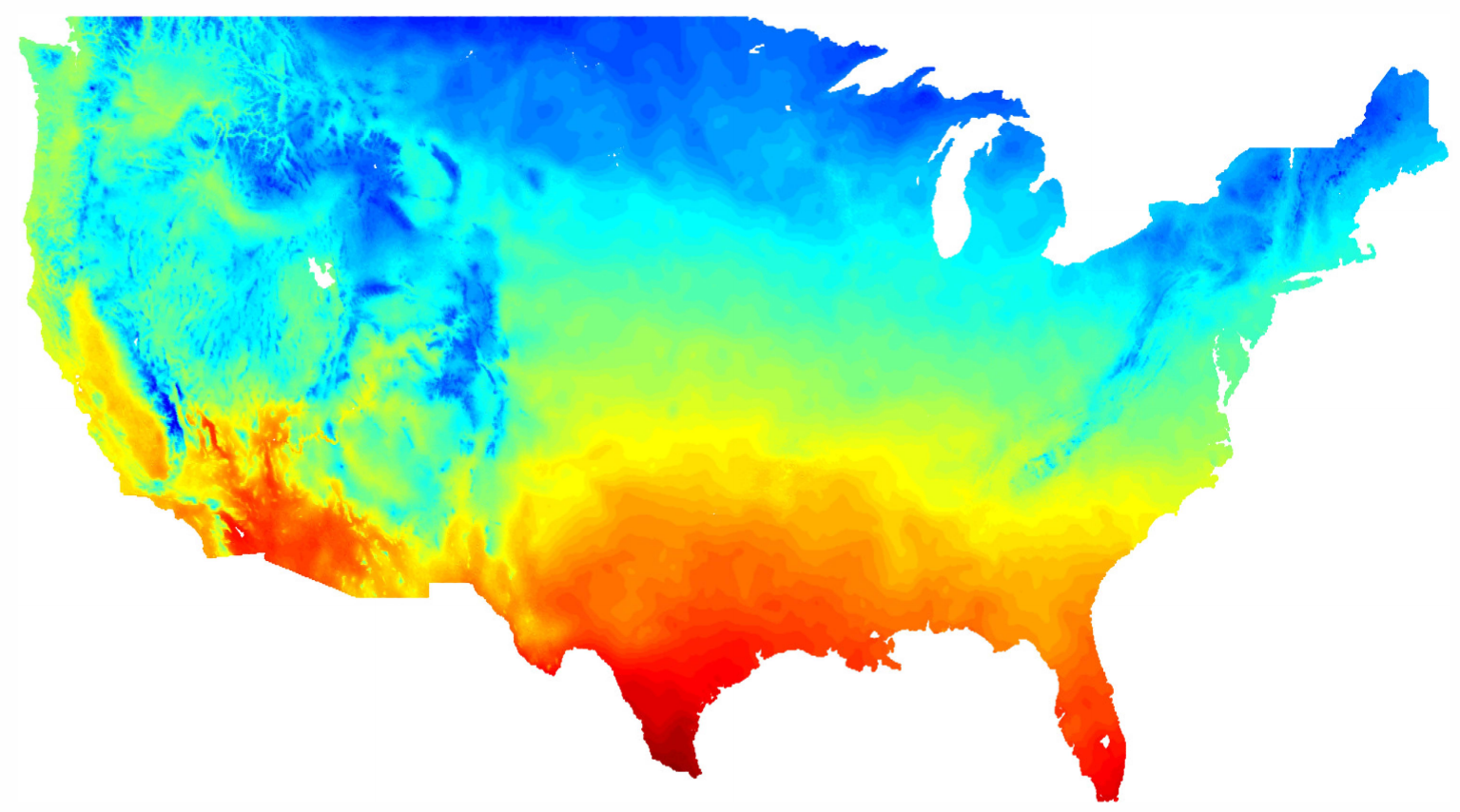}}
\end{minipage}
\begin{minipage}[m]{0.35\linewidth}
\centerline{\includegraphics[width=1\linewidth]{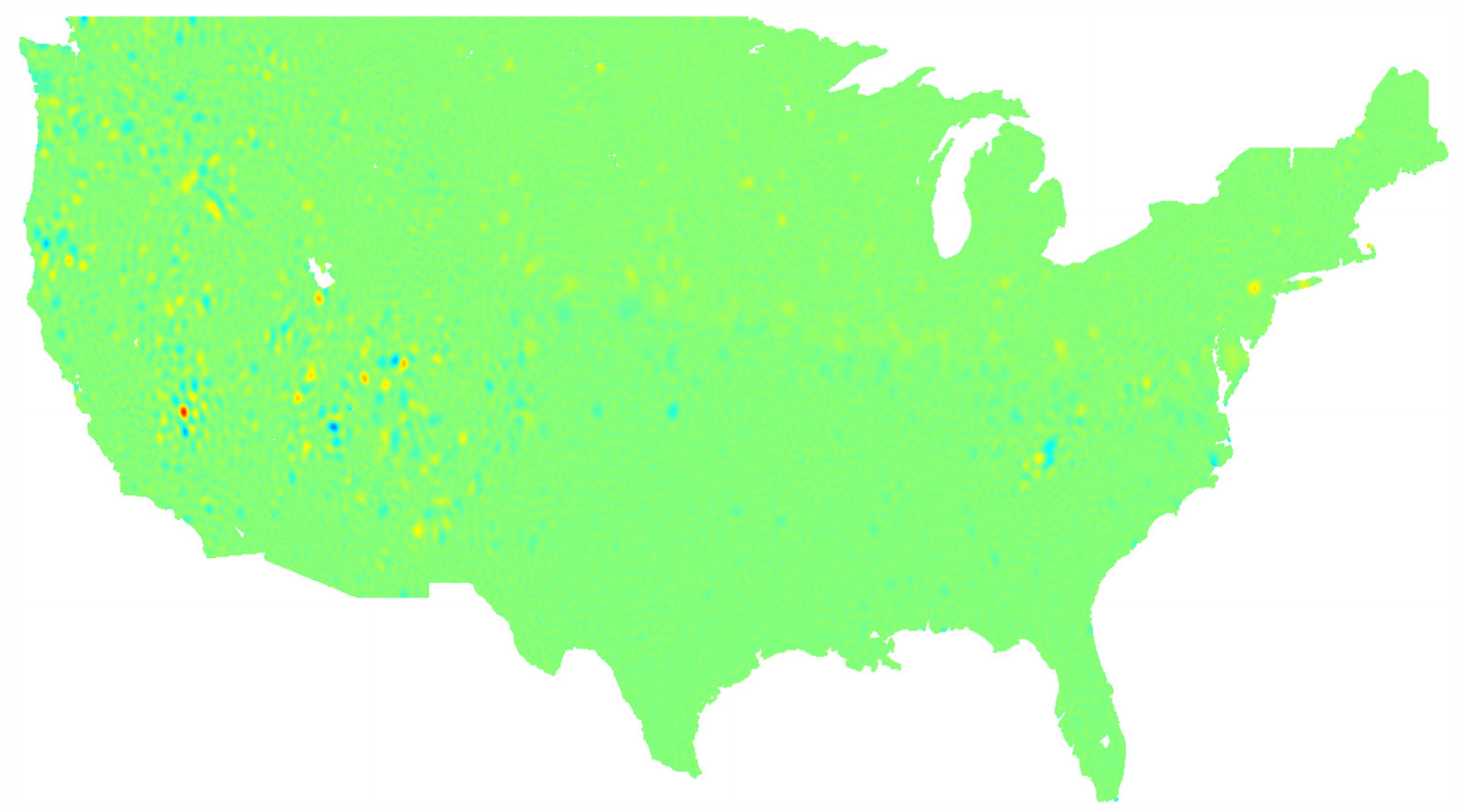}}
\end{minipage} 
\begin{minipage}[m]{0.08\linewidth}
\centerline{\footnotesize{~7.41e-4}}
\end{minipage}
\\
\begin{minipage}[m]{0.15\linewidth}
\centerline{\footnotesize{Compression}}
\centerline{\footnotesize{Ratio}}
\centerline{\footnotesize{5:1}}
\end{minipage}
\begin{minipage}[m]{0.35\linewidth}
\centerline{\includegraphics[width=1\linewidth]{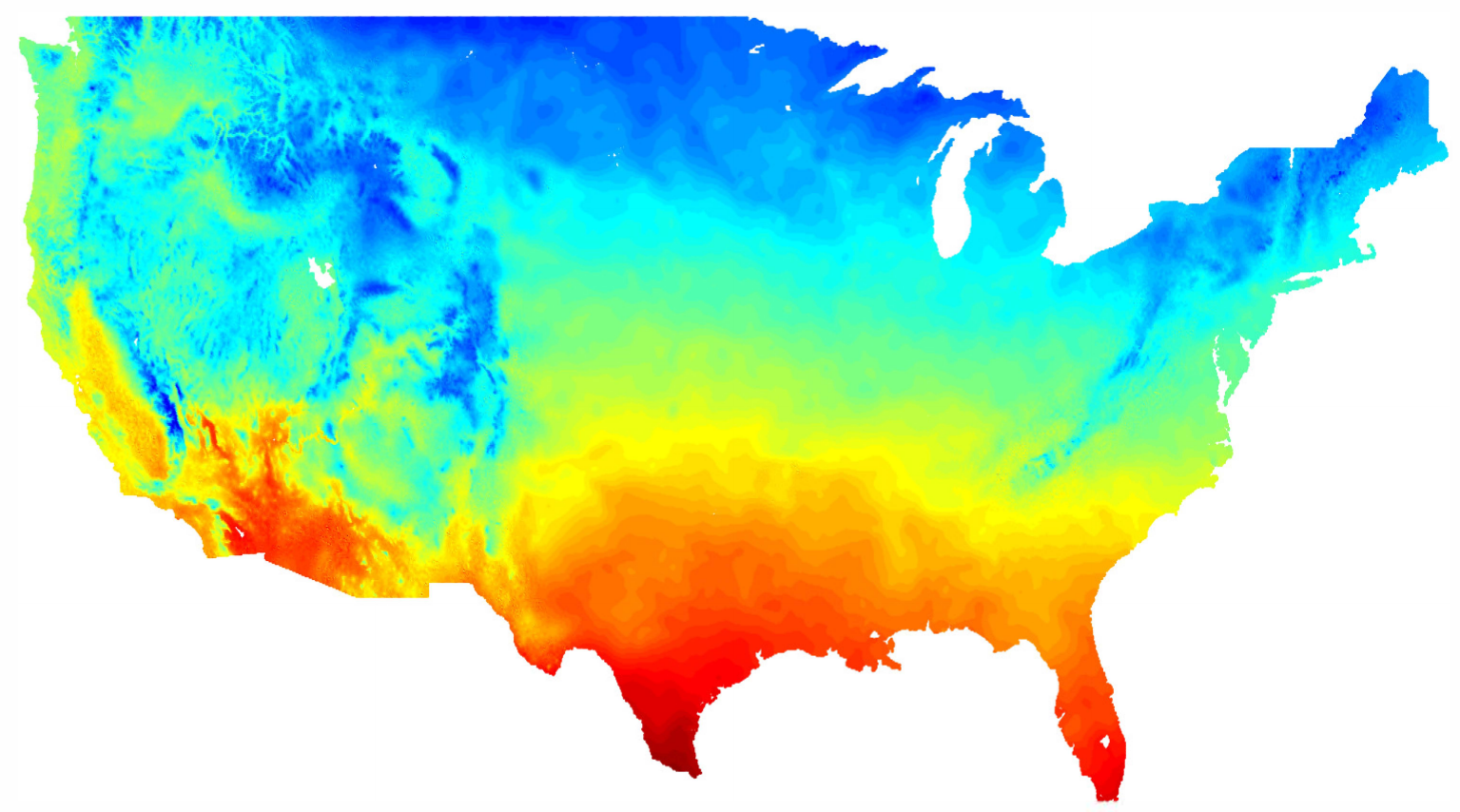}}
\end{minipage}
\begin{minipage}[m]{0.35\linewidth}
\centerline{\includegraphics[width=1\linewidth]{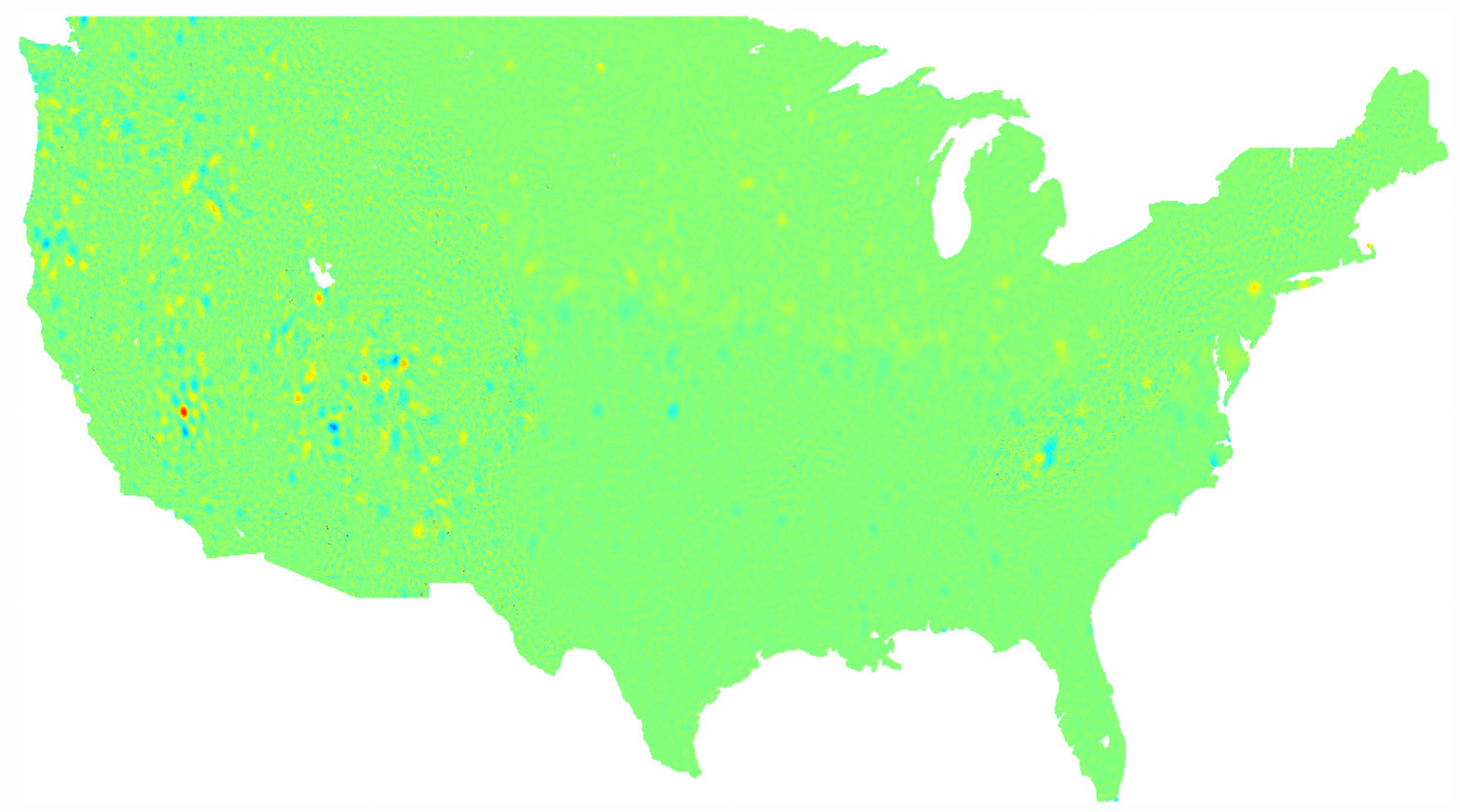}}
\end{minipage} 
\begin{minipage}[m]{0.08\linewidth}
\centerline{\footnotesize{14.78e-4}}
\end{minipage}
\\
\begin{minipage}[m]{0.15\linewidth}
\centerline{\footnotesize{Compression}}
\centerline{\footnotesize{Ratio}}
\centerline{\footnotesize{10:1}}
\end{minipage}
\begin{minipage}[m]{0.35\linewidth}
\centerline{\includegraphics[width=1\linewidth]{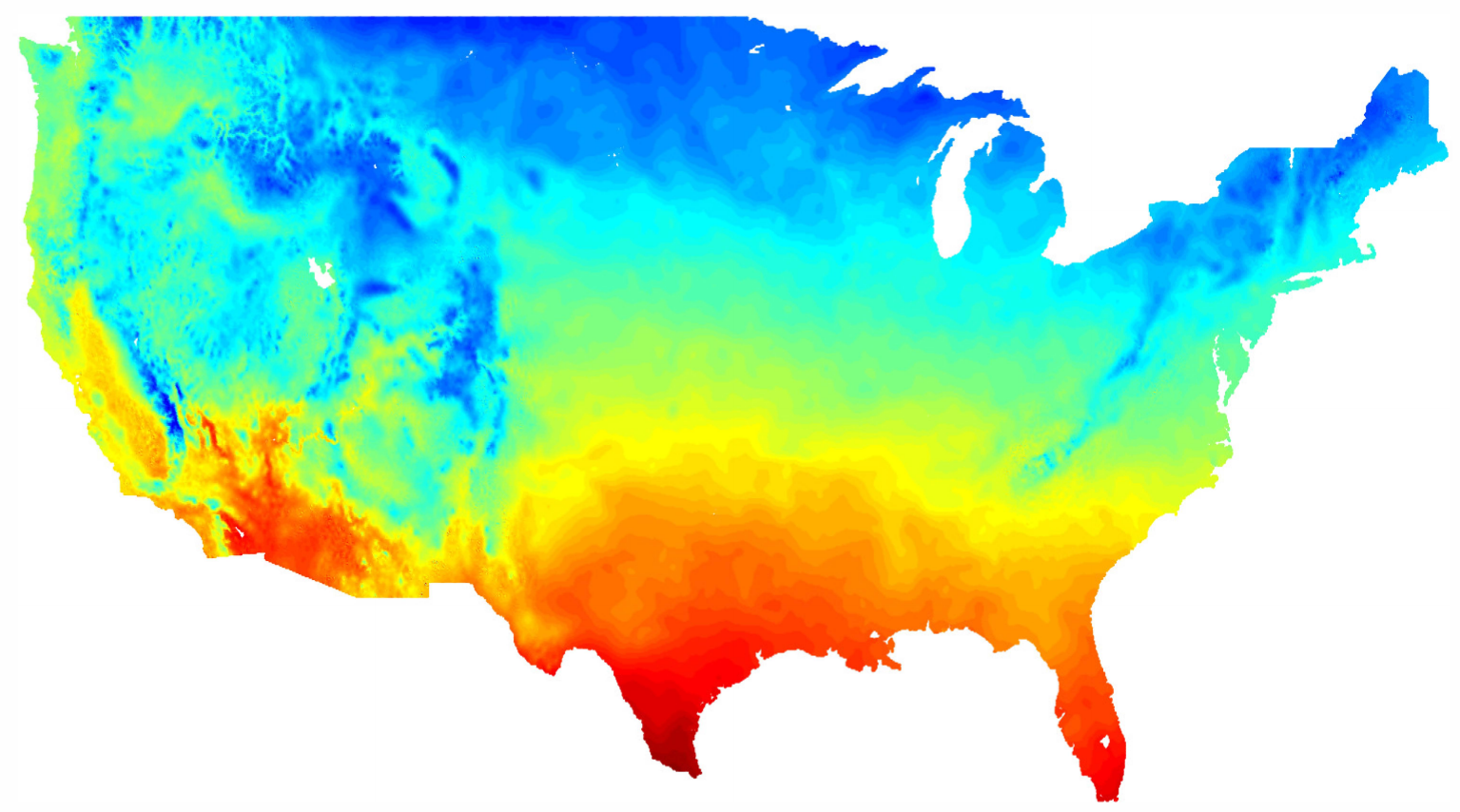}}
\end{minipage}
\begin{minipage}[m]{0.35\linewidth}
\centerline{\includegraphics[width=1\linewidth]{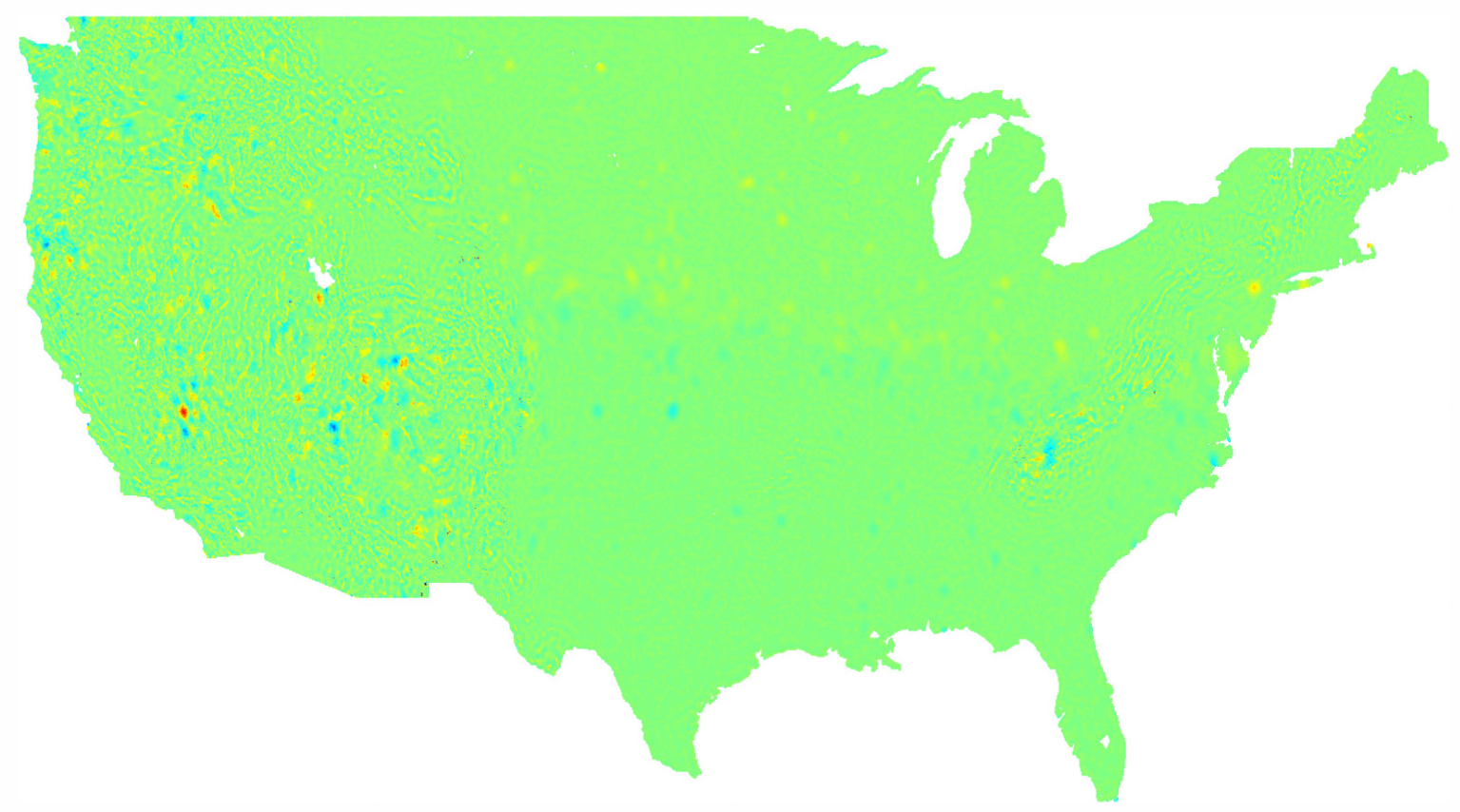}}
\end{minipage} 
\begin{minipage}[m]{0.08\linewidth}
\centerline{\footnotesize{21.69e-4}}
\end{minipage}
\\
\vspace{-.025in}

\begin{minipage}[m]{0.15\linewidth}
\centerline{\footnotesize{~}}
\end{minipage}
\begin{minipage}[m]{0.35\linewidth}
\centerline{\includegraphics[width=1\linewidth]{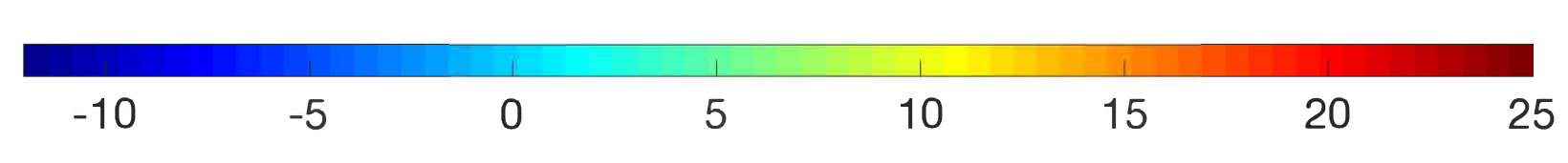}}
\end{minipage}
\begin{minipage}[m]{0.35\linewidth}
\centerline{\includegraphics[width=1\linewidth]{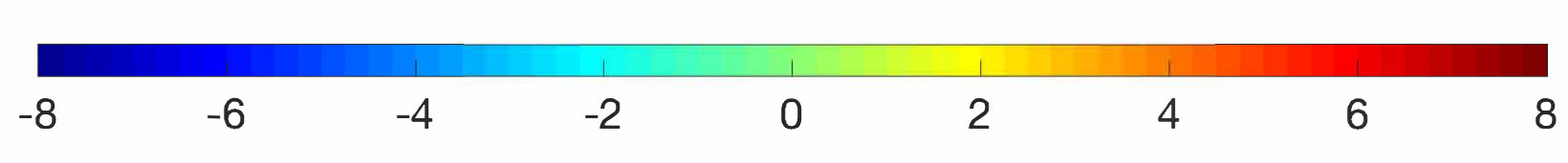}}
\end{minipage} 
\begin{minipage}[m]{0.08\linewidth}
\centerline{\footnotesize{~}}
\end{minipage}
\caption{Compression of the average temperature signal from Section \ref{Se:temp} with a signal-adapted fast $M$-CSFB transform with $M=5$ bands. From top to bottom, we keep 100\%, 80\%, 50\%, 20\%, and 10\% of the analysis coefficients, setting the rest to 0 before synthesis. In all cases, we keep all 28,022 scaling coefficients (the first band, about 6\% of the overall coefficients), and then use the remaining budget to store the wavelet coefficients with the largest magnitudes.}\label{Fig:compression_temp}
\vspace{-.3cm}
\end{figure}

\subsection{Compression Examples} \label{Se:compression}
Next, we compress 
a piecewise-smooth graph signal ${\bf f}$ via the sparse coding optimization
\begin{align}\label{Eq:sparse_coding}
\argmin_{\bf x}   || {\bf f}-\boldsymbol{\Phi}{\bf x} ||_2^2 \hbox{~~~subject to } ||{\bf x}||_0 \leq T,
\end{align}
where $T$ is a predefined sparsity level. After 
normalizing the atoms of various critically-sampled dictionaries, we use the greedy orthogonal matching pursuit (OMP) algorithm \cite{tropp2004greed,elad_book} to approximately solve \eqref{Eq:sparse_coding}. 
  We show the normalized mean square reconstruction errors (NMSE) $\frac{\left|\left|{\bf f}_{{rec}}-{\bf f}\right|\right|_2^2}{||{\bf f}||_2^2}$ in Fig. \ref{Fig:comp}(d). Note that for a fair comparison, we use the exact computation of ${\bf U}$ in the design of all four dictionaries that utilize it.
 For the $M$-CSFB, 
 Fig. \ref{Fig:part_examples} shows 
 the partition into uniqueness sets, and  
 Fig. \ref{Fig:fb} shows the filter bank.

In Fig. \ref{Fig:compression_temp}, we compress the average temperature signal from Section \ref{Se:temp}, which has 469,404 values. We use a signal-adapted fast $M$-CSFB transform with $M=5$ bands and the same parameter settings as Scenario B above. Fig. \ref{Fig:compression_temp} captures the tradeoff between the normalized mean square reconstruction error and the compression ratio. Even using only 10\% of the coefficients, it is difficult to visually identify errors between the reconstruction and the original signal. 

\subsection{Fast Approximate Graph Fourier Transform} \label{Se:fgft}

If the graph Laplacian eigenvalues are distinct, using the exact M-CSFB transform of Section \ref{Se:fb_design} with $M=N$ and the filter endpoints $\tau_0=0$, $\tau_M = \lambda_{\max}+1$, and $\tau_m=\frac{\lambda_{m-1}+\lambda_m}{2}$ for $m=1,2,\ldots,N-1$ yields exactly the graph Fourier transform. The proposed fast $M$-CSFB transform can therefore be used as a fast approximate graph Fourier transform, with a coarser resolution in the spectral domain.  Namely, we use the approximation
\begin{eqnarray} \label{Eq:fgft}
\hat{f}_{approx}(\lambda)={\frac{1}{\sqrt{\eta_{m_{\lambda}}}}} ||\tilde{h}_{m_{\lambda}}(\L){\bf f} ||_2 ,
\end{eqnarray}
where $m_{\lambda}$ is the index of the band containing $\lambda$ and $\eta_{m_{\lambda}}$ is the approximate number of eigenvalues contained in the band containing $\lambda$. This approximation is motivated by the fact that $||{h}_{m}(\L){\bf f} ||_2^2=\sum_{\{\lambda: h_{m}(\lambda)=1\}} |\hat{f}(\lambda)|^2$, by Parseval's equality.

\begin{figure}[t]
\begin{minipage}[m]{0.48\linewidth}
\centerline{\includegraphics[width=.9\linewidth]{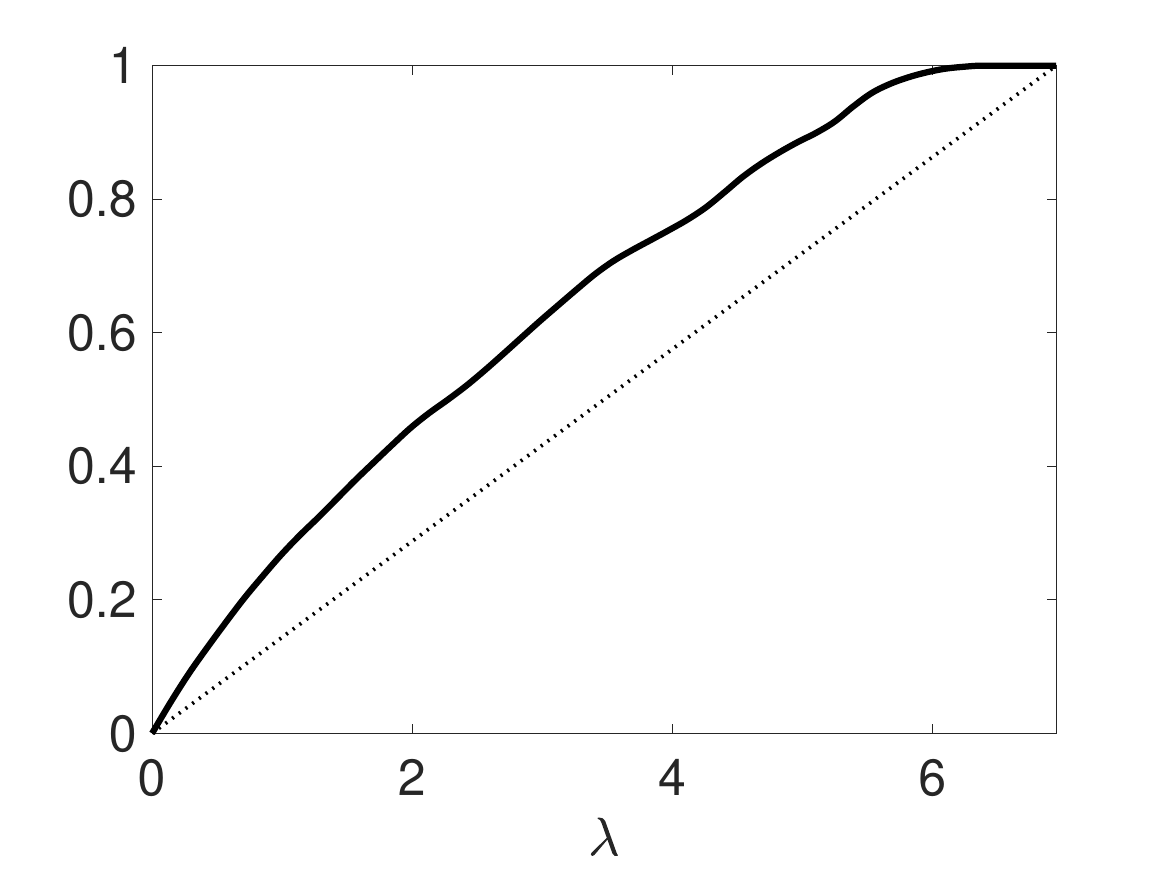}}
\centerline{~\small{(a)}}
\end{minipage}
\begin{minipage}[m]{0.48\linewidth}
\centerline{\includegraphics[width=.9\linewidth]{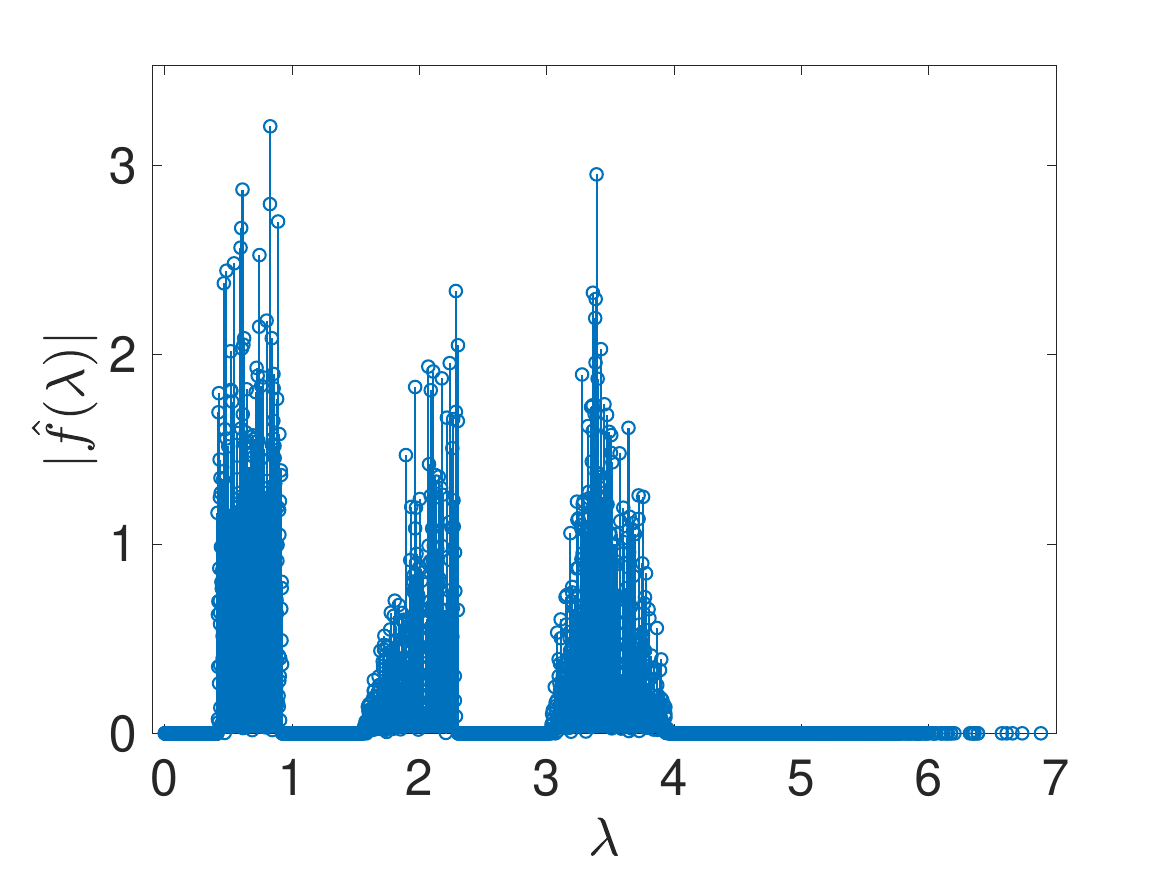}}
\centerline{~\small{(b)}}
\end{minipage} 
\caption{(a) Estimate of the cumulative spectral distribution of the Minnesota road network shown in Fig. \ref{Fig:part_examples}. (b) Synthetic signal in the graph spectral domain of the Minnesota graph, generated with the method from \cite{le2018approximate}.}\label{Fig:fgft1}
\vspace{-.1cm}
\end{figure}

\begin{figure*}[tb]
\begin{minipage}[m]{0.24\linewidth}
\centerline{\small{Equal Length}}
\centerline{\includegraphics[width=.95\linewidth]{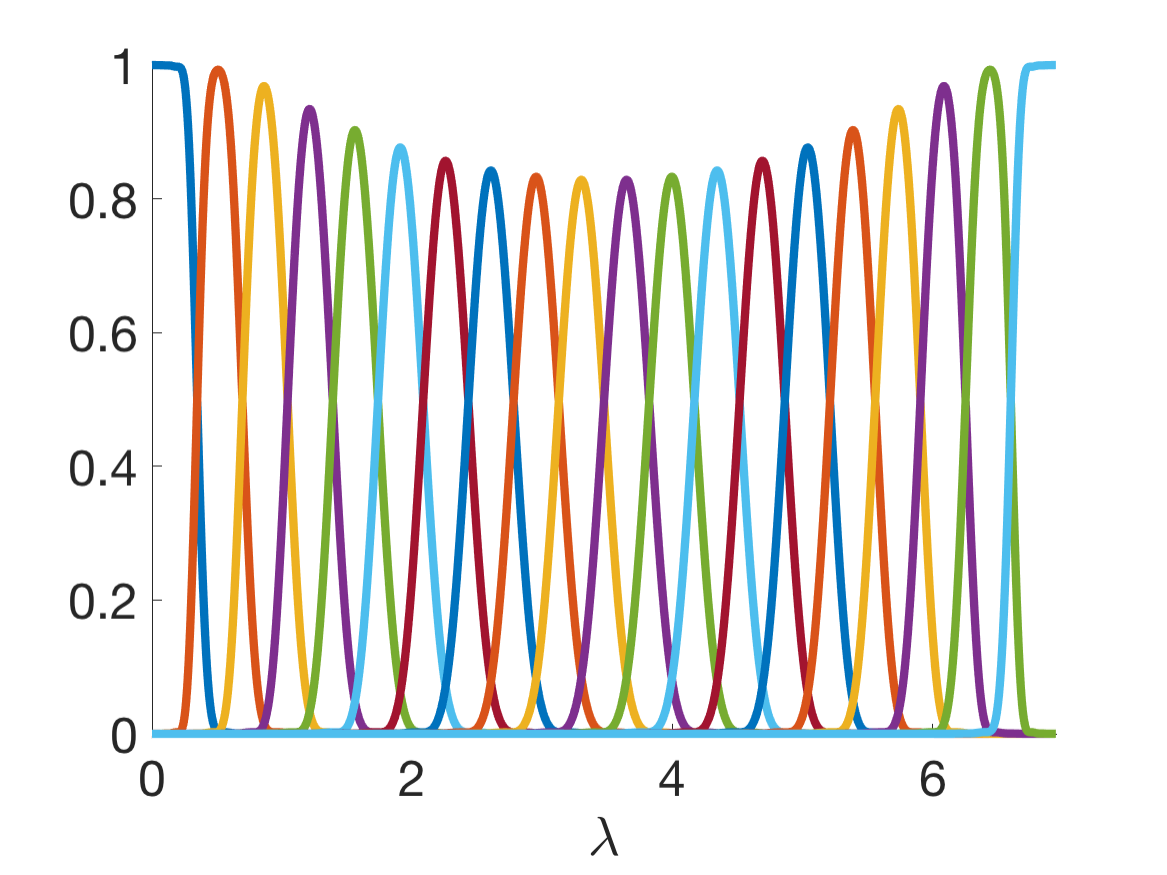}}
\centerline{\includegraphics[width=.95\linewidth]{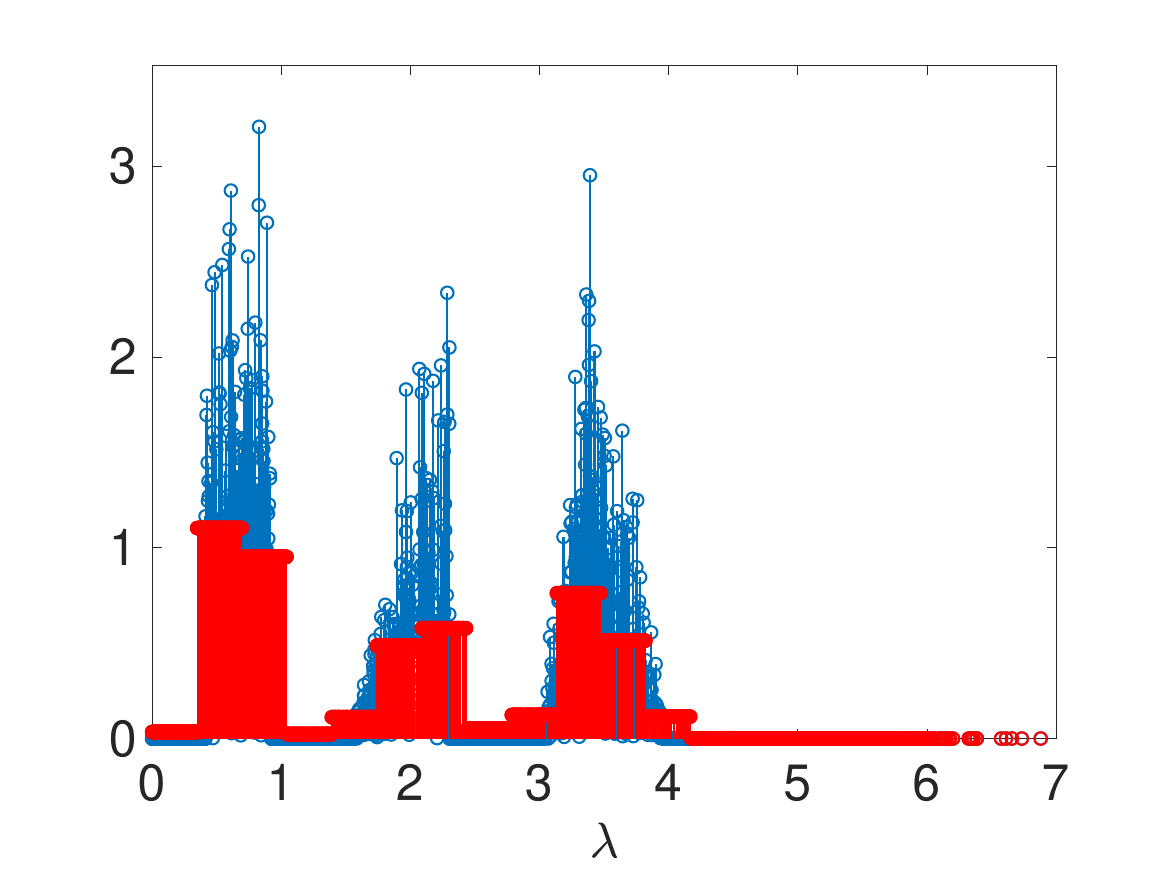}}
\centerline{~~\small{(a)}}
\end{minipage}
\begin{minipage}[m]{0.24\linewidth}
\centerline{\small{Equal Number of Eigenvalues}}
\centerline{\includegraphics[width=.95\linewidth]{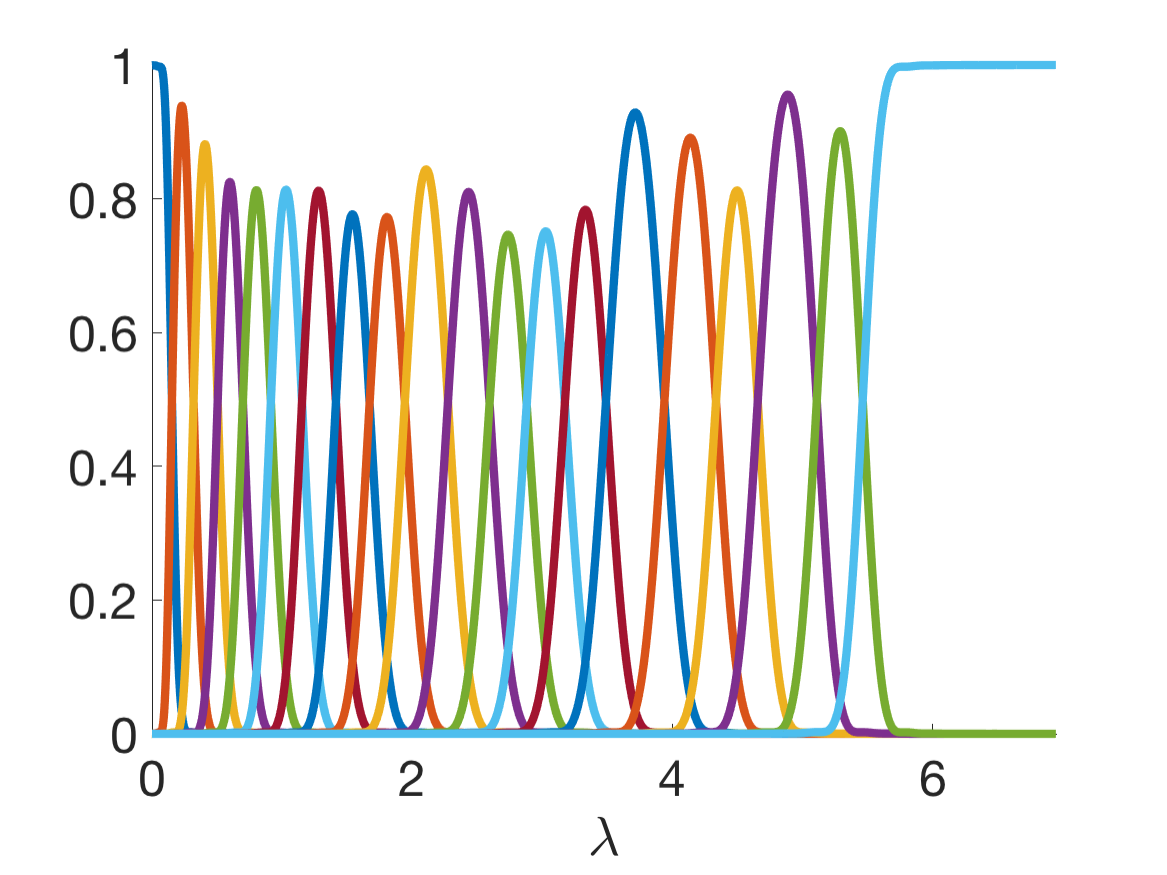}}
\centerline{\includegraphics[width=.95\linewidth]{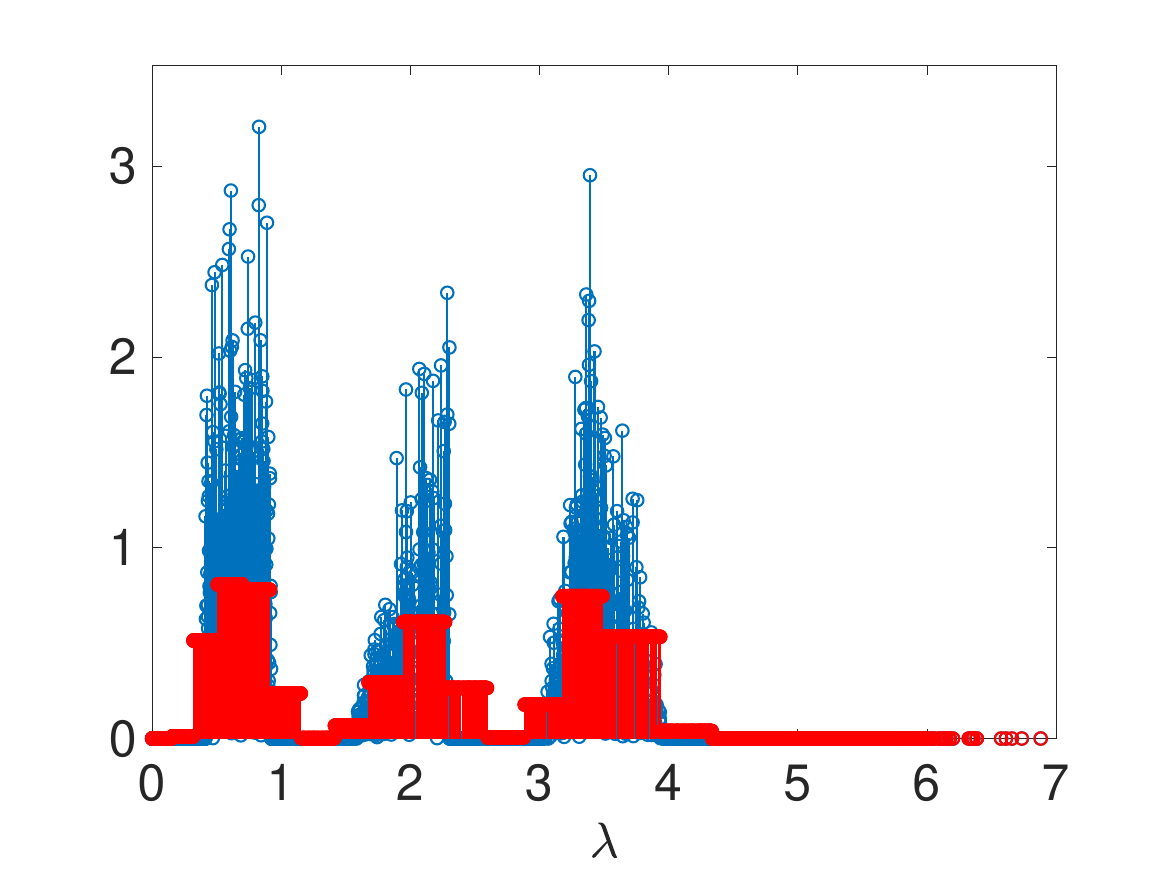}}
\centerline{~~\small{(b)}}
\end{minipage} 
\begin{minipage}[m]{0.24\linewidth}
\centerline{\small{Shifted ($K=80$)}}
\centerline{\includegraphics[width=.95\linewidth]{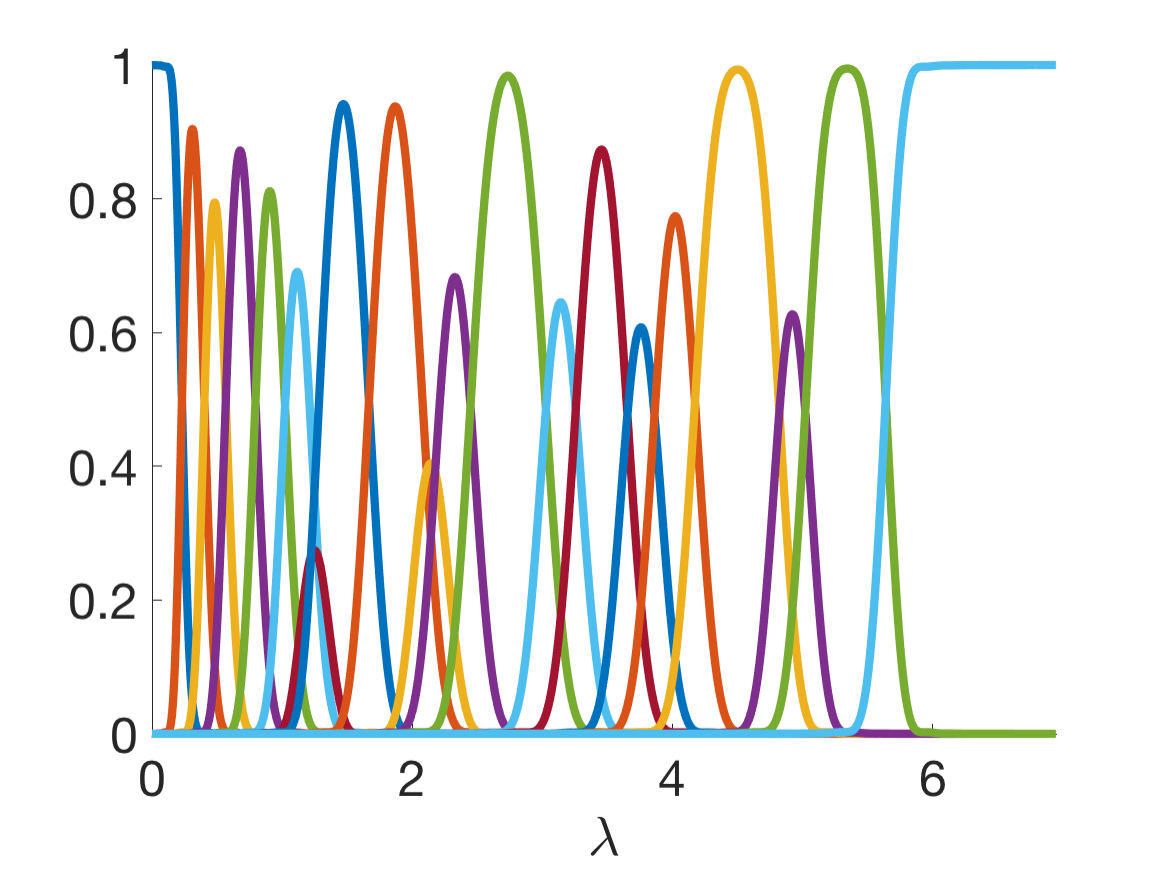}}
\centerline{\includegraphics[width=.95\linewidth]{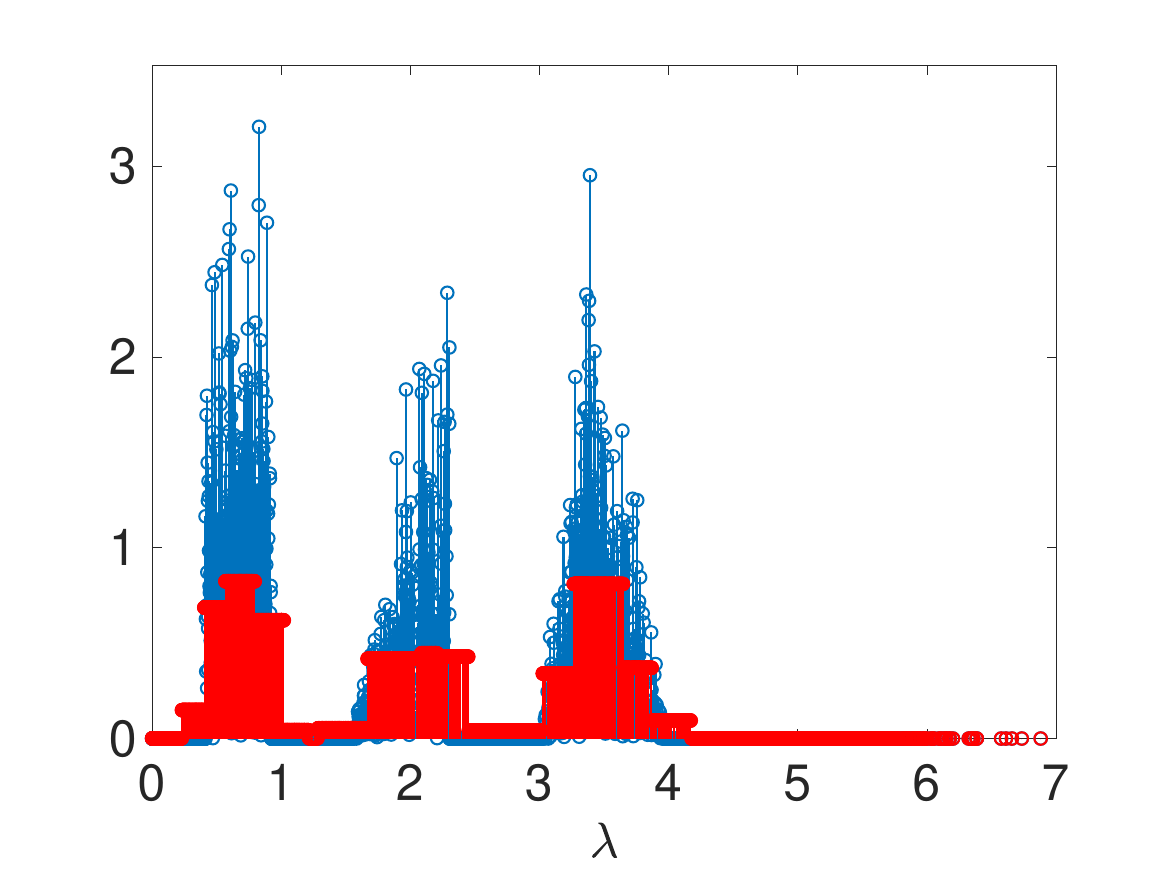}}
\centerline{~~\small{(c)}}
\end{minipage} 
\begin{minipage}[m]{0.24\linewidth}
\centerline{\small{Shifted ($K=250$)}}
\centerline{\includegraphics[width=.95\linewidth]{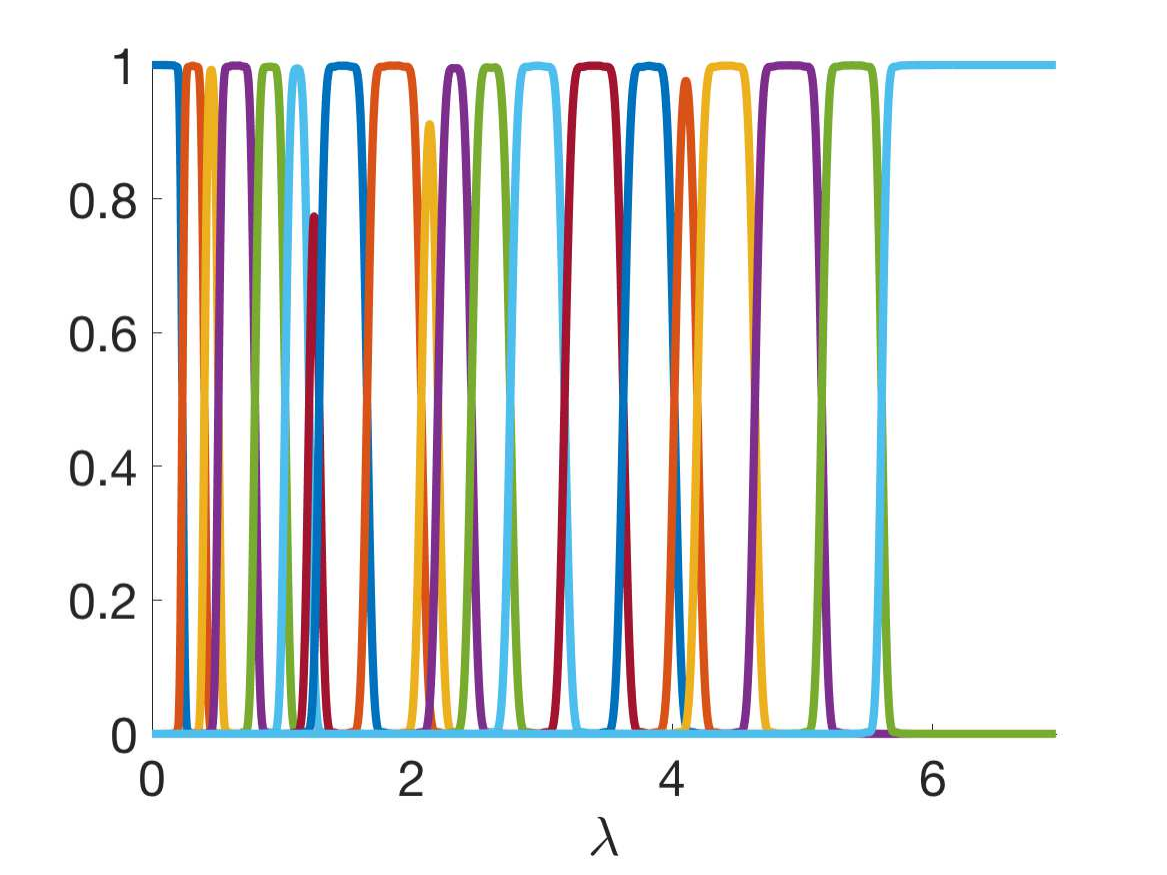}}
\centerline{\includegraphics[width=.95\linewidth]{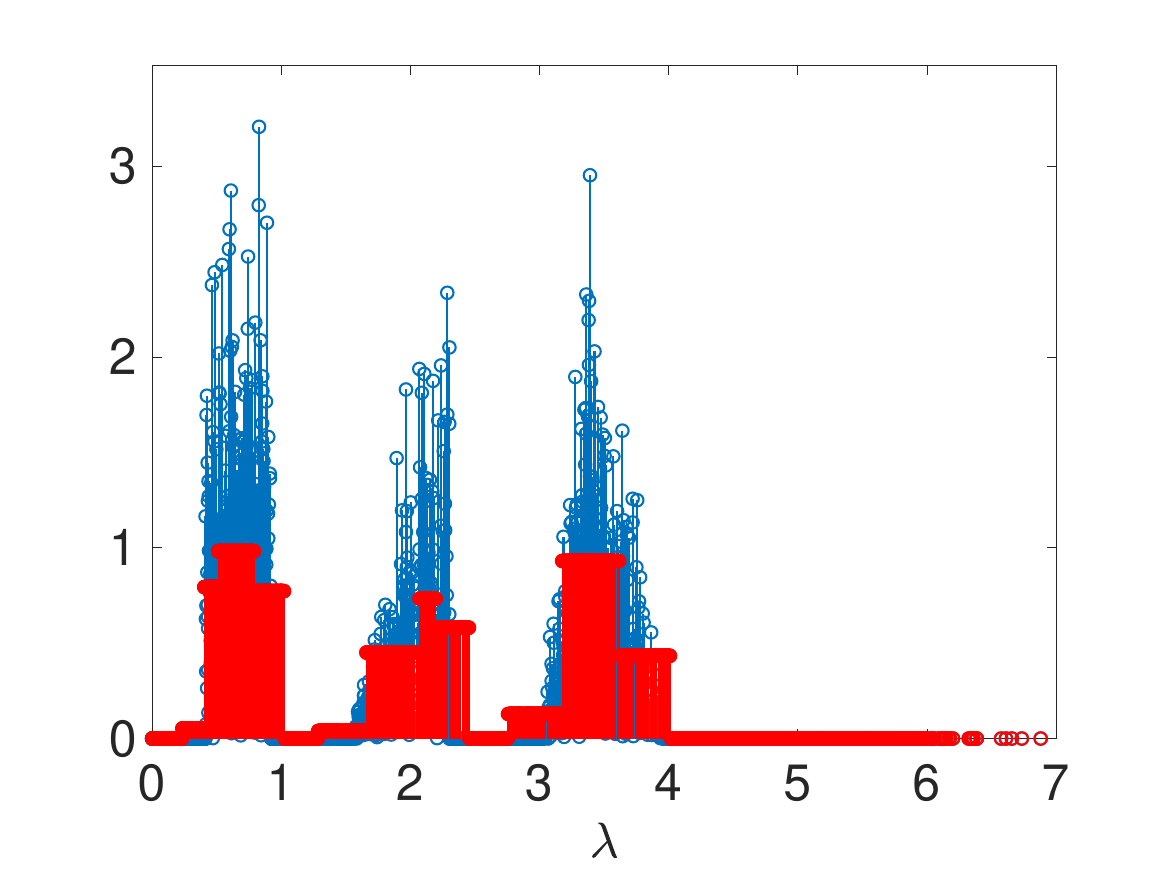}}
\centerline{~~\small{(d)}}
\end{minipage} 
\caption{Fast approximate graph Fourier transform example. The bottom row shows the approximations of the graph Fourier transform of the synthetic signal in Fig. \ref{Fig:fgft1}(b), based on \eqref{Eq:fgft}, using the fast $M$-CSFB transform with four different choices of 20-channel filter banks. The spectrum-adapted filter bank in (b) results in the closest approximation in terms of mean square error, but the filters based on shifted endpoints in (c) and (d) appear to more accurately identify the support of the signal. Ignoring the spectral distribution and taking the filters to have equal length in (a) yields the worst approximation.}\label{Fig:fgft2}
\vspace{-.4cm}
\end{figure*}

In \cite{le2018approximate}\nocite{le2016there}-\cite{le2016flexible}, Le Magoarou, Gribonval, and Tremblay propose approximate fast graph Fourier transform methods that either exactly compute ${\bf U}$ or approximate it by a product of sparse and orthogonal matrices. These methods reduce the complexity of applying the approximate graph Fourier transform from ${\cal O}(N^2)$ to ${\cal O}(N \log N)$; however, they still incur the significant upfront computational cost to compute ${\bf U}$ or approximate it (e.g., on the Minnesota graph, the parallel truncated Jacobi method takes approximately one hour, compared to the five seconds required to perform an exact eigendecomposition). The fast $M$-CSFB method \eqref{Eq:fgft}, on the other hand, provides a coarser approximation to the graph Fourier transform, but scales to much larger graphs.

As a first 
example, we generate a synthetic signal in the graph Fourier domain of the Minnesota road network using the method of \cite[Fig. 1]{le2018approximate}. Fig. \ref{Fig:fgft1}(a) shows the estimated cumulative spectral distribution of the Minnesota graph, and Fig. \ref{Fig:fgft1}(b) shows the synthetic signal in the spectral domain. In Fig. \ref{Fig:fgft2}, we apply the fast $M$-CSFB approximation \eqref{Eq:fgft}, using four different choices of Jackson-Chebyshev polynomial filter banks, each with $M=20$ filters. The first  
chooses the bands to have equal length; the second chooses the bands to have an approximately equal number of eigenvalues; the third shifts the ends of the second slightly according to the procedure outlined in the second for loop of Algorithm \ref{Al:filter_bank}; and the fourth is the same as the third, except 
with a polynomial approximation order of $K=250$, as compared to $K=80$ for the first three filter banks. Quantitatively, the normalized mean square errors between the estimates and the actual $\hat{\bf f}$ are 2.33e-04, 1.64e-04, 1.68e-04, and 1.83e-04, respectively. If we let $M=50$ and $K=250$ for the shifted filters, the 
NMSE drops to 1.59e-04. Qualitatively, the shifted filters seem to 
 better 
 identify  the support of the signal, but perform worse on the magnitudes. 

In identifying the support of the signal in the spectral domain, there is a tradeoff between resolution and computational cost. The resolution of the approximation is controlled by the number of spectral bands $M$. As we add more bands for finer resolution in the spectral domain, however, the filters become narrower (less smooth). We therefore need higher degree polynomials to accurately approximate the filters, which in turn slows down the computations and may result in atoms whose energy is more spread in the vertex domain. We have seen experimentally that choosing $M$ in the 20-50 range and 
$K$ to be  
4-5 times $M$ results in reasonable approximations. 

In Fig. \ref{Fig:fgft3}, we apply the same approximate graph Fourier transform \eqref{Eq:fgft} to the average temperature signal from Fig. \ref{Fig:temperature}(e), with $M=20$, $K=80$ and the shifted filters. 
Because we cannot compute the exact eigenvalues in this case, we show the approximation as a continuous function of $\lambda$ on the interval $[0,\lambda_{\max}]$. A coarse approximation like this can confirm that the signal's energy is concentrated on the low end of the spectrum. We are not aware of any other methods to approximate the graph Fourier transform of a signal on a graph of this size.

\begin{figure}[tbh]
\begin{minipage}[m]{0.48\linewidth}
\centerline{\includegraphics[width=1\linewidth]{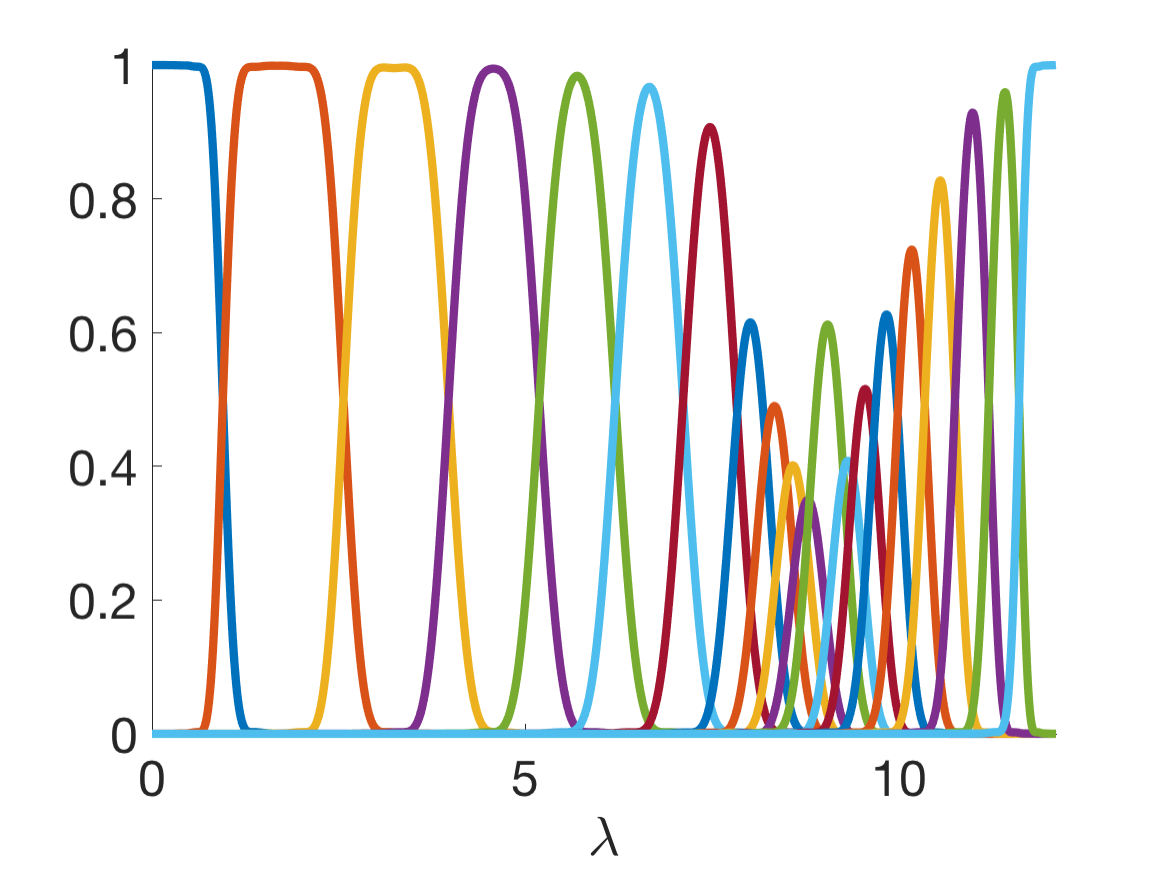}}
\centerline{~\small{(a)}}
\end{minipage}
\begin{minipage}[m]{0.48\linewidth}
\centerline{\includegraphics[width=1\linewidth]{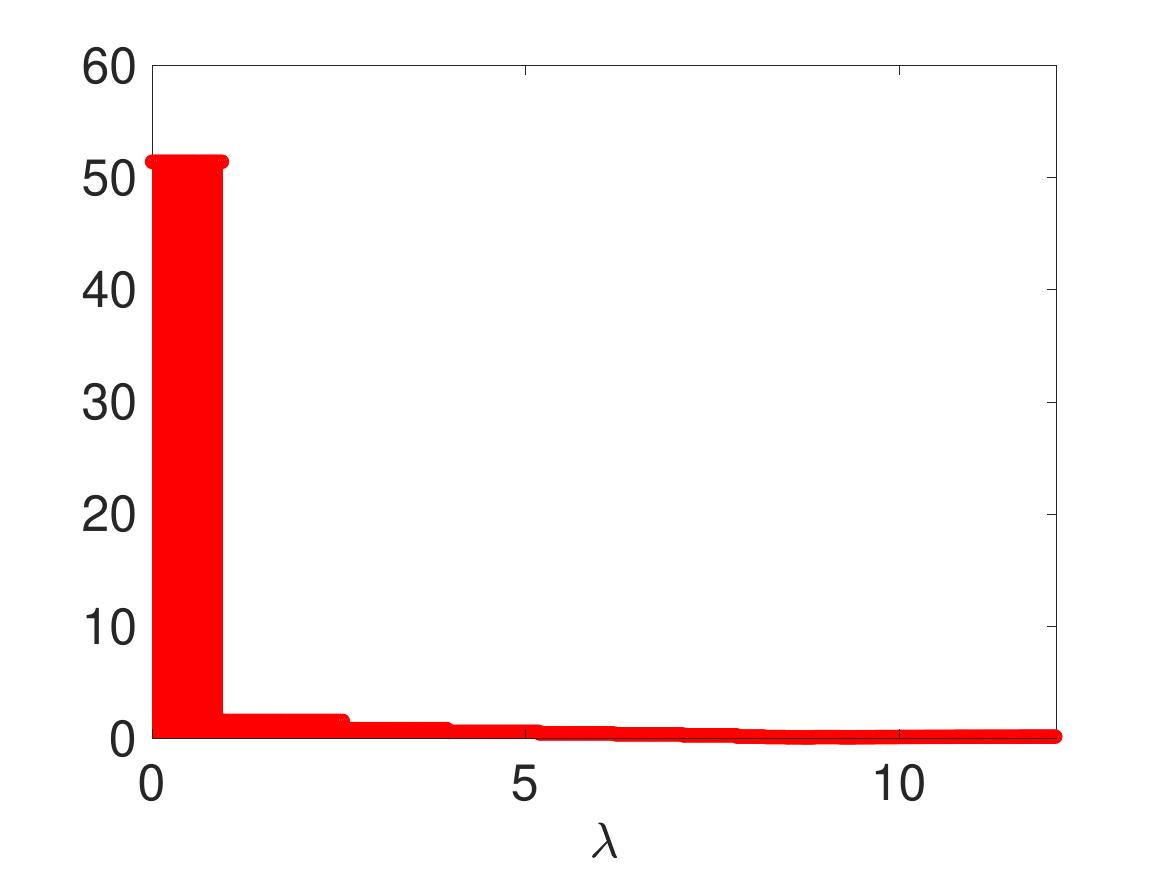}}
\centerline{~~\small{(b)}}
\end{minipage} 
\caption{Approximate graph Fourier transform of the average temperature signal from Fig. \ref{Fig:temperature}(e). (a) The 20 spectrum-adapted filters used in the fast $M$-CSFB. (b) The approximate graph Fourier transform confirms this is a smooth signal with its energy concentrated on the low end of the spectrum.}\label{Fig:fgft3}
\vspace{-.5cm}
\end{figure}

\section{Conclusion and Extensions}
\label{Sec:ongoing}
We have proposed the first critically sampled transforms for graph signals that scale to graphs with hundreds of thousands to millions of vertices. The fast $M$-CSFB transform approximately projects a graph signal onto different bands of the graph Laplacian spectrum. To improve computational efficiency, we leverage the computation of $\{\bar{T}_k(\L){\bf X}\}$ in multiple ways: to estimate the spectral density for the design of the filter bank, to estimate the number of samples for each band, and to estimate the non-uniform sampling distributions. 

The key idea behind the filter bank design is to choose the end points of each band to be in less dense regions of the spectrum so that the resulting filters are more amenable to polynomial approximation. Adapting the non-uniform sampling distribution and allocation of the samples across the bands to the specific signal being analyzed improves the accuracy of the synthesis process without adding to the computational complexity of the setup and analysis steps. 

On one hand, the proposed transform can be seen as a fast approximation of the graph Fourier transform with a coarser resolution in the spectral domain, as discussed in Section \ref{Se:fgft}. On the other hand, the atoms of the proposed transform can also 
be viewed as a subset of the atoms of a spectral graph wavelet transform  \cite{hammond2011wavelets}, albeit with a different set of filters. Both transforms yield atoms of the form $\tilde{h}_m(\L){\boldsymbol \delta}_i$, but the spectral graph wavelet transform includes every vertex $i$ as a center vertex for every scale $m$. 

As with the classical wavelet construction, it is possible to iterate the filter bank on the output from the lowpass channel.
This could be beneficial, for example, in the case that we want to visualize the graph signal at different resolutions on  a sequence of coarser and coarser graphs. An interesting question for future work is how iterating the filter bank with fewer channels at each step compares to a single filter bank with more channels 
supported on a smaller spectral intervals. 

Another complementary direction for future investigation is the specific form of the filters. Within the same
construction, we could use types of filters other than the Jackson-Chebyshev filters (e.g., \cite{tay2015design,liu2018filter}), or adapt the filters to the energy distribution of the signal or an ensemble of signals \cite{behjat2016signal}.

\section{Appendix}
{
\begin{proof}[Proof of Proposition \ref{Le:highpass_uniqueness}]
We assume without loss of generality that ${\mathcal T}=\{0,1,2,\ldots,k-1\}$. 
Suppose first that the set $\mathcal{S}$ is a uniqueness set for $\mbox{col}({\mathbf{U}}_{\mathcal T})$, but $\mathcal{S}^c$ is not a uniqueness set for $\mbox{col}({\mathbf{U}}_{{\mathcal T}^c})$. Then by Lemma \ref{Le:eq_uniq}, the matrix $$\mathbf{A}=
 \left[ \begin{array}{cccccccc}
{\bf u}_{0} & {\bf u}_{1} & \cdots & {\bf u}_{k-1} & {\boldsymbol \delta}_{{\mathcal{S}}^c_1} & {\boldsymbol \delta}_{{\mathcal{S}}^c_2} \cdots & {\boldsymbol \delta}_{{\mathcal{S}}^c_{N-k}} \end{array} \right]$$
has full rank, and the matrix $$\mathbf{B}=
 \left[ \begin{array}{cccccccc}
{\bf u}_{k} & {\bf u}_{k+1} & \cdots & {\bf u}_{N-1} & {\boldsymbol \delta}_{\mathcal{S}_{1}} & {\boldsymbol \delta}_{\mathcal{S}_{2}} \cdots & {\boldsymbol \delta}_{\mathcal{S}_k} \end{array} \right]$$
is singular, implying
\begin{align}\label{Eq:span}
\mbox{span}({\bf u}_{k}, {\bf u}_{k+1},\ldots, {\bf u}_{N-1}, {\boldsymbol \delta}_{\mathcal{S}_{1}}, {\boldsymbol \delta}_{\mathcal{S}_{2}}, \ldots, {\boldsymbol \delta}_{\mathcal{S}_k}) \neq \mathbb{R}^N.
\end{align}
Since 
$\mbox{dim}(\mbox{span}({\bf u}_{k}, {\bf u}_{k+1},\ldots, {\bf u}_{N-1}))=N-k$ and $\mbox{dim}(\mbox{span}({\boldsymbol \delta}_{S_{1}}, {\boldsymbol \delta}_{S_{2}}, \ldots, {\boldsymbol \delta}_{\mathcal{S}_k}))=k,$ equation 
\eqref{Eq:span} implies that
there must exist a vector ${\bf x} \neq {\bf 0}$ such that 
${\bf x} \in  \mbox{span}({\bf u}_{k}, {\bf u}_{k+1},\ldots, {\bf u}_{N-1})$ and  
$ {\bf x} \in \mbox{span}({\boldsymbol \delta}_{\mathcal{S}_{1}}, {\boldsymbol \delta}_{\mathcal{S}_{2}}, \ldots, {\boldsymbol \delta}_{\mathcal{S}_{k}})$. 
Yet, ${\bf x} \in \mbox{col}({\mathbf{U}}_{{\mathcal T}^c})$ implies ${\bf x}$ is orthogonal to  ${\bf u}_0, {\bf u}_1, \ldots, {\bf u}_{k-1}$, and, similarly, 
${\bf x} \in \mbox{span}({\boldsymbol \delta}_{\mathcal{S}_{1}}, {\boldsymbol \delta}_{\mathcal{S}_{2}}, \ldots, {\boldsymbol \delta}_{\mathcal{S}_{k}})$ implies ${\bf x}$ is orthogonal to ${\boldsymbol \delta}_{\mathcal{S}^c_{1}}, {\boldsymbol \delta}_{\mathcal{S}^c_{2}}, \ldots, {\boldsymbol \delta}_{\mathcal{S}^c_{N-k}}.$
In matrix notation, we have $\mathbf{A}^{\top}{\bf x}={\bf 0}$, so $\mathbf{A}^{\top}$ has a non-trivial null space, and thus the square matrix $\mathbf{A}$ is not full rank and $\mathcal{S}$ is not a uniqueness set for $\mbox{col}({\mathbf{U}}_{\mathcal T})$, a contradiction. We conclude that if $\mathbf{A}$ is full rank, then $\mathbf{B}$ must be full rank and $\mathcal{S}^c$ is a uniqueness set for $\mbox{col}({\mathbf{U}}_{{\mathcal T}^c})$, completing the proof of sufficiency. Necessity follows from the same argument, with the roles of $\mathbf{A}$ and $\mathbf{B}$ interchanged.
\end{proof}

\vspace{-.2in}

\balance
\bibliographystyle{IEEEtran}
{\bibliography{mcsfb_refs}}

\end{document}